\documentclass{article}
\usepackage{latexsym, amssymb, enumerate, amsmath}
\usepackage{setspace}
\usepackage{enumitem}

\usepackage{amsthm}
\usepackage{bbm}
\usepackage[margin=1in]{geometry}

\usepackage{graphicx}
\usepackage[labelformat=empty]{subfig}
\usepackage{wrapfig}
\usepackage{color}
\renewcommand{\theequation}{\arabic{section}.\arabic{equation}}
\makeatletter
\@addtoreset{equation}{section}
\makeatother

\newcommand{\N}{\mathbb{N}}

\newcommand{\q}[0]{f_0 +\frac{1}{2} h_{0}^2}
\makeatletter
    
    \newcommand{\Rmnum}[1]{\expandafter\@slowromancap\romannumeral #1@}
  \makeatother

\def\({\left(}
\def\){\right)}
\def\[{\left[}
\def\]{\right]}

\newcommand{\one}[0]{\mathbbm{1}}

\newcommand{\matr}[4]{\left( \begin{array}{ll} #1 & #2 \\ #3 & #4 \end{array}\right)}
\newcommand{\vect}[2]{\left( \begin{array}{l} #1 \\ #2 \end{array} \right)}

\newcommand{\wt}[1]{\widetilde{#1}}

\newcommand{\R}[0]{\wt{R}}

\newtheorem{thm}{Theorem}[section]
\newtheorem{prop}[thm]{Proposition}
\newtheorem{lem}[thm]{Lemma}
\newtheorem{cor}[thm]{Corollary}

\newtheorem{defn}[thm]{Definition}

\newtheorem{conj}[thm]{Conjecture}

\newtheorem{theorem}{Theorem}
\newtheorem{lemma}{Lemma}
\newtheorem{corollary}[lemma]{Corollary}

\newcommand{\beann}{\begin{eqnarray*}}
\newcommand{\eeann}{\end{eqnarray*}}
\newcommand{\bea}{\begin{eqnarray}}
\newcommand{\eea}{\end{eqnarray}}

\newcommand{\ip}[1]{\left\langle #1 \right\rangle}
\newcommand{\dmu}[0]{\exp\left[-N V(\lambda)\right]\,d\lambda}

\newtheorem{rem}[thm]{Remark}

\begin{document}
\date{}

\title{RELATING RANDOM MATRIX MAP ENUMERATION TO A UNIVERSAL SYMBOL CALCULUS FOR RECURRENCE OPERATORS IN TERMS OF BESSEL-APPELL POLYNOMIALS}
\author{Nicholas M. Ercolani
\thanks {Department of Mathematics, The University of Arizona, Tucson, AZ 85721--0089, ({\tt ercolani@math.arizona.edu}). Supported by NSF grant DMS-1615921.}
\and Patrick Waters \thanks {Liftoff Mobil, Inc., 
900 Middlefield Road, 2nd Floor
Redwood City, CA 94063}}


\maketitle


\begin{abstract}  
Maps are polygonal cellular networks on Riemann surfaces. This paper analyzes the construction of {\it closed form} general representations for the enumerative generating functions associated to maps of fixed but arbitrary genus. The method of construction developed here involves a novel asymptotic symbol calculus for difference operators based on the relation between spectral asymptotics for Hermitian random matrices and asymptotics of orthogonal polynomials with exponential weights. These closed form expressions have a universal character in the sense that they are independent of the explicit valence  distribution of the cellular networks within a broad class. Nevertheless the valence distributions may be recovered from the closed form generating functions by a remarkable {\it unwinding identity} in terms of Appell polynomials generated by Bessel functions. Our treatment reveals the generating functions to be solutions of nonlinear conservation laws and their prolongations. This characterization enables one to gain insights that go beyond more traditional methods that are purely combinatorial. 
Universality results are connected to stability results for characteristic singularities of conservation laws that were studied by Caflisch, Ercolani, Hou and Landis \cite{CEHL93} as well as directly related to universality results for random matrix spectra. 
\end{abstract}

{\bf keywords:}
random matrices, Toda lattices, Motzkin paths, string equations, conservation law hierarchies, map enumeration, 
combinatorial generating functions, Hopf algebras, orthogonal polynomials, Appell polynomials

\smallskip


\tableofcontents

\section{Introduction} \label{sec:0} 

This paper focusses on the use of random matrix theory methods to derive explicit closed form expressions of generating functions for map enumeration.  There is a large literature of works on map enumeration including those that use a matrix integral approach (a brief review is given in Section \ref{approaches}). What is novel in what we present here is the development of a broad set of universal structures that underlie these generating functions and provide insight into their calculation. These structures involve ideas from the semiclassical symbol calculus for difference operators and associated Poisson structures, the umbral calculus of Appell polynomials, singularity formation in conservation laws, and the arithmetic geometry of rational ruled surfaces. By {\it universal} here we mean properties that are valence independent as regards map enumeration within a broad class of valence distributions (or, equivalently, independent of the polynomial random matrix potential within a corresponding broad class). 

Within this opening introduction, Section \ref{motivation} provides basic background relating random matrix theory to map enumeration with further elaboration on enumeration given in Section \ref{statement}. Section \ref{key} develops the key ingredients related to the new ideas in this paper. Section 2 will lay out a concise summary of our results and concludes with an outline of the remainder of the paper.

\subsection{Motivation} \label{motivation} 

The intuitive notion of a {\it map} is that of a graph embedded in a compact connected oriented surface such that the complement of the graph is a cellular decomposition of the surface. A depiction of a map in a {\it local chart} on a surface is illustrated by the solid graph embedded in a planar region shown in Figure \ref{dualmap}. Note that in this example all vertices have valence 3 (in the graph-theoretic sense). Maps whose vertices all have the same valence, $j$, are referred to as $j-regular$ maps in analogy with the terminology for graphs. Figure \ref{dualmap} also (locally) illustrates 
the dual map (depicted in terms of the dashed graph). The 3-regularity of the original map results in the dual map being a triangulation of the surface. 
\begin{figure}[h] 
\begin{center} 
\resizebox{3in}{!}{\includegraphics{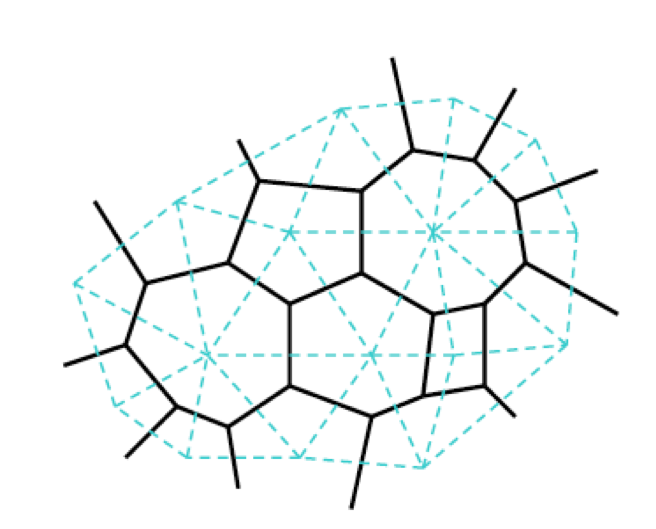}}
\caption{3-Valent Map and its Triangular Dual}  \label{dualmap}
\end{center}
\end{figure}

Maps are combinatorial objects and, as with graphs, it is natural to try to enumerate them (up to some appropriate notion of equivalence). {Map enumeration} has a long history involving several different approaches which we will briefly review at the end of this introduction. But for now it will suffice to state that the approach discussed in this paper is primarily analytical, rather than combinatorial, based on a remarkable connection to random matrix theory. The main thrust of the analytical approach to map enumeration centers on the large-$n$ expansion of Hermitian random matrix ensembles that extend the { \it Gaussian Unitary Ensemble} (GUE) model defined on the space of $n \times n$ Hermitian matrices, $\mathcal{M}_n$, with the family of probability measures
\begin{eqnarray} \label{RMmeasures}
d\mu_{\bf t} &=& \frac1{\tilde{Z}_{n}} \exp\left\{-n \text{Tr } [ V_{\bf t}(M)] \right\} dM
\end{eqnarray}
for $M \in \mathcal{M}_n$ where $V_{\bf t}$ is typically a polynomial 
\beann
V_{\bf t}(\lambda) &=& \frac1{2x} \lambda^2 + \frac1x\sum_{j=1}^J t_j \lambda^j
\eeann
referred to as the {\it potential} of the ensemble. The parameter $x$ (positive) is referred to as the t'Hooft parameter. (This parameter may be introduced in 
another way which will be used in later sections of this paper -- see section \ref{sec:3.1}.)

The large $n$ asymptotics of the {\it partition function} 
\beann
\tilde{Z}_{n}({\bf t}) = \int_{\mathcal{M}_n} \exp\left\{-n \text{Tr } V_{\bf t} \right\} dM,
\eeann 
which serves to normalize $d\mu_{\bf t}$ to be a probability measure, plays the foundational role in constructing generating functions for map enumeration.  (For $J$ odd one needs to consider analytic deformations in order to make sense of this integral and the associated measure--see section \ref{OddValentExt}.) Due to the unitary invariance of the partition function it, as well as expectations of unitarily invariant random variables with respect to \ref{RMmeasures}, may be reduced to integrals over eigenvalues of Hermitian matrices. For this reason this approach to map enumeration is sometimes referred to as the {\it matrix integral method}.
This too has a long history going back to Bessis, Itzykson and Zuber who, in \cite{BIZ80},  first introduced the idea of relating map enumeration to a path integral type formulation of these integrals that could be evaluated via an idea due to Dyson using orthogonal polynomials with exponential weights. The proper formulation of the large $n$ asymptotic analysis in this approach was placed on the firm footing of Riemann-Hilbert analysis in the work of Fokas, Its and Kitaev \cite{FIKII}. Ercolani and McLaughlin in \cite{EM03} provided a rigorous verification of the $1/n^2$ structure of the full asymptotic expansions of generating functions related to the partition functions $\tilde{Z}_n({\bf t})$. 

This last result is expressed in terms of the one point function for random matrix eigenvalues, 

\begin{eqnarray} \label{eq:onepoint}
\rho^{(1)}_n(\lambda) &=& \frac{d}{d\lambda} \mathbb{E}_{\vec{{\bf t}}}\left( \frac1n \# \left\{ k : \lambda_k \in (-\infty, \lambda) \right\}\right)\\
&=& \frac1n K_n(\lambda, \lambda)
\end{eqnarray}
where $\lambda_k$ runs over the eigenvalues of $M$ and 
\begin{eqnarray*}
K_n(\lambda,\eta) &=& e^{ -(n/2) ( V_{\vec{{\bf t}}}(\lambda) + V_{ \vec{{\bf t}}} (\eta)   ) } \sum_{\ell = 0}^{n-1} \pi_\ell(\lambda) \pi_\ell(\eta) 
\end{eqnarray*}
is the Mercer kernel associated to the family of monic orthogonal polynomials $\{ \pi_\ell\}$ with exponential weight 
$e^{-n V_{\bf t}(\lambda)}$ on the real line. 
 
Precisely, the fundamental analytical result states that 
\begin{theorem} \cite{EM03, BD16} \label{Workhorse}
There exist $T >0$ and $\gamma > 0$ such that one has an asymptotic expansion, uniformly valid for $x$ sufficiently close to 1 and all  ${\bf t} \in \mathbb{T}(T, \gamma) = \left\{  {\bf t} \in \mathbb{R}^J : \left|{\frac1x {\bf t}}\right| \leq T , \,\, t_J > \gamma \sum_{j=1}^{J-1} |t_j| \right\}$,  of the form
\beann
\int_{-\infty}^\infty F(\lambda) \rho^{(1)}_n (\lambda) d\lambda = F_0(x, {\bf t}) + n^{-2} F_1(x, {\bf t}) + n^{-4} F_2(x, {\bf t}) + \cdots , 
\eeann
provided the function $F(\lambda)$ is $C^\infty$ and grows no faster than polynomially. The coefficients $F_m$ depend analytically on $x$ and ${\bf t}$ for ${\bf t} \in \mathbb{T}(T, \gamma)$ and the asymptotic expansion may be differentiated 
term by term with respect to $x$ and $\bf t$. 
\end{theorem}

The proof of this result relies centrally on analysis of a Riemann-Hilbert problem for the orthogonal polynomials 
$\{ \pi_\ell\}$ mentioned above. Clearly this makes apriori sense only if the dominant term in the polynomial $V_{\bf t}(\lambda)$ is even, in order to ensure existence of the orthogonal polynomials. This is the case that was treated in  \cite{EM03}. However,  Bleher and Dea\~{n}o \cite{ BD16} showed how the result may be extended to the case where the dominant term in $V_{\bf t}(\lambda)$ is odd by using {\it non-Hermitian} orthogonal polynomials (see Section \ref{OddValentExt}).

Applying Theorem \ref{Workhorse} for monomial $F(\lambda) (= \lambda^m)$ and making appropriate anti-differentiations one derives the asymptotic expansion of the logarithm of the partition function \cite{EM03}:

\bea \label{TopRec}
\log \frac{ \tilde{Z}_n(x, {\bf t})}{\tilde{Z}_n(x, 0)} &=&
n^{2} E_{0}(x, {\bf t}) + E_{1}(x, {\bf t}) + \frac{1}{n^{2}} E_{2}(x, {\bf t}) + \cdots +  \frac{1}{n^{2g-2}} E_{g}(x, {\bf t}) + \dots .
\eea

This expansion provides the fundamental link between random matrix theory and map enumeration due to the remarkable fact that $E_{g}(x,{\bf t})$ is the exponential generating function whose $n^{th}$ term enumerates labelled maps with $n$-vertices on a compact surface of genus $g$ with valence distribution characterized by the parameter values $\bf t$. (This is defined precisely at the start of Section \ref{statement}.)

In \cite{EMP08, Er09} the orthogonal polynomial methodology underlying Theorem \ref{Workhorse} was used to derive closed form expressions for the $E_{g}(x, {\bf t})$ in the case of even potentials. (See Appendix \ref{evenvalence}). This was extended to odd order dominated potentials in \cite{EP11}.  
\medskip

This paper is based on the matrix integral/orthogonal polynomial approach. This provides the direct connections that are fundamental to what we do: the difference operators (the recursion operators for orthogonal polynomials) and the symbol calculus for the continuum limits, via Theorem \ref{Workhorse}, of these operators. In the remainder of this Introduction we outline the key novel ingredients for our analysis.  

 \subsection{Key Ingredients} \label{key}
      
\subsubsection{Pivotal Characterization of the Equilibrium Measure}

All methods of map enumeration ultimately come to rest upon two fundamental generating functions, (\ref{endpoint1}) and (\ref{endpoint2}) below. In the analytical approaches these are naturally interpreted in terms of the support of the equilibrium measure for the asymptotic density of states for the random matrix ensemble (\ref{RMmeasures}) which also turns out to equal the asymptotic distribution of the zeroes of the orthogonal polynomials $\{ \pi_\ell\}$, \cite{Deift}, introduced ealier. 

By definition, the equilibrium measure is the large $n$ limit of the one point function (\ref{eq:onepoint}),

\bea
d\mu_{eq} &=& \lim_{n \to \infty} \rho^{(1)}_n(\lambda)
\eea
which can be shown to be absolutely continuous with respect to Lebesgue measure with Radon-Nikodym derivative determined by solving a scalar Riemann-Hilbert problem \cite{DKM}:

\bea \label{dmu}
d\mu_{eq} &=& \Lambda \,\, d\lambda\\ \nonumber
\Lambda(\lambda) &=& \frac1{2 \pi } \mathbbm{1} _{[r_-, r_+]} (\lambda) \sqrt{(\lambda - r_-)(r_+ - \lambda)} H(\lambda) 
\eea
As a consequence of elementary Riemann-Hilbert analysis $r_\pm$ depends locally analytically on $x$ and $\bf t$ as do the following symmetric combinations
\bea \label{endpoint1}
h_0(x, {\bf t}) &=& \frac{r_+ + r_-}{2}\\ \label{endpoint2}
f_0(x, {\bf t}) &=& \frac{(r_+ -r_-)^2}{16}
\eea 

$H(\lambda)$ is a polynomial of degree $J-2$ which is completely determined by $r_\pm$ \cite{DKM, EMP08}
and in fact may be explicitly written solely in terms of $h_0$ and $f_0$ \cite{Er09, W2}:
\beann
H(\lambda) = \sum_{k=0}^{J-2} (\lambda - r_+)^k   \sum_{m=0}^{k} \sqrt{f_0}^{m-k-1}\left(c^{(\phi)}_{k,m} \sqrt{f_0}\phi_{V,m} + c^{(\psi)}_{k,m} \psi_{V,m}\right)
\eeann
where $c^{(\phi)}_{k,m}$ and $c^{(\psi)}_{k,m}$ are constants independent of the choice of potential $V$, and $\phi_{V,\ell}, \psi_{V,\ell}$ are explicit polynomial functions of $h_0$ and $\sqrt{f_0}$ referred to as {\it string polynomials} for reasons that will be explained shortly. Sending $\sqrt{f_0}$ to $-\sqrt{f_0}$ (which interchanges $r_+$ with $r_-$ ) yields an equivalent representation of $H(\lambda)$. $H$ may be thought of as a function that modulates the classical {\it semicircle law} of the 
{GUE} ensemble.

One of the first key results of this paper is the discovery of a fundamental relation between the string polynomials  and Bessel functions $I_0$ and $I_1$, the {\it modified Bessel functions} of orders $0$ and $1$ respectively: 
\begin{thm} \label{thm:1.1}
\beann
\phi_{V,m} &=& \frac1x \sum_{j=1}^J t_j j!  [s^{j-1-m}] e^{s h_0} I_0(2 s \sqrt{f_0})\\
\psi_{V,m} &=&  \frac1x \sum_{j=1}^J t_j j!  [s^{j-1-m}] e^{s h_0} f_0^{1/2} I_1(2 s \sqrt{f_0})
\eeann
where $[s^\ell]F(s)$ denotes the coefficient of $s^\ell$ in the Maclaurin series for $F$. 
\end{thm}
A corollary of this is that the polynomial factor in (\ref{dmu}) may be written as
\bea \label{Hlambda}
H(\lambda) &=& \frac1{\sqrt{f_0}} \sum_{k=0}^{J-2}  \Xi_k(\sqrt{y_0}) (\lambda - r_+)^k \,\,\,\, \mbox{where}\\ \label{normmean}
y_0 &=& \frac{h^2_0}{f_0}
\eea
and $\Xi_k(\zeta)$ is a certain linear combination of {\it Appell polynomials} (Section \ref{AppellBessel}.) 
The generating functions of  the Appell polynomials are the reciprocals of $I_0$ and $I_1$. Bessel functions arise in several limiting cases of random matrix kernels but their appearance here is something new. 
\smallskip

The observation of the central role played by the ratio $y_0$ in this determination of the equilibrium measure also appears to be something new. To better appreciate this we note that $h_0$ and $f_0$ arise as limits of natural pivotal statistics in random matrix theory:
We will see in Section \ref{background},  that these generating functions can be recovered from the basic genus 0 map generating function, $E_0$ itself as
\begin{eqnarray}  \label{var} 
f_0({\bf t}, x) &=& \lim_{n \to \infty}  \mathbb{E}_{\bf t}\left( \left(\text{Tr }  V_{\bf t}(M)\right)^2\right) \\ \label{e0tt}
 &=& \frac{\partial^2}{\partial t_1^2} E_0({\bf t}, x)\\ \label{xmean}
h_0({\bf t}, x) &=& \lim_{n \to \infty}  \mathbb{E}_{\bf t}\left( \left(\text{Tr } V_{\bf t}(M)\right)_x\right)  \\ \label{e0tx}
&=& \frac{\partial^2}{\partial t_1 \partial x} E_0({\bf t}, x).
\end{eqnarray}

One sees from (\ref{var}) and (\ref{xmean}) that $f_0$ is the asymptotic expectation of the random matrix variance, 
$ \left(\text{Tr }  V_{\bf t}(M)\right)^2$ while $h_0$ is the asymptotic expectation of the $x$-differentiated random matrix mean. These matrix statistics are very much in the character of what one sees with {\it non-commutative expectations}. 
In this setting $\sqrt{y_0}$, appearing as the argument in the Appell polynomials, has the character of an asymptotic normalized mean (mean divided by standard deviation). And indeed, if one rescales $\lambda$ in  (\ref{dmu}) as $\lambda = \sqrt{f_0} \sigma$ then $d\mu_{eq}$ as a function of $\sigma$ has precisely the form of the classical semicircle law centered at 
$\sqrt{y_0}$ and modulated by $H$ as a function of $\sigma$. 

There is a self-similar scaling structure in the equilibrium measure which stems from  fundamental relations between the t'Hooft parameter $x$ and the parameters $\xi_j$ known as {\it string relations}: 
\bea \label{scalingvars}
\xi_j &=& t_j x^{j/2 - 1}\\ \label{fscale}
f_0(x, \vec{t}) &=& x z_0(\dots,  \xi_j ,\dots)\\ \label{hscale}
 h_0(x, \vec{t}) &=& x^{1/2} u_0(\dots,  \xi_j ,\dots)
\eea
Though not immediately evident from the random matrix perspective, this scaling structure is fundamentally related to map combinatorics basically as a consequence of Euler's relation as will be explained in section \ref{background}. In the setting of {\it string theory} \cite{Pol}, $x$ plays the role of a fundamental coupling parrameter.  This is the origin of some of the nomenclature we use (such as {\it string polynomials}).
We see from (\ref{fscale}) and (\ref{hscale}) that the pivotal quantity, $y_0$, on which the modulation factor $H$ solely depends, has no external $x$-weight; i.e., it depends on the parameter $x$ only through the scaling variables (\ref{scalingvars}).

\subsubsection{Symbols of Difference Operators and Semi-Classical Limits} \label{sec:diffsymb}

To develop the monic orthogonal polynomials used in the Riemann-Hilbert analysis, one begins with a weighted Hilbert space 
$ \mathbb{H} = L^2\left(\mathbb{R}, e^{-n V(\lambda)}\right)$
where $V$ is a polynomial of the form $V_{\bf t}$ considered in (\ref{RMmeasures}). Again for $J$ odd one must consider complex deformations of $\mathbb{H}$ (see Section \ref{OddValentExt}); for simplicity of exposition we will initially here take $J$ to be even.  Next consider the basis of monic orthogonal polynomials with respect to this measure.

\begin{eqnarray*}
\pi_\ell(\lambda) &=& \lambda^\ell+ \,\,\mbox{lower order terms}\\
\int \pi_\ell(\lambda) \pi_m(\lambda) e^{-nV(\lambda)} d\lambda &=& 0\,\, \mbox{for}\,\, \ell \ne m.
\end{eqnarray*}
Let $\pi= ( \pi_\ell(\lambda))_{\ell\geq 0}$ be the column vector of all orthogonal polynomials for the potential $V$; then the operator of multiplication by $\lambda$ is representable as 

$$\lambda \pi = L \pi,$$
in terms of the semi-infinite tri-diagonal matrix (three term recurrence relation) , 
\begin{equation} \label{multop}
L = \left(\begin{array}{cccc} a_0 & 1 &  &\\
                              b^2_1 & a_1 & 1 &\\
			          & b^2_2 & a_2  & \ddots \\
                                  &      & \ddots & \ddots 
\end{array}\right)\,.
\end{equation}
$L$ is commonly referred to as the {\it recursion operator} for the orthogonal polynomials and its entries as {\it recursion coefficients}. (For $V$  an entirely even potential, it follows from symmetry that $a_\ell = 0$ for all $\ell$.)

In \cite{EMP08,EP11}, Ercolani, McLaughlin and Pierce analyzed the recurrence coefficients $a_\ell$ and $b^{2}_\ell$ using 
the {\it string equations} 
\begin{align}		
\frac{1}{n}I &= \left[V'(L )_-, L \right].	\label{string1}	
\end{align}
and the {\it Toda lattice equations}
\begin{align}
\frac{x}{n}\partial_{t_j} L =& [L ^{j}_-,L] \label{Toda1}
\end{align}
Here the subscript in $V'(L )_-$ and $L^j_-$ denotes the strictly lower triangular part of the semi-infinite matrices $V'(L )$ and $L^j$, respectively. Equations (\ref{string1}) and (\ref{Toda1}) are, respectively,  direct coordinate realizations of the canonical Heisenberg relation, $[\partial_\lambda, \lambda] = 1$ , and the variational equations for deformations of the exponential weights.

In this paper the explicit calculation of expansions, such as (\ref{TopRec}), that are guaranteed by Theorem \ref{Workhorse} is based on the asymptotic analysis of equations (\ref{string1}) and (\ref{Toda1}). Conceptually it is helpful to view this analysis in terms of a type of WKB {\it symbol calculus}. Such a perspective has been adopted in other settings such as for {\it continuum convexity theorems} \cite{BFR93} and {\it Toeplitz quantization} \cite{BGPU03}. What we do here differs from those cases firstly in that we consider a thermodynamic limit rather than a hydrodynamic limit. Second we explicitly calculate our symbol expansions to higher order rather than to just leading order. To formulate this perspective one makes use of the infinite order pseudo-differential operator (sometimes also called a {\it vertex operator}) defined by the obvious Taylor expansion relation:
\begin{eqnarray} \label{boson1}
e^{c\partial_x} f(x) &\doteq& \sum_{n \geq 0} \frac{c^n}{n!} \frac{\partial^n}{\partial w^n} f(w)|_{w = x} = f(w)
\end{eqnarray}
where $w = x + c$. Then the semi-classical limit of the recursion operator $L$ may be expressed as

\bea \label{symbol}
L &\sim & \mathcal{L} = f(x) e^{-\frac1n\partial_x} + h(x) + e^{\frac1n\partial_x}
\eea
where 
\beann
f(x) &=& f_0({\bf t}, x) + \dots + \frac1{n^{2g}}f_{g}({\bf t}, x) + \dots \\
 h(x) &=&  h_0({\bf t}, x) + \dots + \frac1{n^{g}}h_{g}({\bf t}, x) + \dots 
\eeann
are the asymptotic expansions of $b_{n}^2$ and $a_{n}$ determined by Theorem \ref{Workhorse} (described further in Section \ref{workhorse}); in this context they describe the symbol sequence for $L$ associated to a WKB representation of its eigenvectors with respect to the basis $\{\pi_\ell\}$. To get the idea one may formally take these eigenvectors to be asymptotically of the form
$e^{n S(x)}$. A straightforward expansion of 
\beann
\mathcal{L} \,\,\, e^{n S(x)}
\eeann
yields a leading order or {\it principal symbol} of the form
\bea \label{symbol0}
\mathcal{L}_0 = f_0(x) \,\,\, e^{-\theta} + h_0(x) + e^\theta 
\eea
where $\theta = S_x$. This further reveals the leading order form of (\ref{string1}) and (\ref{Toda1}) to be
\bea \label{PoissonString}
1 &=& \left\{ V'(\mathcal{L}_0)_- , \mathcal{L}_0 \right\} \\ \label{PoissonToda}
x \partial_{t_j} \mathcal{L}_0 &=& \left\{ \mathcal{L}^j_{0_- }, \mathcal{L}_0 \right\}
\eea
where the subscript here denotes truncation to terms with negative powers of $e^\theta$ and the bracket is the 
canonical Poisson bracket   
\beann
\left\{ \mathcal{F}(x,\theta) , \mathcal{G}(x, \theta) \right\}  &=& \frac{\partial \mathcal{F}}{\partial \theta} \frac{\partial \mathcal{G}}{\partial x} - \frac{\partial \mathcal{F}}{\partial x}\frac{\partial \mathcal{G}}{\partial \theta}. 
\eeann
For the 3-regular case, from these equations, by matching the coefficients of 1 and $e^{-\theta}$ respectively (all other coefficients cancel) one gets the systems 

\bea \label{pathstring}
\left( \begin{array}{c} 0 \\ 1 \end{array} \right) &=& \left( \begin{array}{cc} 1 + 6 t h_0  & 6t\\ 6t f_0 & 1 + g t h_0  \end{array} 
\right) \left( \begin{array}{c} h_{0 x} \\ f_{0x} \end{array} \right)
\eea
and 
\bea
\left( \begin{array}{c} h_{0 t} \\ f_{0 t} \end{array} \right) &=& \frac3{x} \left( \begin{array}{cc} 2 h_0 f_0  & 2 f_0 + h_0^2\\ 
 (2 f_0 + h_0^2) f_0 & 2 h_0 f_0  \end{array} 
\right) \left( \begin{array}{c} h_{0 x} \\ f_{0x} \end{array} \right)
\eea
This leading order limit is in accord with the standard quantum mechanical correspondence principle. The interesting relation with Poisson structures, which also relates to quantum groups, will be taken up elsewhere. 
We will more typically be using a path-integrated version ($x$-antiderivative) of (\ref{PoissonString}) which for the 3-regular 
case of (\ref{pathstring}) has the form
\bea \label{3string}
\left( \begin{array}{c} 0 \\ x \end{array} \right) &=& \left( \begin{array}{c} h_0 + 3t(2f_0 +h_0^2) \\ f_0 + 6 t h_0 f_0 \end{array} \right)
\eea
\smallskip

The fundamental connection between these symbol expansions and the map generating functions in (\ref{TopRec}) follows from the semi-classical expansion of a Hirota relation (\ref{Hirota}). This is explicitly given by
\begin{prop} \cite{EMP08} 
\begin{eqnarray} \label{boson12}
\left(e^{\frac1{n} \partial_x} - 2 + e^{- \frac1{n} \partial_x}\right) \frac1{n^2}\log \frac{ \tilde{Z}_n(x, {\bf t})}{\tilde{Z}_n(x, 0)}  &=& \left(e^{\frac1{n} \partial_x} - 2 + e^{- \frac1{n} \partial_x}\right)\sum_{g \geq 0}  \frac{1}{n^{2g}} E_{g}(x, \bf t\;) \\ \nonumber 
= \sum_{\ell \geq 1} \frac{2 n^{-2\ell}}{2\ell !} \frac{\partial^{2\ell}}{\partial w^{2\ell}}\left(\sum_{g \geq 0}  \frac{1}{n^{2g}} E_{g}(w, {\bf t}\;)\right)\Big|_{w=x}  
&=& \sum_{g \geq 0} \sum_{\ell \geq 1} \frac{2 n^{-2(g +\ell)}}{2\ell !} \frac{\partial^{2\ell}}{\partial w^{2\ell}} E_{g}(w, {\bf t}\;)\Big|_{w=x}\\ \nonumber
= \sum_{g \geq 0} \sum_{\ell = 0}^g \frac{2 n^{-2g}}{(2\ell + 2) !} \frac{\partial^{2\ell + 2}}{\partial w^{2\ell + 2}} E_{g-\ell}(w, {\bf t}\;)\Big|_{w=x}
&=& \log f_0(x, {\bf t}) + \sum_{g \geq 1} \frak{C}_g(f_0, f_1, \dots, f_g)(x, {\bf t}) n^{-2g}
\end{eqnarray}
where 
\bea \label{Cum}
\frak{C}_g(f_0, f_1, \dots, f_g) =  \mbox{the}\,\,\,  \mathcal{O}\left(n^{-2g}\right) \,\,\, \, \mbox{terms of}\,\,\, \log\left( 1 + \sum_{m=1}^\infty \frac1{n^{2m}} \frac{f_m}{f_0}\right). 
\eea
and $\frak{C}_0 = \log \frac{f_0}{x}$.
\end{prop}

\noindent The formal manipulations above are justified by the validity of the large $n$ expansion (\ref{TopRec}). This proposition may be rephrased as a recursive scheme for calculating the map generating functions:

\begin{eqnarray} \label{EggDE}
\frac{\partial^2}{\partial x^2} \widehat{E}_g &=& \frak{C}_g(f_0, f_1, \dots, f_g) \,\,\,\, \mbox{where} \\ \nonumber
\widehat{E}_g  &=&  E_g + \sum_{\ell = 1}^g
\frac{2}{(2\ell + 2)!} \frac{\partial^{2\ell}}{\partial x^{2\ell}} E_{g - \ell}
\end{eqnarray}

\begin{rem}
If one views the coefficients $\frac{f_m}{f_0}$ in (\ref{Cum}) as moments in a moment generating function, then 
the $\frak{C}_g$ are just the {\rm cumulants} for those moments which, by (\ref{EggDE}), are also given by recursive 
anti-differentiations in $x$ of those cumulants.
\end{rem}

\medskip

In this paper we will apply a slight modification of the symbol calculus just described in order to align with our full asymptotic expansions from Theorem \ref{Workhorse} and to describe how higher order equations deriving from (\ref{string1}) and (\ref{Toda1}) relate to the string polynomials and Bessel functions described in the previous sub-section. In section \ref{sec:44}
we will develop a path integrated version of the string equations which may be compactly expressed in terms of the resolvent 
of the recursion operator as

\beann
0 &=& a_n + \sum_{j=1} ^ J j t_j [z^{j-1}] \left[ (I - z L)^{-1}\right]_{n,n}\\
x  &=& b^{2}_n + \sum_{j=1} ^ J j t_j [z^{j -1}] \left[ (I - z L)^{-1}\right]_{n,n-1}.
\eeann
where $[z^{k}]$ picks out the $k^{th}$ coefficient in the geometric expansion of the resolvent. The semi-classical expansion of these equations has the form
\bea \label{scstring}
\left( \begin{array}{c} 0 \\ x \end{array} \right) &=&  \left( \begin{array}{c} h\\ f \end{array}  \right) + \sum_{j=1} ^ J j t_j [z^{j-1}] \left( \begin{array}{c} [\eta^0] \cr [\eta^{-1}] \sqrt{f}\end{array} \right) \left[ (1 - z \mathcal{L})^{-1}\right],
\eea
where here
\bea \label{symbol1}
\mathcal{L} &=& \sqrt{f(x)} \eta + h(x) + \sqrt{f(x)} \eta^{-1}.
\eea
Its principal symbol is
\bea \label{symbol1-0}
\mathcal{L}_0 &=& \sqrt{f_0(x)} \eta + h_0(x) + \sqrt{f_0(x)} \eta^{-1}\\ \nonumber
&=& \sqrt{f_0(x)} \left( \eta + \sqrt{y_0} + \eta^{-1}\right)
\eea
which provides the connection to the arguments of the linear combination of Appell polynomials appearing in (\ref{Hlambda}). Comparing this to (\ref{symbol0}) one sees that $\eta$ here plays a role similar to that of $e^\theta$ and this $\mathcal{L}_0$ is the symmetrized version of (\ref{symbol0}) corresponding to the symmetric recursion operator one would have for orthonormal rather than monic orthogonal polynomials. 

The deeper connection to Bessel functions will be seen, in Section \ref{StringApple}, to arise from the identity 
\beann
e^{X (\eta + \eta^{-1})} &=& \sum_{n = - \infty}^\infty I_n(2X) \eta^n
\eeann
where $I_n$ is the modified Bessel function of order $n$.

\subsubsection{Conservation Laws }

In addition to having a semi-classical limit (\ref{PoissonToda}) that is Hamiltonian, the Toda equations (\ref{Toda1}) have a semiclassical limit via (\ref{symbol1-0}) that has the form of a commuting hierarchy of conservation laws. Such multiple representations are not uncommon in integrable systems theory but the additional structure of the string equations in our systems is a special feature that we want to understand better and make use of.

These conservation laws are two dimensional systems of nonlinear PDE whose coefficients are rational functions of $f_0$ and $h_0$ which take the form

\begin{eqnarray}  \label{SYSTEM}
\frac{\partial}{\partial t_j} 
\left(
\begin{array}{c}
h_0 \\ f_0 + \frac12 h_0^2
\end{array}
\right) &+& 
\frac{\partial}{\partial x} 
\left(
\begin{array}{c}
\mathcal{F}_{1}\\ 
\mathcal{F}_{2} + h_0 \mathcal{F}_{1}
\end{array}
\right)  = 0,
\end{eqnarray}
with the fluxes defined by
\begin{eqnarray}
\mathcal{F}_{1} &=&    f_0^{\frac{j+1}{2}} \partial_\zeta R_{j}(\zeta)|_{\zeta = \sqrt{y_0}} \\
 \mathcal{F}_{2} &=&   f_0^{\frac{j}{2} +1} \left(S_{j}(\sqrt{y_0}) - R_{j}(\sqrt{y_0})\right).
\end{eqnarray}
An explanation of the origin and derivation of these conservation laws will be given in Corollary \ref{cor02}.
The functions $S_j$ and $R_j$ are polynomials of degree $j$ with coefficients in $\mathbb{Q}$. They are in fact the Appell polynomials that appeared in the description of the equilibrium measure, (\ref{Hlambda}), whose generators are reciprocal Bessel functions (see section \ref{AppellBessel}). In Theorem \ref{UNIVCONSLAW} we present a universal structure for this hierarchy. 
\medskip

We will see a further connection between the equilibrium measure and these conservation laws in Corollary \ref{RIEMINVTS} which establishes that the latter may be placed in Riemann invariant form with the Riemann invariants being the centered endpoints, $r_\pm = h_0 \pm \sqrt{f_0}$,  for the support of the measure given by (\ref{endpoint1}) and (\ref{endpoint2}).
\medskip

Finally, this will help to explain the unique role that the string equations play here. We will see in Lemma \ref{lem:extdiff} that the semiclassical limit of the path integrated form of those equations implicitly define, as a hodograph transformation, the integral surfaces of the conservation laws. 

\subsubsection{Arithmetic Rational Ruled Surfaces}. \label{RatRuled}

The leading order symbol of the semiclassical string equations (\ref{scstring}) yields equations of the form
\begin{eqnarray} \label{STRINGEQUATIONS}
\left(
\begin{array}{c}
0 \\ x
\end{array}
\right)
 &=& 
\left(
\begin{array}{c}
h_0 +  B_{12} \\ f_0 +  B_{11}
\end{array}
\right)
\end{eqnarray}
generalizing (\ref{3string}) and implicitly defining an algebraic surface  ${\mathcal{S}}$ 
in the four dimensional space with coordinates $(\xi, x, f_0, h_0)$ where 
\begin{eqnarray*}
B_{12} &=& \sum_{j=1}^J j t_j f_0^{\frac{j-1}{2}} S_{j-1}({\zeta})|_{\zeta = \sqrt{y_0}}\\
B_{11} &=&  \sum_{j=1}^J j t_j f_0^{\frac{j}{2}} \partial_\zeta R_{j-1}(\zeta)|_{\zeta = \sqrt{y_0}}.
\end{eqnarray*}
(The coefficients $B_{ij}$ in equations (\ref{STRINGEQUATIONS}) involve the irrationalities $\sqrt{f_0}$ and $\sqrt{y_0}$; however when the Appell polynomials are expanded and $\sqrt{y_0}$ is set to $h_0/\sqrt{f_0}$, these become regular polynomials in $f_0$ and $h_0$.) 

We note that (\ref{STRINGEQUATIONS}) describes a multi-parameter solution of the systems (\ref{SYSTEM}) in the sense 
that it is a solution of the $j^{th}$ conservation for initial data corresponding corresponding to all values of $t_k. \, k \ne j$ held fixed. 
 
${\mathcal{S}}$ is in fact a {\it ruled surface}, with ruling parametrized by $x$, over a rational curve.  It is an integral surface of the the conservation law (\ref{SYSTEM}) .  We illustrate this in the case of a regular map where the integral surface $\mathcal{S}$ --see section \ref{chargeom}--reduces to a cone over a rational algebraic curve $\mathcal{C}$. 
This is a consequence of the fact that for regular maps the parameter dependence reduces to a single self-similar variable  
$\xi = \xi_j = t_j x^{j/2 - 1}$. In the case that $j$ is odd, $\mathcal{C}$ may be explicitly parametrized as
\begin{eqnarray} \label{CURVE}
\xi^2  &=& \frac1{j^2} y_0 \frac{(S_{j-1}(\sqrt{y_0}) - y_0^{1/2} \partial R_{j-1}(\sqrt{y_0}))^{j-2}}{S_{j-1}(\sqrt{y_0})^{j}}
\end{eqnarray}
in the $(\xi^2, y_0)$-plane, where again $y_0 = h_0^2/f_0$, and the numerator and denominator of this function are both polynomials in just $y_0$ with rational numbers as coefficients.  This rationality means that the surface $\mathcal{S}$ is arithmetic, a feature that will imply that the map generating functions are definable as rational functions of $h_0$ and $f_0$ with coefficients in $\mathbb{Q}$. (In fact for at least the regular case the generating functions, with $x=1$  can be expressed purely in terms of the pivotal quantity $y_0$--see Theorem \ref{REGMAINTHEOREM}).

While the equation for $\mathcal{C}$ depends on $j$, the geometric form of this curve is independent of $j$(odd) having exactly two finite (real) turning points, or branch points, with respect to projection onto the $\xi^2$-coordinate and a single positive $y_0$-intercept which is an inflection point for $j > 3$ (see Figure \ref{5curve}), though not inflectionary for $j=3$ .  This is a first instance, a geometric one, of universality with respect to valence that we mentioned at the outset. 
(In a related vein, since for the mixed valence case the the string equations (\ref{STRINGEQUATIONS}) that define $\mathcal{S}$ are multiscale with respect to the various weighted quantities, $\xi_j$, they would appear to lack any kind of global homogeneity. However, if one assigns to $t_j$ an $x$-weight of $x^{1-j/2}$ (which is natural given its relation to $\xi_j$) then the first component of these string equations has global homogeneous weight $x^{1/2}$ and the second component has weight $x$. This is in fact stemming from a general homogeneity property of the string polynomials.)

For reasons indicated in Corollary \ref{RIEMINVTS} we refer to $\mathcal{C}$ as the {\em spectral curve}.
For the sake of concreteness we may illustrate this in the case of tri-valent maps ($j=3$):
The curve $\mathcal{C}$ (\ref{CURVE}) in this case becomes
\begin{eqnarray} \label{3CURVE}
\xi^2 &=& \frac1 9\,\frac {y_0 \left( 2-y_0 \right) }{ \left( 2+y_0 \right) ^{3}},
\end{eqnarray}
\begin{figure}[h] 
\begin{center}
\resizebox{2in}{!}{\includegraphics{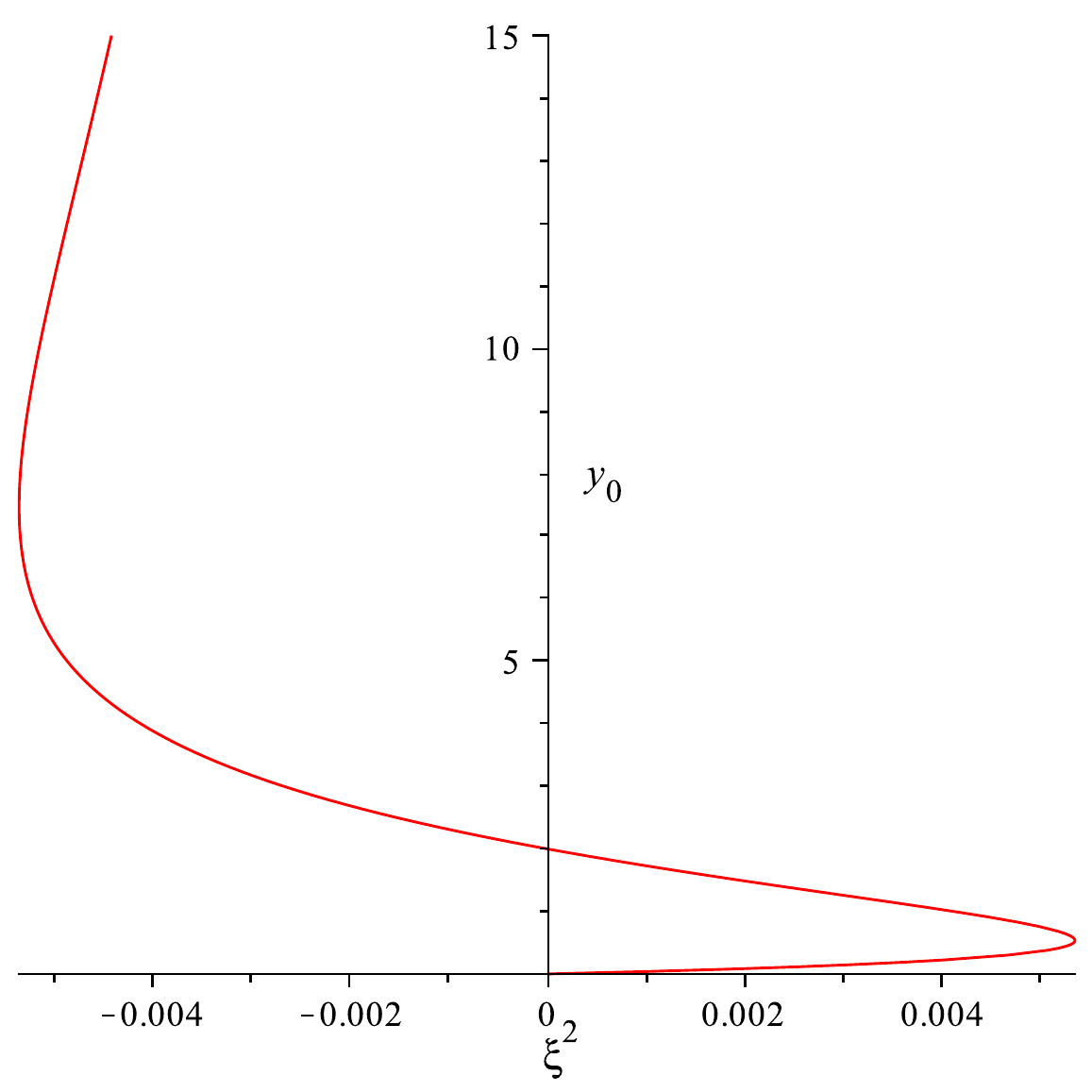}}
\caption{\label{GEOMUNIV} Spectral curve for 3-valent maps }
\end{center}
\end{figure}
a plot of whose real points for $y_0 > 0$ is given in Figure \ref{GEOMUNIV}.

As mentioned above the generating functions for odd regular maps with $x=1$ are rational functions on the spectral curve expressible in the single variable $y_0$. For $j=3$ this is illustrated by 

\beann
E_0|_{x=1} &=& \frac12 \log\left( \frac{2 + y_0}{2 - y_0}\right) - \frac{3 y_0^3-4y_0^2+8y_0-64}{30 (2-y_0)^2} \\
E_1|_{x=1} &=& -\frac1{24} \log\left( \frac{y_0^2 - 8 y_0 + 4}{(2 - y_0)^2}\right)\\
E_2|_{x=1} &=& {\frac { \left( 2800-4240\,y_0+2712\,{y_0}^{2}-1060\,{y_0}^{3}+175\,{y_0
}^{4} \right) {y_0}^{3}}{ 30 \left( 4-8\,y_0+{y_0}^{2} \right) ^{5}}}.
\eeann
As far as we we know this is the first exhibition in the literature of global rational {\it univariate} (i.e. expressible just in terms of the single pivotal quantity $y_0$) representations of odd valence generating functions.

For $j$ even, the conservation law collapses, due to reasons of symmetry, to a scalar conservation law, independent of $h_0$ (which may effectively be set to zero) and the equation for $\mathcal{C}$ takes the form (in terms of $z_0 = f_0|_{x=1}$ )
\begin{eqnarray*}
\xi &=& \frac{1 - z_0}{c_j z_0^{j/2}}\\
c_j &=& j {j-1 \choose j/2 -1}
\end{eqnarray*}
which has just a single finite turning point under projection to the $\xi$-coordinate. (The geometric form is again universal for $j$ even.)

Here too one has a global, rational, univariate representation of generating functions 
(see Appendix \ref{evenvalence}). For instance for $j=4$ one has
\beann
E_0|_{x=1} &=& \frac12 \log(z_0) + \frac1{24} (z_0 - 1) (z_0 - 9)\\
E_1|_{x=1} &=& -\frac1{12} \log(2 - z_0)\\
E_2|_{x=1} &=& \frac1{720} \frac{(z_0 -1)^3 (3 z_0^2 - 21 z_0 - 82)}{(2 - z_0)^5}.
\eeann
\smallskip

More generally, the setting of generating functions for mixed valence maps corresponds to considering solutions of the conservation law (\ref{SYSTEM}) for more general initial data that determines an arithmetic ruled surface more general than a cone. The leading order string equation (\ref{STRINGEQUATIONS}) again implicitly defines a mapping: $(x, \xi) \to (h_0, f_0)$. From the classical viewpoint of the singularity theory of mappings, going back to at least Whitney, one knows that the stable singularities of analytic mappings from the plane to the plane amount to just folds and cusps. However, the mappings here are not general as mappings; they are constrained to be solutions of (\ref{SYSTEM}). In this case it was shown in \cite{CEHL93} that such a constrained mapping could be viewed as a solution of (\ref{SYSTEM}) which is an analytic function on a Riemann surface for which the branch points move as part of the solution. There it is also found that the generic singularity types (i.e., those that are stable with respect to perturbations of initial data) are folds, cusps and nondegenerate umbilic points with non-zero 3-jet. An isolated singularity is generically a square root branch point corresponding to a fold. Since the fold is an envelope of characteristics it travels with the characteristic velocity of that fold. This summary is all that we shall essentially require for our discussion in Section \ref{stability}, but for precise definitions and details of proof we refer the reader to section 3 of \cite{CEHL93}.

The deeper relevance of this geometry for map enumeration is seen, for instance, in the fact (Theorem \ref{MAINTHEOREM}) that for $g > 1$, the map generating functions may be expressed in the form
\beann
E_g &=& \partial^{-2}_x \frac{P_g(h_0, f_0, h_{0x}, f_{0x}, h_{0xx}, \dots)}{f_0^g(f^2_{0x} - f_0 h^2_{0x})^{8g - 3}}
\eeann
where $P_g$ is a valence independent polynomial and $\partial_x^{-2}$ denotes an exact second anti-derivative of the given expression. In other words the higher map generating functions are expressible in terms of higher jets along the characteristic directions (rulings) of the integral surface $\mathcal{S}$. Again, the valence independence of this expression is another, even more fundamental instance of universality.

\section{Statement of Results} \label{statement}

Map counts are generally formulated in terms of generating functions. These are analytic functions whose $n^{th}$ Taylor-Maclaurin coefficient counts the number of labelled $g$-maps with exactly $n$ vertices in a way that we will now make precise.

The exponential generating functions, $E_g$, for enumerating labelled maps, $\gamma$, of genus $g$ are prima facie defined as
\begin{eqnarray} \label{GenFunc}
E_g(x,\vec{t}) &\doteq& \sum_{\gamma \in \{ g- \,\, {\rm maps}| {\rm valence} \leq J \}} x^{{|\rm Faces}(\gamma)|} \frac{\prod_{j=1}^J(-t_j)^{{n_j}(\gamma)}}{{{n_j}(\gamma)!}},
\end{eqnarray}
where $n_j(\gamma)$ is the number of $j$-valent vertices in the map $\gamma$. To reduce such enumerations to a combinatorial question, one needs to define when two maps are equivalent. (One counts maps modulo equivalence which makes the set of equivalence classes finite.) On a genus $g$ surface, two maps are equivalent if there is an orientation-preserving homeomorphism from the surface to itself that induces a homeomorphism of the graph to itself preserving the sets of vertices and edges but possibly respectively permuting them. Non-trivial equivalences, for which such a permutation is non-trivial, do arise (though for $n$ sufficiently large with respect to a fixed $g$, this does not occur). To avoid potential ambiguities one begins with enumerating labelled maps. By {\it labelling} one means that each map carries a distinct labelling, $1, \dots, n$ of its vertices and a cyclic labelling of the edges around each vertex (see \cite{EM03} for a detailed description). Division by the factorials then has the effect of removing the distinctions of labelling so that the enumeration becomes one of just geometric objects, modulo symmetries. In this regard $t_j$ should be replaced by $\frac{t_j}{j}$, to remove the edge labelings, but for notational convenience we absorb this $j$ into $t_j$; the pure geometric count can be recovered later by dividing coefficients by the appropriate factors of $j$.

A word of caution is in order here. Other treatments, such as \cite{EBook}, of map generating functions consider the dual case mentioned earlier in which the Taylor series terms correspond to holding the number of {\it faces} in the map fixed rather than the number of vertices. The two cases are of course completely symmetrical, by duality; however, specific counts can superficially appear to be different unless one is careful to make clear what is being held fixed. Throughout this paper we will always be holding the number of vertices fixed for each term.

As a consequence of Euler's formula, which in our case states that, for a $g$-map 
\beann
2-2g &=&  |{\rm Faces}(\gamma)| - |{\rm Edges}(\gamma)| + |{\rm Vertices}(\gamma)|\\
&=& |{\rm Faces}(\gamma)| - 1/2 \sum_{j=1}^J j n_j + \sum_{j=1}^J n_j
\eeann
 the terms in the above series for $E_g$ may be rewritten, via
\begin{eqnarray} \label{GenFunc2}
\prod_{j=1}^J t_j^{{n_j}(\gamma)} x^{{|\rm Faces}(\gamma)|}
 &=& x^{2 - 2g}  \prod_{j=1}^J \left(t_j x^{j/2 - 1}\right)^{n_j(\gamma)}, \,\, \mbox{as}\\ \label{GenFunc3}
E_g(x,\vec{t}) &=&   x^{2 - 2g} e_g (\vec{\xi})\\ \label{GenFunc4}
e_g(\vec{\xi}) &=&\sum_{\gamma \in \{{\rm unlabelled} \,\, g- \,\, {\rm maps}\}}{\frac{\prod_{j=1}^J\xi_j^{n_j(\gamma
)}}{\left|{\rm Aut}(\gamma)\right|}}
\end{eqnarray}
where $\xi_j = t_j x^{j/2 - 1}$ and the sum is taken over maps with maximum valence $\leq J$. ${\rm Aut}(\gamma)$ is the finite group of map automorphisms 
given as the quotient of equivalences of the $g$-map $\gamma$ by trivial equivalences as defined above.

\medskip

We have shown \cite{EM03} that these generating functions are analytic  in $\vec{\xi}$ near $0$ with Taylor-MacLaurin expansions of the form
\begin{eqnarray} \label{TMexp}
E_g(x,\vec{t}) &=& x^{2-2g} \sum_{n_j \geq 1} \frac1{n_1! \cdots n_J !} \prod_{j=1}^J\left(-\xi_j^{n_j}\right) \kappa_{g}(n_1, \cdots n_J).
\end{eqnarray}
where
\medskip

$\kappa_{g}(n_1, \cdots n_J) =$ the number of labelled $g$-maps with $n_j$ $j$-valent vertices for $j = 1, \dots, J$. 
\medskip

We refer to the overall factor $x^{2-2g}$ of $E_{g}$ as the {\it external weight} of the expression.
\bigskip

Our first main result, Theorem \ref{MAINTHEOREM} (just below), will effectively state that the generating functions for maps with $g >0$ can be expressed in terms of {\it universal} rational functions of the two fundamental combinatorial generating functions, $f_0$ and $h_0$, and their respective $x$-derivatives of bounded order. By {\it universal} here we mean that these representations of the $E_g$ ($g > 0$) have no explicit dependence on the valence parameters $\vec{t}$. In other words the  $\vec{t}$-dependence enters only through the dependence of the fundamental generating functions, $f_0$ and $h_0$, on these parameters. 

The coefficients of the conservation laws we have previously discussed do depend explicitly on the valence coefficients, $t_j$; however, Theorem \ref{UNIVCONSLAW} will show that this hierarchy can be written in a more universal form, in terms of the generating functions for Appell polynomials, with the valence parameters entering in the initial data for these PDE. 

We can now turn to the precise statement of our main results.
\medskip

All the generating functions we consider are homogeneous with respect to two weight gradings that reflect the self-similar scaling structure mentioned just above and in (\ref{scalingvars}):

\begin{defn}
The {\rm polynomial weight} of $f_0$ or any $x$-derivative $f^{(m)}_0 (= d^m f_0/dx^m)$ is 1. The  polynomial weight of $h_0$ or any $x$-derivative $h^{(m)}_0$ is $1/2$. The {\rm differential weight} of any $x$-derivative, $h^{(m)}_0$ or  $f^{(m)}_0$ is $m$.
\end{defn}
\begin{defn}
A polynomial in $h_0, f_0, h_{0x}, f_{0x}, h_{0xx}, \dots $ is {\rm homogeneous} if all its terms have the same polynomial weight and the same differential weight. The polynomial weight of a product  (or quotient) of homogeneous polynomials is given by adding (or subtracting) the polynomial weights of the factors (numerator - denominator) and similarly for differential weights of products and quotients. 
\end{defn}

\begin{theorem} \label{MAINTHEOREM} 
\begin{enumerate}[label=\alph*)]

\item All of the generating functions, $E_g$ for $g > 1$, have closed form expressions that are rational in terms of the auxiliary functions $h_0, f_0$ and their $x$-derivatives. As was also seen in  (\ref{TMexp}), these expressions have an overall external weight of the form $x^{2-2g}$. $E_g$ has polynomial weight 0 and differential weight $2g-2$. For genus 1, one has
\begin{eqnarray} \label{Eone} 
E_1 &=& \frac1{24} \log\left( \frac{x^2(f^2_{0x} - f_0 h^2_{0x})}{f^2_0}\right).
\end{eqnarray}
so that in this case the expression inside the logarithm is rational with polynomial weight 0, differential weight 2 and external weight 2; so this is at least an analogue of the weight relation for the higher genus generating functions.

\item The closed form expressions for $E_g$, $g > 0$, are universal in the sense that they have no explicit dependence on the valence parameters. We have just seen this for $g=1$. For $g > 1$ one has,
\beann
E_g &=& \partial^{-2}_x \frac{P_g(h_0, f_0, h_{0x}, f_{0x}, h_{0xx}, \dots)}{f_0^g(f^2_{0x} - f_0 h^2_{0x})^{8g - 3}}
\eeann
where $P_g$ is a valence independent polynomial.  The symbol $\partial^{-2}_x$  is an exact second antiderivative of the given expression which has the form of a rational funciton of $f_0, h_0$ and their x-derivatives, which is again universal. 
\item Alternatively, $E_g$, for $g>1$ is expressed as a rational function of just $f_0$ and $h_0$ (without any $x$-derivatives) but which {\it does} depend on valence. More explicitly the $E_g$  are functions on the rational ruled surface ${\mathcal{S}}$ 

(The coefficients $B_{ij}$ in equations (\ref{STRINGEQUATIONS}) involve the irrationality $\sqrt{f_0}$ and so one might, prima facie, expect the expressions for $E_g$ to involve this irrationality; however, that is not the case.)

\item $E_0$ has a {non-locally}  universal expression 
\beann
E_0 &=& \partial^{-2}_x \log\left( \frac{f_0}{x}\right).
\eeann
In the regular case of single valence $j$, this takes the form 
\begin{eqnarray}  \label{Ezero}
E_0 &=& \frac{1}{2} x^2 \log \frac{f_0}{x} + \frac{3(j-2) x^2}{4j}-\frac{(j-2)(j+1)x}{j(j+2)}\left(\q\right)+\frac{(j-2)^2}{2j(j+2)}\left(\frac{1}{2} f_{0}^2 + h_{0}^2 f_0 \right).
\end{eqnarray}
When $j$ is even, in (\ref{Ezero}), $h_0$ is set identically to zero.
\end{enumerate}
 \end{theorem}
 The proofs of parts (a) and (b) of Theorem \ref{MAINTHEOREM} are given in Section \ref{UnivEg}. Part (c) is established in Section \ref{RatEg} and the formulae in part (d) are derived in Section \ref{DerE0}. The universal differential expressions described in part (b) of this theorem become quickly quite complicated as $g$ increases. However, remarkably, those expressions can be iteratively "integrated up" in $x$ to get less complicated new expressions for $E_g$  that depend functionally only on $h_0, f_0$ and $x$ (and not their $x$-derivatives). However, in doing this, as one would expect, one picks up "integration constants" in our formulae, and {\it these} can depend on the valences in terms of the initial data (classically known as the {\it initial strip}) for the conservation law. This situation is somewhat reminiscent of solving a differential equation in which the equation itself is independent of parameters (i.e., it is autonomous) but solutions will depend on conditions initially or at the boundary. The fact that we can carry out these integrations explicitly is reminiscent of integrable differential equations. However these differential expressions are not differential equations in the ordinary sense and one does not have integrability here in the ordinary sense. Nevertheless they are calculated by a recursive scheme that is derived as a continuum limit of a hierarchy of integrable systems, the Toda Lattice equations to be precise.  The ability to integrate up is based on what we refer to as an {\it unwinding identity} (section \ref{Unwinding}). 

\begin{theorem} \label{UNIVCONSLAW}
The {\it hierarchy} of conservation law initial value problems, stated in (\ref{SYSTEM}), that uniquely determine $(f_0, h_0)$ may be re-expressed as the following universal hyperbolic system, 

\begin{eqnarray}  \label{UnivConsLaw}
t \frac{d}{dV} 
\left(
\begin{array}{c}
h_0 \\ f_0 
\end{array}
\right) &+& 
\left(x - \left(
\begin{array}{cc}
f_0 & h_{0} \\ f_{0}h_0 & f_0
\end{array}
\right) \right)  \frac{\partial}{\partial x} \left(
\begin{array}{c}
h_0 \\ f_0 
\end{array}
\right) = 0,
\end{eqnarray}
where $\frac{d}{dV}$ denotes the directional  derivative through the space of polynomial potentials in the "direction" of the fixed general background potential, 
$V = \sum_{j=1}^J \alpha_j \lambda^j$, initialized at the Gaussian potential $\frac12 \lambda^2$.  The PDE is a hyperbolic system for the flow in that direction. 
\end{theorem}
\noindent The proof of Theorem \ref{UNIVCONSLAW} is given in Section \ref{MV}, based on the conservation laws developed in Section \ref{sec:LOCL}. As given, the universal system (\ref{UnivConsLaw}) may appear to be just a formal linear combination of flows. However, its meaning for map enumeration arises when considering asymptotic map enumeration as the number of vertices limits to infinity (see Section \ref{sec:UnivAsympEnum}). In this mixed valence case, the asymptotics of the now multi-scale generating functions describes the behavior of maps as the number of vertices goes to infinity but with relative valence distribution fixed by the coefficients in the directional derivative. A more systematic discussion of mixed valence asymptotics will be deferred to a subsequent paper; however, we will comment on its larger significance in sections \ref{stability} and \ref{asympenum}.
\medskip

We now illustrate the rational, completely anti-differentiated, form of our generating functions in just the regular map case. For $j$-regular maps, the genus $0$ generating function has the form, already given in \eqref{Ezero}, of a rational function in $f_0, h_0$ and $x$ modulo a single logarithmic term. Genus 1 represents 
a critical value. All calculations of higher genus generating functions are based on its explicit evaluation (as well as that of $E_0$). Therefore a significant portion of this paper will be based on deriving  the first part of the following result. The remaining parts of the theorem describe deep and useful Laurent expansions of the higher generating functions that we will derive.. 

\begin{theorem} \label{REGMAINTHEOREM}
The genus 1 generating function in the $j$-regular case becomes
\begin{eqnarray*}
E_1 &=& \frac1{24} \log\left( \frac{x^2}{f^2_0} \frac{4 - y_0}{\widehat{D}} \right)\\
\widehat{D} &=& \left(2 - j + j\left( \frac{x}{f_0}\right)\right)^2 - (j-2)^2 y_0.
\end{eqnarray*}
$\widehat{D}$ is the discriminant of the projection from $\mathcal{C}$ onto the $\xi^2$-coordinate and so can be expressed as a polynomial in $y_0$.  It is divisible by $4-y_0$; the ratio within the above logarithm is a polynomial in $y_0$ of degree $j$. 

For $g>1$, $e_g$ (where $ e_g = x^{2g-2}E_g $) is  a rational funciton of just $f_0$ and $h_0$ on $\mathcal{C}$.  In the regular case the higher $e_g$ may be expressed as rational functions of $y_0$ whose poles are located  at  the zeroes of $\Pi(y_0)$ which is a polynomial of degree $j-1$ in $y_0$, defined in terms of the discriminant of $\mathcal{C}$ by the relation 
\beann 
\Pi(y_0)/\left(S_{j-1}(\sqrt{y_0})\right)^2 &=& \widehat{D}/(4-y_0).
\eeann 
The order at these poles is exactly $5g-5$. In particular, near the real turning point nearest to the origin, with coordinate $(\xi_c^2, y_{0c})$ where $\xi_c$ is the common radius of convergence 
for the Taylor-Maclaurin expansions of all of the  $e_g$, one has a partial  fractions expansion in $\xi$ of the form
\beann
e_g(y_0) &=& \widehat{C}^{(g)} + \cdots + \frac{c_g(j)}{(y_0 - y_{0c})^{5g-5}} \\
e_g(z_0) &=& C^{(g)} + \frac{c_0^{(g)}(\nu)}{(j/2 - (j/2-1)z_0)^{2g-2}} + \cdots + \frac{c_{3g-3}^{(g)}(j/2)}{(j/2 - (j/2-1)z_0)^{5g-5}}
\eeann
for odd and even valence $j$ respectively and with $z_0 = f_0/x$. $c_g$ and $c_{3g-3}^{(g)}$ are non-zero for all $j$.
$C^{(g)}$ are valence independent constants and we expect the same to be true for $\widehat{C}^{(g)}$.
\end{theorem}
The derivation of the formulae for $E_1$ is carried out in Section \ref{toprec} and that of the Laurent expansions in section \ref{sec:Laurent}.

\begin{rem} \label{remark23}
More generally, in the case of mixed valence the poles of the $E_g$, for $g>1$, are confined to the discriminant locus of the projection of $\mathcal{S}$ onto the $\xi^2$ coordinate.
\end{rem}
 
Again, we illustrate these expressions in the case that $j=3$:
\beann
\widehat{D} &=& \frac{-y_0^3 +12 y_0^2 - 36 y_0 + 16}{(2 + y_0)^2}\\
S_2(\sqrt{y_0}) &=& 2 + y_0 \\
\Pi &=& y_0^2 - 8 y_0 + 4 \\
e_0 &=& \frac12 \log\left( \frac{2 + y_0}{2 - y_0}\right) - \frac{3 y_0^3-4y_0^2+8y_0-64}{30 (2-y_0)^2} \\
e_1 &=& -\frac1{24} \log\left( \frac{y_0^2 - 8 y_0 + 4}{(2 - y_0)^2}\right)\\
e_2 &=& {\frac { \left( 2800-4240\,y_0+2712\,{y_0}^{2}-1060\,{y_0}^{3}+175\,{y_0
}^{4} \right) {y_0}^{3}}{ 30 \left( 4-8\,y_0+{y_0}^{2} \right) ^{5}}}.
\eeann
\bigskip

The utility of these results for enumeration may now be described. Our stated goal has been to construct generating functions for the enumeration of maps from which explicit counts can be made while at the same time revealing, hopefully, more general combinatorial and probabilistic insight about these random structures and their universal characteristics. We have already seen some of these universal characteristics and there will be others to follow but we want to first describe the practical utility of our compact generating functions for extracting counts.  To explain this we continue to restrict attention to the case of regular maps. 

As just stated, we started with the goal of deriving explicit expressions for the $e_g$ as functions of $\xi^2$. By means of the PDE system (\ref{SYSTEM}) we were able to re-express these generating functions as comparatively much simpler elementary functions of the variable $y_0$.  The relation between $\xi^2$ and $y_0$ is mediated by the spectral curve (\ref{CURVE}). For each specific value of $j$ this is also a rather simple relation (see for example the tri-valent case, (\ref{3CURVE})). So to relate $e_g$ as a function of $y_0$ to its form in terms of $\xi^2$ comes down to eliminaitng     $y_0$ in terms of $\xi^2$  along the spectral curve (Figure \ref{GEOMUNIV}).  However, the larger point we wish to make is that for the purpose of extracting key combinatorial information from the generating functions, it is not necessary to make this last step.  A primary example of this is the Taylor-Maclaurin coefficients of $e_g$ for regular maps that provide the basic enumerations, (\ref{TMexp}), that motivated the original studies of maps. So, for instance, we find in Section \ref{integform} that,
\begin{eqnarray} \label{KAPPA}
\frac{\kappa^{(j)}_{1}(2m)}{(2m)!} &=& \frac1{2\pi i} \frac{j^{2m}}{24 m} \oint \frac{(2 (\partial S_{j-1}/ \partial y_0) \Pi  - S_{j-1} (\partial \Pi / \partial y_0 ) ) S_{j-1}^{jm-1}}{\Pi (S_{j-1} - (j-1) y_0^{1/2}R_{j-2})^{(j-2)m}} \frac{dy_0}{y_0^m}
\end{eqnarray}
This rational representation allows for efficient and elementary evaluation of the counts. In the tri-valent case this integral directly yields
\beann
\left[\frac{\kappa^{(3)}_{1}(2)}{3 \cdot 2!}, \frac{\kappa^{(3)}_{1}(4)}{3 \cdot 4!}, \frac{\kappa^{(3)}_{1}(6)}{3 \cdot 6!}, \frac{\kappa^{(3)}_{1}(8)}{3 \cdot 8!}, \frac{\kappa^{(3)}_{1}(10)}{3 \cdot 10!}, \dots\right] 
= \left[\frac32, 135, 16524, 2291976, \frac {1701555984}{5}, \dots \right]
\eeann
with odd coefficients all zero for topological reasons. These counts are further illustrated in Section \ref{integform}. 
\medskip

In the general mixed valence setting, the $x$ dependence does not scale out self-similarly; or, more precisely, one has functions of several distinct self-similar variables. However, a generic transverse $x$-slice of the integral surface $\mathcal{S}$ is a curve analogous to $\mathcal{C}$. In complete generality these slices may have more than just two real turning points. However, if one chooses valence parameters to be near a single pure valence, i.e., to be in a class of $\vec{t} = (t_1, \dots , t_J)$ sufficiently near $\vec{t}^{(j)} = (0, \dots, 0 , t_j, 0, \dots, 0)$ then the corresponding generic $x$-slices will have only two real turning points and be {\it close} to the regular $\mathcal{C}$. More precisely, one says that $\mathcal{C}$ is {\it stable} within the class of mappings that are solutions to conservation laws of the form that will be described in Corollary \ref{RIEMINVTS}.  Further commentary on this is provided in subsection \ref{stability}. 

In the particular case of an {\it even} regular map,  the integral surface itself becomes independent of $h_0$ and the $\xi^2$-projection has a unique turning point. This significantly simplifies the closed form of the generating functions: 

\begin{corollary} \label{EVENVALENCE}
 In the case of regular even valence $j = 2\nu$, the system of conservation laws (\ref{SYSTEM}) collapses to a scalar conservation law, the closed form expressions for $E_g$ become functions of just $f_0$ and they are realized, robustly, as the limit of the general expressions as $h_0 \to 0$.
\end{corollary}

For instance, in this case the genus 1 generating function bcomes
\beann
E_1 &=& \frac1{24} \log\left( \left(\frac{x f_{0x}}{f_0}\right)^2\right)\\
&=& \frac1{24} \log\left(   \frac{4x}{f_0 \widehat{D}}  \right) \\
\widehat{D} &=& \left( 2 - j + j \left( \frac{x}{f_0}\right)\right)^2\\
&=& 4 \left( 1 - \nu + \nu \left( \frac{x}{f_0}\right)\right)^2.
\eeann
Simplifying it follows that
\beann
E_1 &=& \frac1{12} \log\left( \frac{x}{\nu x - (\nu -1) f_0}\right)
\eeann
or eqivalently, setting $f_0 = x z_0$,
\beann
e_1 &=& -\frac1{12} \log\left({\nu  - (\nu -1) z_0}\right)
\eeann
which coincides with our earlier calculations in the case of even regular maps (see Appendix \ref{evenvalence}).

For genus 0 the regular generating function (\ref{Ezero}) with even valence $j = 2 \nu$ also follows this rule:

\beann
E_0 &=& \frac{1}{2} x^2\log \frac{f_0}{x} +\frac{(\nu-1)^2}{4\nu(\nu+1)}\left(f_0- x\right)\left(f_0-\frac{3(\nu+1)x}{\nu-1}\right)\\
e_0 &=& \frac{1}{2} \log z_0 +\frac{(\nu-1)^2}{4\nu(\nu+1)}\left(z_0- 1\right)\left(z_0-\frac{3(\nu+1)}{\nu-1}\right)
\eeann
which again agrees with our earlier calculations.

\begin{rem} \label{REM1.1} 
The rationality results stated in Theorem \ref{MAINTHEOREM} (c) resolve and significantly extend early conjectures in \cite{BIZ80} and \cite{BC91} about the possible structure of closed form expressions for the $E_g$. 
In particular the results here extend to the fully general case what we had already established in \cite{Er09} for regular even valence maps.
\end{rem}

\begin{rem} \label{REM1.2}
The universal form of expressions like (\ref{Eone}) is completely new to the literature on map enumeration. The universality class here is that corresponding to the class of random matrix ensembles given by (\ref{RMmeasures}) whose eigenvalue spacing statistics Deift et al \cite{DKM, DKMVZ} showed to be universal. The connection of our generating function representations to this spectral statistics is  indicated in Corollary \ref{RIEMINVTS} below and further elaborated on in section \ref{evolvemeas}. 
\end{rem}

\begin{rem} \label{REM1.3}
For scalar conservation laws the construction of an integral surface is a standard exercise in the method of characteristics. However, for higher dimensional systems this is far from being that straightforward. So the existence of the closed implicit solutions (\ref{STRINGEQUATIONS}) indicates that our system is exceptional. There are some general descriptions of exceptional hyperbolic systems of this type \cite{T85}; in our case this special structure stems from the fact that the generating functions are solutions of a continuum limit of the integrable Toda Lattice hierarchy.
\end{rem}

\begin{corollary} \label{RIEMINVTS}
In the case when the valences are not all even, the conservation law (\ref{SYSTEM}) may be effectively diagonalized with distinct eigenvalues (i.e., placed in {\it Riemann invariant form} \cite{Wh}) as
\begin{eqnarray} 
\frac{\partial}{\partial \xi} 
\left(
\begin{array}{c}
r_+ \\ r_-
\end{array}
\right) &+&  
\left(
\begin{array}{cc}
\lambda_+ & 0\\ 
0 & \lambda_-
\end{array}
\right)  \frac{\partial}{\partial x} 
\left(
\begin{array}{c}
r_+ \\ r_-
\end{array}
\right) = 0, \\
r_\pm &=& h_0 \pm 2 \sqrt{f_0}\\ \label{endpoints}
&=& \pm \sqrt{f_0} \left(2 \pm \sqrt{y_0}\right)\\
\lambda_\pm &=& B_{11} \pm \sqrt{f_0} B_{12}
\end{eqnarray}
where $r_\pm$ are the Riemann invariants and $\lambda_\pm$ are the characteristic speeds.  The invariants $r_\pm$ are the endpoints of the support of the spectral equilibrium measure for the random matrix ensemble with probabiity measure (\ref{RMmeasures}). These endpoints in fact completely determine the equilibrium measure. 
\end{corollary}

In the case of even valence, $\lambda_+ = \lambda_-$ and the system collapses to a scalar conservation law as mentioned earlier. In this case the connnection to random matrix spectral statistics is different and we have detailed that elsewhere \cite{Er09}.

\begin{rem} 
Another universal feature of our results is {\it arithmetic} in character. As we have seen, the closed form generating functions will also reduce to expressions in terms of polynomials in $y_0$ with rational coefficients. It follows that all statistical quantities related to these generating functions will be defined over algebraic number fields. Certainly what these number fields are depends on the particular valences involved, but the splitting structure of the enumerative formulas over the real numbers has a universal character (see Section \ref{arithmetic}). 
\end{rem}

\subsection{Universal Enumerative Asymptotics} \label{sec:UnivAsympEnum}

One final universal characteristic of the map counts is their asymptotic form for large $n$. Most, if not all, prior work on map enumeration has focussed on asymptotic enumeration (see section \ref{asympenum}). Because we have closed form expressions for our generating functions we are able to extract this information simply and in great generality. We find, for instance (Corollary \ref{cor:egasymp}), that for $j$-regular maps with $j$ odd, one has 
\begin{eqnarray} \label{KAPPA-ASYMP}
\frac{\kappa^{(j)}_{g}(2m)}{(2m)!} &\sim& \frac{c_g}{ \Gamma(\frac{5g-5}{2})} t_{c}^{1/2 - 5g/2 -2m} (2m)^{(5g-7)/2} \qquad \text{as }m\rightarrow \infty,
\end{eqnarray}
where $\xi_c = x^{j/2 - 1}t_c$ is the common radius of convergence for the Taylor-Maclaurin series of all the $E_g$.
The coefficient $c_g$ contains the information related to the limit $\xi \to \xi_c$ and is determined by a recursion of the form
\begin{eqnarray} \label{RECURRENCE}
0 &= & c_{g+1} + C_1     C_2 z_{0,c}^{2g-1} (25g^2-1)   c_{g } +  6C_1  \sum_{m=1}^{g} c_{m} c_{g+1-m}. \qquad  (g\geq 1)
\end{eqnarray}
The explicit valence dependence here is contained only in the coefficients $C_k$ (see (\ref{p872})) and $z_{0,c}$ is the value of the generating function $z_0$ (where $f_0 = x z_0$) at $\xi = \xi_c$. We note that the form of this asymptotics is quite similar to that of the even valence case (see (\ref{eg-asymp})). Moreover, the above recurrence mirrors the universal form of the recurrence in the even valence case which calculates the coefficients in an asymptotic expansion of the tritronqu\'ee form of the first Painlev\'e transcendent \cite{CLRM}.
As described in Appendix \ref{POne}, this further advances, to the odd valence case, the connection between the role that the Painlev\'e hierarchy and its extensions play in relating the double scaling limit of random matrix theory to 2D quantum gravity.

\subsection{Other Approaches to Map Enumeration} \label{approaches} 
 
The original approaches to map enumeration, going back to Tutte \cite{Tu}, used purely combinatorial methods.
More recent developments in this vein \cite{JV, BGR, CMS, Chap}. have discovered remarkable and elegant combinatorial bijections, in some cases between maps and group generators related to the character theory of the symmetric group and in other cases between maps and decorated trees that reduce map enumeration to other more tractable enumeration problems.
\bigskip

There is an alternative, formal, approach to map enumeration based on resolvent identities for random matrices, as opposed to recurrence operators, that goes back to Ambjorn, Chekov, Kristjansen, and Makeenko \cite{A} which also yields a recursive scheme for calculating the $e_{g}({\bf{t}})$. This approach is referred to as the method of {\it loop equations}. Eynard  \cite{Eynard, EBook} later streamlined the calculation of generating functions from loop equations and established a direct connection between loop equations and Tutte's equations that are key to the combinatorial method mentioned earlier. Using Theorem \ref{Workhorse} this formal method can also be derived rigorously \cite{EM7}. See also \cite{GS} for another approach to rigorously establish loop equations.

\subsection{Outline}

In section \ref{background} we provide a more technically precise but self-contained summary of the foundations
underlying the key ingredients outlined in the Introduction, related to random matrix theory, orthogonal polynomials and Riemann-Hilbert analysis.  Subsection \ref{workhorse} amplifies the statement of Theorem \ref{Workhorse} as it relates to the map generating functions, $E_g$ as well as the fundamental generating functions, $f_0, h_0$, and their higher genus analogues $f_g, h_g$. In subsection \ref{OddValentExt} we outline how our orthogonal polynomial asymptotic methods get extended to the case of odd degree dominated exponential weights via use of non-Hermitian orthogonal polynomials.

Section \ref{sec:44} introduces the Toda lattice and String equations that the recurrence operators satisfy and uses these to develop a symbol calculus based on the contnuum limits of these equations. The leading order forms of the Toda and String equations are then described in detail. In particular, their conservation law and universal form, as stated in (\ref{SYSTEM}) and Theorem \ref{UNIVCONSLAW} respectively are  revealed as well as their Riemann invariant form as it was stated in Corollary \ref{RIEMINVTS} (see subsection \ref{RINVARIANTS}). In subsection \ref{ODEreduction} it is shown how, on the integral surfaces corresponding to regular maps, these equations reduce to ODEs. 

Section \ref{Hopf-String} starts with a brief introduction to the binomial Hopf algebra and its general connection to Appell polynomials. We then build the fundamental link between Appell polynomials associated to Bessel functions and the conservation laws and string polynomials introduced in Section \ref{sec:44}. In particular, the crucial unwinding identity for string polynomials is proven in section \ref{Unwinding}.

Sections \ref{chargeom} - \ref{arithmetic} develop the characteristic geometry of the conservation laws on integral surfaces $\mathcal{S}$ associated to regular maps introduced in subsection \ref{RatRuled}. In this regular case we demonstrate that  the surface $\mathcal{S}$ is a cone over a rational algebraic curve $\mathcal{C}$, explictly given by (\ref{CURVE}), which we refer to as the spectral curve.   In particular Section \ref{evolvemeas} illustrates the relation between $\mathcal{C}$ and the equilbrium measures of the associated matrix ensembles, thereby justifying its name. Section \ref{gaussleaf} relates the characteristic geometry of the conservation laws to the algebraic geometry of the spectral curve in terms of the Appell polynomials introduced in section \ref{Hopf-String}. Section \ref{arithmetic} details the arithmetic structure of $\mathcal{C}$ that explains the local partial fractions expansions of the map generating functions at branch points stated in Theorem \ref{REGMAINTHEOREM}. This section also illuminates Galois symmetries and divisibility properties of the generating functions.

In Sections \ref{toprec} - \ref{closedforms} we establish the results that underlie the statements in Theorem \ref{MAINTHEOREM} as well as the pole order result of Theorem \ref{REGMAINTHEOREM} (see Theorem \ref{Eg}). 
Section \ref{toprec} provides a detailed account of how the higher genus generating functions are calculated and the source of their universality. Section \ref{finestruc} collects  information about the closed form structure of the higher genus auxiliary map generating functions $h_g$ and $f_g$.  
Section \ref{closedforms} focusses on the explicit structure of $E_0$ and $E_1$ and basic properties of the closed form expressions for the remaining $E_g$. Subsection \ref{UnivEg} establishes the universality results stated in Theorem \ref{MAINTHEOREM} and subsection \ref{RatEg} establishes the rationality results of that theorem. Finally in subsection \ref{integform} we show how these closed forms are used to extract explicit map counts and explain the form of these results that were illustrated in (\ref{KAPPA-ASYMP}) 
and (\ref{RECURRENCE}).

The results summarized in Corollary \ref{EVENVALENCE} are more fully explained in Appendix \ref{evenvalence}.

\section{Background} \label{background}

This section explains the origin and background for deriving (in section  \ref{sec:44} ) the conservation laws (and their higher genus prolongations) that were presented in the Introduction.  Much of this background will just be summarized; fuller details may be found in \cite{EM03, EMP08, BD10, EP11}.
The combinatorial significance of the fundamental generating functions, $f_0$ and $h_0$  introduced in section \ref{sec:0} (along with, again, their higher genus prolongations) will be explained here.

\subsection{Random Matrix Partition Functions} \label{sec:3.1}

Starting with this chapter we will introduce a useful change of notation related to the t'Hooft parameter initially presented in the introduction. Precisely we introduce a new parameter, $N$, implicitly defined by 
\bea \label{tHooft2}
x = \frac{n}{N}.
\eea 
This has the effect of separating out two physical parameters that comprise $x$: the system size $n$ and an effective "inverse temperature" $N$. In our asymptotic statements we will always take $n$ and $N$ to go together to $\infty$ so that $x$ remains 
$\mathcal{O}(1)$. This has the effect of modifying the random matrix measure  to read 
\begin{eqnarray} \label{RMmeasures2}
d\mu_{\bf t} &=& \frac1{\tilde{Z}_n} \exp\left\{-N \text{Tr } [ V_{\bf t}(M)] \right\} dM
\end{eqnarray}
for $M \in \mathcal{M}_n$ where $V_{\bf t}$ is now the polynomial 
\begin{eqnarray} \label{eq:genpot}
V_{\bf t}(\lambda) &=& \frac1{2} \lambda^2 + \sum_{j=1}^J t_j \lambda^j.
\end{eqnarray}
Note that this removes the dependence of $V_{\bf t}$ on $x$; however, the dependence of our asymptotic expansions on $x$ continue to hold as described in the Introduction. For much of this section we will, for simplicity, restrict attention to the special case
where
\beann
V(\lambda) &=& \frac12 \lambda^2 + t_1 \lambda + t \lambda^j.
\eeann
This potential corresponds to the case of regular maps.  Extensions to more general potentials will be discussed in section \ref{MV}.
There are of course some potential convergence issues here; these will be addressed later in Section \ref{OddValentExt}.
Due to the unitary invariance of the partition function  it, as well as expectations of unitarily invariant random variables, may be reduced to integrals over eigenvalues of Hermitian matrices. The latter is the form of the partition function we will typically work with (see (\ref{GenPot}) ).

\subsection{Orthogonal Polynomials} \label{Orth}

One may apply standard methods of orthogonal polynomial theory that go back to Szeg\"o (see also \cite{EP11}), to deduce the
existence of a semi-infinite lower unipotent matrix $A$ such that the recursion operator $L$, (\ref{multop}), may be conjugated to a shift operator,
\begin{eqnarray*}
{L} = A^{-1}\epsilon A \,\,\,\,\,&\mbox{where}&\,\,\,\,\,
\epsilon = \left(\begin{array}{cccc} 0 & 1 &  &  \\
                              0 & 0 & 1 &\\
			          & 0 & 0  & \ddots \\
                                  &      & \ddots & \ddots 
\end{array}\right)\,.
\end{eqnarray*}

This is related to the Hankel matrix 
\begin{eqnarray*}
\mathcal{H} =  \left(\begin{array}{cccc} m_{0} & m_1 &  m_2 & \dots\\
                              m_1 & m_2 & m_3 & \dots \\
			      m_2    & m_3 & m_4  & \dots \\
                                \vdots &  \vdots   & \vdots & \ddots 
\end{array}\right),
\,\,\,\,\, & \mbox{where} & \,\,\,\,\,\,
m_k = \int_\mathbb{R} \lambda^k e^{-N V(\lambda)} d \lambda
\end{eqnarray*}
is the $k^{th}$ moment of the measure, by
\begin{eqnarray*}
A D A^{\dagger} =  \mathcal{H} , \qquad 
D = \mbox{diag}\,\left\{ d_{0}, d_{1} \dots \right\} , \qquad
d_n = \frac{\det \mathcal{H}_{n+1}}{\det \mathcal{H}_{n}}
\end{eqnarray*}
where $\mathcal{H}_n$ denotes the $n \times n$ principal sub-matrix of $\mathcal{H}$ whose determinant may be expressed as 
\begin{eqnarray}
\det \mathcal{H}_n &=& n! {Z}^{(n)} \left(\xi_1,  \xi_{2\nu} \right)\\ \nonumber
\label{szego} {Z}^{(n)} \left(\xi_1, \xi_{2\nu}\right) &=& \int_\mathbb{R} \cdots \int_\mathbb{R} \exp\left\{
-N^2\left[\frac1{N}\sum_{m=1}^{n} V(\lambda_{m}; \xi_1, \ \xi_{2\nu})  - \right.\right. \\
\label{GenPot}
&& \phantom{\int_\mathbb{R} \cdots \int_\mathbb{R} \exp}\hspace{2cm}
\left. \left.
\frac1{N^2} \sum_{m\neq \ell} \log{|
\lambda_{m} -
\lambda_{\ell} | } \right] \right\}  d^{n} \lambda,
\end{eqnarray}
where $V(\lambda; \xi_1, \ \xi_{2\nu + 1}) = \frac12 \lambda^2 + \xi_1 \lambda +  \xi_{2\nu} \lambda^{2\nu}$.
Set $\det \mathcal{H}_0 = 1$.
\medskip

\begin{rem}
$\det \mathcal{H}_n$ is referred to as a Hankel determinant ($\mathcal{H}_n$ is a Hankel matrix). 
In a set of celebrated works, Szeg\"o analyzed the leading asymptotic behavior of Toeplitz determinants. The above identity shows that in studying the asymptotic behavior of the partition function, we are also exploring the analogue of Szeg\"o's theorems but for Hankel determinants. Indeed this provides another perspective on what was done in \cite{EM03}.  
\end{rem}
\smallskip

One sometimes needs to extend the domain of the partition function to include more than one parameter, such as $\xi_1$, as was done here. (See (\ref{ett}) - (\ref{hcount}) to understand the combinatorial motivation for this.) Doing this presents no difficulties in the prior constructions as long as the highest order term is even. 
\smallskip

The diagonal elements of $D$ may in fact be expressed as
\begin{eqnarray} \label{szego1}
d_n = \frac{\tau^2_{n+1, N}}{\tau^2_{n, N}} d_n(0) \,\,\,\,\,& \mbox{where} &  \,\,\,\,\,
\tau^2_{n,N} = \frac{{Z}^{(n)}\left(\xi_1, \xi_{2\nu}\right)}{{Z}^{(n)}\left(0,0\right)}
\end{eqnarray}
Tracing through these connections, from $\mathcal{L}$ to $D$, one may derive the fundamental identity relating the random matrix partition function to the recurrence coefficients,
\begin{eqnarray}\label{Hirota}
b^2_{n,N} = \frac{d_n}{d_{n-1}}&=& \frac{\tau^2_{n+1, N}\tau^2_{n-1, N}}{\tau^4_{n, N}} b^2_{n,N}(0)
\end{eqnarray}
which is the basis for analyzing continuum limits. (Note that $b^2_{0,N}(0) = 0$ and therefore $b^2_{0,N} \equiv 0$.) A differential version of these relations and also of the diagonal recursion elements when the external potential is not even (as above when we add the parameter $\xi_1$) is given by:

\begin{eqnarray} \label{RecCoeff}
\label{a} a_{n, N} &=& -\frac1{N} \frac{\partial}{\partial\xi_1} \log \left[ \frac{\tau^2_{n+1, N}}{\tau^2_{n, N}} \right]  
                                  =  -\frac1{N} \frac{\partial}{\partial \xi_1} \log \left[ \frac{Z^{(n+1)}(\xi_1,\xi_{2\nu})}{Z^{(n)}(\xi_1, \xi_{2\nu})} \right]\\ \nonumber
\label{b}  b_{n, N}^2 &=& \frac1{N^2} \frac{\partial^2}{\partial \xi_1^2} \log \tau^2_{n, N}
                                       = \frac1{N^2} \frac{\partial^2}{\partial \xi_1^2} \log Z^{(n)}(\xi_1, \xi_{2\nu})\,,
\end{eqnarray}

All of the above can be extended to the case when the polynomial potential's leading term is an odd power (as in $V(\lambda) = \frac12 \lambda^2 + \xi_j \lambda^j$ where $j$ may be odd) by deforming the contour of integration, for the measure of orthogonality, away from the negative real axis. This was first descrirbed in \cite{BD10} and later extended in \cite{EP11}.
See also section \ref{OddValentExt}.

The structure of the recurrence coefficients may be studied in depth via the {\it String Equations}:
\begin{eqnarray}		
0 =&  V'(L)_{n,n}	\label{UnDiff1}		\\
\frac{n}{N} =&  V'(L)_{n,n-1}	\label{UnDiff2}		
\end{eqnarray}  
Equation (\ref{UnDiff1}) can be derived as follows.  Integrating by parts one find that
\begin{eqnarray*}		
0 =\int \pi'_n(\lambda) \pi_n(\lambda)\, e^{-NV(\lambda)}\,d\lambda	
&=& N \int \pi_n(\lambda)^2  V'(\lambda)\, e^{-NV(\lambda)}\,d\lambda 	
= N ||{\pi_n}||^2 V'(L)_{n,n}.
\end{eqnarray*}
An analogous calculation, starting with $n ||{\pi_{n-1}}||^2 = \int \pi'_{n}(\lambda) \pi_{n-1}(\lambda)e^{-NV(\lambda)} \, d\lambda$ establishes (\ref{UnDiff2}).

For the elementary potential, $V = \frac12 \lambda^2 + \xi \lambda^j$ the string equations become
\begin{eqnarray}
0 &=& \left(L +j \xi L^{j-1}\right)_{n, n}\\
\frac{n}{N} &=& \left(L +j \xi L^{j-1}\right)_{n, n-1}.
\end{eqnarray}

\subsection{Passage to the Continuum Limit} \label{workhorse}

We now turn to the continuum limits of the fundamental equations.  The key element here is the following asymptotic result, derived in \cite{EM03}, for the logarithm of the tau-function (\ref{szego1})

\begin{eqnarray} \label{I.002} 
\log \tau^2_{n, N}({\xi}\;) &=&
N^{2} E_{0}(x, {\xi \;) + E_{1}(x, \xi}\;) + \frac{1}{N^{2}} E_{2}(x, \xi \;) + \cdots +  \frac{1}{N^{2g-2}} E_{g}(x, \xi\;) + \dots \,\,\,\,\,\,\,\,\,\,
\end{eqnarray}
as $n, N \to \infty$ with  the {\it 't Hooft parameter}, $x = \frac{n}{N}$, held fixed. Moreover, for $Re \,\, \xi > 0$ and $x \sim 1$,
\begin{enumerate}

\item[i.] \label{unif} the expansion is uniformly valid on compact subsets in $(x,\xi)$;

\item[ii.] \label{analyt} $e_g(\xi)$ has a finite radius of convergence, $\xi_c$, in a neighborhood of $\xi = 0$ determined by an algebraic singularity at the point $\xi_c$ on the real axis 
(for fixed valence, $\xi_c$ is the {\it same} for all $g$);

\item[iii.] \label{diff} the expansion may be differentiated term by term in $(x,\xi\;)$ with uniform error estimates for these derivatives similar to those given in (\ref{errorest}).

\end{enumerate}

The meaning of (i) is that for each $g$ there is a constant, $K_g$, depending only on the compact subset  and $g$ such that 
\begin{eqnarray} \label{errorest}
\left| \log \tau^2_{n, N} \left(\xi\;\right) - N^2 E_0(x, \xi\;)  -  \dots  - \frac{1}{N^{2g-2}} E_{g}(x, \xi\;) \right|  \leq \frac{K_g}{N^{2g}}
\end{eqnarray}
for all $(x, \xi\;)$ in the compact subset. The estimates referred to in (iii) have a similar form with $\tau^2_{n, N}$ and $E_{j}(x, \xi\;)$ replaced by their mixed derivatives (the same derivatives in each term) and with a possibly different set of constants $K_g$.
The radius of convergence mentioned in property (ii) is related to the critical singularity mentioned earlier.

The next observation is that because of (iii) and because of representations (\ref{RecCoeff}) in terms of the tau-function, the recurrence coefficients themselves inherit a genus expansion from that of (\ref{I.002}):

\begin{eqnarray}   \label{b-shift}
b_{n, N}^2 &\doteq& f(x) = f_0(\xi, x) + \dots + \frac1{n^{2g}}f_{g}(\xi, x) + \dots 
\qquad f_{g}(\xi, x) = x^{1-2g} z_g(\xi)\\
\label{a-shift}
a_{n, N} &\doteq& h(x) = h_0(\xi, x) + \dots + \frac1{n^{g}}h_{g}(\xi, x) + \dots 
\qquad h_{g}(\xi, x) = x^{1/2-g} u_g(\xi)
\end{eqnarray}
where here one needs to bear in mind that $\xi$ depends on $x$ as $\xi = tx^{j/2 - 1}$. The representations (\ref{RecCoeff})
also impart a combinatorial interpretation to the coefficients $f_g(\xi, x)$ and $h_g(\xi, x)$ as generating functions. The continuum limits of (\ref{RecCoeff}) imply that
\begin{eqnarray} \label{ett}
f_g(\xi, x) &=& \frac{\partial^2}{\partial t_1^2} E_g(\xi, x, \xi_1)|_{\xi_1 = 0}\\ \label{tx}
h_g(\xi, x) &=& \frac{\partial^2}{\partial t_1 \partial x} E_g(\xi, x, \xi_1)|_{\xi_1 =0}.
\end{eqnarray}
These parameter derivatives for maps play the same role as they do in graphical enumeration where combinatorially this is referred to as {\it pointing} \cite{FS}. Here this implies that
\begin{eqnarray} \label{fcount}
z^{(m)}_g(0) = \frac{d^m f_g(\xi,x)}{d\xi^m}\Big|_{x =1, \xi =0} &=& ^\sharp\left\{ \mbox{two-legged $g$-maps  with $m$ $j$-valent vertices} \right\}\\ \label{hcount}
u^{(m)}_g(0) = \frac{d^m h_g(\xi,x)}{d\xi^m}\Big|_{x =1, \xi =0} &=& ^\sharp\left\{ \mbox{one-legged $g$-maps  with $m$ $j$-valent vertices and one marked face} \right\}\,\,\,\,\,\,\,\,\,\,\,\,\,\,\,\,\,\,\,\,\,\,\,\,\,\,\,\,\,\,
\end{eqnarray}
(A {\it leg} is an edge emerging from a univalent vertex, so that the leg is the only edge incident to that vertex.  So, for instance, the maps being counted in (\ref{fcount}) have exactly two vertices of valence one; all other vertices have valence $j$.)

\begin{rem}
In the mixed valence case these combinatorial correspondences extend naturally; so, for example,
\beann
\frac{\partial^m f_g(\xi,x)}{\partial\xi_1^{m_1} \cdots \partial\xi_J^{m_J}}\Big|_{x =1, \xi_j =0} &=& ^\sharp\left\{ \mbox{two-legged $g$-maps  with $m_j$ $j$-valent vertices}, j = 1 \dots J\right\}
\eeann
\end{rem}
\medskip

To further analyze the continuum limits, we observe that for $n >>1$, one is far out along the diagonal in the recurrence matrix  which takes the form

\begin{eqnarray}
\left(\begin{array}{ccccc} \ddots & \ddots &  & &\\
                              b^2_{n\left(1 - 1/n\right)} & a_{n\left(1 - 1/n\right)} & 1 & &\\
			          & b^2_n & a_n  & 1 & \\
                                  &      & b^2_{n\left(1 + 1/n\right)} & a_{n\left(1 + 1/n\right)} & 1\\
                                  &    &   &  \ddots & \ddots \end{array} \right)
\end{eqnarray}


To describe local variations along the diagonal that are small in comparison to $n$ one introduces the variable $w = x\left(1 + \ell/n\right) = (n + \ell)/N$. (This will be motivated in more detail in section \ref{sec:44}).

\subsection{Extensions to Odd Valence via Non-Hermitian Orthogonal Polynomials} \label{OddValentExt}
 In \cite{BD10} an extension of the analysis in \cite{EM03} and \cite{BI} to the case of matrix ensembles with a cubic dominant weights was carried out using non-Hermitian orthogonal polynomials \cite{DHK}. Using this key idea, the Toda continuum limits derived in \cite{EMP08} were extended, in \cite{EP11}, to the case of {\it odd} dominant weights thereby filling in the missing odd times in the previous analysis. That paper also derived the continuum limits of the difference string equations for these orthogonal polynomials. 

For valence $j$ odd, we again consider a simple potential function of the form
\begin{align} 
V(\lambda;t) =& \frac{1}{2} \lambda^2 + t \lambda^j . \label{VDef1}
\end{align}
The parameter $j$ can be regarded as fixed; this potential will correspond to the enumeration of $j$-valent maps.
 We consider a partition function defined by 
\begin{align}
Z_{n,N}(t)=&  \int_{\Gamma^n}\exp\left[-N\sum_{i=1}^n  V(\lambda_i ;t)\right]\prod_{1\leq i<j\leq n} (\lambda_i-\lambda_j)^2 \, d^n \lambda . \label{ZDef1}
\end{align}
Here the contour $\Gamma$ is the union of a ray from $e^{i\theta_j}\infty$ to zero and a ray from zero to positive infinity, where
$\theta_j$ is chosen so that the integral converges for $t\geq 0$:
\begin{align*}
\theta_j =\begin{cases} 19\pi/24 & \text{ if }j=3 \\ (j-1)\pi/j & \text{ if }j=5,7,9,\ldots \end{cases}
\end{align*}
The case $j=3$ is handled separately because $\theta_j$ must be greater than $3\pi/4$ for the integral to converge when $t=0$.  
The partition function defined in (\ref{ZDef1}) can be thought of as the normalization constant for an induced measure on eigenvalues arising from a unitarily invariant matrix ensemble of matrices with eigenvalues on the contour $\Gamma$ \cite{BD10}, \cite{W}.  However, our discussion in this article will not rely on this interpretation.

The existence of an expansion of the form
\begin{align}
\log \frac{Z_{n,N}(t)}{Z_{n,N}(0)} \sim & \sum_{g\geq 0} e_g(x,t)N^{2-2g}, \label{TopExp1}
\end{align}
where the asymptotic expansion is taken as $n,N \rightarrow \infty$ with the ratio $n/N=x$ held fixed, has been proven in the case $j=3$ by Bleher and Dea\~{n}o \cite{BD10}.  The arguments in \cite{BD10} extend to handle arbitrary odd values of $j$.  Precise statements giving analytical properties of the topological expansion can be found in \cite{EP11}.  In particular, expansions of derivatives of $\log Z_{n,N}(t)/Z_{n,N}(0)$ with respect to $t$ can be computed by termwise differentiation of the series (\ref{TopExp1}).  The same is true for differentiation with respect to $x$, but for more complicated reasons; see \cite{W} for a proof that this property is equivalent to the string equations.

As with the case of even valence, for the purpose of computing the generating functions $e_g$, it is difficult to work with the partition function directly.  Szeg\"o's relation (\ref{SzegoRel1}) and the Hirota-Szeg\"o relation (\ref{SzegoRel2}) \cite{EMP08} give two fundamentally different relations between the partition function and recurrence coefficients for a family of orthogonal polynomials:
\begin{align}
\frac{Z_{n+1,N} Z_{n-1, N}}{Z_{n,N}^2} =& \frac{n+1}{n} b^{2}_{n,N}. \label{SzegoRel1} \\
b^{2}_{n,N} \sim & \frac{1}{N^2}\frac{\partial^2}{\partial t_{1}^2} \log \frac{Z_{n,N}(t)}{Z_{n,N}(0)}   \label{SzegoRel2}
\end{align}
To make sense of (\ref{SzegoRel2}), one must consider a potential $V(\lambda)=t_1 \lambda + \frac{1}{2} \lambda^2 +t_j \lambda^j$.  After the derivatives with respect to $t_1$ are resolved, one may then set $t_1=0$.  The coefficients $b^{2}_{n,N}$ are defined by the three term recurrence
\begin{align}
\lambda p_{n,N}(\lambda) =& p_{n+1,N}(\lambda) + a_{n,N} p_{n,N}(\lambda) + b^{2}_{n,N} p_{n-1,N}(\lambda), \label{ThreeTermRec1}
\end{align}
where the monic orthogonal polynomials $p_{n,N}$ arise from the following (non-hermitian) inner product and orthogonality relation:
\begin{align}
\ip{F(\lambda),G(\lambda)}=& \int_{\Gamma} F(\lambda) G(\lambda)\exp\left(-N V(\lambda;t)\right)\, d\lambda . \label{OrthRel1} \\
\ip{p_{n,N}(\lambda),\lambda^m}=&0 &\text{ for all }m<n \label{OrthRel2}
\end{align}
Let us note that because the inner product is non-hermitian, denominators may vanish in the formula
\begin{align*}
p_{n,N}(\lambda)=& \lambda^n - \frac{\ip{\lambda^n, p_{n-1,N}}}{\ip{p_{n-1,N},p_{n-1,N}}} p_{n-1,N}(\lambda)-\ldots -\frac{\ip{\lambda^n, p_{0,N}}}{\ip{p_{0,N},p_{0,N}}} p_{0,N}(\lambda).
\end{align*}
Thus it is not clear that $p_{n,N}(\lambda)$ exists.  However, it has been shown \cite{BD10} in the case $j=3$ that for sufficiently large $n$, and $t$ in an open set with $0$ on its boundary, that there exists a polynomial $p_{n,N}$ satisfying (\ref{OrthRel2}).  Although we have not explicitly carried out the similar analysis for higher odd valences $j$, it is clear that this will extend (see \cite{EP11}).

\section{Fundamental Equations in the Continuum Limit and Symbol Calculus} \label{sec:44} 

In the first four subsections of this section we will be restricting attention to $j$-regular potentials, i.e., those of the form (\ref{VDef1}) and in particular for the 
case of $j$ odd. However, starting in subsection \ref{MV} we extend our attention to more general potentials of mixed valence which build on what we find for the regular case.

\subsection{Fundamental Equations} \label{FundEqns}

\subsubsection{String equations}
In the case that orthogonal polynomials $p_n(\lambda)$ exist for all $n$, the three term recurrence (\ref{ThreeTermRec1}) can be expressed in terms of a Jacobi matrix $L$:
\begin{align} \label{jacobiMatrix}
L=\left( \begin{array}{ccccc}
a_0 & 1 & 0 & 0 &  \\
b^{2}_1 & a_1 & 1 & 0 &  \\
0 & b^{2}_2 & a_2 & 1 & \\
 & & \ddots & \ddots & \ddots
 \end{array}\right)
 \end{align}
Let $p= ( p_n(\lambda))_{n\geq 0}$ be the column vector of all (potentially non-hermitian) orthogonal polynomials for the potential $V$; then the three term recurrence can be encoded as a matrix equation
\begin{align}
\lambda p =& L p. \label{matrixRecurrence}
\end{align}
If orthogonal polynomials $p_n$ only exist for sufficiently large $n\geq n^{\ast}$, then the first $n^{\ast}+2$ rows of the matrix equation (\ref{matrixRecurrence}) are invalid; however the remaining rows still hold, and this is sufficient for our purposes.

Higher order asymptotics of the recurrence coefficient, $a_{n}$ and $b^{2}_n$, in (\ref{jacobiMatrix}) can be computed using the string equations (\ref{string1}).  Since the partition function is related to the recurrence coefficients by (\ref{SzegoRel1}), this facilitates the calculation of the generating functions $e_g$.
An undifferenced form of these equations was used by Bleher and Dea\~{n}o in \cite{BD10}:
\begin{align}		
0 =&  V'(L)_{n,n}	\label{unDiff1}		\\
\frac{n}{N} =&  V'(L)_{n,n-1}	\label{unDiff2}		
\end{align}  
The derivation of these equations is completely analogous to that of (\ref{UnDiff1}) and (\ref{UnDiff2}).

\subsubsection{Motzkin Paths}
We define a {\it  Motzkin path} of length $r$ to be a function $p$ from the discrete set $0,1,2,\ldots ,r$ to $\mathbb{Z}$ such that $p(i+1)-p(i) \in \{-1,0,1\}$.  The Motzkin path may also be thought of as the piecewise linear graph  determined by the step height increments of $p$.  We use the notation $\mathcal{P}^{r}(m,n)$ for the set of Motzkin paths of length $r$ which begin at height $m$ (that is $p(0)=m$), and end at height $n$.  It is typically the case that Motzkin paths are required to never cross the $x$-axis; i.e., they are what one refers to in probability theory as {\it excursions}. However in our work we must allow the more general setting where such zero-crossings are permitted. Motzkin paths allow a combinatorial interpretation of the entries of powers of a tridiagonal matrix:
\begin{align} \label{Motzkincount}
L^{r}_{n,m} =& \sum_{p\in \mathcal{P}^{r}(n,m)} C(p) \\
C(p)=& \prod_{i=0}^{r-1} \begin{cases} 
1 & \text{ if }p(i+1)=p(i)+1 \\
a_{p(i)} & \text{ if }p(i+1)=p(i) \\
b^{2}_{p(i)} & \text{ if }p(i+1)=p(i)-1
\end{cases} .
\end{align}
This may be more compactly expressed in terms of the algebraic resolvent operator for $L$:
\beann
L^{r}_{n,m} &=& [z^r] \left[ (I - z L)^{-1}\right]_{n,m}
\eeann
We call $C(p)$ the contribution of the Motzkin path $p$.  The string equations can be expressed in Motzkin path form:
\begin{eqnarray}\label{stringpath1}
0 &=& a_n + j t \sum_{p\in \mathcal{P}^{j-1}(n,n)} C(p) \\  \nonumber
&=& a_n + j t [z^{j-1}] \left[ (I - z L)^{-1}\right]_{n,n}\\
\label{stringpath2}
\frac{n}{N} &=& b^{2}_n + j t \sum_{p\in \mathcal{P}^{j-1}(n,n-1)} C(p)\\
&=& b^{2}_n + j t [z^{j -1}] \left[ (I - z L)^{-1}\right]_{n,n-1}.
\end{eqnarray}

This description has some relation to combinatorial path processes that are comprised of excursions. The general description of this and connections to continued fractions (which are also related to orthogonal polynomial recursions) was already noticed by Flajolet in his seminal paper \cite{Flaj}. More recently Bouttier and Guitter \cite{BG} applied this to the study of geodesic distance in planar maps (related to random surfaces) in terms of a combinatorial objects called mobiles which can be encoded by restricted Motzkin paths that, in particular, are excursions. However, as we have mentioned, in the work here the Motzkin paths must be allowed to have zero-crossings, as will become evident in the subsequent sections, and in this generality Flajolet's method breaks down. Nevertheless there are some suggestive analogies between our work and that of \cite{BG}. Their {\it characteristic equation} ( linear recurrence equation) appears to play a related role to that of the nonlinear PDE, (\ref{TODA}) which underlies the characteristic geometry (Section \ref{chargeom}) in our work.

\subsubsection{Toda equations} \label{sec:Toda}
 Equations of motion for the recurrrence coefficients are obtained by extracting matrix entries from the diagonal or first subdiagonal in (\ref{Toda1}).  For example, the $t_1$-Toda equations are
\begin{align}
N^{-1}\partial_{t_1} a_n =& b^{2}_{n} -b^{2}_{n+1} \label{t1TodaA} \\
N^{-1} \partial_{t_1} b^{2}_n =& b^{2}_{n}\left(a_{n-1} -a_{n}\right) \label{t1TodaB}
\end{align}
Let us briefly sketch a derivation of the Toda equations (\ref{Toda1}).
Fix a value of $j$; we will use an overdot to represent differentiation with respect to $t_j$.  If $m>n$ then
\begin{align}
0=&\partial_{t_j}\int_{\Gamma} p_m(\lambda) p_n(\lambda) \dmu \nonumber \\
=&\int_{\Gamma} \dot{p}_m(\lambda) p_n(\lambda) \dmu-\int_{\Gamma} (N\lambda^j)p_m(\lambda) p_n(\lambda) \dmu.
\end{align}
If $m\leq n$ the polynomial $\dot{p}_m$ is of strictly smaller degree than that of $p_n$ and so their inner product vanishes. Hence,
\begin{align}
\int_{\Gamma} \dot{p}_m(\lambda) p_n(\lambda) \dmu  
=& \begin{cases} 0 & \text{ if $m\leq n$}\\N\int_{\Gamma} \lambda^j p_m(\lambda)  p_n(\lambda) \dmu  & \text{ if $m>n$}\end{cases} \nonumber \\
=& N(L^{j}_{-})_{m,n} \left\|p_n\right\|^2. 
\end{align}
The Toda equation (\ref{Toda1}) may now be obtained by calculating $\int \lambda \dot{p}_m(\lambda) p_n(\lambda)\dmu$ in two ways: by applying the three term recurrence either to $\lambda p_n(\lambda)$ or to $\partial_t(\lambda p_m(\lambda))$.

 We use the identity $\left\|p_{n}\right\|^2=b^{2}_n\left\|p_{n-1}\right\|^2$.  The first evaluation is
\begin{align}
\int\lambda \dot{p}_m p_n \,d\mu =& \int \dot{p}_m \left(p_{n+1}+a_np_n+b^{2}_np_{n-1}\right)\,d\mu \nonumber \\
=&N(L^{j}_{-})_{m,n+1}\left\|p_{n+1}\right\|^2+Na_n(L^{j}_{-})_{m,n}\left\|p_{n}\right\|^2+Nb^{2}_n(L^{j}_{-})_{m,n-1}\left\|p_{n-1}\right\|^2 \nonumber \\
=& N\left\|p_n\right\|^2\left((L^{j}_{-})_{m,n+1}b^{2}_{n+1}+(L^{j}_{-})_{m,n}a_n+(L^{j}_{-})_{m,n-1}\right) \nonumber \\
=& N \left\|p_n\right\|^2 \left(L^{j}_{-}L\right)_{m,n}. \label{Toda pf 1}
\end{align}
The second way is
\begin{align}
\int\lambda \dot{p}_m p_n \,d\mu =& \int \partial_{t_j}\left(\lambda p_m\right) p_n \,d\mu \nonumber \\
=& \int \partial_{t_j}\left(p_{m+1}+a_mp_m+b^{2}_mp_{m-1}\right) p_n \,d\mu \nonumber \\
=&\dot{a}_m\left\|p_n\right\|^2\delta_{m,n}+\dot{b}^{2}_m \left\|p_{n}\right\|^2\delta_{m-1,n} +\int \left(\dot{p}_{m+1}+a_m\dot{p}_m+b^{2}_m\dot{p}_{m-1} \right)p_n\,d\mu_{\bf t} \nonumber \\
=&\dot{a}_m\left\|p_n\right\|^2\delta_{m,n}+\dot{b}^{2}_m \left\|p_{n}\right\|^2\delta_{m-1,n} \nonumber \\
& +N\left\|p_n\right\|^2\left((L^{j}_{-})_{m+1,n} + a_m (L^{j}_{-})_{m,n} +b^{2}_m (L^{j}_{-})_{m-1,n}\right)\nonumber \\
=&\left\|p_{n}\right\|^2\dot{L}_{m,n} +N\left\|p_n\right\|^2\left(L L^{j}_{-}\right)_{m,n}.\label{Toda pf 2}
\end{align}
Combining (\ref{Toda pf 1}) and (\ref{Toda pf 2}) proves (\ref{Toda1}).
\bigskip
 
In general the Toda equations can be written in Motzkin path form as
\begin{align}
N^{-1} \partial_{t_j} a_n =& \sum_{p\in \mathcal{P}^{j}(n,n-1)} C(p) - \sum_{p\in \mathcal{P}^{j}(n+1,n)} C(p) \label{MotzkinTodaA} \\
=&   [z^{j}] \left[ (I - z L)^{-1}\right]_{n,n-1} - [z^{j}] \left[ (I - z L)^{-1}\right]_{n+1,n}\nonumber \\
N^{-1} \partial_{t_j} b^{2}_n =&  \sum_{p\in \mathcal{P}^{j}(n,n-2)} C(p) -  \sum_{p\in \mathcal{P}^{j}(n+1,n-1)} C(p) +(a_{n-1}-a_n) \sum_{p\in \mathcal{P}^{j}(n,n-1)} C(p) \label{MotzkinTodaB}\\ 
=& [z^{j}] \left[ (I - z L)^{-1}\right]_{n,n-2} 
- [z^{j}] \left[ (I - z L)^{-1}\right]_{n+1,n-1} 
+ (a_{n-1} - a_n) [z^{j}] \left[ (I - z L)^{-1}\right]_{n,n-1}. \nonumber
\end{align}

\subsection{Continuum Limits}
The continuum limits of the string and Toda equations will be described in terms of certain scalings of the independent variables, both discrete and continuous. The positive parameter $\frac1N$ sets the scale for the potential in the random matrix partition function and is taken to be small. The discrete variable $n$ labels the lattice {\it position} on $\mathbb{Z}^{\geq 0}$ that marks, for instance, the $n^{th}$ orthogonal polynomial and recurrence coefficients.  We also always take $n$ to be large and in fact to be of the same order as $N$; i.e., as $n$ and $N$ tend together to $\infty$ , they do so in such a way that their ratio, the t'Hooft parameter introduced in section \ref{background},
\begin{equation}\label{xdef} 
x \doteq  \frac{n}{N}
\end{equation}
remains fixed at a value close to $1$. 

In addition to the {\it global} or {\it absolute} lattice variable $n$, we also introduce a {\it local} or {\it relative} lattice variable denoted by $\ell$. It varies over integers but will always be taken to be small in comparison to $n$ and independent of $n$. The Motzkin lattice paths naturally  introduce the composite discrete variable $n + \ell$ into the formulation of the difference string and Toda equations which we think of as small discrete variations around a large value of  $n$. The spatial homogeneity of those equations manifests itself in their all having the same form, independent of what $n$ is, while $\ell$ in those equations has a bounded 
\emph{bandwidth}. For instance, when $j = 2\nu +1$, this varies over $\{-\nu - 1, \dots, -1,0,1, \dots, \nu + 1\}$. In that case taking $\nu + 1 << n$ will insure the necessary separation of scales between $\ell$ and $n$. We define 
\begin{eqnarray}
\label{wdef}{w} &\doteq& (n+\ell)/N\\
&=& x +   \frac{\ell}{N} = x\left( 1 + \frac{\ell}{n}\right).
\end{eqnarray} 
as a {\em spatial} variation close to $x$ which will serve as a continuous analogue of the lattice location along a Motzkin path relative to the starting location of the path.  

We also recall $\xi_j$ and introduce $\tilde{w}$: 
\begin{eqnarray}
\label{nuscaling2} \xi_{j} &\doteq& x^{j/2 - 1} {t_j}\\
\label{nuscaling2.5} \widetilde{w} &\doteq& \left( 1 + \frac \ell n\right)
\end{eqnarray}
In terms of these scalings, 
the large $n$ expansions of the recursion coefficients may be rewritten, \cite{EMP08} and \cite{EP11}, as 
\begin{eqnarray}\label{bs-asymp}  
b_{n, N}^2 &=& x\left( z_0(\xi_j) + \dots + \frac1{n^{2g}}z_{g}(\xi_j) + \dots\right)\\
z_g(\xi_j) &=& \frac{d^2}{d\xi_1^2} e_g(\xi_1, \xi_j)|_{\xi_1 = 0}\\ \label{as-asymp}
a_{n, N} &=&  x^{1/2}\left( u_0(\xi_j) + \dots + \frac1{n^{g}}u_{g}(\xi_j) + \dots\right)\\ 
\label{us-asymp}
u_g(\xi_j) &=& \sum_{\begin{matrix} 2 g_1 + m = g+1 \\  g_1 \geq 0\,, m>0\end{matrix}}  - \frac1{m!} \frac{\partial^{m+1}}{\partial \xi_1\partial x^m} \bigg[ 
x^{2-2g_1} e_{g_1}\left(\xi_1, \xi_j\right)\bigg]_{\xi_1 = 0}.
\end{eqnarray}
The expansion (\ref{bs-asymp}) is derived by setting $F(\lambda)= \lambda$ in Theorem \ref{Workhorse}, differentiating the resulting expansion term by term with respect to $t_1$ and then setting $t_1 = 0$.
We further define 
\begin{eqnarray} \label{eq:b-shift}
b_{n+\ell, g_s}^2 &=& f^{[\ell]}_0(\xi_j, w) + \dots + \frac1{n^{2g}}f^{[\ell]}_{g}(\xi_j, w) + \dots \\
\label{b-shift_g} f^{[\ell]}_{g}(\xi_j, w) &=& w^{1-2g} z_g(\xi_j\widetilde{w}^{j/2-1})\\
 \label{eq:a-shift}
a_{n+\ell, g_s} &=& h^{[\ell]}_0(\xi_j, w) + \dots + \frac1{n^{g}}h^{[\ell]}_{g}(\xi_j, w) + \dots \\
\label{a-shift_g} h^{[\ell]}_{g}(\xi_j, w) &=& w^{1/2-g} u_g(\xi_j\widetilde{w}^{j/2-1}). 
\end{eqnarray}
Note that when $\ell=0, f^{[0]}_{g}(\xi_j, w) = x^{1-2g} z_g(\xi_j)$ and 
$h^{[0]}_{g}(\xi_j, w) = x^{1/2-g} u_g(\xi_j)$; so we set
\begin{eqnarray} \label{f-coeffs}
f_g(\xi_j,w) &=& f^{[0]}_{g}(\xi_j, x) = x^{1-2g}z_g(x^{j/2-1} t_j/j)\\ \label{h-coeffs}
h_g(\xi_j,w) &=& h^{[0]}_{g}(\xi_j, x) = x^{1/2-g} u_g(x^{j/2-1}t_{j}/j).
\end{eqnarray}
For future notational convenience we will denote the asymptotic expansions in (\ref{bs-asymp}) and (\ref{as-asymp})  by $f(\xi_j, w)$ and $h(\xi_j, w)$ respectively so that
\begin{eqnarray} \label{f-exps}
f(\xi_j, w) &=& f_0(\xi_j, w) + \dots + \frac1{n^{2g}}f_{g}(\xi_j, w) + \dots\\ \label{h-exps}
h(\xi_j, w) &=& h_0(\xi_j, w) + \dots + \frac1{n^{g}}h_{g}(\xi_j, w) + \dots.
\end{eqnarray}
By our prior results, $f(\xi_j, w)$ and $h(\xi_j, w)$ may be differentiated term by term as uniformly valid asymptotic expansions. Hence one may construct expansions of $f^{[\ell]}$ (resp. $h^{[\ell]}$) in terms of {\it Taylor type} expansions of $f$ (resp. $h$) in $w$ around $x$; or, equivalently, in terms of the vertex operator notation introduced in (\ref{boson1}).
\begin{eqnarray} \label{f1k}
f^{[\ell]}(\xi_j, x) &=& \sum_{m=0}^\infty \frac{f_{w^{(m)}}|_{w=x}}{m!} \left( \frac{\ell}{N} \right)^m = e^{\frac{\ell}{N} \partial_x} f(x)\\
\label{h1k}
h^{[\ell]}(\xi_j, x) &=& \sum_{m=0}^\infty \frac{h_{w^{(m)}}|_{w=x}}{m!} \left( \frac{\ell}{N} \right)^m =  e^{\frac{\ell}{N} \partial_x} h(x)
\end{eqnarray}
where the subscript $w^{(m)}$ denotes the operation of taking the $m^{th}$ term-by-term derivative of the asymptotic expansion of $h(\xi_j, w)$ or $f(\xi_j, w)$ with respect to $w$:
\begin{eqnarray*}
h_{w^{(m)}} &=& \sum_{g \geq 0} \frac{\partial^{m}}{\partial w^m}  h_g(\xi_j,w) \frac{1}{n^g}\\
f_{w^{(m)}}   &=& \sum_{g \geq 0} \frac{\partial^{m}}{\partial w^m}  f_g(\xi_j,w) \frac{1}{n^{2g}}.
\end{eqnarray*}
As valid asymptotic expansions the representations, (\ref{h1k}) or (\ref{f1k}), denote the asymptotic series whose successive terms are gotten by first setting $1/N = x/n$, as given by (\ref{xdef}), and then collecting all terms with a common power of $1/n$. 

In what follows we will sometimes abuse notation, drop the evaluation at $w=x$ and manipulate the series $h^{[\ell]}(\xi_j, w), f^{[\ell]}(\xi_j, w)$.
In doing this these series must now be regarded as formal but whose orders are still defined by collecting all terms in $1/n$ and $1/N$ of a common order. (Recall that $1/N \sim \frac1n$ so that $n^{-\alpha} N^{-\beta} = \mathcal{O}(n^{-(\alpha + \beta)})$).  They will be substituted into the difference string and the Toda equations to derive the respective continuum equations. At any point in this process, if one evaluates these expressions at $w=x$ and $\frac1N = \frac xn$, one  recovers valid asymptotic expansions of the $a_{n+k, N}$ and $b^2_{n+k, N}$. 

\subsubsection{Continuum Limit of the Toda Equations} 

We recall from Section \ref{sec:Toda} that the {\it Toda equations} have the form:
\begin{eqnarray} \label{an}
-\frac{1}{N}  \frac{d a_n}{dt_{j}} &=& \left( {L}^{j}\right)_{n+1, n} -  \left( {L}^{j}\right)_{n, n-1}\\ \nonumber
&=& [z^{j}] \left\{ \left[ (I - z L)^{-1}\right]_{n+1,n} - \left[ (I - z L)^{-1}\right]_{n,n-1} \right\} \\ \label{bn} 
-\frac{1}{N}   \frac{d b^2_n}{dt_{j}} &=&  (a_n - a_{n-1}) \left({L}^{j}\right) _{n, n-1}
+ \left({L}^{j}\right)_{n+1, n-1} - \left( {L}^{j}\right)_{n, n-2}\\ \nonumber
&=& [z^{j}] \left\{ (a_{n} - a_{n-1})  \left[ (I - z L)^{-1}\right]_{n,n-1}    
+  \left[ (I - z L)^{-1}\right]_{n+1,n-1}  -    \left[ (I - z L)^{-1}\right]_{n,n-2}  
 \right\}.
\end{eqnarray}
Unlike the equation for the time derivative of the $a_n$, the Toda equation for the $b_n^2$ does not have a RHS which is an exact difference. 
However, one may replace the latter by an equation for the combinations $b_n^2 + 1/2 a_n^2$ which yields an equivalent system that is an exact difference,
\begin{eqnarray} \label{cn}
-\frac{1}{N}   \frac{d (b^2_n + \frac12 a_n^2)}{dt_{j}} &=& 
a_n \left( {L}^{j}\right)_{n+1, n} -  a_{n-1} \left( {L}^{j}\right)_{n, n-1}
+ \left( {L}^{j}\right)_{n+1, n-1} -  \left( {L}^{j}\right)_{n, n-2} \,\,\,\,\,\,\, \\
\nonumber
&=& 
a_n \left[ (I - z L)^{-1}\right]_{n+1,n} -  a_{n-1} \left[ (I - z L)^{-1}\right]_{n,n-1}
+ \left[ (I - z L)^{-1}\right]_{n+1,n-1} -  \left[ (I - z L)^{-1}\right]_{n,n-2}.
\end{eqnarray}
\medskip

The continuum limit equations are realized by replacing the recursion operator $L$ in the above equations by its asymptotic symbol expansion $\mathcal{L}$, (\ref{symbol1}).
Equations (\ref{an}) and (\ref{cn}), for the case of $j = 2\nu+1$ , then transform to

\begin{eqnarray} \label{toda-anew}
-\frac{1}{N}  \frac{d a_n}{dt_{2\nu+1}} &=&
\sum_{P \in \mathcal{P}^{2\nu+1}(1,0)}\left(\prod_{p_a=1}^{2\mu(P)} a_{n+\ell_{p_a}(P)+1}\prod_{p_b=1}^{\nu-\mu(P)+1} b^2_{n+\ell_{p_b}(P)+1} \right.\\
\nonumber  &-& \left. \prod_{p_a=1}^{2\mu(P)} a_{n+\ell_{p_a}(P)}\prod_{p_b=1}^{\nu-\mu(P)+1} b^2_{n+\ell_{p_b}(P)}\right)\\ \label{toda-bnew}
-\frac{1}{N}  \frac{d (b^2_n + \frac12 a_n^2)}{dt_{2\nu+1}} &=& a_n \sum_{P \in \mathcal{P}^{2\nu+1}(1,0)}  \left(\prod_{p_a=1}^{2\mu(P)} a_{n+\ell_{p_a}(P)+1}\prod_{p_b=1}^{\nu-\mu(P)+1} b^2_{n+\ell_{p_b}(P)+1}\right)\\
 \nonumber &-& a_{n-1} \sum_{P \in \mathcal{P}^{2\nu+1}(1,0)}\left(\prod_{p_a=1}^{2\mu(P)} a_{n+\ell_{p_a}(P)}\prod_{p_b=1}^{\nu-\mu(P)+1} b^2_{n+\ell_{p_b}(P)}\right)\\
\nonumber &+&\sum_{P \in \mathcal{P}^{2\nu+1}(2,0)}\left(\prod_{p_a=1}^{2\mu(P)+1} a_{n+\ell_{p_a}(P)+1}\prod_{p_b=1}^{\nu-\mu(P) + 1} b^2_{n+\ell_{p_b}(P)+1} \right. \\
\nonumber &-& \left. \prod_{p_a=1}^{2\mu(P)+1} a_{n+\ell_{p_a}(P)}\prod_{p_b=1}^{\nu-\mu(P) + 1} b^2_{n+\ell_{p_b}(P)}\right),
\end{eqnarray}
in  which $b_n^2$ and $a_n$ should be replaced, respectively, by the expansions (\ref{bs-asymp}) and (\ref{as-asymp}) and their differenced extensions by (\ref{eq:b-shift}) and (\ref{eq:a-shift}).
\medskip

Alternatively, this conservative form of the equations (for general $j$) may be conveniently represented in terms of the asymptotic recurrence operator, 
$\mathcal{L} = \sqrt{f(x)} \eta + h(x) + \sqrt{f(x)} \eta^{-1}$ discussed at the end of section \ref{sec:diffsymb}:
\begin{eqnarray} \label{CptToda1}
-\partial_{t_j} h &=& \partial_x \left[\eta^{-1}\right] \sqrt{f} \left[z^j \right]   (I - z \mathcal{L})^{-1} \\ \label{CptToda2}
-\partial_{t_j} (f + \frac12 h^2) &=& \partial_x \left\{ \left[\eta^{-2}\right] f \left[z^j \right]   +   \left[\eta^{-1}\right] h \sqrt{f} \left[z^j \right]  \right\} (I - z \mathcal{L})^{-1}. 
\end{eqnarray}
This representation essentially follows directly from the previous one once one makes the following observations.  The exact difference form of the RHS's of the previous equations leads to cancellation of an overall factor of $1/N$ on both sides of those equations with the difference replaced by $\partial_x$. The symbol 
$[\eta^{-m}]$ picks off the coefficient of $\eta^{-m}$ in the subsequent expression. The particular indices, $m = -1, -2$ appearing here correlate to the weighted path spaces we sum over in (\ref{toda-anew}) - (\ref{toda-bnew}), $\mathcal{P}^{j}(1,0)$ and $\mathcal{P}^{j}(2,0)$, respectively. This kind of calculus is familiar from combinatorial graph theory or Markov chains on graphs where path counts can be calculated in terms of the resolvent of the adjacency matrix of the graph.
The factors of $\sqrt{f}$ appearing in the above equation which comes from the weighting: the number factors is the difference between the initial and terminal
points of the path. We also note that, as was discussed in section \ref{sec:diffsymb}, the asymptotic operator $\mathcal{L}$ being used here corresponds to the symmetric recurrence operator for  orthonormal polynomials as opposed to the Hessenberg form (associated to monic polynomials) which we had at the start of this section. However, since these two operators are conjugate to one another by a diagonal matrix, this distinction does not enter into the continuum limit of the Toda  equations.

\subsubsection{Continuum Limit of the String Equations} \label{motzstring}

Based on substituting $\mathcal{L}$ into the path formulation of the string equations given in (\ref{stringpath1}) and (\ref{stringpath2}), and simnilar to the discussion at the end of the previous section, the continuum limit of these equations may be expressed in a similarly compact form. Since these equations will
be used in this paper for detailed higher order calculations we will present them here in more detailed form than was done for the Toda equautions:

\begin{eqnarray} \label{ContString}
\left( \begin{array}{c} 0 \\ x \end{array} \right) &=&  \left( \begin{array}{c} h\\ f \end{array}  \right) + j t [z^{j-1}] \left( \begin{array}{c} [\eta^0] \cr [\eta^{-1}] \sqrt{f}\end{array} \right) \left[ (1 - z \mathcal{L})^{-1}\right]\\ \nonumber
&=& \left( \begin{array}{c} h\\ f \end{array}  \right) + t \sum_{(\alpha, \beta) } n^{-(|\alpha| + |\beta|)} \left( \begin{array}{c}\tilde{P}^{(a)}_{\alpha, \beta, j} (h, f)  \partial_x^\alpha h \partial_x^\beta f  \cr \tilde{P}^{(b)}_{\alpha, \beta, j} (h, f) \partial_x^\alpha h \partial_x^\beta f\end{array} \right)
\end{eqnarray}
where
\begin{eqnarray} \nonumber
\partial_x^\alpha F(x) &=& \prod_{i=1}^{\ell(\alpha)} \partial_x^{\alpha_i} F(x) \\ \nonumber
\partial_x^\beta F(x) &=& \prod_{i=1}^{\ell(\beta)} \partial_x^{\beta_i} F(x) \\ \nonumber
\left( \begin{array}{c}\tilde{P}^{(a)}_{\alpha, \beta, j} (h, f) \\ \tilde{P}^{(b)}_{\alpha, \beta, j} (h, f) \end{array} \right) 
&=&    \left( \begin{array}{c}   [\eta^{\ell(\beta)} z^{j-1-\ell(\alpha) - \ell(\beta)}] \frac{1}{f^{\ell(\beta) /2}} \cr [\eta^{\ell(\beta) - 1} z^{j-1-\ell(\alpha) - \ell(\beta)}] \frac{1}{f^{((\ell(\beta)-1)/2)}}  \end{array} \right)   \prod_{q=1}^{\ell(\alpha) + \ell(\beta)} \frac1{m_q !} \,\,\,\,\, \\ \label{StringCoeff}
&& \left\{\left[\sum_{\gamma=0}^{q-1}  \left( \eta_\gamma \partial_{\eta_\gamma}  -   \sigma_\gamma    \right)  \right]^{m_q} \,\,\,\,\,\,\,\, 
\prod_{\gamma = 0}^{\ell(\alpha) + \ell(\beta)} \left(1 - z \mathcal{L}_\gamma \right)^{-1}
\right\}_{\eta_\gamma = \eta} 
\end{eqnarray}
and
\beann
\mathcal{L}_\gamma &=& \sqrt{f^{(\gamma)}} \eta_\gamma + h^{(\gamma)} + \sqrt{f^{(\gamma)}} \eta_\gamma^{-1}.
\eeann
The values of $m: \{ 1,2, \dots, \ell(\alpha) + \ell(\beta)\} \rightarrow \mathbb{Z} $
are in 1-1 correspondence with the parts of $\alpha$ and $\beta$ and 
$\sigma: \{ 1,2, \dots, \ell(\alpha) + \ell(\beta)\} \rightarrow \{ 0,1\}$ is such that if $\sigma(i) = 0$, then $m(i)$ is a part of $\alpha$ and if $\sigma(i) = 1$, then $m(i)$ is a part of $\beta$.
Also, when $\alpha = \emptyset = \beta$ the first product in the above formula is set to 1. The derivation of (\ref{StringCoeff}) is explained in \cite{W} extending a result for the even regular valence case established in \cite{EMP08}. 

\subsubsection{String Polynomials}  \label{sec:StringPoly}
We observe that $\mathcal{L}$ has the form of a Lie algebra operator in which $\eta$ plays the role of a raising operator and $\eta^{-1}$ plays the role of a lowering operator, thus preserving the "tridiagonal" structure from which this asymptotic formula is derived. This leads us to introduce continuum analogues of the resolvent expressions for the Motzkin path enumerations given by (\ref{Motzkincount}). For general potentials $V_{\bf t}(\lambda)$ of the form (\ref{eq:genpot}) we define
\begin{eqnarray} \label{PHIm}
{\phi}_{V,m} &=& [\eta^{0}] V_{\bf t}^{(m+1)}(\mathcal{L}_0) \\ \label{PSIm}
{\psi}_{V,m} &=& [\eta^{-1}] \sqrt{f_0}V_{\bf t}^{(m+1)}(\mathcal{L}_0)
\end{eqnarray}
where  $\mathcal{L}_0 = \sqrt{f_0} \eta + h_0 + \sqrt{f_0} \eta^{-1}$ and $V_{\bf t}^{(k)}$ denotes the $k^{th}$ partial derivative of $V_{\bf t}$ with respect to $\lambda$. (Note that these expressions differ slightly from those used in the statement of Theorem \ref{thm:1.1} since the potential used there is a slight modification of the potential (\ref{eq:genpot}) we are using now.)  It should be clear that such expressions will arise naturally in the analysis of the String and Toda equations. However it will be more practical to work with simplified analogues that are more directly related to the continuum limit of the resolvent of $L$. We will refer to these analogues as {\it string polynomials} and define them, relative to a fixed valence $j$ (which replaces the role of the fixed potential $V_{\bf t}$), by
\begin{eqnarray} \nonumber
\phi_m &=& (j)_{m+1} [z^{j-m-1}] [\eta^{0}] \left( 1 - z \mathcal{L}_0 \right)^{-1}\\ \label{phim}
&=& (j)_{m+1} [\eta^{0}] \left(\sqrt{f_0} \eta + h_0 + \sqrt{f_0} \eta^{-1}\right)^{j-m-1}\\ \nonumber
\psi_m &=& (j)_{m+1} [z^{j-m-1}] [\eta^{-1}] \sqrt{f_0}\left( 1 - z \mathcal{L}_0\right)^{-1}\\
\label{psim}
&=& (j)_{m+1} [\eta^{-1}]  \sqrt{f_0}\left(\sqrt{f_0} \eta + h_0 + \sqrt{f_0} \eta^{-1}\right)^{j-m-1}
\end{eqnarray}
where $(j)_k = j (j-1) \cdots (j-k+1)$. Note that these polynomials do not depend on any ${\bf t}$ parameter. We suppress indicating the dependence of $\phi$
and $\psi$ on $j$ since this will generally be clear from context.
Clearly these functions are polynomials in $h_0$ and $\sqrt{f_0}$ but as we shall see in section \ref{Hopf} they have a great deal more structure.

\subsection{Continuum Equations at Leading Order} \label{sec:5}
We are now in a position to describe in more detail the leading order behavior of the continuum limit equations we derived in the previous subsections.  We note that the leading order Toda equations (\ref{TODA}) have the structure of a hyperbolic system. They in fact are equivalent to a system of conservation laws with flux vectors given in terms of polynomial expressions in $h_0$ and $f_0$.  In section \ref{StringApple} we will see how these polynomials are directly related to the
Bessel-Appell polynomials. 

\subsubsection{Leading Order Conservation Laws for general potentials} \label{sec:LOCL}
The conservation stated in (\ref{SYSTEM}) follows directly from the continuum Toda equations, (\ref{CptToda1} - \ref{CptToda2}) by replacing $\mathcal{L}$ in those equations with $\mathcal{L}_0$ and making use of the string polynomials described in the previous section. Because of that we will restate those equations here as a corollary of the continuum equations and give a brief description of their derivation.

\begin{cor} \label{cor02}
The $j^{th}$ continuum Toda lattice equation at leading order may be expressed as the conservation law:
\begin{eqnarray} \label{conservation}
\frac{\partial}{\partial t_j} 
\left(
\begin{array}{c}
h_0\\ f_0 + \frac12 h_0^2
\end{array}
\right) &+& 
\frac{\partial}{\partial x} 
\left(
\begin{array}{c}
\mathcal{F}_{1}\\ 
\mathcal{F}_{2}  + h_0 \mathcal{F}_{1}\
\end{array}
\right)  = 0.
\end{eqnarray}

\begin{eqnarray}  \nonumber
\mathcal{F}_{1}(j) &=&  \psi^{(j)}_{-1}\\ \label{flux1}
&=& \sum_{m = 0}^{\lfloor \frac{j}2 -1\rfloor } {j \choose 2m+1} {2m+1 \choose m+1} h_0^{j -2m -1} f_0^{m + 1}\\ \nonumber
\mathcal{F}_{2}(j) &=&  f_0 \phi^{(j)}_{-1} - \frac1{j+1}\psi^{(j+1)}_{-1} \\ \label{flux2}
&=& \sum_{m = 0}^{\lfloor \frac{j}2 \rfloor }{j \choose 2m} {2m \choose m+1}  h_0^{j -2m} f_0^{m + 1}\, ,
\end{eqnarray}
where the argument $j$ (usually suppressed) denotes that these are the fluxes for the conservation law corresponding to the $t_{j}$ flow; i.e., corresponding to deforming the potential through the coefficient of $\lambda^{j}$.
\end{cor}
\begin{rem} \label{remcor02}
We remark that these equations do not show any explicit dependence on the potential which, in general, will be of mixed valence. That dependence comes through the densities $h_0$ and $f_0$ which are functions of the ${\bf t}$-parameters in the potential (\ref{eq:genpot}). However, the $j^{th}$ conservation law only describes the effect on those densities of varying the coefficient, $t_j$, of $\lambda^j$ in the potential. In other words these conservation laws are autonomous in the potential $V$. So from the PDE perspective the background potential specifies initial conditions which in general may be posed on some initial strip.
\end{rem}
\begin{proof}
Consider the flux for (\ref{CptToda1}) with $\mathcal{L}$ evaluated at $\mathcal{L}_0$: 
$$
 \left[\eta^{-1}\right] \sqrt{f_0} \left[z^j \right]   (I - z \mathcal{L}_0)^{-1} = \psi^{(j)}_{-1}
$$
by definition from (\ref{psim}) with the explicit polynomial expression (\ref{flux1}) coming from applying the binomial expansion twice in the second line of (\ref{psim}). the polynomial (\ref{flux1}) follows in a similar way. The linear combination on the first line in that case follows from an identity (\ref{eq:BessId1}) that serves to relate the projection 
$[\eta^{-2}]$ to a linear combination of $[\eta^{-1}]$ and $[\eta^{0}]$. That identity stems from the Bessel generating function (\ref{BessGen}) which in this case states that $I_{2}(X) + I_0(X) = \frac1{X} I_1(X)$. The relation to string polynomials is explained in section \ref{StringApple}.
\end{proof}

\subsubsection{Leading Order Conservation Laws for potentials with regular odd valence}

For simplicity and because it will be the focus of many of our topics in this paper we now restrict attention to the case of potentials with regular odd valence, $j = 2\nu+1$,   So $t$ in the theorem below stands for the time parameter $t_j$ with $j = 2 \nu +1$. 

\begin{thm}
The continuum equations at leading order for the Toda, differenced string and undifferenced string equations, at valence $j = 2\nu+1$ may respectively  be written in matrix form as
\begin{eqnarray} \label{TODA}
\frac{\partial}{\partial t} 
\left(
\begin{array}{c}
h_0 \\ f_0
\end{array}
\right) &+& (2\nu+1)
\left(
\begin{array}{cc}
B_{11} & B_{12}\\ 
f_0 B_{12} & B_{11}
\end{array}
\right)  \frac{\partial}{\partial x} 
\left(
\begin{array}{c}
h_0 \\ f_0
\end{array}
\right) = 0\\
\label{STRING} 
\left(
\begin{array}{c}
0 \\ 1
\end{array}
\right) &=& 
\left(
\begin{array}{cc}
A_{11} & A_{12}\\ 
f_0 A_{12} & A_{11}
\end{array}
\right)  \frac{\partial}{\partial x} 
\left(
\begin{array}{c}
h_0 \\ f_0
\end{array}
\right)\\
\label{UNSTRING}
\left(
\begin{array}{c}
0 \\ x
\end{array}
\right)
 &=& 
\left(
\begin{array}{c}
h_0 + (2\nu + 1 ) t B_{12} \\ f_0 + (2\nu+1) t B_{11}
\end{array}
\right)
= \left(
\begin{array}{c}
\phi_{V,0}\\ \psi_{V,0}
\end{array}
\right)
\end{eqnarray}
where 
\begin{eqnarray}  \nonumber
B_{11} &=&  \psi_0/(2\nu+1)\\ \label{B11}
&=& \sum_{\mu = 1}^\nu {2\nu \choose 2\mu - 1, \nu - \mu, \nu - \mu +1} h_0^{2\mu-1} f_0^{\nu-\mu+1}\\  \nonumber
B_{12} &=&  \phi_0/(2\nu+1) \\ \label{B12}
&=& \sum_{\mu = 0}^\nu {2\nu \choose 2\mu, \nu - \mu, \nu - \mu} h_0^{2\mu} f_0^{\nu-\mu},
\end{eqnarray}
and
\begin{eqnarray}  \nonumber
A_{11} &=& {\phi}_{V,1} = 1 + t \phi_1\\ \label{A11}
&=& 1 + (2\nu+1) t \sum_{\mu = 0}^{\nu-1} {2\nu \choose 2\mu + 1, \nu - \mu - 1, \nu - \mu} (\nu - \mu) h_0^{2\mu+1} f_0^{\nu-\mu-1}\\ \nonumber
A_{12} &=&  {\psi}_{V,1}/f_0 = t \psi_1/f_0  \\ \label{A12}
&=& (2\nu+1) t \sum_{\mu = 0}^{\nu-1} {2\nu \choose 2\mu, \nu - \mu -1, \nu - \mu +1} (\nu - \mu +1) h_0^{2\mu} f_0^{\nu-\mu -1}\, ,
\end{eqnarray}
in which the string polynomials here are all defined relative to $j = 2\nu+1$. (The form of $A_{11}$ and $A_{12}$ here stems from the fact that $V = \frac12 \lambda^2 + t \lambda^j$.)
\end{thm}
\noindent The derivation of the expressions in the theorem follows from the same kinds of considerations applied in the previous subsection and in particular explicit calculations on the polynomial expressions in (\ref{flux1}) and (\ref{flux2}) (see section \ref{MV} for the basic idea). These were also derived by an independent method in \cite{EP11}.
\medskip

\noindent We set
\beann
{\bf B} &=& \left(
\begin{array}{cc}
B_{11}(\nu) & B_{12}(\nu)\\
B_{21}(\nu) & B_{22}(\nu)
\end{array}
\right)\\
{\bf A} &=& \left(
\begin{array}{cc}
A_{11}(\nu) & A_{12}(\nu)\\
A_{21}(\nu) & A_{22}(\nu)
\end{array}
\right)
\eeann
where the dependence of these functions on $\nu$ is made to indicate the dependence on the odd valence $2\nu+1$. For future reference we record the following useful identities.
\begin{lem}
\begin{eqnarray} \label{grad1}
(2 \nu +1) B_{11} &=&  \partial_{h_0} \mathcal{F}_{1}  \\
\label{grad2}  (2 \nu +1)  B_{12}  &=&  \partial_{f_0} \mathcal{F}_{1} \\
\label{grad3} (2 \nu +1) B_{21} &=&  \left\{ \mathcal{F}_{1} + \partial_{h_0} \mathcal{F}_{2}\right\}  \\
\label{grad4} (2 \nu +1)  B_{22} &=& \partial_{f_0} \mathcal{F}_{2}   \\
\label{symm1} B_{11} &=& B_{22}\\
\label{symm2} B_{21} &=& f_0 B_{12}\\
\label{hess1} A_{11} &=& 1 + (2\nu+1) t \partial_{f_0}B_{11}\\
\label{hess2} A_{12} &=& (2\nu+1) t \partial_{f_0}B_{12}\\
\label{hess3} A_{21} &=& (2\nu+1) t \partial_{h_0}B_{11}\\
\label{hess4} A_{22} &=& 1 +(2\nu+1) t \partial_{h_0}B_{12}\\
\label{symm3} A_{11} &=& A_{22}\\
\label{symm4} A_{21} &=& f_0 A_{12}.
\end{eqnarray}
\end{lem}
\begin{proof}
The relations (\ref{grad1}) and (\ref{grad2}) follow from comparison of  (\ref{B11}), (\ref{B12}) with (\ref{flux1}).  Relations (\ref{grad3}) and (\ref{grad4}) follow from a similar comparison with (\ref{flux2}). Equation (\ref{hess1}) (resp. (\ref{hess2})) follows from comparing (\ref{B11}) with (\ref{A11}) (resp.  (\ref{B12}) with (\ref{A12}). (\ref{hess3}) and (\ref{hess4}) follow from (\ref{grad1}), (\ref{grad2}) and (\ref{grad4}). Finally (\ref{symm1} - \ref{symm4}) follow by direct calculation from (\ref{B11}, \ref{B12}, \ref{A11}, \ref{A12}). Note that the coefficients of $\mathbf B$ are expressible in terms of the gradient of $\mathcal{F}_{1}$ and those of $\mathbf A$ in terms of the hessian of $\mathcal{F}_{1}$.
\end{proof}
\medskip

\subsubsection{Conservation Law Structure at Leading Order: Hodograph Transform and Riemann Invariants} \label{RINVARIANTS}

We will now show that the equations (\ref{STRING}) are in fact a differentiated form of the hodograph solution of (\ref{TODA}). To this end first note that the pointwise eigenvalues of
$\mathbf B$ are given by 
\begin{eqnarray}\label{evs}
\lambda_\pm &=& B_{11} \pm \sqrt{f_0} B_{12}
\end{eqnarray}
with corresponding left eigenvectors
\begin{eqnarray}\label{evecs}
(\pm \sqrt{f_0}, 1).
\end{eqnarray}
We introduce the hodograph relations \cite{Kev} among the variables $(t, x, h_0, f_0)$,
\begin{eqnarray*}
(\pm \sqrt{f_0}, 1)  \left(
\begin{array}{c}
h^{(in)}_0(x)\\ f^{(in)}_0(x)
\end{array}
\right) &=& (\pm \sqrt{f_0}, 1) \left[ \left(
\begin{array}{c}
h_0 \\ f_0
\end{array}
\right) +(2\nu+1) t \left(
\begin{array}{c}
B_{12} \\ B_{11}
\end{array}
\right)\right]\\
&=& (\pm \sqrt{f_0} h_0 + f_0) +(2\nu+1)  t \lambda_\pm 
\end{eqnarray*}
in which the left hand side gives the initial values of $h_0$ and $f_0$ giving the initial curve, at $t=0$ that determines a unique integral surface a solution to the leading order PDE. In our case these initial values are $h^{(in)}_0(x) = 0, f^{(in)}_0(x) = x$ corresponding to the Gaussian weight.
In this case the hodograph solution reduces to 
\begin{eqnarray}\label{genhod}
x &=& (\pm \sqrt{f_0} h_0 + f_0) +(2\nu+1)  t \lambda_\pm\, .
\end{eqnarray}
For this choice of initial values we refer to this integral surface, and its corresponding solution, as a {\it regular surface} or the {\it $(2\nu+1)$-regular surface} when we want to specify the particular equation in the continuum Toda hirearchy for which this is an integral surface. 

Note that by subtracting and adding the two ($\pm$) equations in (\ref{genhod}) one recovers an integrated form of (\ref{STRING}),
\begin{eqnarray} \label{GaussHod1}
 h_0 + (2\nu+1)  t B_{12} &=& 0\\ \label{GaussHod2}
 f_0 + (2\nu+1)  t B_{11} &=& w.
\end{eqnarray}
Conversely, differentiating these equations with respect to $w$ and using the identities
(\ref{grad1} - \ref{symm4}) directly yields the equations (\ref{STRING}). Thus the leading order continuum difference string equations are equivalent to the hodograph relations (\ref{genhod}). We now show independently that these hodograph relations implicitly give the solution of the continuum Toda equations.  

\begin{lem} \label{lem:extdiff}
A local solution of (\ref{TODA}) is implicitly defined by (\ref{GaussHod1}) and (\ref{GaussHod2}).
\end{lem}
\begin{proof}
The annihilator of the differentials of (\ref{genhod})  satisfy
\begin{eqnarray} \nonumber
\left(
\begin{array}{c}
dx -  (2\nu+1) \lambda_+ dt \\ dx -  (2\nu+1) \lambda_- dt
\end{array}
\right)
 &=& \left[ \left(
\begin{array}{cc}
\sqrt{f_0} & 1 + \frac12 \frac{h_0}{\sqrt{f_0}}\\ \label{extdiff}
-\sqrt{f_0} & 1 + \frac12 \frac{h_0}{\sqrt{f_0}}
\end{array}
\right) 
+ (2\nu+1) t\left(
\begin{array}{cc}
\frac{\partial \lambda_+}{\partial h_0} & \frac{\partial \lambda_+}{\partial f_0}\\ 
 \frac{\partial \lambda_-}{\partial h_0} & \frac{\partial \lambda_-}{\partial f_0}
\end{array}
\right) \right] \left(
\begin{array}{c}
d h_0 \\ d f_0
\end{array}
\right)  \\
&=& \left[ \begin{array}{cc}
\sqrt{f_0} \left(A_{11} + \sqrt{f_0} A_{12} \right) & \left(A_{11} + \sqrt{f_0} A_{12} \right)\\ 
-\sqrt{f_0} \left(A_{11} - \sqrt{f_0} A_{12} \right) & \left(A_{11} - \sqrt{f_0} A_{12} \right)
\end{array} \right] 
\left(
\begin{array}{c}
d h_0 \\ d f_0
\end{array}
\right) 
\end{eqnarray}
which determines a two-dimensional distribution locally on the space $(t, x, h_0, f_0)$. In deriving the second line of this annihilator equation, use was made of the identities (\ref{hess1} - \ref{hess4}) and the first of the hodograph equations (\ref{GaussHod1}). This may be rewritten in diagonal form as
\begin{eqnarray*}
\frac{dx}{dt} -  (2\nu+1) \lambda_\pm  &=& \sqrt{f_0} \left(A_{11} \pm \sqrt{f_0} A_{12} \right) \frac{dr_\pm}{dt}
\end{eqnarray*}
where 
\begin{eqnarray} \label{char}
r_\pm = h_0 \pm 2\sqrt{f_0}.
\end{eqnarray}
From this one sees that away from where the matrix on the right hand side fails to have maximal rank; i.e., away from the locus where $A_{11}^2 - f_0 A_{12}^2 = 0$, this exterior differential system determines a well-defined integral surface over the $(t, x)$ plane whose characteristic curves are given by the left hand side: $\frac{dx}{dt} = (2\nu+1)\lambda_\pm$ which is equivalent to the Toda equations (\ref{TODA}) in Riemann inavariant form,
\begin{eqnarray}\label{RinvtToda}
\frac{\partial}{\partial t} 
\left(
\begin{array}{c}
r_+ \\ r_-
\end{array}
\right) &+&  (2\nu+1)
\left(
\begin{array}{cc}
\lambda_+ & 0\\ 
0 & \lambda_-
\end{array}
\right)  \frac{\partial}{\partial x} 
\left(
\begin{array}{c}
r_+ \\ r_-
\end{array}
\right) = 0.
\end{eqnarray}
\end{proof}

We now state and prove a standard result which is, nevertheless, interesting in out setting.

\begin{prop}
The Riemann invariants, $r_\pm$ are respectively constant along the integral curves (characteristics) of the respective ode's $\frac{dx}{dt} = j \lambda_\pm(r_+, r_-)$, where $j = 2\nu+1$.
\end{prop}
\begin{proof}
We have from (\ref{RinvtToda}) that
\beann
\partial_t r_\pm + j \lambda_\pm \partial_x r_\pm &=& 0
\eeann
So, if $x_\pm(t)$ satisfies 
\beann
\frac{dx_\pm}{dt} = j \lambda_\pm,
\eeann
then
\beann
\partial_t r_\pm + \frac{dx}{dt} \partial_x r_\pm &=& 0\\
\frac{d}{dt} r_\pm (t, x_\pm(t)) &=& 0
\eeann
where the second equation follows from the first by the chain rule. Hence $r_\pm$ is indeed constant along the respective charaacteristics $x_\pm(t)$.
\end{proof}

Of course the precise form of these characteristics are determined by the curve of initial data one chooses for the characteristic ode's. In the case of regular valence this curve is determined by the self-similar scaling associated to Euler's relation:
\beann
(x, x^{\frac{j}{2}-1}_\pm).
\eeann
As we will see, in section \ref{chargeom} the inegral surface $\mathcal{S}$ in this case has the form of a cone over an algebraic curve $\mathcal{C}.$

\subsubsection{ODE reduction in the case of regular maps} \label{ODEreduction}
We continue to restrict our attention to a $j$-regular surface (for $j = 2\nu + 1$), defined by (\ref{GaussHod1}) and (\ref{GaussHod2}) in the 0-jet space $\left\{t, x, h_0, f_0\right\}$ to write the leading order Toda equations (\ref{TODA})  in the form
\begin{eqnarray} \label{GaussToda}
t \left(
\begin{array}{c}
h_{0,t } \\ f_{0,t }
\end{array}
\right) &+& \left(x - \left(
\begin{array}{cc}
f_0 & h_{0} \\ f_{0}h_0 & f_0
\end{array}
\right) \right) 
\left(
\begin{array}{c}
h_{0,x} \\ f_{0,x}
\end{array}
\right) = 0 \, .
\end{eqnarray}
One may also evaluate these equations on $h_0, f_0$ in their self-similar form:
\begin{eqnarray*}
h_0(t ,x) &=& x^{1/2} u_0(t x^{\nu - 1/2})\\
f_0(t ,x) &=& x z_0(t x^{\nu - 1/2}).
\end{eqnarray*}
From this scaling relation one may deduce the {\it zero-order exchange relations} 
\begin{eqnarray} \label{0exchange}
x\partial_x f_0 &=& f_0 + \frac{j-2}{2} t\partial_t f_0\\
x\partial_x h_0 &=& (1/2)h_0 + \frac{j-2}{2} t\partial_t h_0.
\end{eqnarray}
It is straightforward to then calculate that the induced ODE reduction, with respect to the similarity variable $\xi = tx^{\nu -1/2}$ of (\ref{TODA}) takes the form
\begin{eqnarray} \label{GaussTodaSS} 
\left(
\begin{array}{c}
\xi h^\prime_{0} \\ \xi f^\prime_{0} 
\end{array}
\right) &=& - \frac1 D \left(jx -(j-2) \left(
\begin{array}{cc}
f_0 & -h_{0} \\ - f_{0}h_0 & f_0
\end{array}
\right) \right)
\left(x - \left(
\begin{array}{cc}
f_0 & h_{0} \\ f_{0}h_0 & f_0
\end{array}
\right) \right) 
\left(
\begin{array}{c}
h_{0} \\ 2 f_{0}
\end{array}
\right)\\
&=& \frac{-jt} D \left(2x + j(j-2)t \left(\begin{array}{cc}
B_{11} &- B_{12}\\- f_{0} B_{12}& B_{11}
\end{array}\right)\right){ \bf B}
\left(
\begin{array}{c}
h_{0} \\ 2 f_{0}
\end{array}
\right)
\end{eqnarray}
where $\,^\prime = d/d\xi$ and
\begin{eqnarray} \label{discrim}
D &=& (jx - (j-2)f_0)^2 - f_0(j-2)^2h_0^2\\
\nonumber &=& d_+ \cdot d_- \,;\\ \label{discrimpm}
d_\pm &=& (jx - (j-2) f_0) \pm (j-2) \sqrt{f_0} h_0.
\end{eqnarray}
Combining (\ref{GaussToda}) and the above exchange relations one derives
\begin{eqnarray*}
 \left( j x - (j-2) \left(
\begin{array}{cc}
f_0 & h_{0} \\ f_{0}h_0 & f_0
\end{array}
\right) \right) 
\left(
\begin{array}{c}
h_{0,x} \\ f_{0,x}
\end{array}
\right) &=& \left(
\begin{array}{c}
h_{0} \\2  f_{0}
\end{array}
\right).
\end{eqnarray*}
Inverting this gives
\begin{eqnarray} \label{odvalinv}
\left(
\begin{array}{c}
h_{0,x} \\ f_{0,x}
\end{array}
\right) &=& \frac1D \left( j x + (j-2) \left(
\begin{array}{cc}
- f_0 & h_{0} \\ f_{0}h_0 & - f_0
\end{array}
\right) \right) \left(
\begin{array}{c}
h_{0} \\2  f_{0}
\end{array}
\right).
\end{eqnarray}
From this one may directly derive that the $x$-derivatives of the Riemann invariants (\ref{char}) are concisely given by
\begin{eqnarray} \label{wrpm}
\partial_x r_{\pm} &=& \frac{r_\pm}{d_\mp}
\end{eqnarray}
Combining this with (\ref{RinvtToda}) one has the diagonal system

\begin{eqnarray}\label{RinvtToda2}
\partial_t r_\pm &=& - j \frac{\lambda_\pm r_\pm}{d_\mp} \,;\\
\lambda_\pm &=& B_{11} \pm \sqrt{f_0}B_{12}\\
&=& \frac{(x-f_0) \mp \sqrt{f_0} h_0}{j s}
\end{eqnarray}
where the third line follows from the hodograph equations (\ref{GaussHod1} - \ref{GaussHod2})
so that (\ref{RinvtToda2}) may be rewritten as
\begin{eqnarray}\label{RinvtToda3}
t\partial_t r_\pm &=& \frac{(f_0 - x) \pm \sqrt{f_0} h_0}{d_\mp}  r_\pm \,.
\end{eqnarray}
We note that (\ref{wrpm}) and (\ref{RinvtToda3}) only depend on the valence through $d_\pm$.
\medskip

Multiplying out (\ref{odvalinv}) one can prolong the regular surface to the 1-jet space in the patch $\left\{x, h_0, f_0\right\}$ via the formula
\begin{eqnarray} \nonumber 
\left(
\begin{array}{c}
h_{0, x} \\ f_{0,x}
\end{array}
\right)&=& \frac1 D \left(jx -(j-2) \left(
\begin{array}{cc}
f_0 & - h_{0} \\ - f_{0}h_0 & f_0
\end{array}
\right) \right)
\left(
\begin{array}{c}
h_{0} \\ 2 f_{0}
\end{array}
\right)\\ \label{jet1}
&=& \left(
\begin{array}{c}
 h_0 \frac{j(f_0 + x) -2 f_0}{D}\\ - f_0 \frac{\delta_{f_0} D}{(j-2) D} 
\end{array}\right).
\end{eqnarray}

We note that, in the coordinate patch $\left\{s, h_0, f_0\right\}$, this extension takes the form
\begin{eqnarray} \label{jet1Alt}
\left(
\begin{array}{c}
h_{0, x} \\ f_{0,x}
\end{array}
\right) &=& \left(
\begin{array}{c}
 \frac{-A_{12}}{A^2_{11} - f_0 A^2_{12}}\\ \frac{A_{11}}{A^2_{11} - f_0 A^2_{12}}
\end{array}\right). 
\end{eqnarray}

We next introduce a relation that will be important for the further topics that we examine in this paper. We already noted that, with regard to (\ref{extdiff}), well-defined solution curves of the continuum Toda system could be constructed away from the locus where $A^2_{11} - f_0 A^2_{12} = 0$. From (\ref{wrpm}) and (\ref{RinvtToda3}) it seems reasonable that this locus should bear some relation to that of $D = d_+ \cdot d_-$. In fact one has
\begin{eqnarray} \label{basic1}
A^2_{11} - f_0 A^2_{12} &=& \frac{1}{f_{0x}^2 - f_0 h_{0x}^2}\\
\label{basic2} &=& \frac{D}{f_0 (4 f_0 -  h_0^2)}\\
&=& \frac{-1}{f_0}\frac{d_+ \cdot d_-}{r_+ \cdot r_-}\\
&=& \frac{-1}{f_0 \partial_x r_+ \cdot \partial_x r_-}
\end{eqnarray}
where (\ref{basic1}) follows from the fact that
\begin{eqnarray} \label{Ainv}
{\bf A}^{-1} &=&  \left[\begin{array}{cc}  f_{0x} & h_{0x} \\ f_0 h_{0x} & f_{0x} \end{array} \right] \,,
\end{eqnarray}
which in turn follows from (\ref{STRING}), by taking determinants;  (\ref{basic2}) from the last line of (\ref{jet1}) by direct calculation and the next two equations by definition and (\ref{wrpm}) respectively.

Finally, comparing the last line of (\ref{jet1}) to (\ref{jet1Alt}) one also has
\begin{eqnarray}
A_{12} &=& \frac{h_0(jx + (j-2)f_0)}{f_0 (h_0^2 -4 f_0)}\\
A_{11} &=& \frac{2jx +(j-2) (h_0^2 -2f_0)}{4 f_0 - h_0^2}.
\end{eqnarray}

For subsequent use we introduce some rescalings which will be further motivated in Section \ref{chargeom}
\begin{eqnarray*}
y_0 &=& \frac{h_0^2}{f_0}\\
\widehat{\mathcal{F}}_{1} &=&  \frac1{2\nu+1}\sum_{\mu = 0}^\nu {2\nu + 1 \choose 2\mu, \nu-\mu, \nu - \mu +1} y_0^\mu\\
 \widehat{B}_{11} &=& 2 y_0^{1/2} \partial_{y_0} \widehat{\mathcal{F}}_{1}\\ 
 \widehat{B}_{12} &=& (\nu+1) \widehat{\mathcal{F}}_{1} - y_0 \partial_{y_0} \widehat{\mathcal{F}}_{1}.
\end{eqnarray*}

\subsection{Universal System for Mixed Valence} \label{MV}

We now return to the case of a general potential having the form (\ref{eq:genpot}). We have already noted in (\ref{scstring}) that general string equations take the form

\begin{eqnarray} \label{eq:symbol}
\left( \begin{array}{c} 0 \\ x \end{array} \right) &=&  \left( \begin{array}{c} h\\ f \end{array}  \right) +\sum_{j=1}^J  j t_j [z^{j-1}]\left( \begin{array}{c} [\eta^0] \cr [\eta^{-1}] \sqrt{f} \end{array} \right) \left[ (1 - z \mathcal{L})^{-1}\right].
\end{eqnarray}
This is consistent with the general form of the string polynomials given by (\ref{PHIm}) and (\ref{PSIm}) which replaces $\mathcal{L}$ above by its principal symbol $\mathcal{L}_0$ (and $h, f$ by $h_0, f_0$). Making that replacement, one finds at leading order the undifferenced string equations become
\begin{eqnarray} \label{eq:prinsymbol}
\left( \begin{array}{c} 0 \\ x \end{array} \right) &=&  \left( \begin{array}{c} h_0 \\ f_0 \end{array}  \right) + \sum_{j=1}^J   t_j \left( \begin{array}{c} \phi^{(j)}_0 \cr \psi^{(j)}_0 \end{array} \right)
\end{eqnarray}
from the definition of string polynomials introduced in section \ref{sec:StringPoly}.
\smallskip

Similarly, as pointed out in Remark \ref{remcor02}, the form of the $j^{th}$ conservation law remains unchanged in the mixed valence case:

\begin{eqnarray} \label{genflow}
\frac{\partial}{\partial t_j} 
\left(
\begin{array}{c}
h_0 \\ f_0 +\frac12 h_0^2
\end{array}
\right) &+& 
\frac{\partial}{\partial x} 
\left(
\begin{array}{c}
\mathcal{F}_{1}{(j)}\\ 
\mathcal{F}_{2}{(j)} + h_0 \mathcal{F}_{1}{(j)}
\end{array}
\right)  = 0.
\end{eqnarray}
Converting to their form as hyperbolic systems these become
\begin{eqnarray} \label{eq:hyp}
\frac{\partial}{\partial t_j} 
\left(
\begin{array}{c}
h_0 \\ f_0 
\end{array}
\right) &+& 
\left(
\begin{array}{cc}
\partial_{h_0} \mathcal{F}_{1}{(j)} & \partial_{f_0} \mathcal{F}_{1}{(j)}\\ 
\partial_{h_0} \mathcal{F}_{2}{(j)} + \mathcal{F}_{1}{(j)} & \partial_{f_0} \mathcal{F}_{2}{(j)} 
\end{array}
\right) \frac{\partial}{\partial x} \left(
\begin{array}{c}
h_0 \\ f_0 
\end{array}
\right) = 0.
\end{eqnarray}
Moreover, by direct calculation on the polynomial fluxes in Corollary \ref{cor02} one finds that 
\begin{eqnarray} 
\partial_{h_0} \mathcal{F}_{1}{(j)} &=& \partial_{f_0} \mathcal{F}_{2}{(j)} = \psi^{(j)}_0 \\
\partial_{h_0} \mathcal{F}_{2}{(j)} + \mathcal{F}_{1}{(j)} &=& f_0 \partial_{f_0} \mathcal{F}_{1}{(j)} = f_0 \psi^{(j)}_0,
\end{eqnarray}
so that the system (\ref{eq:hyp}) simplifies to the form
\begin{eqnarray} 
\frac{\partial}{\partial t_j} 
\left(
\begin{array}{c}
h_0 \\ f_0 
\end{array}
\right) &+& 
\left(
\begin{array}{cc}
\psi^{(j)}_0 & \phi^{(j)}_0\\ 
f_0 \phi^{(j)}_0 & \psi^{(j)}_0 
\end{array}
\right) \frac{\partial}{\partial x} \left(
\begin{array}{c}
h_0 \\ f_0 
\end{array}
\right) = 0.
\end{eqnarray}

Now we consider the general linear combination of conservation laws,
\begin{eqnarray} 
t \sum_{j=1}^J \alpha_j \frac{\partial}{\partial t_j} 
\left(
\begin{array}{c}
h_0 \\ f_0 
\end{array}
\right) &+& 
t \sum_{j=1}^J \alpha_j \left(
\begin{array}{cc}
\psi^{(j)}_0 & \phi^{(j)}_0\\ 
f_0 \phi^{(j)}_0 & \psi^{(j)}_0 
\end{array}
\right) \frac{\partial}{\partial x} \left(
\begin{array}{c}
h_0 \\ f_0 
\end{array}
\right) = 0,
\end{eqnarray}
where the $\alpha_j$ are the coefficients of a fixed background potential $V - \frac12 \lambda^2$ of generally mixed valence. We recall that $\frac{d}{dV}$ denotes the derivative combination on the LHS which is a directional  derivative in the "direction" of the background potential (initialized at the Gaussian potential).  The PDE is a hyperbolic system for the flow in that direction. It is straightforward to see, by working backward, that this equation can be rewritten as a conservation law for that flow. Now, substituting from the leading order continuum string equation (\ref{eq:prinsymbol}), with $t_j = \alpha_j t$, into this system yields a system of {\it universal} form.
\begin{eqnarray} 
t \frac{d}{dV} 
\left(
\begin{array}{c}
h_0 \\ f_0 
\end{array}
\right) &+& 
\left(x - \left(
\begin{array}{cc}
f_0 & h_{0} \\ f_{0}h_0 & f_0
\end{array}
\right) \right)  \frac{\partial}{\partial x} \left(
\begin{array}{c}
h_0 \\ f_0 
\end{array}
\right) = 0.
\end{eqnarray}

\subsection{Some bosonic versions of string equations} \label{bosestrings} 

Using some of the continuum limit formulations we have derived in the general valence setting, one may further derive a number of elegant and useful relations. For instance, one may invert (\ref{boson12}) to express the coefficients of the genus expansion in terms of these cumulants:

\begin{align*}
\log b^{2}_k(x,\vec{t}) 
=& \left(e^{\frac{1}{k}\partial_x}-2+ e^{-\frac{1}{k}\partial_x}\right) \frac{1}{k^2} \log \tau^{2}_{k}(x,\vec{t}) \\
=& 4\,\text{sinh}^2 \left(\frac{1}{2k}\partial_x \right)\frac{1}{k^2} \log \tau^{2}_{k}(x,\vec{t})
\end{align*}
The Taylor expansion of hyperbolic cosecant squared is
\begin{align*}
\text{csch}^2(X) =& \sum_{m\geq 0} \frac{(1-2m)2^{2m} B_{2m}}{(2m)!} X^{2m-2}
\end{align*}
where the coefficients $B_m$ are the Bernoulli numbers. This establishes an elegant relation that had been formally derived, using Euler-Maclaurin expansions, in \cite{BIZ80}:

\begin{align}
\frac{1}{k^2} \log \tau^{2}_{k}(x,\vec{t})
=& \frac{1}{4} \text{csch}^2\left(\frac{1}{2k}\partial_x\right) \log b^{2}_k(x,\vec{t}) \nonumber \\
=& \sum_{m\geq 0} \frac{(1-2m) B_{2m}}{(2m)!} \left(\frac{1}{k}\partial_x\right)^{2m-2} \log b^{2}_k(x,\vec{t}) \nonumber \\
=& \sum_{g\geq 0} N^{-2g} \sum_{m=0}^g \frac{(1-2m)B_{2m}}{(2m)!} \partial_{x}^{2m-2} \mathcal{C}_{g-m}(f_0,f_1,\ldots ,f_g)(x,\vec{t})
\end{align}

\smallskip

Next, we recall from Theorem 3.5 of \cite{EP11} the formula relating the coefficients of the asymptotic partition function to the asymptotic $a$ recursion coefficients (see also (\ref{as-asymp})):
\begin{eqnarray} \label{BOSON}
h_g(x, \vec{t}) &=& \sum_{\begin{matrix} 2 g_1 + m = g+1 \\  g_1 \geq 0\,, m>0\end{matrix}}  - \frac1{m!} \frac{\partial^{m+1}}{\partial \xi_1\partial w^m} \bigg[ 
E_{g_1}\left(w, \xi_1, t_2, t_3, \dots\right)\bigg]_{\xi_1 = t_1, w = x}
\end{eqnarray}
Multiplying by $k^{-g}$ and summing one has 
\begin{eqnarray}
\sum_{g \geq 0} h_g(x, \vec{t}) k^{-g} &=& \sum_{g \geq 0} k^{-g} \sum_{\begin{matrix} 2 g_1 + m = g+1 \\  g_1 \geq 0\,, m>0\end{matrix}}  - \frac1{m!} \frac{\partial^{m+1}}{\partial \xi_1\partial w^m} \bigg[ E_{g_1}\left(w, \xi_1, t_2, t_3, \dots\right)\bigg]_{\xi_1 = t_1, w = x}\\
\nonumber &=& -\sum_{g_1 \geq 0} \sum_{m \geq 0}  \frac1{(m+1)!} k^{-(m+1)} \frac{\partial^{m+1}}{\partial w^{m+1}} \bigg[ k^{-2 g_1 + 1} \frac{\partial}{\partial s_1}E_{g_1}\left(w, s_1, t_2, t_3, \dots\right)|_{s_1 = t_1}\bigg]_{ w = 1}\\ 
&=& - \left(e^{\frac1{k} \partial_w} - 1 \right) \sum_{g \geq 0} k^{-(2 g - 1)} \bigg[\frac{\partial}{\partial \xi_1}E_{g}\left(w, \xi_1, t_2, t_3, \dots\right)|_{\xi_1 = t_1}\bigg]_{ w = x}.
\end{eqnarray}

The bosonic operator (see (\ref{boson1})) appearing in the last line here is invertible, in the sense of pseudo-differential operators, modulo "constants" (in $w$). As seen from (\ref{boson1}) the only functions in the kernel are those taken into themselves under translation in $w$ by $1/k$ for general $k$; this amounts to saying that the recursion operators are constant in $n$ at least asymptotically. However this is not the case for the exponentially weighted orthogonal polynomials we consider.
(Alternatively, one could argue that there are no terms constant in $w$ in the coefficients of the genus expansion since these would correspond to faceless maps, but every map has at least one face.) Hence we may work orthogonal to the kernel. To make this effective, we open up the $w$-differentiation to one further order before evaluating at $x$. This is possible by the analyticity with respect to parameters of the genus expansion that was established in \cite{EM03}. Thus one may write
\begin{eqnarray*}
\frac{\partial}{\partial w}\sum_{g \geq 0} h_g(w, \vec{t})|_{w=x} k^{-g}
&=& - \left(e^{\frac1{k} \partial_w} - 1 \right) \sum_{g \geq 0} k^{-(2 g - 1)} \frac{\partial^2}{\partial w \partial \xi_1}E_{g}\left(w, \xi_1, t_2, t_3, \dots\right)|_{\xi_1 = t_1,  w = x}
\end{eqnarray*}
so that the bosonic operator is acting purely on gradients.
One may  now write
\begin{eqnarray*}
\frac{\frac{1}{k} \partial_w}{e^{\frac{1}{k} \partial_w} - 1} \sum_{g \geq 0} h_g(w, \vec{t})\Big|_{ w = x} k^{-g} &=& - \sum_{g \geq 0}  k^{-2 g} \frac{\partial^2}{\partial w\partial \xi_1}E_{g}\left(w, \xi_1, t_2, t_3, \dots\right)|_{s_1 = t_1, w=x}.
\end{eqnarray*}
But now we recognize from the generating function for Bernoulli numbers that by the pseudo-differential calculus one may expand the operator on the left-hand side of the previous equation to get
\begin{eqnarray} \nonumber
 \sum_{m=0}^\infty \frac{B_m}{m!} \frac{1}{k^{m}} \frac{\partial^{m}}{\partial w^{m}}\sum_{g \geq 0} h_g(w, \vec{t})\Big|_{ w = x} k^{-g} &=& - \sum_{g \geq 0} k^{-2 g} \frac{\partial^2}{\partial w\partial \xi_1}E_{g}\left(w, \xi_1, t_2, t_3, \dots\right)|_{\xi_1 = t_1, w=x}\\ \label{boson22}
\end{eqnarray}
where the coefficients $B_m$ are again the Bernoulli numbers. Since the right hand side of this equation is an expansion in even powers of $k$, one may deduce a useful corollary by collecting coefficients of odd powers of $k$ on the left-hand side and setting them equal to zero:
\begin{eqnarray} \label{hoddident}
\sum_{m = 0}^{2g + 1} \frac{B_m}{m!} \frac{\partial^m}{\partial w^m} h_{2 g +1 - m}(w, \vec{t})|_{w = x} &=& 0.
\end{eqnarray}
This identity was also derived by purely elementary combinatorial arguments in \cite{W}.
\medskip

Again, all seemingly formal manipulations made above may be justified by the results of \cite{EM03}.

\section{The Hopf Algebra Characterization of String Polynomials} \label{Hopf-String}
This section will reveal a remarkable connection between Bessel functions, Appell polynomials and explicit solutions to the Continuum Toda-String Equations. At some level this is perhaps not surprising given the seminal role that Bessel functions play in the work of Baik, Deift, and Johansson \cite{AD}. However, there it came out of a structural relation with Toeplitz determinants, whereas here it stems from Hankel determinants. Nevertheless, this may suggest deeper relations.
\subsection{The Binomial Hopf Algebra and its Umbral Calculus} \label{Hopf}
We briefly review here the elements of the theory of Hopf algebras, and in particualr the binomial Hopf algebra, that we will need. For background, and more details we refer the reader to \cite{Roman} and \cite{JR}. For us the binomial Hopf algebra is the ring of polynomials $\mathbb{Q}[\zeta]$ with co-multiplication given by 
\begin{eqnarray*}
\Delta \zeta^n &=& \sum_{k=0}^n {n \choose k} \zeta^k \otimes \zeta^{n-k}.
\end{eqnarray*}

One considers the dual vector space, $\mathbb{Q}[\zeta]^*$ of all linear functionals on this algebra
and denotes the action of such a functional, $\Lambda$, on a polynomial $p(\zeta)$ by $\langle \Lambda | p(\zeta)\rangle$. An element of this dual space is determined by its values on a basis and so, in particular one may uniquely associate the functional to a sequence of numbers
\begin{eqnarray*}
a_k &=& \langle \Lambda | \zeta^k \rangle.
\end{eqnarray*}
There is then a naturally induced convolution on any two functionals, $\Lambda, \Xi$ with respective sequences, $a_k, b_k$, given in terms of the co-multiplication by
\begin{eqnarray*}
\langle \Lambda*\Xi | \zeta^n \rangle &=&( \Lambda \otimes \Xi) \Delta \zeta^n \\ 
& = & \sum_{k=0}^n {n \choose k}  \langle \Lambda | \zeta^k \rangle  \langle \Xi | \zeta^{n-k} \rangle\\
&=& \sum_{k=0}^n {n \choose k} a_k b_{n-k}.
\end{eqnarray*}
One may then associate to each element of the dual space a formal exponential generating function $g(t) = \sum_{k=0}^n \frac{a_n}{n!} t^n$. This sets up a one-to-one correspondence between the dual space of the binomial Hopf algebra and the space of formal exponenetial generating functions in which convolution of functionals corresponds to the product of associated generating functions. In other words the algebra structure on $\mathbb{Q}[\zeta]^*$ induced by convolution is isomorphic to the natural algebraic structure on the space of formal generating functions.   So from now on, discussions of this dual space will largely 
be phrased solely and simply in terms of generating functions by writing 
\begin{eqnarray*}
\langle g(t) | \zeta^k \rangle &=& a_k.
\end{eqnarray*}

Given all this it is natural to consider, more generally, linear operators on the binomial Hopf  algebra and their relation to linear operators on $\mathbb{Q}[\zeta]^*$. An example of this is the differentiation operator $\partial^k$ acting on the algebra of polynomials. With respect to the diagonal inner product for the monomial basis of $\mathbb{Q}[\zeta]$ 
\beann
(\zeta^j, \zeta^k) &=& \delta_{jk}
\eeann
this is represented as
\beann
\partial^k p(\zeta) &=& \sum_{n \geq 0} (n)_k (p(\zeta), \zeta^n) \zeta^{n-k}
\eeann
where $(n)_k = n(n-1)\cdots (n-k+1)$. Acting on the basis element $p(\zeta) = \zeta^n$, this becomes the linear functional
\beann
\partial^k \zeta^n &=& \left\{ \begin{array}{cc} (n)_k \zeta^{n-k} & 0 \leq k \leq n \cr 0 & k > n\end{array} \right.
\eeann
One may use this to introduce another role for the algebra of generating functions: take $t^k$ to denote the symbol for $\partial^k$, so that
\beann
t^k \zeta^n &=& \left\{ \begin{array}{cc} (n)_k \zeta^{n-k} & 0 \leq k \leq n \cr 0 & k > n\end{array} \right.
\eeann
Then one may extend this to define a generating function 
\beann
f(t) &=& \sum_{k=0}^\infty \frac{a_k}{k!} t^k
\eeann
as a linear  operator on  $\mathbb{Q}[\zeta]$ by
\begin{eqnarray} \label{genderiv}
f(t) \zeta^n &=& \sum_{k = 0}^n {n \choose k} a_k \zeta^{n-k}.
\end{eqnarray}
We note that we are now using a generating function operationally in two different ways as
representing a linear functional and as representing a linear operator on $\mathbb{Q}[\zeta]$, or notationally as $\langle f(t) | p(\zeta) \rangle$ and $f(t) p(\zeta)$. This notational distinction and context will hopefully keep these different roles clear. The utility of these two different usages is given by the following result which intertwines them:
\begin{thm} \cite{Roman}
\beann
\langle f(t) g(t) | p(\zeta) \rangle &=& \langle g(t) | f(t) p(\zeta) \rangle.
\eeann
\end{thm} 
The product of generating functions here corresponds to composition of operators. An immediate consequence of the theorem is 
\beann
\langle f(t) | p(\zeta) \rangle &=& \langle t^0 | f(t) p(\zeta) \rangle;
\eeann
in other words, applying the functional $f(t)$ to $p(\zeta)$ is the same as applying the operator and then evaluating at $\zeta =0$.
This operator representation gives us another way to approach the bosonic operators introduced in section \ref{bosestrings}. One has
\beann
e^{\sigma t} \zeta^n = \sum_{k=0}^\infty t^k \zeta^n =  \sum_{k=0}^n {n \choose k} \sigma^k \zeta^{n-k} = (\zeta + \sigma)^n.
\eeann
By linearity it follows that 
\beann
e^{\sigma t} p(\zeta) &=& p(\zeta+ \sigma)
\eeann
Since $t$ here stands for $\partial = \partial_\zeta$, this coincides with equation (\ref{boson1}) for polynomials and will therefore extend to anlaytic functions by  the Weierstrass approximation theorem (which applies in the case of our generating functions for map enumeration).

Finally note that one can take $\frac{t^k}{k!}$ as a dual basis in the binomial Hopf algebra; indeed,
\begin{eqnarray*}
\langle t^k | \zeta^n \rangle &=& n! \delta_{k,n}.
\end{eqnarray*}
Moreover, now that the role of generating functions in defining linear operators has been established we will often adopt the notation $g(\partial)$ in place of $g(t)$ which will also help to distinguish the usage between operators and linear functionals. One sees that the form (\ref{genderiv}) may also be written as
\begin{eqnarray*}
g(\partial) \zeta^n &=& \sum_{j=0}^n {n \choose j} \langle g(t) | \zeta^j \rangle \cdot | \zeta^{n-j} \rangle\\
&=& \sum_{k=0}^n {n \choose k} a_{n-k} \zeta^k.
\end{eqnarray*}

\subsection{Appell Polynomials Generated by Reciprocal Bessel Functions} \label{AppellBessel}
From this point on we are going to restrict attention to invertible $g(t)$ (i.e., $g(0) \ne 0$). Now multiply each $\frac{t^k}{k!}$  by $g(t)$ to get the basis $g(t) \frac{t^k}{k!}$. It clearly has a dual basis of polynomials given by $s_n(\zeta) = g(\partial)^{-1} \zeta^n$. These are called the {\it Appell polynomials determined by $g(t)$}. We list here some of the more salient properties of the Appell polynomials.
\begin{eqnarray} \label{genfcn}
\sum_{k = 0}^\infty \frac{s_k(\zeta)}{k!} t^k &=& \frac1{g(t)} e^{\zeta t}\\ \label{binom}
s_n(\zeta + \sigma) &=& \sum_{k=0}^n {n \choose k} s_k(\sigma) \zeta^{n-k}\\ \label{diffrecur}
\partial s_n(\zeta) &=& n s_{n-1}(\zeta)\\
h(t) &=& \sum_{k = 0}^\infty \frac{\langle h(t) | s_k(\zeta) \rangle}{k!} g(t) t^k
\end{eqnarray}
\begin{eqnarray}
p(\zeta) &=& \sum_{k \geq 0} \frac{\langle g(t) |\partial^k p(\zeta) \rangle}{k!} g(t) s_k(\zeta)\\
\zeta s_n(\zeta) &=& s_{n+1}(\zeta) + \sum_{k=0}^n {n \choose k} \langle g^\prime(t) | s_{n-k}(\zeta) \rangle s_k(\zeta)\\ \label{scaling}
\langle h(t) | p(a \zeta) \rangle &=& \langle h(a t) | p(\zeta) \rangle\\ 
\langle h(t) | \zeta p(\zeta) \rangle &=& \langle h^\prime(t) | p(\zeta) \rangle.
\end{eqnarray}
There are many classical examples of Appell polynomials. Among these are the Hermite polynomials associated to $g(t) = e^{t^2/2}: H_n(\zeta) = e^{- \partial^2/2} \zeta^n$. 

Pertinent to this paper, we will now show that the sequence of polynomials, $\widehat{B}_{12}(\nu)$, indexed by $\nu$ and regarded as functions of $\zeta = \sqrt{y_0}$, that were introduced at the end of Section \ref{sec:44}, are the even elements of a sequence of Appell polynomials. A related correspondence will also be shown to hold for $\widehat{B}_{11}(\nu)$. The exponential generating function which generates the $\widehat{B}_{12}(\nu)$ is the reciprocal of the modified Bessel function of the first kind of order 0,
\begin{eqnarray*}
g(t) &=& \frac1{I_0(2t)}\\
I_0(2t) &=& \sum_{\mu=0}^\infty \frac{t^{2\mu}}{(\mu!)^2}\\
&=& \sum_{\mu=0}^\infty {2\mu \choose \mu}\frac{t^{2\mu}}{2 \mu!}\\
S_n(\zeta) &\doteq& I_0(2\partial) \zeta^n
\end{eqnarray*}
\begin{eqnarray} \label{s2n}
S_{2\nu}(\zeta) &=& \sum_{\mu = 0}^{\nu} {2\nu \choose 2\mu} {2(\nu - \mu) \choose \nu  - \mu} \zeta^{2\mu} = \sum_{\mu = 0}^{\nu} {2\nu \choose 2\mu, \nu - \mu, \nu - \mu} \zeta^{2\mu}\\
\label{s2n+1}
S_{2\nu + 1}(\zeta) &=& \sum_{\mu = 0}^{\nu} {2\nu + 1 \choose 2\mu + 1} {2(\nu - \mu) \choose \nu - \mu}\zeta^{2\mu + 1} = \sum_{\mu = 0}^{\nu} {2\nu + 1\choose 2\mu + 1, \nu - \mu, \nu - \mu} \zeta^{2\mu + 1}\\ \label{b12}
\widehat{B}_{12}(\nu) &=& S_{2\nu}(\sqrt{y_0}).
\end{eqnarray}

Similarly $\widehat{B}_{11}(\nu)$ is associated to the reciprocal of the modified Bessel function of order 1.
\begin{eqnarray*}
g(t) &=& \frac{t}{I_1(2t)}\\
I_1(2t) &=& \frac12 \frac{d}{dt} I_0(2t) = \frac12 \sum_{\mu=1}^\infty \frac{ 2\mu\,\, t^{2\mu - 1}}{(\mu!)^2} = \sum_{\mu=1}^\infty \frac{t^{2\mu - 1}}{(\mu -1)! \mu!}\\
\frac{I_1(2t)}{t} &=&   \sum_{\mu=0}^\infty \frac1{2\mu + 1} {2\mu +1 \choose \mu}\frac{t^{2\mu}}{2 \mu !}\\
R_n(\zeta) &\doteq& I_1(2\partial) \partial^{-1} \zeta^n  \\
\partial R_{2\nu}(\zeta) &=& \partial \sum_{\mu = 0}^{\nu} {2\nu \choose 2\mu } \frac1{2(\nu - \mu) + 1 } {2(\nu - \mu) + 1 \choose \nu  - \mu} \zeta^{2\mu} = \sum_{\mu = 1}^{\nu} {2\nu \choose 2\mu } \frac{2 \mu}{2(\nu - \mu) + 1 } {2(\nu - \mu) + 1 \choose \nu  - \mu} \zeta^{2\mu - 1} \\ 
&=& \sum_{\mu = 1}^{\nu} {2\nu \choose 2\mu - 1, \nu - \mu, \nu - \mu + 1} \zeta^{2\mu - 1}\\
&=& 2 \nu R_{2\nu - 1}(\zeta)
\end{eqnarray*}
\begin{eqnarray} \nonumber
\partial R_{2\nu + 1}(\zeta) &=& \partial \sum_{\mu = 0}^{\nu} {2\nu +1\choose 2\mu + 1 } \frac1{2(\nu - \mu) + 1 } {2(\nu - \mu) + 1 \choose \nu  - \mu}\zeta^{2\mu + 1}\\ \nonumber
 & =&  \sum_{\mu = 0}^{\nu} {2\nu + 1 \choose 2\mu + 1} \frac{2 \mu + 1}{2(\nu - \mu) + 1 } {2(\nu - \mu) + 1 \choose \nu  - \mu} \zeta^{2\mu } \\ \label{r2n}
&=& \sum_{\mu = 0}^{\nu} {2\nu + 1\choose 2\mu, \nu - \mu, \nu - \mu + 1}\zeta^{2\mu}\\ \nonumber
&=& (2 \nu + 1) R_{2 \nu}(\zeta)\\ \label{b11}
\widehat{B}_{11}(\nu) &=& \partial R_{2\nu}(\sqrt{y_0}) = 2\nu R_{2\nu - 1}(\sqrt{y_0})\\ \nonumber
(2\nu+1)\widehat{\mathcal{F}}_{1}(\nu) &=&  R_{2\nu}(\sqrt{y_0}).
\end{eqnarray}

\begin{rem} In \cite{EMP08} and  \cite{Er09} it was observed that for even valence, the key equations at leading order have a fundamental connection to the polynomial relations satisfied by $z_0$ as a consequence of its characterization as a generating function for Catalan numbers. Something like that is going on with $\widehat{B}_{11}(\nu) $ and $\widehat{B}_{12}(\nu) $ in the case of odd valence as well, although it is a bit more intricate to observe. Applying the scaling relation (\ref{scaling}) one has that $\widehat{B}_{12}(-y_0)$ (which is relevant for studying the behavior of $\widehat{B}_{12}(y_0)$ along the negative real axis) corresponds to the generating function
$I_0(-2 i t) = J_0(2t)$ where $J_0(t)$ is the order zero Bessel function of the first kind.   
The Bessel function of order $n$ (respectively modified Bessel funciton of order $n$) satisfies the second order linear ode
\begin{eqnarray} \label{Bess}
t^2 Y^{\prime \prime} + t Y^\prime + (t^2 - n^2) Y &=& 0\\ \label{MBess}
t^2 Y^{\prime \prime} + t Y^\prime - (t^2 + n^2) Y &=& 0.
\end{eqnarray}
Applying the Fourier transform, $\mathcal{F}$, to these ode's transforms them to algebraic equations which are readily solved. One may then determine that the Fourier transform of $J_0$ satisfies
\begin{eqnarray} \label{FourierBessel}
\frac12 \mathcal{F}[J_0(t)](2\sqrt{w}) &=& (1 - 4w)^{-1/2}\\ \nonumber
&=& (1 - w \partial_w) C(w)\\ \nonumber
&=& f_{0 w}(t_4 = 1) \,\,\, \mbox{when the valence equals 4}
\end{eqnarray}
where $C(w)$ is the generating function for the standard Catalan numbers. 
\end{rem}
\bigskip

We conclude this subsection with a derivation of some striking properties of $\widehat{B}_{12}$ and $\widehat{B}_{11}$ and relations between them. First, note that

\begin{lem}
\begin{eqnarray} \label{AppDiff}
\widehat{B}_{12} - y_0^{1/2} \widehat{B}_{11} &=& S_{2\nu}(\sqrt{y_0}) - \sqrt{y_0} \partial_\zeta R_{2\nu}(\sqrt{y_0})
\end{eqnarray}
has a corresponding generating function which is not invertible. Hence, this sequence of polynomials cannot be Appell. 
\end{lem}
\smallskip

\begin{proof}
From our prior derivations we have
\begin{eqnarray*}
\widehat{B}_{12} - y_0^{1/2} \widehat{B}_{11} &=& S_{2\nu}(\sqrt{y_0}) - \sqrt{y_0} \partial_\zeta R_{2\nu}(\sqrt{y_0})
\end{eqnarray*}
Applying the umbral calculus shows that the polynomials on the RHS correspond to the generating function
\begin{eqnarray*}
I_0(2t) - \partial_t t \frac{I_1(2t)}{t} &=& I_0(2t) - \partial_t I_1(2t)\\
&=& \sum_{\mu=0}^\infty {2\mu \choose \mu}\frac{t^{2\mu}}{2 \mu!} - \sum_{\mu=0}^\infty {2\mu + 1\choose \mu}\frac{t^{2\mu}}{2 \mu!}\\
&=& - \sum_{\mu=1}^\infty {2\mu \choose \mu} \frac{\mu}{\mu + 1} \frac{t^{2\mu}}{2 \mu!}
\end{eqnarray*}
which is not invertible.
\end{proof}
\bigskip

We next turn to an important characterization of the roots of the Appell polynomials.

\begin{prop} \label{mainprop}
$S_n(\zeta)$, viewed as a function of a complex $\zeta$ variable, has exactly $n$ zeroes (which must therefore be simple) along the imaginary axis in the complex $\zeta$ plane which are symmetric about the origin and the zeroes of $S_{n-1}(\zeta)$ interlace those of $S_n(\zeta)$.  $\widehat{B}_{12}(\nu)$ is inductively given by 
\begin{eqnarray} \label{B12Rec}
\widehat{B}_{12}(\nu) &=& \partial^{-2} S_{2 \nu -2}(\sqrt{y_0}) + C_{2\nu}
\end{eqnarray}
where $C_{2\nu} = {2\nu \choose \nu}$ are the coefficients of the generating function $(1 - w \partial_w) C(w)$  ($C(w)$ here is the generating function for the standard Catalan numbers)  which in turn is related to the Fourier transform of the Bessel function of the first kind as in (\ref{FourierBessel}). Consequently, one further has that $\widehat{B}_{12}(\nu)(y_0)$ has $\nu$ negative real zeroes as a function of $y_0$.
\end{prop}

We defer the proof to Appendix \ref{Appellzeroes}

There is an entirely analogous result, with analogous proof, concerning the zeros of $y_0^{1/2}\widehat{B}_{11}(\nu)$:
\begin{prop}
$R_n(\zeta)$ has exactly $n$ zeroes (which must therefore be simple) along the imaginary axis in the complex $\zeta$ plane which are symmetric about the origin and the zeroes of $R_{n-1}(\zeta)$ interlace those of $R_n(\zeta)$.  $\widehat{B}_{11}(\nu)$ is explicitly given by 
\begin{eqnarray} \label{B11Rec}
\widehat{B}_{11}(\nu) &=& \partial^{-1} R_{2 \nu -2}(\sqrt{y_0}) 
\end{eqnarray}
with constant of integration set to zero.  Consequently, one has that $y_0^{1/2}\widehat{B}_{11}(\nu)(y_0)$ has $\nu$ negative real zeroes as a function of $y_0$.
\end{prop}
\begin{rem}
The previous two propositions establish properties for the two respective sequences of Appell polynomials that are standard properties of orthogonal polynomials; namely, that their zeroes are distinct and real and the zeroes of successive polynomials in each sequence interlace. For orthogonal polynomials this is a direct consequence of the fact that the recurrence formulae for the polynomials is a symmetric tri-diagonal matrix \cite{Deift}. There is an extension of this to a subclass of Appell polynomials known as multiple orthogonal polynomials (the term {\it multiple} here refers to the fact that these polynomials are orthogonal with respect to a multiple collection of measures on the line rather than just one). In these cases the recurrence matrix is of bounded bandwidth (or {\it height}) about the diagonal independent of $n$.  (See \cite{VA} for more details and \cite{EM01} for an example related to random matrix theory.) However, examination of the recurrence formula in the cases of $S_n$ and $R_n$ shows the recurrence matrices here are far from being bandwidth bounded. So this appears to be a genuinely new type of zero-interlacing result within the class of Appell polynomials.
\end{rem}

\begin{rem}
The polynomials $\widehat{B}_{12}(y_0)$ and $y_0^{1/2}\widehat{B}_{11}(y_0)$ can be rewritten in the form 
\begin{eqnarray} \label{trigonometric}
\zeta^\nu + a_1 \zeta^{\nu -1} + \cdots + a_{\nu - 1} \zeta + c + b_{\nu -1} \zeta^{-1} + \cdots + b_1 \zeta^{-\nu - 1} + b_0 \zeta^{- \nu}
\end{eqnarray}
by multiplying through by $\zeta^{-\nu}$ where $\zeta = \sqrt{y_0}$.
In this form they are referred to as {\it trigonometric polynomials} which, thought of as mappings from $\mathbb{P}^1$ 
to $\mathbb{P}^1$, have a point in the image ($\infty$) with only two preimages (poles). The topological classification of such mappings and description of their combinatorial significance was carried out by Arnold and his school in a series of papers (see \cite{Arnold}).
\end{rem}

\subsection{String Polynomials and Appell Polynomials} \label{StringApple}
In the previous subsections we have seen connections between the Appell polynomials $S_n$ and $R_n$ and certain string polynomials, $\phi_0$ and $\psi_0$. In this subsection we  will develop this connection systematically and in fuller generality. We first recall the well-known generating function for modified Bessel functions of order $n, I_n(t)$, which satisfy (\ref{MBess}) and are unbounded at $+ \infty$:
\begin{eqnarray} \label{BessGen}
e^{X (\eta + \eta^{-1})} &=& \sum_{n = - \infty}^\infty I_n(2X) \eta^n.
\end{eqnarray}
We frist show how to relate the string polynomials to the Bessel functions $I_0$ and $I_1$ using the definitions (\ref{phim} - \ref{psim})

\begin{eqnarray} \nonumber
\psi_m 
&=&  [\eta^{-1}] \sqrt{f_0}  (j)_{m+1} \left(\sqrt{f_0}\eta + h_0 +\sqrt{f_0}\eta^{-1}\right)^{j-m-1} \\
\nonumber
&=& \sqrt{f_0} j!  [s^{j-1-m}]\,[\eta^1] \sum_{\ell\geq 0} \frac{s^{\ell}}{\ell !}\left(\sqrt{f_0}\eta + h_0 +\sqrt{f_0}\eta^{-1}\right)^{\ell} \\ \nonumber
&=& \sqrt{f_0} j! [s^{j-1-m}]\left(e^{sh_0}\,[\eta^1]  e^{s\sqrt{f_0} (\eta +\eta^{-1})}\right) \\
\label{sys1}
&=& \sqrt{f_0} j!  [s^{j-1-m}] e^{s{h_0}}I_1(2s\sqrt{f_0}) \label{asdf3434}\\ 
\label{sys2}
\phi_m &=& j!  [s^{j-1-m}] e^{s{h_0}}I_0(2s\sqrt{f_0})
\end{eqnarray}
where in the last line of the derivation for $\psi_m$ we have made use of (\ref{BessGen}).
The derivation for $\phi_m$ is entirely similar. {\it We note that the variable $s$ introduced here is an independent variable; in particular, it has no relation to the variable introduced in (\ref{nuscaling2})}.

We introduce an extension of the string polynomials that will prove quite useful in Section \ref{toprec}:
\begin{eqnarray} \label{PHI}
\Phi_n &\doteq&  I_n(2s\sqrt{f_0}) e^{sh_0}\\
\psi_n &=&  \sqrt{f_0} j!  [s^{j-1-n}] \Phi_1\\
\phi_n &=& j!  [s^{j-1-n}] \Phi_0 
\end{eqnarray}

The string polynomials can be directly related to our Appell polynomials. This is seen in the following direct calculation where we make use of the change of variables from the coordinates
$(h_0, \sqrt{f_0})$ to the characteristic coordinates $(y_0, \sqrt{f_0})$ where $y_0 = h_0/ \sqrt{f_0}$ and with $j = 2\nu+1$:

\begin{eqnarray*}
\phi_{m} &=& f_0^{\frac{2\nu - m}{2}}  [\eta^0] (2\nu+1)_{m+1} \left(\eta + \sqrt{y_0} + \eta^{-1}\right)^{2\nu-m}\\
&=& (2\nu+1)_{m+1} f_0^{\frac{2\nu - m}{2}} \sum_{\mu \geq 0} {2\nu - m \choose 2(\nu - \mu) - m, \mu, \mu} \left(\eta\right)^\mu \left(\eta^{-1} \right)^\mu (\sqrt{y_0})^{2(\nu -\mu) -m}\\
&=& (2\nu+1)_{m+1} f_0^{\frac{2\nu - m}{2}} \sum_{\mu \geq 0} {2\nu - m \choose 2(\nu - \mu) - m, \mu, \mu}  (\sqrt{y_0})^{2(\nu -\mu) -m}\\
&=& (2\nu+1)_{m+1} f_0^{\frac{2\nu - m}{2}} S_{2 \nu -m}(\sqrt{y_0})\\
&=& f_0^{\frac{2\nu - m}{2}} \partial^{m+1}S_{2 \nu +1}(\zeta)\Big|_{\zeta = \sqrt{y_0}}\\
&=& f_0^{-\frac{m+1}{2}} (\partial_{\sqrt{y_0}})^{m+1} \widehat{\mathcal{F}}_{2}(\sqrt{y_0})\\
\psi_{m} &=& f_0^{\frac{2\nu - m+1}{2}} [\eta^{-1}] (2\nu+1)_{m+1} \left(\eta + \sqrt{y_0} + \eta^{-1}\right)^{2\nu-m}\\
&=& (2\nu+1)_{m+1} f_0^{\frac{2\nu - m+1}{2}} \sum_{\mu \geq 0} {2\nu - m \choose 2(\nu - \mu) - m - 1, \mu, \mu+1}   (\sqrt{y_0})^{2(\nu -\mu) -m-1}\\
&=& (2\nu+1)_{m+1}  f_0^{\frac{2\nu - m + 1}{2}} \partial R_{(2\nu - m)}(\sqrt{y_0}) \\
&=& f_0^{\frac{2\nu - m+1}{2}}\partial^{m+2} R_{2\nu +1}(\zeta)\Big|_{\zeta = \sqrt{y_0}}\\
&=& f_0^{-\frac{m+1}{2}}(\partial_{\sqrt{y_0}})^{m+1} \widehat{\mathcal{F}}_{1}(\sqrt{y_0}).
\end{eqnarray*}
In particular, one has:

\begin{eqnarray} \label{s2nu+1}
\phi_{-1} &=& f_0^{\frac{2\nu + 1}{2}} S_{2 \nu + 1}(\sqrt{y_0}) \\ \label{s2nu}
\phi_{0} &=& (2\nu + 1)f_0^{\nu} S_{2 \nu}(\sqrt{y_0}) \\ \label{dr2nu+1}
\psi_{-1} &=& f_0^{\frac{2\nu + 2}{2}} \partial R_{2 \nu + 1}(\sqrt{y_0})\\ \label{dr2nu}
\psi_{0} &=& (2\nu+1) f_0^{\frac{2\nu + 1}{2}} \partial R_{2 \nu}(\sqrt{y_0})
\end{eqnarray}

The following proposition summarises the above calculation and also relates it to our earlier defined fluxes and flux gradients.
\begin{prop} \label{prop01} With $y_0 = \frac{h_0^2}{f_0}$  one has for $j = 2 \nu +1 $,
\begin{eqnarray*}
\left(\begin{array}{c}
\phi_m \\ \psi_m
\end{array}\right) &=& (2 \nu +1)_{m + 1}\left(\begin{array}{c}
f_0^{\frac{2\nu -m}{2}} S_{(2\nu - m)}(\sqrt{y_0}) \\ f_0^{\frac{2\nu - m + 1}{2}} \partial R_{(2\nu - m)}(\sqrt{y_0})
\end{array}\right) \\
&=& (2 \nu +1)_{m + 1}\left(\begin{array}{c}
\mathcal{F}_{2}(\nu - \frac{m+1}{2}) \\ \mathcal{F}_{1}(\nu - \frac{m+1}{2}) 
\end{array}\right) \qquad {\rm for}\,\,\,  m \,\,\, {\rm odd}\\
&=& (2 \nu +1)_{m + 1}\left(\begin{array}{c}
B_{12}(\nu - \frac{m}{2}) \\ B_{11}(\nu - \frac{m}{2}) 
\end{array}\right) \qquad {\rm for}\,\,\,  m \,\,\, {\rm even}
\end{eqnarray*} 
where $(j)_m = j \cdot (j-1) \cdots (j - m + 1) $ is the descending factorial.  
\end{prop}
\bigskip

This can also be extended to the case of even valence:

\begin{prop} \label{prop01A}  
For $j = 2\nu \,\,  (\nu > 1)$, $h_0 \equiv 0$ and $m$ even,
\begin{eqnarray*}
\left(\begin{array}{c}
\phi_m \\ \psi_m
\end{array}\right)_{j = 2\nu}&=&  (2\nu)_{m+1} \left\{\begin{array}{c}  { 2\nu - m \choose \nu - \frac{m}{2}} f_0^{\nu - \frac{m}{2}}  \\ 0  \end{array} \right. ;
\end{eqnarray*}
and for $j = 2\nu \,\,  (\nu > 1)$, $h_0 \equiv 0$ and $m$ odd,
\begin{eqnarray*}
\left(\begin{array}{c}
\phi_m \\ \psi_m
\end{array}\right)_{j = 2\nu} &=&   (2\nu)_{m+1} \left\{\begin{array}{c} 0 \\ { 2\nu - m \choose \nu - \frac{m - 1}{2}} f_0^{\nu - \frac{m - 1}{2}}    \end{array} \right. .
\end{eqnarray*}
\end{prop}

\subsection{Unwinding Identity} \label{Unwinding}
Finally we examine how the string polynomials transform under differentiation with respect to $x$ and derive a key identity for future calculations. Here $h_0$ and $f_0$ depend on $x$ as specified in (\ref{f-coeffs} - \ref{h-coeffs}) and $^\prime = \partial_x$.
\begin{prop} (Unwinding Identity) \label{prop05}
\begin{eqnarray*}
\partial_x \left(
\begin{array}{c}
\phi_{m-1}\\ 
\psi_{m-1}
\end{array}
\right) &=& \left(\begin{array}{cc}
h'_0 & f'_0/f_0\\ 
f'_0 & h'_0
\end{array}\right)  \left(\begin{array}{c}
\phi_{m}\\ 
\psi_{m}
\end{array}\right)
\end{eqnarray*}
\end{prop}
\smallskip

\noindent{\bf Proof:}
We will use the following identities which are easily obtained from the Taylor series for Bessel functions:
\begin{eqnarray*}
I'_0(2X) &=& I_1(2X) \\
I'_1(2X) &=& I_0(2X)-\frac{1}{2X} I_1(2X)
\end{eqnarray*}
Now the calculation:
\begin{eqnarray*}
\partial_x \left(
\begin{array}{c}
\phi_{m-1}\\ 
\psi_{m-1}
\end{array}
\right) =  \sum_j j!\,\widetilde{t}_j [s^{j-m}] \partial_x \left(
\begin{array}{c}
e^{sh_0}I_0(2s\sqrt{f_0})\\ 
\sqrt{f_0}e^{sh_0}I_1(2s\sqrt{f_0})
\end{array}
\right)\\
=\sum_j j!\,\widetilde{t}_j [s^{j-m}] \left(
\begin{array}{c}
sh'_0 e^{sh_0}I_0(2s\sqrt{f_0}) +e^{sh_0} \frac{sf'_0}{\sqrt{f_0}}I'_0(2s\sqrt{f_0})\\ 
\frac{f'_0}{2\sqrt{f_0}} e^{sh_0}I_1(2s\sqrt{f_0})+\sqrt{f_0} sh'_0 e^{sh_0}I_1(2s\sqrt{f_0})+\sqrt{f_0}e^{sh_0}\frac{sf'_0}{\sqrt{f_0}} I'_1(2s\sqrt{f_0})
\end{array}
\right)
\nonumber\\
= \sum_j j!\,\widetilde{t}_j [s^{j-m}]
s \left(
\begin{array}{c}
h'_0 e^{sh_0}I_0(2s\sqrt{f_0}) +\frac{f'_0}{f_0} \sqrt{f_0}e^{sh_0} I_1(2s\sqrt{f_0}) \\ 
\sqrt{f_0} h'_0 e^{sh_0}I_1(2s\sqrt{f_0})+\sqrt{f_0}e^{sh_0}\frac{f'_0}{\sqrt{f_0}} I_0(2s\sqrt{f_0})
\end{array}
\right)\\
= \sum_j j!\,\widetilde{t}_j [s^{j-m-1}]
\left(
\begin{array}{cc}
h'_0 & f'_0/f_0 \\ 
f'_0 & h'_0
\end{array}
\right)
\left(
\begin{array}{c}
e^{sh_0}I_0(2s\sqrt{f_0})\\ 
\sqrt{f_0}e^{sh_0}I_1(2s\sqrt{f_0})
\end{array}
\right)\\
= \left(\begin{array}{cc}
h'_0 & f'_0/f_0\\ 
f'_0 & h'_0
\end{array}\right)  \left(\begin{array}{c}
\phi_{m}\\ 
\psi_{m}
\end{array}\right).
\end{eqnarray*}

\section{Characteristic Geometry } \label{chargeom}
In the proof of Lemma \ref{lem:extdiff} we observed that the system (\ref{extdiff}) determines an integral surface and presents it, locally, as a graph over the $(x,t)$-plane away from the caustics contained in the locus defined by $A_{11}^2 - f_0 A_{12}^2 = 0$. However one can take a more global point of view and define this integral surface as the zero set of the hodograph equations (\ref{GaussHod1} - \ref{GaussHod2}), in which case the surface may be smoothly coordinatized even in a neighborhood of the caustics and so is, in fact, a two-dimensional manifold.

Given this integral surface's definition as the zero set of the hodograph equations, it is natural to view our results within a commutative algebra framework. Our study of the leading order continuum equations in the previous subsection may be posed within the coordinate ring
\begin{eqnarray}\label{coordr1}
{\bf S} &=& \frac{\mathbb{Q}[t, x, h_0, f_0]}{\mathcal{I}}
\end{eqnarray} 
where $\mathcal{I}$ is the {\it string ideal} generated by the leading order continuum string (hodograph) equations, (\ref{GaussHod1} - \ref{GaussHod2}).
One may also think of 
${\bf S}$ as a module over the local ring $\mathbb{Q}[x,x^{-1}]$. This gives ${\bf S}$ a natural grading by powers of $x$. With respect to the self-similar variable $\xi = tx^{\nu - 1/2}$ one may rewrite this ring as
\begin{eqnarray} \label{coordr2}
{\bf S}_x &=& \frac{\mathbb{Q}[x,x^{-1}][\xi, u_0, z_0]}{\mathcal{I}_x}
\end{eqnarray}
where $\mathcal{I}_x$ is the ideal generated by
\begin{eqnarray} \label{tildestring}
1-z_0 &=& (2\nu + 1) \xi \tilde{B}_{11} = \xi \tilde{\psi}_0\\  \nonumber
-u_0 &=& (2\nu + 1) \xi \tilde{B}_{12} = \xi \tilde{\phi}_0,  
\end{eqnarray}
where $\tilde{B}_{11} = x^{-(\nu + 1/2)}B_{11}$ and $\tilde{B}_{12} = x^{-\nu }B_{12}$,
which is manifestly contained in the level zero component of the grading. The generating functions $h_g, f_g$ and $E_g$ are elements of the function field of ${\bf S}_x$ which we will denote ${\bf K}$ . The grading on ${\bf S}_x$ extends to those elements of ${\bf K}$; in particular, the generating functions respectively lie in unique graded components of non-positive weight.

With respect to this structure, the integral surface $\mathcal{S}$ associated to ${\bf S}_x$ may be viewed as a family of algebraic curves over a base $\mathbb{C}^*$ which has coordinate $x$. Differentiation with respect to $w$ (along the base) corresponds automorphically to an affine differentiation with respect to $\xi$ (along the fiber). The precise form of this correspondence on our generating functions is given by the {\it exchange relations},
\begin{eqnarray*}
x\partial_x f_g &=& (1-2g)f_g + \frac{j-2}{2} \xi\partial_\xi f_g\\
x\partial_x h_g &=& (1/2-g)h_g + \frac{j-2}{2} \xi\partial_\xi h_g
\end{eqnarray*}
which is an immediate consequence of their self-similar structure (\ref{f-coeffs}, \ref{h-coeffs}).
This viewpoint will prove to be useful in later sections. We shall see, for instance, that  $E_g(t,x) = x^{2-2g}e_g(\xi)$ lies in the graded component of level $2-2g$ so that differentiation of this element with respect to $x$ lowers this degree by 1.

The reduction of the string and Toda equations to Riemann invariant form, (\ref{wrpm}) and (\ref{RinvtToda2}), suggests a corresponding reduction of the characteristic geometry. The natural variable here, again, is
\begin{eqnarray*}
y_0 &=& \frac{h_0^2}{f_0}\\
&=& \frac{u_0^2}{z_0}
\end{eqnarray*}
with respect to which the string relations combine to give the single relation 
\begin{eqnarray} \label{curve}
\xi^2 &=& \frac{1}{(2\nu+1)^2} y_0 \frac{(\widehat{B}_{12} - y_0^{1/2}\widehat{B}_{11})^{2\nu-1}}{\widehat{B}_{12}^{2\nu+1}},
\end{eqnarray}
which can be related to both the conservation law fluxes as well as the Bessel-Appell polynomials (see section \ref{AppellBessel}): 
\begin{eqnarray}
\widehat{B}_{11} &=& 2 y_0^{1/2} \partial_{y_0} \widehat{\mathcal{F}}_1 
= \partial R_{2\nu}(\sqrt{y_0})\\ \label{B-identity}
\widehat{B}_{12} &=& (\nu+1)\widehat{\mathcal{F}}_1 - y_0 \partial_{y_0} \widehat{\mathcal{F}}_1
= S_{2\nu}(\sqrt{y_0})\\
\widehat{\mathcal{F}}_1 &=& \frac1{2\nu+1} \sum_{\mu = 0}^\nu {2\nu + 1 \choose 2\mu, \nu-\mu, \nu - \mu +1} y_0^\mu =  \frac1{2\nu+1} R_{2\nu}(\sqrt{y_0}).
\end{eqnarray}
For example in the trivalent case this becomes
\begin{eqnarray} \label{3curve}
\xi^2 &=& \frac1 9\,\frac {y_0 \left( 2-y_0 \right) }{ \left( 2+y_0 \right) ^{3}}.
\end{eqnarray}
We will also make use of the following identities which are deducible from the string equations (\ref{tildestring}),
\begin{eqnarray} \label{red1}
1 - \frac1{z_0} &=& \frac{y_0^{1/2} \widehat{B}_{11}}{\widehat{B}_{12}}\\ \label{red2}
\widehat{d}_\pm &=& \left(2 - j \left(1 - \frac1{z_0}\right)\right) \pm (j-2) \sqrt{y_0} \\ \label{red3}
\widehat{D} &=& \widehat{d}_+ \widehat{d}_- .
\end{eqnarray}
We will see that the $e_g$, for $g \geq 2$,  are in fact defined in terms of rational functions in the funciton field of the coordinate ring
\begin{eqnarray*}
{\bf \widehat{S}} &=& \frac{\mathbb{Q}[\xi^2,{y_0}]}{\widehat{\mathcal{I}}}.
\end{eqnarray*} 
Here, $\widehat{\mathcal{I}}$ is the principal ideal generated by (\ref{curve}). ${\bf \widehat{S}} $ is the coordinate ring of a rational algebraic curve $\mathcal{C}$ which is in 1:1 correspondence with any of the $x$-slices of $\mathcal{S}$. By {\it rational curve} here we mean that $\mathcal{C}$ is isomorphic to the projective line, $\mathbb{P}^1$. This fact is immediate from  (\ref{curve}) which presents $\mathcal{C}$ as the graph of a function, $\xi^2$ of $y_0$; i.e., $y_0$ is a global uniformizing parameter for $\mathcal{C}$.

For future use we also define 
\begin{eqnarray} \label{tildeS}
{ \bf \tilde{S}} &=& \frac{\mathbb{Q}[x, x^{-1}][\xi^2,{y_0}]}{\widehat{\mathcal{I}}}
\end{eqnarray}
which is the coordinate ring of $\mathcal{C}$ base-extended over the $x$-line, thus incorporating all $x$-slices in one coordinate ring.

Fig. \ref{tricurvebranch} describes aspects of this curve in the trivalent case. Part (a) of this figure shows the branching of the curve ($x=1$ slice of the integral surface) occuring at $(\xi^2, y_0) = (-\frac1{108\sqrt{3}}, 2(2+\sqrt{3}))$ and $(\frac1{108\sqrt{3}}, 2(2-\sqrt{3}))$. These correspond respectively to the caustics $d_\pm = 0$. Part (b) shows how the $\xi^2$ coordinate of the caustic scales with change of the slice location $x$, for positive, real values of $x$. Under this change the form of the curve remains essentially the same as shown in (a), the only difference being a self-similar scaling.  

\begin{figure}[h] 
\begin{center}
\resizebox{2in}{!}{\includegraphics{tricurvebranch-eps-converted-to.pdf}}
\resizebox{2in}{!}{\includegraphics{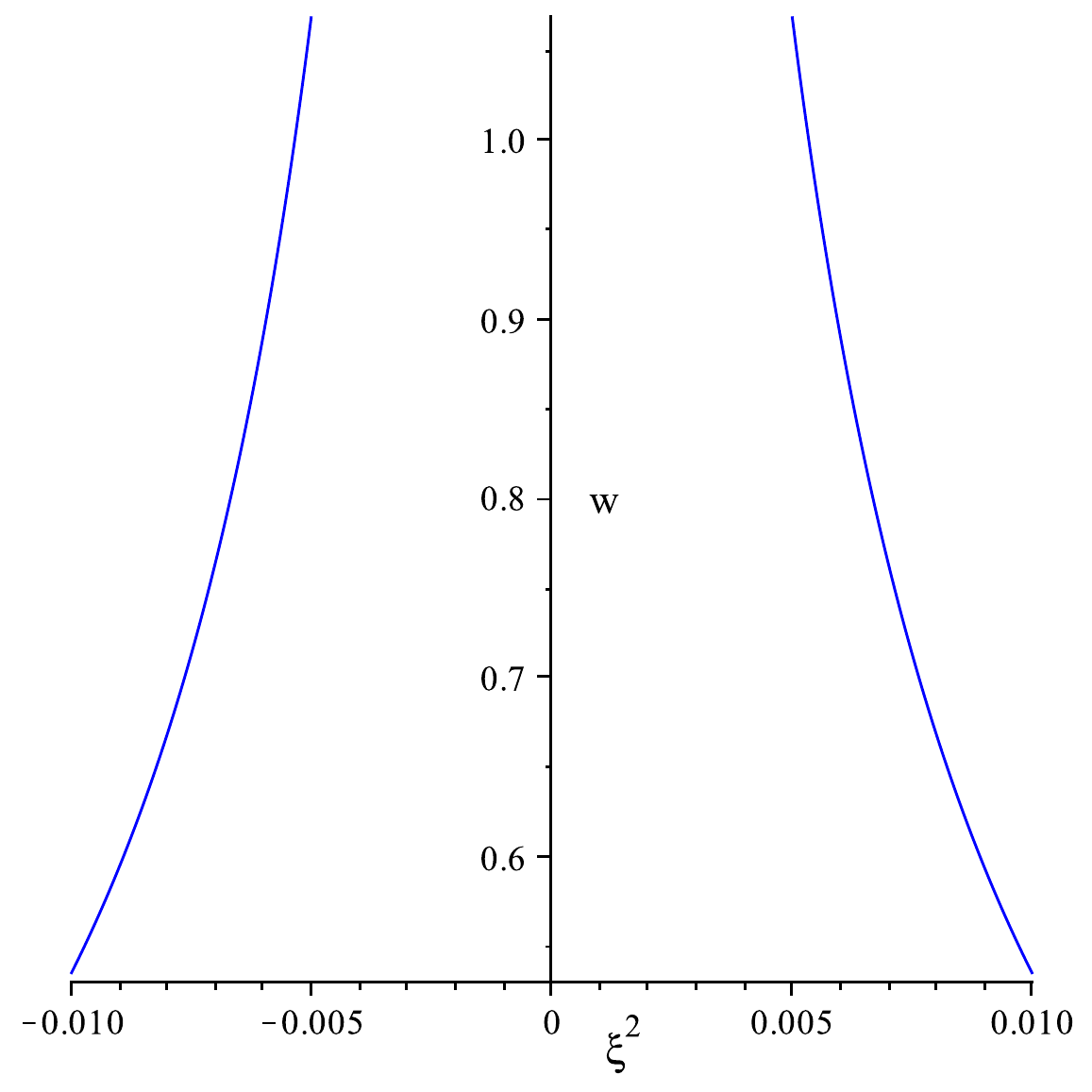}}
\caption{\label{tricurvebranch} a) Branches near $(\xi^2, y_0) = (0,0)\,\,\, (x=1)$ \,\,\, b) Branch point locations in $(w(=x), \xi^2)$}
\end{center}
\end{figure}

The geometry of the curve, (\ref{curve}), would seem to get much more complicated as the (odd) valence increases. However, a straightforward calculation reveals that, in general,
\begin{lem} \label{yderiv}
\begin{eqnarray} \label{crit}
y_0 \frac{d \xi^2}{d y_0} &=&  z_0 \xi^2 \left({A}^2_{11} - f_0 {A}^2_{12}\right)\\ \label{crit2}
\frac{d \xi^2}{d y_0} &=& -z_0 \xi^2 \frac{\widehat{D}}{y_0(y_0 - 4)}.
\end{eqnarray}
\end{lem}
\begin{proof}
Recall that the terms used on the RHS's of the above equations were defined in (\ref{A11}, \ref{A12}, \ref{red3}). Equation (\ref{crit2}) is equivalent to 
\begin{eqnarray*}
\xi y_0^\prime &=& - \frac{2 y_0  (y_0 - 4)  }{z_0 \widehat{D}}.
\end{eqnarray*}
which follows straightforwardly from (\ref{GaussTodaSS}) and the observation that
\begin{eqnarray} \nonumber
\xi y_0^\prime &=& \xi \frac{z_0 2u_0 u_0^\prime -u_0^2 z_0^\prime}{z_0^2}\\ \label{y0deriv}
&=& \frac{u_0}{z_0^2} \left(
\begin{array}{cc}
2z_{0}, & -u_{0}
\end{array}
\right) \cdot
\left(
\begin{array}{c}
\xi u_{0}^\prime \\ \xi z_{0}^\prime
\end{array}
\right).
\end{eqnarray}
Equation (\ref{crit}) then follows by (\ref{basic2}).
\end{proof}

\noindent We note that $(y_0 - 4) = 0$  if and only if ${r}_+ = 0 \,\, (y_0^{1/2} = - 2)$ or $ {r}_- =0 \,\,(y_0^{1/2} = 2)$.

For future use we also record here
\begin{eqnarray} \label{crit3}
\frac{dy_0}{d\xi^2} &=& - \frac{y_0 (y_0 - 4)}{z_0 \xi^2 \widehat{D}}\\ \label{crit4}
\frac{dz_0}{d\xi^2} &=& \frac{z_0 (y_0 - 4)}{8 \xi^2 \widehat{D}} \left[\left(1 - \frac1{z_0} + \sqrt{y_0}\right) \frac{\widehat{d}_+}{\sqrt{y_0} - 2} - \left(1 - \frac1{z_0} - \sqrt{y_0}\right) \frac{\widehat{d}_-}{\sqrt{y_0} + 2}\right]\\ \label{crit5}
\frac{dz_0}{dy_0} &=& - \frac{z_0^2}{8 y_0} \left[\left(1 - \frac1{z_0}\right)\left(\frac{\widehat{d}_+}{\sqrt{y_0} - 2} - \frac{\widehat{d}_-}{\sqrt{y_0} + 2}\right) + \sqrt{y_0} \left(\frac{\widehat{d}_+}{\sqrt{y_0} - 2} + \frac{\widehat{d}_-}{\sqrt{y_0} + 2}\right)\right]
\end{eqnarray}

\subsection{Branch Point Analysis}

The zeroes of the right hand side of (\ref{crit2}) are potential locations for branching on the integral surface, These include zeroes of $\widehat{D}$ among which we expect to find branching. The only other places where branching could occur are over $\xi^2 = 0$ or where $z_0 = 0$. However, from (\ref{curve}) we see that the zeroes of $\xi^2$ correspond to the vanishing of $(\widehat{B}_{12} - y_0^{1/2}\widehat{B}_{11})^{2\nu-1}$. If the vanishing of 
$\widehat{B}_{12} - y_0^{1/2}\widehat{B}_{11}$ is simple (as we will show to be the case in Corollary \ref{hypcurve}) then, since it comes with the odd power $2\nu -1$, such a point must be vertically inflectionary for $\nu > 1$ and so does not create a turning point. There are no positive, real zeroes of $z_0$ since, by (\ref{red1}),  these would correspond to zeroes of $\widehat{B}_{12}$ which is positive for positive $y_0$. There are also apparent singularities in (\ref{crit2}). One is at $y_0 = 0$; however, comparison with (\ref{curve}) shows that this is matched by a factor of $y_0$ appearing inside $\xi^2$ and so is removable. The apparent pole at zeores of $y_0 = 4$  could create a vertical tangent to the graph of $\xi^2$ (equivalently horizontal tangents to $\mathcal{C}$ over $\xi^2$); however, these may be canceled by zeroes of the numerator as is evident from the previous example. As mentioned,  $\widehat{B}_{12}$ appearing in $(\ref{curve})$ has no real zeroes for positive $y_0$ and so there are no further possibilities for horizontal tangents. So we see that, in general, genuine branching occurs only within the zeroes of $\widetilde{D}$. As mentioned, there may also be inflectionary points over $\xi^2 =0$. 

We illustrate this in the valence 5 case for which the curve has the equation
\begin{eqnarray} \label{5val}
\xi^2 &=& -{\frac {27}{25}}\,{\frac {y_0 \left( {y_0}^{2}-2 \right) ^{3}}{ \left( {y_0
}^{2}+12\,y_0+6 \right) ^{5}}}.
\end{eqnarray}
This is graphed in Figure \ref{5curve}.
\begin{figure}[h] 
\begin{center}
\resizebox{3in}{!}{\includegraphics{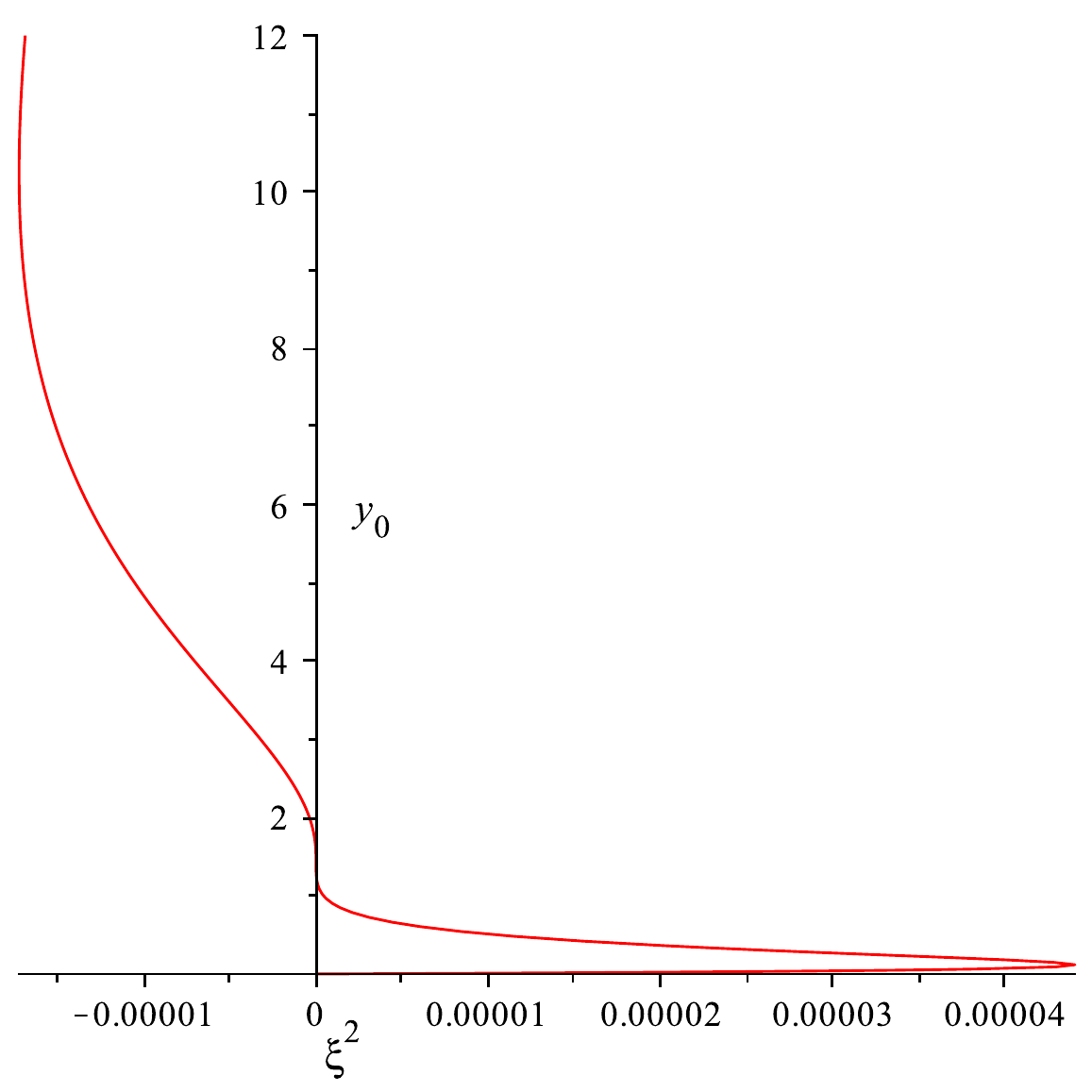}}
\end{center}
\caption{\label{5curve} Spectral curve for 5-valent maps}
\end{figure}
 Again one can discern two real branch points, one to the left and one to the right of $\xi^2 = 0$ interpolated by an inflection point at $y_0 = \sqrt{2}$ over $\xi^2 = 0$. Figure \ref{5curvebranch} shows an enlarged resolution of these two respective turning points. 
\begin{figure}[h] 
\begin{center}
\resizebox{2in}{!}{\includegraphics{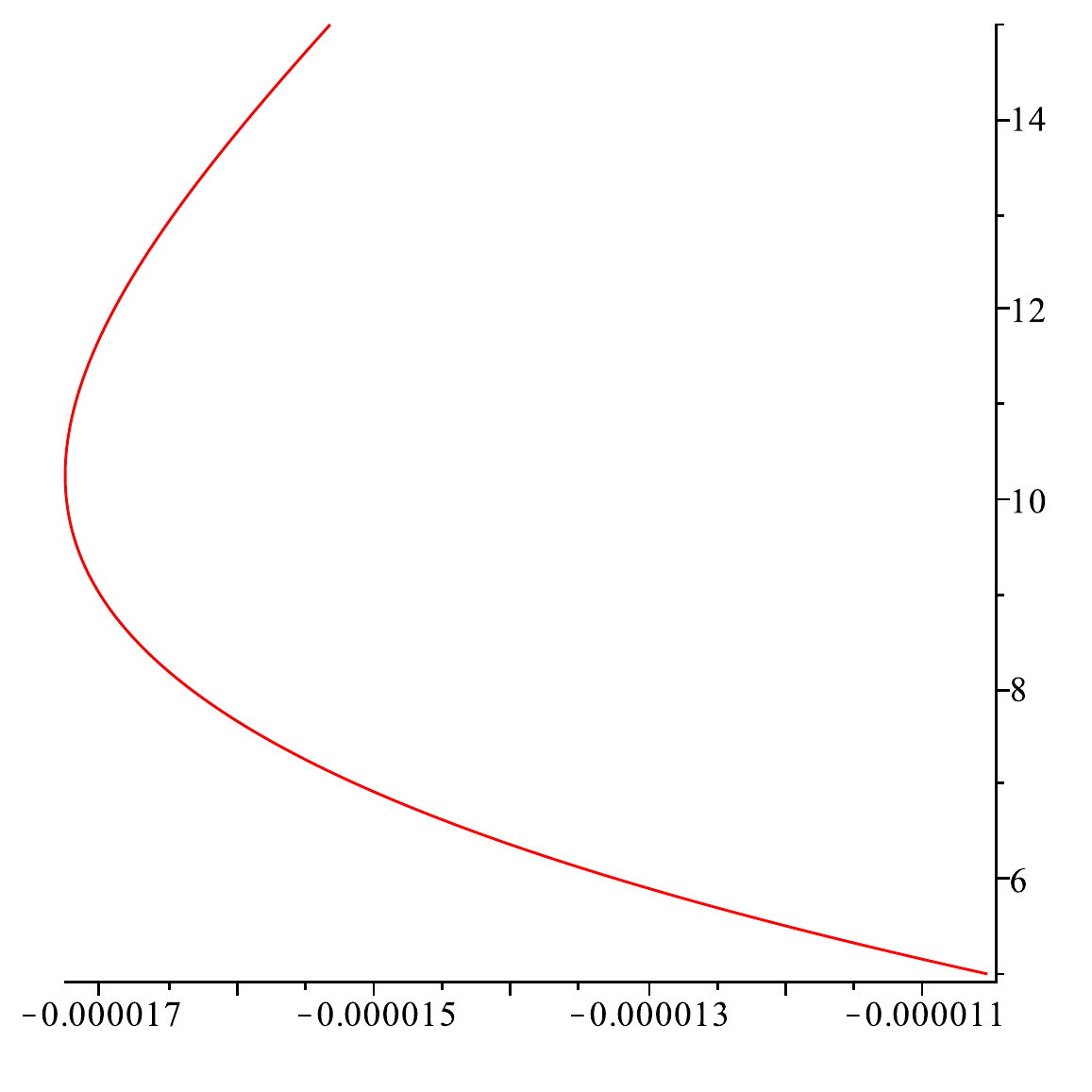}}
\resizebox{2in}{!}{\includegraphics{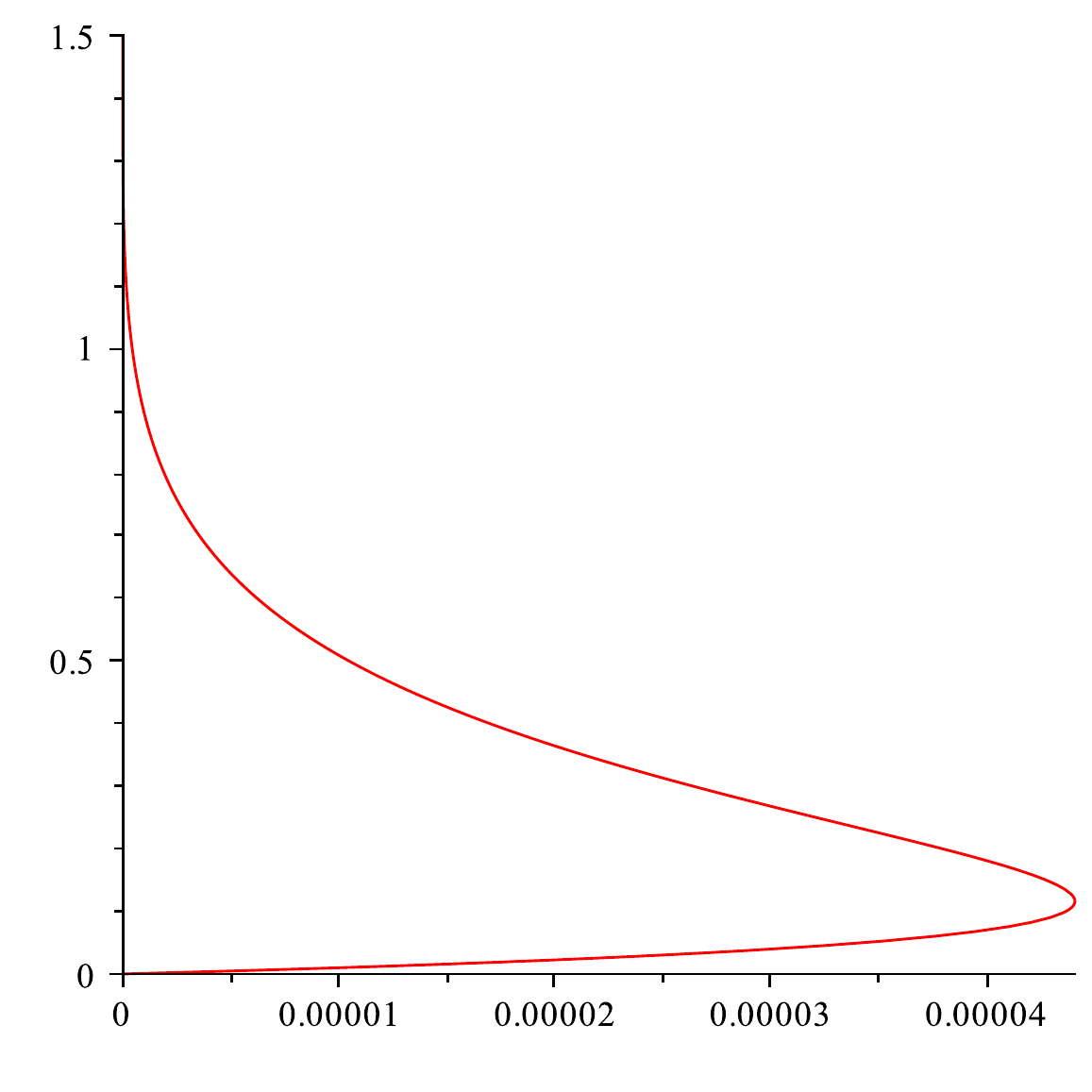}}
\caption{\label{5curvebranch} a) Left branch along $d_+ = 0$ \,\,\, b) Right branch along $d_- = 0$}
\end{center}
\end{figure}
The curve has another inflection point at $y_0 = - \sqrt{2}$ between two horizontal asymptotes at 
$y_0 = -6 + \sqrt{30}$ and $-6 - \sqrt{30}$ (Figure \ref{asympI-II} a). A real component of the curve passes through this inflection between the two asymptotes. Another real component lies below the asymptote $y_0 = -6 - \sqrt{30}$ (Figure \ref{asympI-II} b). In the extended real plane gotten by adding points over $\xi^2 = \infty$ these various "components" are connected to one another through the asymptotes so that the {\it real} curve  is topologically equivalent to an oval. There are no other turning points beyond the two real ones depicted in Figure \ref{5curvebranch}.  Note that again the potential horizontal tangent at $y_0 = 4$, was removable. The two real zeroes of $\widehat{D}/(4 - y_0) = 9 \frac{(y_0^4 - 8 y_0^3 - 20 y_0^2 - 32 y_0 + 4)}{({y_0}^{2}+12\,y_0+6) ^{2}}$ which correspond to the two places where the graph of this polynomial crosses the $y_0$-axis, shown in Figure \ref{Dhat}. It follows that the other two zeroes of this quartic polynomial are complex conjugate. 

\begin{figure}[h] 
\begin{center}
\resizebox{2in}{!}{\includegraphics{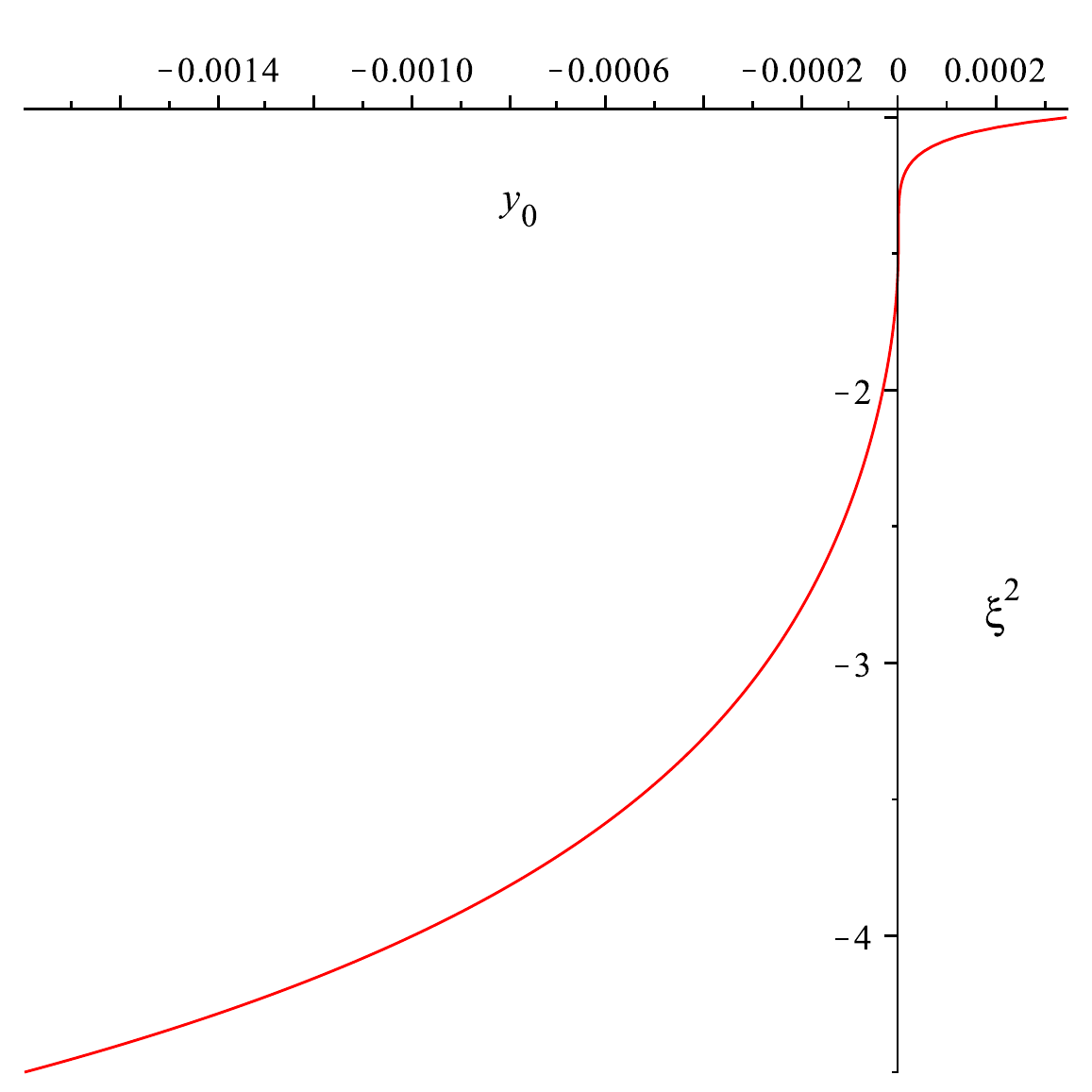}}
\resizebox{2in}{!}{\includegraphics{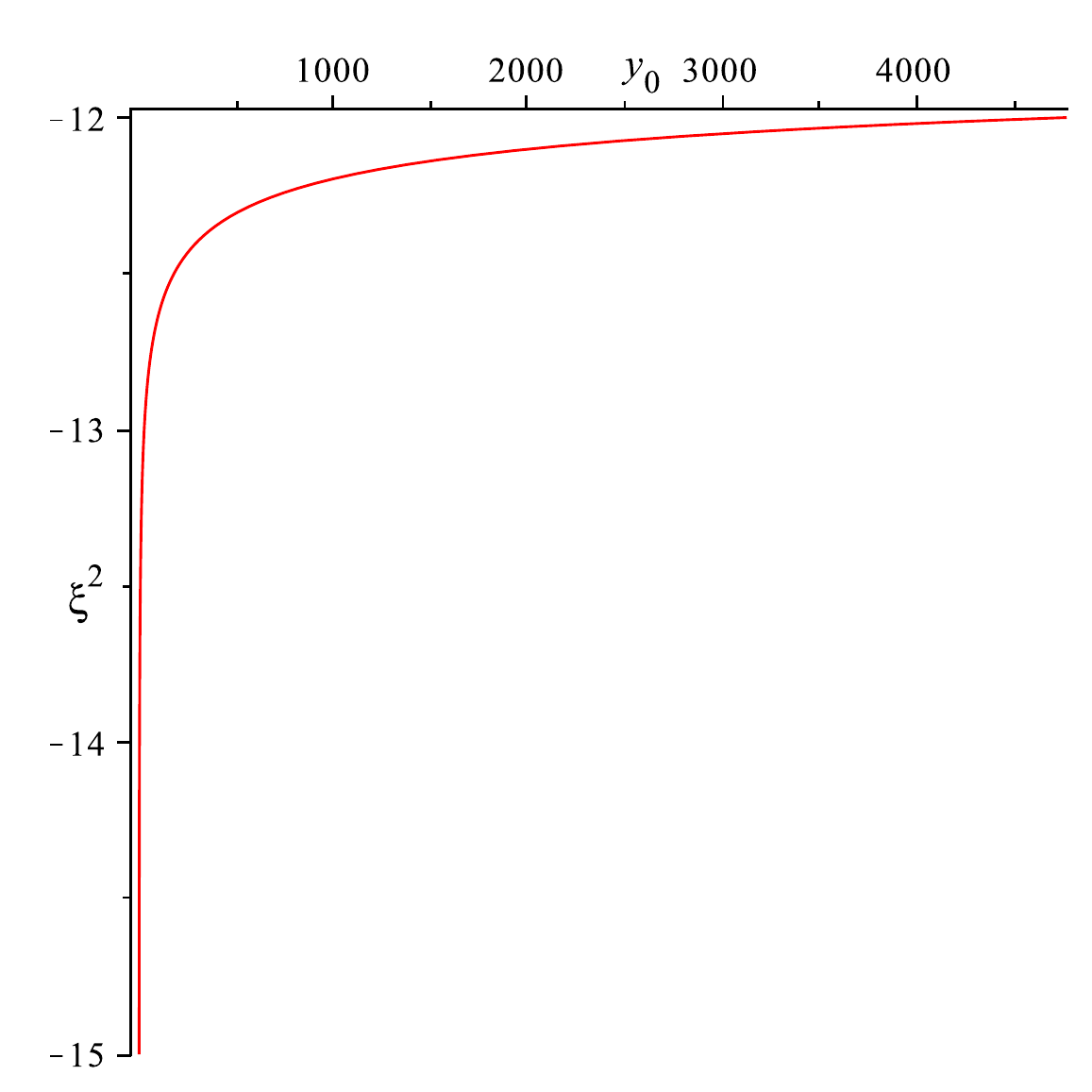}}
\caption{\label{asympI-II} a) portion of $\mathcal{C}$ between negative asymptotes b)  $\mathcal{C}$ below last asymptote}
\end{center}
\end{figure}

\begin{figure}[h] 
\begin{center}
\resizebox{2in}{!}{\includegraphics{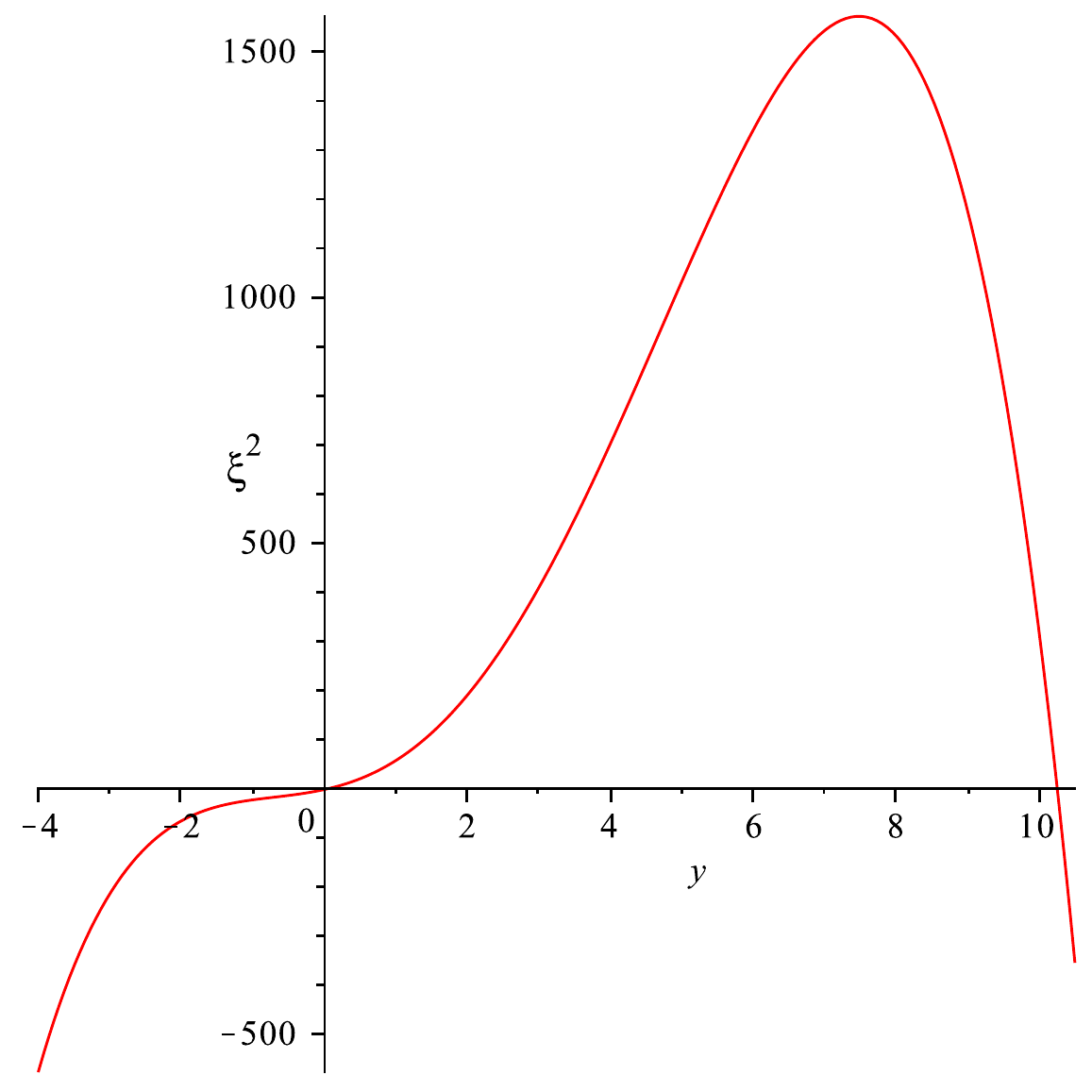}}
\caption{\label{Dhat} a) Graph of real $\widehat{D}/(y_0 - 4)$}
\end{center}
\end{figure}

The upshot is that the relevant portion of the integral surface is the positive component since it is connected to the Gaussian point at $(\xi^2, y_0) = (0,0)$; i.e., this is the {\it real} regular surface. We will see that $ f_g$ and $e_g$ are all functions of $\xi^2$ (i.e., they are even in $\xi$) and so live naturally on this curve. The  generating functions $h_{g}$ are odd functions of $\xi$; hence, up to an overall factor of $\xi$ these also live naturally on the curve. Moreover, since these are enumerative generating funcitons the coefficients in their Taylor-Maclaurin series around $\xi^2 = 0$ are real (positive rational in fact) and so they can indeed be viewed as real-valued functions on the real spectral curve with radius of convergence equal to the 
$\xi^2$-coordinate of the turning point on the right (Figure \ref{5curvebranch} b). 

This picture for the five-valent case extends to the general case of arbitrary odd valence as will be demonstrated in section \ref{gaussleaf}. 

\subsection{Evolution of the Equilibrium Measure} \label{evolvemeas}
Before proceeding to the case of general odd valence, we note that we are now in a better position to visualize the relation between the spectral curve and spectrum. In Figure  \ref{eqm}
we illustrate the point that the endpoints of the support of the equilibrium measure, as a function of the time-like variable $\xi$, uniquely determine the full density of the equilibium measure. At the same time, these endpoints are the Riemann invariants of the leading order continuum Toda equations and so, as such, completely determine the spectral curve. Figure \ref{eqm} illustrates this correspondence. We note that the "equilibrium meaure" determined by the Riemann invariants is only an actual measure over the "physical" $\xi$ interval between $0$ and the turning point on the right. 
\begin{figure}[h] 
\centerline{\includegraphics[width=1.25\linewidth]{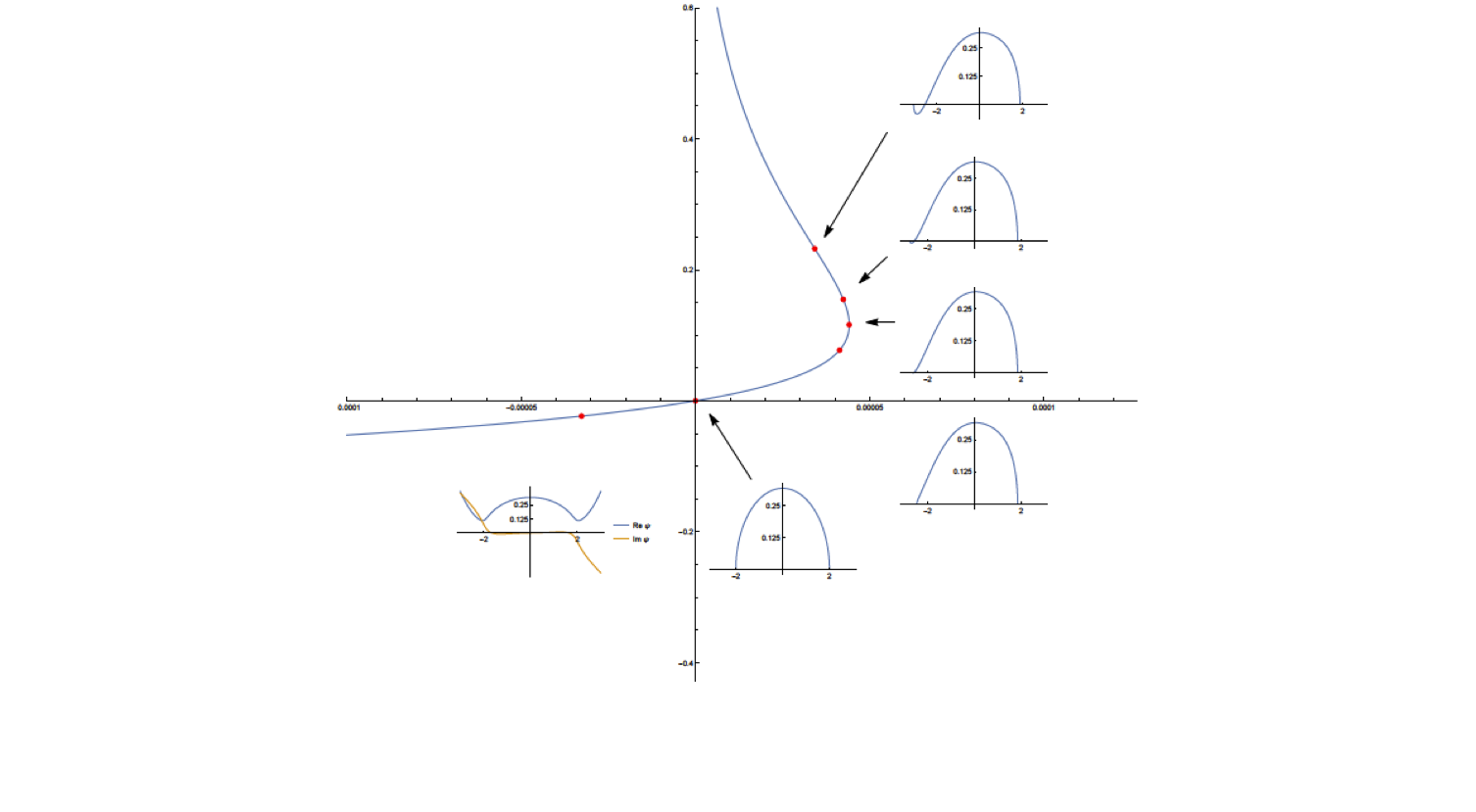}}
\caption{\label{eqm} Correspondence between the Spectral Curve and the Support of the Equilibrium Measure}
\end{figure}
\section{The Spectral Curve, $\mathcal{C}$} \label{gaussleaf}
We now turn to a general description of the algebraic curve $\mathcal{C}$, associated to the integral surface $\mathcal{S}$, that was introduced in section \ref{chargeom}. Going forward we will express everything in terms of a hyperelliptic curve  $\widehat{\mathcal{C}}$ corresponding to the trigonometric double cover (\ref{trigonometric}) of  $\widehat{B}_{12} - y_0^{1/2} \widehat{B}_{11}$ which we will refer to as the {\it spectral curve} and whose origin we now explain. In what follows, we take the subscript $\pm$ to respectively denote the branch of $\sqrt{y_0}$ whose real part is positive (respectively negative). The expression $\widehat{B}_{12} - y_0^{1/2} \widehat{B}_{11}$ we have been cosnidering thus corresponds to the "positive" branch in this sense. We will continue to denote this function in this manner, with the understanding that when we pass to the other branch, $y_0^{1/2}$ changes sign in the argument of $\widehat{B}_{11}$ as well. 

\begin{rem}
We note that, for the map generating series the relevant branch is the one whose real part is negative (i.e. $-\sqrt{y_0}$ by the just stated convention) since the Taylor-McLaurin coefficients of $u_0$ are negative. 
\end{rem}
\bigskip

\begin{prop} \label{baseprop}
The rational algebraic curve $\mathcal{C}$ may be compactly expressed as 
\begin{eqnarray} \label{spectralcurve}
\xi^2 &=& \frac{y_0}{(2\nu +1)^2} \frac{\widehat{\lambda}_\pm^{2\nu - 1}}{\widehat{B}_{12}^{2\nu + 1}}
\end{eqnarray}
where 
\begin{eqnarray}
\lambda_\pm &=& {B}_{11} \pm \sqrt{f_0}{B}_{12} \\
\widehat{\lambda}_\pm &=&  - \left(\widehat{B}_{12}(y_0) \mp y_0^{1/2} \widehat{B}_{11}(\pm \sqrt{y_0})\right)\\
&=& \frac{ - \lambda_\pm}{f_0^{\nu + 1/2}}\\ \label{zzero}
f_0 = x z_0 &=& \frac{x \widehat{B}_{12}}{\widehat{B}_{12} - y_0^{1/2} \widehat{B}_{11}} = \frac{x \widehat{B}_{12}}{\widehat{\lambda}_\pm}
\end{eqnarray}
where $\lambda_\pm$ are the eigenvalues (characteristic speeds) for the conservation law (\ref{conservation}). (Note from this representation that $\lambda_\pm = - \lambda_\mp$ changes depending on the branch of $\sqrt{y_0}$ while $\widehat{\lambda}_\pm$ does not.) Moreover, the pullback of (\ref{spectralcurve}) to the trigonometric double cover, $\widehat{\mathcal{C}}$, has the two branches
\begin{eqnarray*}
\xi &=&  \frac{\pm i}{(2\nu +1) \widehat{\lambda}_\pm}\frac{\sqrt{y_0}}{z_0^{\nu + 1/2}}
\end{eqnarray*}
as does
\begin{eqnarray*}
h_0 &=& x^{1/2} u_0 = - x^{1/2} \left(\pm\sqrt{y_0}\right)   \left(\frac{\widehat{B}_{12}}{\widehat{\lambda}_\pm}\right)^{3/2}\\
\pm \sqrt{f_0} &=& \pm x^{1/2} \sqrt{z_0} = - x^{1/2} \left(\frac{\widehat{B}_{12}}{\widehat{\lambda}_\pm}\right)^{3/2}
\end{eqnarray*}
which defines the left Riemann eigenvectors, $(\pm \sqrt{f_0}, 1)$ and motivates the appellation {\rm spectral curve} for $\widehat{\mathcal{C}}$.
\end{prop}
\begin{proof}
These statements follow directly from definitions together with relations (\ref{curve}) and (\ref{red1}).
\end{proof}
\begin{cor} \label{hypcurve}
The hyperelliptic curve $\widehat{\mathcal{C}}$ has exactly $2 \nu$ branch points (in $\sqrt{y_0}$) all of which are pure imaginary, finite, non-zero and symmetric about 0. Two of these are conjugate with a positive square that we denote by $y_0^*$; all the others come in conjugate pairs whose squares are negative and they interlace the zeroes of $\widehat{B}_{12}$. The genus of $\widehat{\mathcal{C}}$ is $\nu$.
\end{cor}
\begin{proof}
Because $\widehat{B}_{12} - y_0^{1/2} \widehat{B}_{11}$ is polynomial in $y_0$ it suffices to show that this polynomial has one zero which is positive while all the others are negative.
This follows from examination of (\ref{curve}) by which we see that $\xi^2$ has poles (in $y_0$) exactly at the zeros of $\widehat{B}_{12}$. By Proposition \ref{mainprop} there are $\nu$ such zeros and they are all negative. Now consider (\ref{crit2}),
which we rewrite here:
\begin{eqnarray*}
\frac{d \xi^2}{d y_0} &=& -z_0 \xi^2 \frac{\widehat{D}}{y_0(y_0 - 4)}.
\end{eqnarray*}
 Using (\ref{red2})  we observe that $\widehat{D} = \widehat{d}_+\widehat{d}_- = \left(2 - j \left(1 - \frac1{z_0}\right)\right)^2 - (j-2)^2 y_0$ from which we conclude that $\widehat{D} > 0$ when $y_0 < 0$.  One then obtains that the  sign of 
$d\xi^2/dy_0$ is the same as that of $-z_0 \xi^2$. But re-writing this using (\ref{curve}) and (\ref{zzero}) one sees that
\begin{eqnarray} \label{signtest}
\xi^2 z_0 &=& \frac1{(2 \nu + 1)^2} y_0 \frac{(\widehat{B}_{12} - y_0^{1/2} \widehat{B}_{11})^{2\nu-2}}{(\widehat{B}_{12})^{2\nu}}.
\end{eqnarray}
It follows that for $y_0 < 0$, $d\xi^2/dy_0 > 0$ when defined. Hence $\xi^2$ is monotone increasing between consecutive poles. Therefore $\widehat{B}_{12} - y_0^{1/2} \widehat{B}_{11}$ must have at least one negative zero between any two consecutive zeros of $\widehat{B}_{12}$. Hence $\widehat{B}_{12} - y_0^{1/2} \widehat{B}_{11}$ does have $\nu - 1$ negative zeros which interlace the zeroes of $\widehat{B}_{12}$. This accounts for all but one zero of $\widehat{B}_{12} - y_0^{1/2} \widehat{B}_{11}$ which is a polynomial with real coefficients and therefore this last zero must be real. 
Now for $y_0$ negative but larger than the greastest zero of $\widehat{B}_{12}$,
$\widehat{B}_{12} > 0$.  Since $\xi^2 z_0 < 0$ for $y_0 < 0$, it follows from (\ref{signtest}) that $\widehat{B}_{12} - y_0^{1/2} \widehat{B}_{11}$ does not vanish there either. Nor does it vanish at $y_0 = 0$ as can be seen by direct evaluation. Hence the final zero must occur for $y_0 >0$. 

The square roots of these $\nu$ zeroes then yield the $2 \nu$ branch points, symmetric about 0 of $\widehat{\mathcal{C}}$. 
As is the case for all trigonometric polynomial mappings, there is no branch point at $\infty$. It then follows from Hurwitz's formula that $\widehat{\mathcal{C}}$ has genus $\nu$.
\end{proof}
\begin{thm} \label{genthm}
For odd valance, the rational curve, $\mathcal{C}$ in the {\it real} $(\xi^2, y_0)$-plane, has the following geometric properties:
\begin{enumerate}
\item[a)] The curve has $\nu$ horizontal asymptotes at negative values of $y_0$ equal to the $\nu$ zeros (which are all negative) of $\widehat{B}_{12}$. In between consecutive horizontal asymptotes, the curve is monotonically increasing, from the lower asymptote to the upper asymptote with a vertical tangency (inflection point) of order $2 \nu - 1$ where it crosses the $y_0$-axis at the value of $y_0$ equal to the unique zero of $\widehat{B}_{12} - y_0^{1/2} \widehat{B}_{11}$ that lies between the two horizontal asymptotes. (When $\nu = 1$ the crossing of the $y_0$-axis is transversal, rather than tangential.) The order of tangency of the curve to these horizontal asymptotes, as $\xi^2 \to \pm \infty$, is $2 \nu + 1$.
\item[b)] The curve has another component above the highest horizontal asymptote (which lies below the $\xi^2$-axis).  This component increases monotonically from the horizontal asymptote, crosses the $\xi^2$-axis at $0$ and continues to increase monotonically to a simple turning point with positive $\xi^2$ and $y_0$ coordinates corresponding to a zero of $d_-/r_+$ along the curve. It then continues to increase monotonically in $y_0$ but now with decreasing $\xi^2$, crossing the $y_0$-axis at $y_0^*$ with vertical tangency of order $2 \nu -1$. It then continues to the left of the $y_0$-axis until it reaches the zero of $d_+/r_-$ with smallest positive $y_0$ coordinate (and negative $\xi^2$ coordinate). After that, the curve continues to increase monotonically in $y_0$ with increasing $\xi^2$ and approaches $+\infty$ along the $y_0$-axis from the left. The value $y_0 = 4$ lies between $y_0^*$ and the $y_0$-coordinate of the turning point on the left. The curve has a final component below the lowest horizontal asymptote which decreases monotonically in $y_0$ from this asymptote with decreasing $\xi^2$ and approaches $-\infty$ along the $y_0$-axis from the right. The order of tangency of the curve to the $y_0$-axis at $\pm \infty$ is $2 \nu -1$.
\item[c)]  The scaled disciminant $\widehat{D}$ is divisible by $y_0 - 4$ with the ratio equal to a rational funciton of order $\nu$ in $y_0$; i.e. it is the ratio of a polynomial of degree $\nu$ in the numerator to $\widehat{B}_{12}$ in the denominator.
\end{enumerate}
\end{thm}
\smallskip

The proof of this theorem will be presented in Appendix \ref{GenThm}.
\medskip

\begin{rem}
The above description, though seemingly more complicated, does in fact resonate with the characteristic geometry underlying the even valence case \cite{Er09}. For even valence the integral surface is that of the inviscid Burgers equation and there is only one family of characteristics. The curve is expressed in terms of the variables $\xi$ and $z_0$ ($u_0$ is identically 0 in the even valence case).  There is just one real turning point  on the real regular surface with positive $(\xi, z_0)$ coordinates.  All other branch points are off at infinity (more precisely at $(\xi, z_0) = (\infty, 0)$). The Taylor-Maclaurin series of the generating functions, $f_g$ and $e_g$, are function elements on the real regular surface with radius of convergence equal to the $\xi$-coordinate of the unique finite turning point. This analogy will play a guiding role in later sections.
\end{rem}

\section{Arithmetic of Generating Functions on $\mathcal{C}$} \label{arithmetic}
We bring forward here two conjugate polynomials, $\Pi_\mp$, that arose fundamentally in the proof of Appendix \ref{GenThm}. Their definition, determined  from (\ref{PImp}) or (\ref{dmp}), is alternatively implicitly given by
\begin{eqnarray} \label{Pi}
\frac{\widehat{d}_\mp}{2 \pm \sqrt{y_0}} &=& \frac{\Pi_\mp}{\widehat{B}_{12}}\,\,\,\,\,\, \mbox{or}\\
\frac{d_\mp}{r_\pm} &=& \pm \frac{\sqrt{f_0} \,\,\,\, \Pi_\mp}{\widehat{B}_{12}} = \frac{ \pm x^{1/2}\,\,\, \Pi_\mp}{\sqrt{\widehat{B}_{12} \left(\widehat{B}_{12} - y_0^{1/2} \widehat{B}_{11}\right)}}.
\end{eqnarray}
$\Pi_\mp$ are polynomials of degree $j - 1 = 2 \nu$ in $\sqrt{y_0}$ which satisfy the symmetry  
$\Pi_{\mp}(-\sqrt{y_0}) =  \Pi_{\pm}(\sqrt{y_0})$ and $\Pi_-\Pi_+$ is a polynomial of degree $j-1 = 2\nu$ in $y_0$. It follows from this representation that $\Pi_\mp$ has rational coefficients and so it becomes of interest to inquire about the arithmetic properties of its roots.
\medskip

This representation extends to higher derivatives. From the definition (\ref{char}) of the characteristics $r_\pm$ one deduces
\begin{eqnarray*}
2 \partial_x h_0 &=& \partial_x (r_+ + r_-) = \frac{r_+}{d_-} + \frac{r_-}{d_+} = \frac{\widehat{B}_{12}}{\sqrt{f_0}} \left( \frac1{\Pi_-} - \frac1{\Pi_+} \right)\\
2 \frac{\partial_x f_0}{\sqrt{f_0}}  &=& \partial_x (r_+ - r_-) = \frac{r_+}{d_-} - \frac{r_-}{d_+} = \frac{\widehat{B}_{12}}{\sqrt{f_0}} \left( \frac1{\Pi_-} + \frac1{\Pi_+} \right)
\end{eqnarray*}
which yields
\begin{eqnarray} \label{h0x}
h_{0x} &=& \frac{\widehat{B}_{12}}{2 \sqrt{f_0}} \left(\frac1{\Pi_-} - \frac1{\Pi_+}\right)\\ \label{f0x}
f_{0x} &=& \frac{\widehat{B}_{12}}{2} \left(\frac1{\Pi_-} + \frac1{\Pi_+}\right)\\ \label{FundDisc}
f_{0x}^2 - f_0 h_{0x}^2 &=& \frac{\widehat{B}_{12}^2}{\Pi_- \Pi_+} = \frac{4 - y_0}{\widehat{D}}\\ \label{fminush}
f_{0x} - \sqrt{f_0} h_{0x} &=& \frac{2 - \sqrt{y_0}}{\widehat{d}_+}\\ \label{fplush}
f_{0x} + \sqrt{f_0} h_{0x} &=&  \frac{2 + \sqrt{y_0}}{\widehat{d}_-}.
\end{eqnarray}
Differentiating and combining the last two equations yields
\begin{eqnarray*}
f_{0xx} &=& - \frac12 \left\{\frac{2 - \sqrt{y_0}}{\widehat{d}_+} \left[ \partial_w \log \widehat{d}_+ + \left(1 - \frac{\widehat{d}_-}{\widehat{d}_+} \right) \frac{2 + \sqrt{y_0}}{\widehat{d}_-} \right] + \frac{2 + \sqrt{y_0}}{\widehat{d}_-} \left[ \partial_w \log \widehat{d}_- + \left(1 - \frac{\widehat{d}_+}{\widehat{d}_-} \right) \frac{2 - \sqrt{y_0}}{\widehat{d}_+}\right] \right\}\\
\left( \frac12 f_{0x} h_{0x} + f_0 h_{0xx} \right) &=& \frac{\sqrt{f_0}}2 \left\{\frac{2 - \sqrt{y_0}}{\widehat{d}_+} \left[ \partial_w \log \widehat{d}_+ + \left(1 - \frac{\widehat{d}_-}{\widehat{d}_+} \right) \frac{2 + \sqrt{y_0}}{\widehat{d}_-} \right] - \frac{2 + \sqrt{y_0}}{\widehat{d}_-} \left[ \partial_w \log \widehat{d}_- + \left(1 - \frac{\widehat{d}_+}{\widehat{d}_-} \right) \frac{2 - \sqrt{y_0}}{\widehat{d}_+}\right] \right\}
\end{eqnarray*}

It follows from Theorem \ref{genthm} that we have the factorizations
\begin{eqnarray*}
\Pi_- &=& q_-(\sqrt{y_0}) Q_-(\sqrt{y_0})\\
\Pi_+ &=& q_+(\sqrt{y_0}) Q_+(\sqrt{y_0})\\
\Pi_- \Pi_+ &=& q(y_0) Q(y_0)
\end{eqnarray*}
where $q_\mp$ are quadratic polynomials with real coefficients whose roots are real, and $Q_\mp$ are polynomials of degree $2 \nu -2$ with real coefficients whose roots come in complex conjugste pairs.  $q$ is a quadratic polynomial in $y_0$ whose roots are squares of the roots of $q_+$ and also of $q_-$. As a consequence $q_+, q_-$ and $q$ all have a common splitting field extension over $\mathbb{Q}$. The roots of $q(y_0)$ are the $y_0$ coordinates of the real turning points of the spectral curve. 
As an example, consider again the case of valence $j=3$ for which 
\begin{eqnarray*}
q_- &=& (\sqrt{y_0})^2 + 2 \sqrt{y_0} -2 = \left(\sqrt{y_0} - (-1 + \sqrt{3})\right) \left( \sqrt{y_0} - (-1 - \sqrt{3})\right)\\
q_+ &=& (\sqrt{y_0})^2 - 2 \sqrt{y_0} -2 = \left(\sqrt{y_0} - (1 + \sqrt{3})\right) \left( \sqrt{y_0} - (1 - \sqrt{3})\right)\\\\
Q_\mp &=& -1\\
q &=& y_0^2 - 8 y_0 + 4 = \left(y_0 - 2(2 + \sqrt{3})\right) \left(y_0 - 2(2 - \sqrt{3})\right)
\end{eqnarray*}
and as stated above we see that 
\begin{eqnarray*}
\left( \pm (1 + \sqrt{3})\right)^2 &=& 2(2 + \sqrt{3})\\
\left( \pm (1 - \sqrt{3})\right)^2 &=& 2(2 - \sqrt{3})
\end{eqnarray*}
so that, indeed, $\mathbb{Q}(\sqrt{3})$ is a splitting field for $q_+, q_-$ and $q$. This leads to a natural conjecture for the form of the generating functions $f_g, h_{2g}$ and $h_{2 g + 1}$ near the turning points of the spectral curve. To explain this, first consider the form of $f_{0x}$, $j = 3$, near these points:
\begin{eqnarray*}
\frac{f_{0x}}{f_0} &=& -\frac{2-y_0}{2} \left( \frac1{\left(\sqrt{y_0} - (-1 + \sqrt{3})\right) \left( \sqrt{y_0} - (-1 - \sqrt{3})\right)} + \frac1{\left(\sqrt{y_0} - (1 + \sqrt{3})\right) \left( \sqrt{y_0} - (1 - \sqrt{3})\right)}\right)\\
&=& -\frac{2-y_0}{\sqrt{3}} \left( \frac1{\left(\sqrt{y_0} - (-1 + \sqrt{3})\right)} - \frac1{ \left( \sqrt{y_0} - (-1 - \sqrt{3})\right)} + \frac1{\left(\sqrt{y_0} - (1 + \sqrt{3})\right)} - \frac1{ \left( \sqrt{y_0} - (1 - \sqrt{3})\right)}\right)\\
&=& 2 \frac{y_0 - 2}{\sqrt{3}} \left( \frac{(1 + \sqrt{3})}{y_0 - 2(2 + \sqrt{3})} - \frac{(1 - \sqrt{3})}{y_0 - 2(2 - \sqrt{3})} \right)
\end{eqnarray*}
so that the principal part of $\dfrac{f_{0x}}{f_0}$near the turning point on the right is
\begin{eqnarray*}
\frac4{\sqrt{3}} \frac{(1 + \sqrt{3})}{y_0 - 2(2 + \sqrt{3})}
\end{eqnarray*}
while the principal part of $\dfrac{f_{0x}}{f_0}$near the turning point on the right is
\begin{eqnarray*}
-\frac4{\sqrt{3}} \frac{(1 - \sqrt{3})}{y_0 - 2(2 - \sqrt{3})}
\end{eqnarray*}
so we see that the principal parts are interchanged by applying the Galois involution to $\pm \sqrt{3}$.
So in general we conjucture that $f_g/f_0, h_{2g}/h_0$ and $h_{2 g + 1}/h_0$ have partial fractions expansions at the turning points whose coefficients lie in a splitting field for $q_{\mp}$  with maximal pole order in each case equal to a uniform expression in $g$ and with the principal parts at the two turning points interchanged by a Galois involution of this field extension. 
In Section \ref{finestruc} we sill see that this is in fact the case. 

As further illustration of all the above expressions, when $j = 3$, one has
\begin{eqnarray*}
f_0 &=& x \frac{2 + y_0}{2 - y_0}\\
h_0 &=& \sqrt{y_0} \sqrt{f_0}\\
d_\mp &=& (3x - f_0) \mp \sqrt{y_0} f_0\\
\Pi_\mp &=& -\left( \sqrt{y_0}^2 \pm 2 \sqrt{y_0} -2\right) = -\left((y_0 - 2) \pm 2 \sqrt{y_0}\right)\\
\Pi_- \Pi_+ &=& y_0^2 - 8 y_0 + 4\\
\frac{d_\mp}{r_\pm} &=& \pm x^{1/2} \frac{(y_0 - 2) \pm 2 \sqrt{y_0}}{\sqrt{4 - y_0^2}}\\
h_{0x} &=&  \frac{2 \sqrt{y_0}}{\sqrt{f_0}}\frac{y_0 + 2}{y_0^2 - 8 y_0 + 4} = 2  x^{-1/2} \frac{\sqrt{y_0(2 + y_0)(2 - y_0)}}{y_0^2 - 8 y_0 + 4}\\
\frac{h_{0x}}{h_0} &=& - 2 x^{-1} \frac{y_0 - 2}{y_0^2 - 8 y_0 + 4}\\
f_{0x}  &=& - \frac{(y_0 + 2)(y_0 - 2)}{y_0^2 - 8 y_0 + 4}  \\
\frac{f_{0x}}{f_0} &=& 2 x^{-1} \frac{(y_0 - 2)^2}{y_0^2 - 8 y_0 + 4}\\
\frac{h_{0xx}}{h_0} &=& -\frac{24}{x^2} \frac{y_0(y_0 - 2)^3}{(y_0^2 - 8 y_0 +4)^3}\\
\frac{f_{0xx}}{f_0} &=& \frac{8}{x^2}\frac{(2-y_0)^2(y_0^2 -2 y_0 +4)}{ (y_0^2 - 8 y_0 +4)^3} \\
\end{eqnarray*}

\subsection{Symmetries of $\mathcal{C}$}

We consider the projection from $\mathcal{C}$ to the $\xi^2$-line. The critical points (in $y_0$) of this projection are given by the zeroes of $d \xi^2/ dy_0$ which by (\ref{crit2}) and (\ref{red1}) may be written as
\begin{eqnarray*}
\frac{d \xi^2}{d y_0} &=& \left[ \frac1j \frac{\left(\widehat{B}_{12} - y_0^{1/2} \widehat{B}_{11}\right)^{\nu - 1}}{\widehat{B}_{12}^{\nu+1}}\right]^2 \Pi_- \Pi_+.
\end{eqnarray*}
The zeroes of $\widehat{B}_{12} - y_0^{1/2} \widehat{B}_{11}$ correspond to the inflection points which are all tangent to the $y_0$ axis (at real values of $y_0$). There is also an inflection point at $\infty$. The poles at the zeroes of $\widehat{B}_{12}$ correspond to  horizontal asymptotes (which are points on the curve at $\infty$ in $\xi^2$). Finally the zeroes of $\Pi_- \Pi_+$ are branch points of the covering (two of which are the real turning points we have discussed earlier). Blowing up the two-dimensional surface coordinatized by $(\xi^2, y_0)$ at the inflection points yields a (rational) surface on which the spectral curve becomes a $\nu + 1$-degree cover of the $\xi^2$-line with branch point locus  of degree $2 \nu$ coinciding with the zeroes (in $y_0$) of $\Pi_- \Pi_+$. This is consistent with the Riemann-Hurwitz formula and the fact that $\mathcal{C}$ is a rational curve. Thus these critical branch points have their $y_0$-coordinates in the splitting field of  $ \Pi_- \Pi_+$, and, as a consequence of (\ref{curve}), so do their $\xi^2$-coordinates. For instance, in the trivalent case, the critical values are $\xi^2 = \pm \frac1{108 \sqrt{3}}$, which is also an element of $\mathbb{Q}(\sqrt{3})$. As a consequence of (\ref{curve}) these critical values can also be computed as the roots (in $\xi^2$) of the discriminant of 
\begin{eqnarray*}
j^2 \xi^2 \widehat{B}_{12}^{2 \nu + 1} - y_0 \left(\widehat{B}_{12} - y_0^{1/2} \widehat{B}_{11}\right)^{2 \nu - 1}.
\end{eqnarray*} 
Now consider the basic generating functions: $\frac{f_g}{f_0}$ or $E_g$ for $g > 0$. We will see in Section \ref{finestruc} that these are all rational functions in $y_0$ with poles contained in the zeroes of $ \Pi_- \Pi_+$ (which are algebraic numbers in the splitting field of $ \Pi_- \Pi_+$) and with coefficients in the same splitting field. Clearly the Galois group of $ \Pi_- \Pi_+$ (over $\mathbb{Q}$) acts on this expansion by permuting the principal parts associated to the different roots. 

The Galois group here is the same as the monodromy group of the branched covering of the 
$\mathcal{C}$ over the $\xi^2$-plane. In the event that these roots are all simple (which would be the case if the number of critical values for the above-mentioned discriminant was $2 \nu$) then this is the full permutation group on $2 \nu$ roots which induces, as well,  a specific transformation on the coefficients. In this case knowing the principal part of a generating function at one of the simple roots would determine the principal part at all roots via the Galois action. For us it would be most natural to take the root corresponding to the real turning point on the right (i.e., for $\xi^2 > 0$) since the $\xi^2$-coordinate for this turning point is the radius of convergence for the generating function viewed as a Taylor-Maclaurin expansion. If there are multiplicities then the monodromy will not be the full permuation group but there will be subgroup associated to the right turning point that permutes it with a subset of the other simple branch points while preserving $\mathcal{C}$. This point of view is a very natural parallel to what one has in the case of even valence.

\subsection{Divisibility Properties}
We start again with (\ref{Pi}) and transform it:
\begin{eqnarray*}
\frac{\Pi_\mp}{\widehat{B}_{12}} &=& \frac{\widehat{d}_\mp}{2 \pm \sqrt{y_0}}  \\
&=& \frac{2 - j \left(1 - \frac1{z_0}\right) \mp (j-2) \sqrt{y_0}}{2 \pm \sqrt{y_0}} \\
&=& \frac{2 - j \left(1 - \frac1{z_0}\right) \mp (j-2) \sqrt{y_0} + 2(j-2) -2(j-2)}{2 \pm \sqrt{y_0}} \\
&=& \frac{2 - j \left(1 - \frac1{z_0}\right) - (j-2) \left( 2 \pm \sqrt{y_0}\right) + 2(j-2)}{2 \pm \sqrt{y_0}} \\
&=& \frac{(j-2) + j\frac{\left(\widehat{B}_{12} - y_0^{1/2} \widehat{B}_{11}\right)}{\widehat{B}_{12}}}{2 \pm \sqrt{y_0}}  - (j-2)\\
&=& \frac1{\widehat{B}_{12}} \left[\frac{(j-2) \widehat{B}_{12} + j\left(\widehat{B}_{12} - y_0^{1/2} \widehat{B}_{11}\right)}{2 \pm \sqrt{y_0}}  - (j-2) \widehat{B}_{12} \right]\\
&=& \frac1{\widehat{B}_{12}} \left[\frac{(2j-2) \widehat{B}_{12} - j  y_0^{1/2} \widehat{B}_{11}}{2 \pm \sqrt{y_0}}  - (j-2) \widehat{B}_{12} \right]\\
\end{eqnarray*}

where in the fifth line we have applied (\ref{red1}). From this we may then conclude that
\begin{eqnarray*}
\Pi_\mp &=& \frac{(2j-2) \widehat{B}_{12} - j  y_0^{1/2} \widehat{B}_{11}}{2 \pm \sqrt{y_0}}  - (j-2) \widehat{B}_{12}\\
\Pi_\mp + (j-2) \widehat{B}_{12}&=& \frac{(2j-2) \widehat{B}_{12} - j  y_0^{1/2} \widehat{B}_{11}}{2 \pm \sqrt{y_0}}.  
\end{eqnarray*}

It follows that, in fact,
\begin{eqnarray}
\frac{(2j-2) \widehat{B}_{12} - j  y_0^{1/2} \widehat{B}_{11}}{4 - y_0}
\end{eqnarray}
is a polynomial! This may be rephrased in terms of our Appell polynomials (where $\zeta = \sqrt{y_0})$ as stating, equivalently, that

\begin{eqnarray*}
&&\frac{4\nu S_{2\nu}(\zeta) -(2\nu + 1) \partial R_{2\nu}(\zeta)}{4 - \zeta^2}\\
&=& \frac{4\nu S_{2\nu}(\zeta) -(2\nu + 1) 2\nu  R_{2\nu - 1}(\zeta)}{4 - \zeta^2}\\
&=& 2\nu \frac{2 S_{2\nu}(\zeta) -(2\nu + 1)   R_{2\nu - 1}(\zeta)}{4 - \zeta^2}
\end{eqnarray*}
is a polynomial of degree $\nu - 1$ in $\zeta^2$ with rational coefficients.

\subsection{Mixed Valence and Stability} \label{stability}
 
We first consider the regular case of our system (\ref{SYSTEM}). For pure odd valence, we have seen in Theorem \ref{genthm} (and its proof in Appendix \ref{GenThm}) that the left and right real turning points are of square root branch point type located at a real simple zero of $d_+/ r_-$ and $d_-/r_+$ respectively and that they do indeed travel with the respective characteistic speeds $\lambda_-$ and $\lambda_+$. Under the reduction from $(x, \xi, h_0, f_0)$ to self-similar variables $(\xi^2, y_0)$, that can be done in the regular case, we have denoted the branch point location  on the right (the one closest to the origin) by  $(\xi_c^2, y_{0 c})$ where 
\beann
\xi_c^2 &=& x^{j-2} t_c^2 \\
y_{0 c} &=& \frac{h_0^2(\xi_c)}{f_0(\xi_c^2)}.
\eeann
A completely similar description holds in the case of pure even valence except that (\ref{SYSTEM}) now collapses to a scalar conservation law and there is only a single simple real turning point. 

When one passes to the fuller setting of mixed valence, one no longer has the self-similar reduction enabling one to view the integral surface associated to the pure valence initial data as  the  scaling of a spectral curve. Instead, for mixed valence initial data, one must work in the full phase space $(x, \xi, h_0, f_0)$. Nevertheless, the results of \cite{CEHL93} say that a qualitatively similar situation persists concerning branch point behavior. In particular, it is of most interest to focus on the situation which is a generic mixed valence perturbation of a pure valence case; i.e., where then initial $\vec{t}$ parameters are in a small neighborhood of $\vec{t}^{(j)} = (0, \dots, 0 , t_j, 0, \dots, 0)$. Since \cite{CEHL93} shows that the real turning points are stable singularities, the local structure must persist; in particular they remain square-root type branch points. On the other hand, some aspects of the integral surface description in Theorem \ref{genthm} will change. The only qualitatively significant point here concerns the vertical tangencies described in part (a) of Theorem \ref{genthm} which arise for pure valence $j > 3$. Since these types of singularities are not of the stable types mentioned at the start of this subsection (and detailed in Theorem 1.2 of \cite{CEHL93}), these tangencies will break up into a collection of stable types under genric multivalent perturbations. However for many applications this will not affect results for enumerative questions about maps. In particular, for asymptotic (in large size) statements such perturbations will not alter the asymptotic results found in the pure valence case. The reason for this stems from methods of {\it singularity analysis} in analytic combinatorics (see \cite{FS}). The essential point is that all singularities potentially contribute to, and in fact completely determine, asymptotic expansions of Taylor-Maclaurin coefficients in the generating functions centered at $\xi = 0$. However, the contributions of the different singularities are exponentially separated in terms of the distance of the respective singularities from $\xi =0$. Our analysis in Theorem \ref{genthm} showed that the turning point singularity on the right in the case of pure odd valence, which is the only singularity in the case of pure even valence, is the one that is closest to $\xi = 0$. This feature will not change under sufficiently small perturbations; i.e., that turning point is stable and it will remain the closest singularity. The contributions of all other singularities will be beyond all orders in terms of the asymptotic expansion at the right turning point.

We will say more about the relevance of these stability observations in section \ref{asympenum}.

\section{Topological Recursion} \label{toprec} 
Finally we turn to the mechanism for our original goal: the explicit compact calculation
of the map generating functions $E_g$ as well as the higher genus associated generating
functions $h_g, f_g$. This will all be based on the contunuum string equation formulation presented in section \ref{motzstring} but made effective after a number of transformations based on String polynomials and associated Hopf algebra constructions. We outline those stages here first before getting into them in detail.
\begin{enumerate}
\item[I.] {\it Continuum String Equations in Differential-Difference Resolvent Form}
This is the content of equations (\ref{ContString} - \ref{StringCoeff}). The differential operators here are in terms of $\partial_\eta$ where $\eta$ labels the {\it height} associated to the continuum raising operator. 
\item[II.] {\it Continuum String Coefficients in terms of Bessel-type Operators acting on String Polynomials} This is a key step in which the height differentiation operators are transmuted into Bessel-type operators acting on string polynomials. This is where the connection to Hopf-algebra analysis can come in. The differential operators here are in terms of $\partial_f$ and $\partial_h$
where $f$ and $h$ are the continuum genearting functions. The basic ingredients of this transition are described in subsection \ref{BesselRedux}. This is formulated in terms of the extended string polynomials (\ref{PHI}). Then the somewhat lengthy "change of variables" is outlined in subsection \ref{StringOps}.
\item[III.] {\it Order Expansion in $1/N$} As was explained earlier in the paper, the hierarchy of continuum string equations is extracted by sytematically collecting terms of fixed order in $1/N$ from the continuum string equations which here will be the fixed order terms coming from (\ref{finalform}). In subsection \ref{loworder} we will illustrate this in a few low order cases. The key observation here is that the ladder-like structure that the extended String Polynomials inherit from Bessel generating functions enable one to reduce everything down to ordinary string polynomials.
\item[IV.] {\it Unwinding Analysis} In this final step, repeated usage of the unwinding identity (Prop \ref{prop05}) and integration by parts enable one to eliminate higher order string polynomials and reduce everything down to the leading order generating functions and their $x$-derivatives. In section \ref{UnivEg} we will see how, in general, the higher genus generating functions may be compactly expressed in terms of $h_0, f_0$ and their $x$-derivatives in a form that is independent of valence. 
\item[V.] Although our general formula (\ref{finalform}) for the string coefficients is stated for a fixed valence $j$, this only enters through coefficient extraction in terms of $[z^j]$. Hence the formula extends naturally to the case of mixed valence by linearity.
\end{enumerate}

\subsection{Some Fundamental Reduction Formulas based on Bessel Identities} \label{BesselRedux}
We will make use of the following basic Bessel identities:
\begin{eqnarray} \label{eq:BessId1}
\frac12 (I_{n-1}(t) - I_{n+1}(t)) &=& \frac{n}t I_n(t) \\  \label{eq:BessId2}
\frac12 (I_{n-1}(t) + I_{n+1}(t)) &=&  I^\prime_n(t)
\end{eqnarray}
From this, one immediately has the system
\beann
\left( \begin{array}{cc} -t \partial_t  & t \\ -t  & t \partial_t \end{array}\right)
\left( \begin{array}{c}  I_n \\ I_{n-1}\end{array}\right) &=&  \left( \begin{array}{c} n I_n \\ (n-1) I_{n-1}\end{array}\right)\\
\left( \begin{array}{cc} -t \partial_t  & t  \\ - t & t \partial_t \end{array}\right)^2
\left( \begin{array}{c}  I_n \\ I_{n-1}\end{array}\right) &=& 
\left( \begin{array}{cc} (t \partial_t)^2 - t^2  & t \\ -t & (t \partial_t)^2 - t^2\end{array}\right)
\left( \begin{array}{c}  I_n \\ I_{n-1}\end{array}\right)\\
&=&  \left( \begin{array}{c} n^2 I_n -t I_{n-1}\\ (n-1)^2 I_{n-1} + t I_n\end{array}\right)\\
\left[(t \partial_t)^2 - t^2\right] \left( \begin{array}{c}  I_n \\ I_{n-1}\end{array}\right) &=& \left( \begin{array}{c} n^2 I_n \\ (n-1)^2 I_{n-1} \end{array}\right).
\eeann

Now making the substitutions
\begin{eqnarray} \label{extendPHI}
x = 2 s z\sqrt{f} \qquad \partial_x = \frac{\sqrt{f}}{sz} \partial_f \qquad \Phi_n = I_n(2sz\sqrt{f})
e^{szh}
\end{eqnarray}
the above systems become
\beann
\left( \begin{array}{cc} -2f \partial_f   & 2 \sqrt{f} \partial_h \\ -2 \sqrt{f} \partial_h  & 2 f \partial_f\end{array}\right)
\left( \begin{array}{c}  \Phi_n \\ \Phi_{n-1}\end{array}\right) =  \left( \begin{array}{cc} n & 0 \\ 0 & n - 1\end{array}\right)   \left( \begin{array}{c}  \Phi_n \\ \Phi_{n-1}\end{array}\right) \\ 
  \left[ 4 (f \partial_f)^2 - 4 f \partial_h^2  \right] 
\left( \begin{array}{c} \Phi_n \\ \Phi_{n-1}\end{array}\right)
= \left( \begin{array}{cc} n^2  & 0 \\ 0 & (n - 1)^2 \end{array}\right) \left( \begin{array}{c} \Phi_n \\ \Phi_{n-1}\end{array}\right) .
\eeann
\begin{rem}
$\Phi_n$ introduced here extends the definition given in (\ref{PHI}) in that here $h_0$ and $f_0$
are prolonged to the full generating functions $h$ and $f$ and $s$ has been extended to $sz$.
\end{rem}

\begin{cor} Applying the first variation in the case where $n = 1$ we have
\begin{eqnarray} \label{StringRel}
\left( \begin{array}{cc} -2f \partial_f - 1 & 2 \sqrt{f} \partial_h \\ -2 \sqrt{f} \partial_h & 2 f \partial_f\end{array}\right)
\left( \begin{array}{c}  \Phi_1 \\ \Phi_{0}\end{array}\right) &=&  0
\end{eqnarray}
\end{cor}
Note that this generates the ideal of relations in the algebra of string functions since 
$\left( \begin{array}{c} \Phi_1 \\ \Phi_{0}\end{array}\right)$ is the generating function for the $\left( \begin{array}{c}  \psi_m \\ \phi_{m}\end{array}\right)$.
\medskip

\noindent More generally,
\beann
\left( \begin{array}{cc} -t \partial_t  & t \\ -t  & t \partial_t \end{array}\right)
\left( \begin{array}{c}  n^{2m}I_n \\ (n-1)^{2m} I_{n-1} \end{array}\right) &=&  \left( \begin{array}{c} n^{2m+1} I_n + t (n-1)^{2m} I_{n-1} \\ (n-1)^{2m+1} I_{n-1} - t n^{2m}I_n\end{array}\right)\\
\left( \begin{array}{cc} -t \partial_t  & 0 \\ 0  & t \partial_t \end{array}\right) \left[(t \partial_t)^2 - t^2\right]^m \left( \begin{array}{c}  I_n \\  I_{n-1} \end{array}\right) &=&
\left( \begin{array}{c}  n^{2m+1} I_n \\ (n-1)^{2m+1} I_{n-1} \end{array}\right) 
\eeann
which yields
\begin{eqnarray} \label{BesselGenerator}
\left( \begin{array}{cc} -2f \partial_f  & 0 \\ 0  & 2f \partial_f \end{array}\right) \left[ 4 (f \partial_f)^2 - 4 f \partial_h^2  \right]^m \left( \begin{array}{c} \Phi_n \\ \Phi_{n-1}\end{array}\right) &=&
\left( \begin{array}{c}  n^{2m+1} \Phi_n \\ (n-1)^{2m+1} \Phi_{n-1} \end{array}\right).
\end{eqnarray}

\subsection{String Equations in terms of Operators on String Functions} \label{StringOps}

Now we turn to studying the string equations in the string function representation of section \ref{motzstring}, where the sum over distinguished steps $\vec{m}, \vec{\sigma}$ consistent with $\alpha, \beta$ is understood. For some notational compression we set
\beann
\mathcal{O}(p) &=& \left\{ \begin{array}{cc}  \left[ 4 (f \partial_f)^2 - 4 f \partial_h^2  \right]^m & p = 2m\\
\left( \begin{array}{cc} -2f \partial_f  & 0 \\ 0  & 2f \partial_f \end{array}\right) \left[ 4 (f \partial_f)^2 - 4 f \partial_h^2  \right]^m & p =2m+1 > 1 \\
\left( \begin{array}{cc} -2f \partial_f  & 2 \sqrt{f} \partial_h  \\ - \sqrt{f} \partial_h  & 2f \partial_f \end{array}\right) & p = 1\end{array} \right.
\eeann

Also, $\mathcal{M}^{d}_{n}$ denotes the space of monomials $p$ of degree $d$ in $n+1$ variables $p_0, \dots, p_{n}$ which are of the form
\beann
\frac{p_0^{\sum_{q=0}^{n-1} k_0^{(q)}} \cdots p_{n-1}^{\sum_{q=0}^{n-1} k_{n-1}^{(q)}}}{\prod_{q=0}^{n-1} \prod_{0 \leq i \leq q} k_i^{(q)}!} \times 
\left(|\alpha| + |\beta| - \sum_{j=1}^{n-1} M_j\right) ! \frac{p_n^{\sum_{i=0}^{n-1} k_i^{(n)}}}{\prod_{0 \leq i \leq n-1} k_i^{(n)}!}
\eeann
where $k_0^{(0)} = M_0 \leq m_1, k_0^{(1)} + k_1^{(1)} = M_1 \leq m_2, \dots,  k_0^{(n-1)} + \cdots + k_{n-1}^{(n-1)} = M_{n-1} \leq m_n, k_0^{(n)} + \cdots + k_{n-1}^{(n)} = |\alpha| + |\beta| - \sum_{j=1}^{n-1} M_j$, with the stipulation that if $\Sigma_q = 0$ then $M_{q-1} = m_{q}$. Also, we set $\Sigma_q = \sum_{\gamma=0}^{q-1} \sigma_\gamma$. (Note that it is always the case that $\sigma_0 =0$; hence, one always has $\Sigma_0 = 0$.)
\bigskip

We now proceed to re-express (\ref{StringCoeff}) in terms of Bessel-type operators (as defined in the previous section) acting on string functions:
\beann
\left( \begin{array}{c}\tilde{P}^{(a)}_{\alpha, \beta, j} (h, f) \\ \tilde{P}^{(b)}_{\alpha, \beta, j} (h, f) \end{array} \right) 
=    \left( \begin{array}{c}   [\eta^{\ell(\beta)} z^{j-1-\ell(\alpha) - \ell(\beta)}] \frac{1}{f^{\ell(\beta) /2}} \cr [\eta^{\ell(\beta) - 1} z^{j-1-\ell(\alpha) - \ell(\beta)}] \frac{1}{f^{((\ell(\beta)-1)/2)}}  \end{array} \right)^\dagger   \prod_{q=1}^{\ell(\alpha) + \ell(\beta)} \frac1{m_q !} \,\,\,\,\, \\
\left[\sum_{\gamma=0}^{q-1}  \left( \eta_\gamma \partial_{\eta_\gamma}  -   \sigma_\gamma    \right)  \right]_{\eta_\gamma = \eta}^{m_q}
\prod_{\gamma = 0}^{\ell(\alpha) + \ell(\beta)} \frac1{1 - z(\sqrt{f^{(\gamma)}} \eta_\gamma + h^{(\gamma)} + \sqrt{f^{(\gamma)}} \eta_\gamma^{-1})}  \,\,\,\,\, \\
=   [\eta^{\ell(\beta)-1} z^{j-1-\ell(\alpha) - \ell(\beta)}]  \left(\begin{array}{c} \frac{1}{f^{\ell(\beta)/2}}  \frac1{\eta} \\ \frac{1}{f^{((\ell(\beta)-1)/2)}}  \end{array}\right)^\dagger \,\,\,\,\, \\
\prod_{q=1}^{\ell(\alpha) + \ell(\beta)} \frac1{m_q !}
\left[ \sum_{\gamma=0}^{q-1}                                                       \eta_\gamma \partial_{\eta_\gamma}  -   \Sigma_q    \right]^{m_q} \int_0^\infty \cdots \int_0^\infty ds_0 \dots ds_{\ell_\alpha + \ell_\beta} \,\,\,\,\, \\
\left(  \prod_{\gamma = 0}^{\ell(\alpha) + \ell(\beta)}  e^{-s_\gamma\left(1 -z\left(\sqrt{f^{(\gamma)}} \eta_\gamma + h^{(\gamma)} + \sqrt{f^{(\gamma)}} \eta_\gamma^{-1}\right)\right)} \right) \Big|_{\eta_\gamma = \eta}
\eeann

\beann
=  [\eta^{\ell(\beta) - 1} z^{j-1-\ell(\alpha) - \ell(\beta)}] \frac1{f^{\ell(\beta)/2}} \left( \begin{array}{c} \frac{ 1}{\eta} \\ 1 \end{array}\right)^\dagger \left(\begin{array}{cc} 1  & 0 \cr 0 & f^{1/2}  \end{array} \right) \,\,\,\,\, \\ 
\sum_{p \in \mathcal{M}^{|\alpha|+|\beta|}_{\ell(\alpha) + \ell(\beta)}}\prod_{i=0}^{\ell(\alpha) + \ell(\beta) - 1 } \\
\int_0^\infty ds_{i}  \frac{\left(\eta_\gamma \partial_{\eta_\gamma} \right)^{\deg p_i \in p}}{\prod_{q=0}^{\ell(\alpha) + \ell(\beta)-1} k_i^{(q)} ! }
  e^{-s_i \left(1 -z\left(\sqrt{f^{(\gamma)}} \eta_\gamma + h^{(\gamma)} + \sqrt{f^{(\gamma)}} \eta_\gamma^{-1}\right)\right)}\Big|_{\eta_\gamma = \eta} 
\times \\ \left(|\alpha| + |\beta| - \sum_{j=1}^{\ell(\alpha) + \ell(\beta)-1} M_j\right) ! \\ 
\frac{\prod_{q=1}^{\ell(\alpha) + \ell(\beta)} \left(-\Sigma_{q} \right)^{k_{q-1}^{(\ell(\alpha) + \ell(\beta))}}}{\prod_{q=1}^{\ell(\alpha) + \ell(\beta)} k_{q-1}^{(\ell(\alpha) + \ell(\beta))} ! }  \frac1{1 - z(\sqrt{f} \eta + h + \sqrt{f} \eta^{-1})} 
\eeann

\beann
=  [\eta^{\ell(\beta) - 1} z^{j-1-\ell(\alpha) - \ell(\beta)}] \frac1{f^{\ell(\beta)/2}} \left( \begin{array}{c} \frac{ 1}{\eta} \\ 1 \end{array}\right)^\dagger \left(\begin{array}{cc} 1 & 0 \cr 0 & f^{1/2}  \end{array} \right) \,\,\,\,\, \\
\sum_{p \in \mathcal{M}^{|\alpha|+|\beta|}_{\ell(\alpha) + \ell(\beta)}}\prod_{i=0}^{\ell(\alpha) + \ell(\beta) - 1 }\int_0^\infty ds_{i}  \frac{ \left(\eta \partial_{\eta} \right)^{\deg p_i \in p}}{\prod_{q=0}^{\ell(\alpha) + \ell(\beta) - 1} k_i^{(q)} ! }
  \sum_{n = -\infty}^\infty     I_n(2 s_i z \sqrt{f^{(i)}}) e^{s_i z h^{(i)}} \eta^n \\
\times  \left(|\alpha| + |\beta| - \sum_{j=1}^{\ell(\alpha) + \ell(\beta)-1} M_j\right) ! \\
\frac{\prod_{q=1}^{\ell(\alpha) + \ell(\beta)} \left( -\Sigma_{q} \right)^{k_{q-1}^{(\ell(\alpha) + \ell(\beta))}}}{\prod_{q=1}^{\ell(\alpha) + \ell(\beta)} k_{q-1}^{(\ell(\alpha) + \ell(\beta))} ! }  \frac1{1 - z(\sqrt{f} \eta + h + \sqrt{f} \eta^{-1})} \,\,\,\,\, \\
=  [\eta^{\ell(\beta) - 1} z^{j-1-\ell(\alpha) - \ell(\beta)}] \frac1{f^{\ell(\beta)/2}} \left( \begin{array}{c} \frac1{\eta} \\ 1 \end{array}\right)^\dagger \left(\begin{array}{cc} 1 & 0 \cr 0 & f^{1/2}  \end{array} \right) 
\sum_{p \in \mathcal{M}^{|\alpha|+|\beta|}_{\ell(\alpha) + \ell(\beta) }}  \\  \prod_{i=0}^{\ell(\alpha) + \ell(\beta) - 1 }\int_0^\infty ds_{i}  \frac{1}{\prod_{q=0}^{\ell(\alpha) + \ell(\beta)-1} k_i^{(q)} ! }
  \sum_{n = -\infty}^\infty   \left(n\right)^{\deg p_i \in p}   I_n(2 s_i z \sqrt{f^{(i)}}) e^{s_i z h^{(i)}} \eta^n \\
\times  \left(|\alpha| + |\beta| - \sum_{j=1}^{\ell(\alpha) + \ell(\beta)-1} M_j\right) ! \\ 
\frac{\prod_{q=1}^{\ell(\alpha) + \ell(\beta)} \left(-\Sigma_{q} \right)^{k_{q-1}^{(\ell(\alpha) + \ell(\beta))}}}{\prod_{q=1}^{\ell(\alpha) + \ell(\beta)} k_{q-1}^{(\ell(\alpha) + \ell(\beta))} ! }   \frac1{1 - z(\sqrt{f} \eta + h + \sqrt{f} \eta^{-1})}
\eeann

\beann
=  \frac{[\eta^{\ell(\beta) - 1} z^{j-1-\ell(\alpha) - \ell(\beta)}]}{f^{\ell(\beta)/2}}  \times \left(\begin{array}{cc} 1 & 0 \cr 0 & f^{1/2}  \end{array} \right) \\
\sum_{p \in \mathcal{M}^{|\alpha|+|\beta|}_{\ell(\alpha) + \ell(\beta)}} \int_0^\infty ds_{0}  \frac{1}{\prod_{q=0}^{\ell(\alpha) + \ell(\beta)} k_0^{(q)} ! }
  \sum_{n = -\infty}^\infty  \mbox{diag} \left(n + 1, n\right)^{\deg p_0 \in p}  \left( \begin{array}{c}  \Phi_{n+1} \\  \Phi_{n}  \end{array}\right)^\dagger \eta^n  \,\,\,\,\, \\
  \times \prod_{i=1}^{\ell(\alpha) + \ell(\beta) - 1 }\int_0^\infty ds_{i}  \frac{1}{\prod_{q=0}^{\ell(\alpha) + \ell(\beta)-1} k_i^{(q)} ! }
  \sum_{n = -\infty}^\infty  (n)^{\deg p_i \in p}  I_n(2 s_i z \sqrt{f^{(i)}}) e^{s_i z h^{(i)}} \eta^n \,\,\,\,\, \\
\times \left(|\alpha| + |\beta| - \sum_{j=1}^{\ell(\alpha) + \ell(\beta)-1} M_j\right) ! \frac{\prod_{q=0}^{\ell(\alpha) + \ell(\beta)} \left(- \Sigma_q \right)^{k_{\ell(\alpha) + \ell(\beta)}^{(q)}}}{\prod_{q=0}^{\ell(\alpha) + \ell(\beta)} k_{\ell(\alpha) + \ell(\beta)}^{(q)} ! }  \frac1{1 - z(\sqrt{f} \eta + h + \sqrt{f} \eta^{-1})}
\eeann

\beann
=  \frac{[\eta^{\ell(\beta) - 1} z^{j-1-\ell(\alpha) - \ell(\beta)}]}{f^{\ell(\beta)/2}}  \left(\begin{array}{cc} 1 & 0 \cr 0 & f^{1/2}  \end{array} \right) \,\,\,\,\, 
\sum_{p \in \mathcal{M}^{|\alpha|+|\beta|}_{\ell(\alpha) + \ell(\beta)}}   \frac{1}{\prod_{q=0}^{\ell(\alpha) + \ell(\beta)-1} k_0^{(q)} ! }  \\  
\mathcal{O}(\deg p_0 \in p) 
 \frac1{1 - z(\sqrt{f} \eta + h + \sqrt{f} \eta^{-1})}  \\
\times \left( \begin{array}{c} \frac{1}{\eta}  \\ 1 \end{array}\right)   \prod_{i=1}^{\ell(\alpha) + \ell(\beta) - 1 }\int_0^\infty ds_{i}  \frac{1}{\prod_{q=0}^{\ell(\alpha) + \ell(\beta)-1} k_i^{(q)} ! }
  \sum_{n = -\infty}^\infty   \left(n\right)^{\deg p_i \in p}   I_n(2 s_i z \sqrt{f}) e^{s_i z h} \eta^n \,\,\,\,\, \\
\times \left(|\alpha| + |\beta| - \sum_{j=1}^{\ell(\alpha) + \ell(\beta)-1} M_j\right) ! \frac{\prod_{q=1}^{\ell(\alpha) + \ell(\beta)} \mbox{diag}\left(- \Sigma_q\right)^{k_{q-1}^{(\ell(\alpha) + \ell(\beta))}}}{\prod_{q=1}^{\ell(\alpha) + \ell(\beta)} k_{q-1}^{(\ell(\alpha) + \ell(\beta))} ! }   \frac1{1 - z(\sqrt{f} \eta + h + \sqrt{f} \eta^{-1})} \,\,\,\,\, \\
=  \frac{[\eta^{\ell(\beta) - 1} z^{j-1-\ell(\alpha) - \ell(\beta)}]}{f^{\ell(\beta)/2}}  \left(\begin{array}{cc} 1 & 0 \cr 0 & f^{1/2}  \end{array} \right)  \sum_{p \in \mathcal{M}^{|\alpha|+|\beta|}_{\ell(\alpha) + \ell(\beta)}} \,\,\,\,\, 
 \frac{2^{\deg p_0 \in p}}{\prod_{q=0}^{\ell(\alpha) + \ell(\beta)-1} k_0^{(q)} ! }   \\ \mathcal{O}({\deg p_0 \in p})   \frac1{1 - z(\sqrt{f} \eta + h + \sqrt{f} \eta^{-1})}   \,\,\,\, 
 \left( \begin{array}{c} \frac{1}{\eta}  \\ 1 \end{array}\right) \\ 
\prod_{i=1}^{\ell(\alpha) + \ell(\beta) - 1 }\int_0^\infty ds_{i}  \frac{1}{\prod_{q=0}^{\ell(\alpha) + \ell(\beta)-1} k_i^{(q)} ! }
  \sum_{n = -\infty}^\infty   \left(n\right)^{\deg p_i \in p}   I_n(2 s_i z \sqrt{f^{(i)}}) e^{s_i z h^{(i)}} \eta^n \,\,\,\,\, \\
 \times \left(|\alpha| + |\beta| - \sum_{j=1}^{\ell(\alpha) + \ell(\beta)-1} M_j\right) ! \frac{\prod_{q=1}^{\ell(\alpha) + \ell(\beta)} \left(-\Sigma_q \right)^{k_{q-1}^{(\ell(\alpha) + \ell(\beta))}}}{\prod_{q=1}^{\ell(\alpha) + \ell(\beta)} k_{q-1}^{(\ell(\alpha) + \ell(\beta))} ! }   \frac1{1 - z(\sqrt{f} \eta + h + \sqrt{f} \eta^{-1})} \,\,\,\,\, 
\eeann

Iterating this shift, we ultimately get

\beann
=   \frac{[\eta^{\ell(\beta) - 1} z^{j-1-\ell(\alpha) - \ell(\beta)}]}{f^{(\ell(\beta)-1)/2}}  \left( \begin{array}{cc} 1/\sqrt{f} & 0 \\ 0 & 1  \end{array}\right) \sum_{p \in \mathcal{M}^{|\alpha|+|\beta|}_{\ell(\alpha) + \ell(\beta)}} 2^{\sum_{j=1}^{\ell(\alpha) + \ell(\beta)-1} M_j}\\
   \left[ \left(\prod_{i=0}^{\ell(\alpha) + \ell(\beta)-1} \frac{1}{\prod_{q=0}^{\ell(\alpha) + \ell(\beta)-1} k_i^{(q)} ! } \mathcal{O}({\deg p_i \in p}) 
 \frac1{1 - z(\sqrt{f} \eta + h + \sqrt{f} \eta^{-1})} \right) \left( \begin{array}{c} \frac{1}{\eta} \\ 1 \end{array}\right) \right]  \\
\times \left(|\alpha| + |\beta| - \sum_{j=1}^{\ell(\alpha) + \ell(\beta)-1} M_j\right) ! \frac{\prod_{q=1}^{\ell(\alpha) + \ell(\beta)}\left(-\Sigma_q\right)^{k_{q-1}^{(\ell(\alpha) + \ell(\beta))}}}{\prod_{q=1}^{\ell(\alpha) + \ell(\beta)} k_{q-1}^{(\ell(\alpha) + \ell(\beta))} ! }   \frac1{1 - z(\sqrt{f} \eta + h + \sqrt{f} \eta^{-1})} 
\eeann
where
the derivatives, $\partial_f, \partial_h$ act only on their particular associated factor, $\frac1{1 - z(\sqrt{f} \eta + h + \sqrt{f} \eta^{-1})}$.

So, in summary, we have:
\begin{eqnarray} \nonumber
\left( \begin{array}{c}  \tilde{P}^{(a)}_{\alpha, \beta, j} (h, f) \\ \tilde{P}^{(b)}_{\alpha, \beta, } (h, f) \end{array} \right) 
 =   \frac{[\eta^{\ell(\beta) - 1} z^{j-1-\ell(\alpha) - \ell(\beta)}]}{f^{(\ell(\beta)-1)/2}}  \left( \begin{array}{cc} 1/\sqrt{f} & 0 \\ 0 & 1  \end{array}\right) \sum_{p \in \mathcal{M}^{|\alpha|+|\beta|}_{\ell(\alpha) + \ell(\beta)}} 2^{\sum_{j=1}^{\ell(\alpha) + \ell(\beta)-1} M_j}\\ 
\label{finalform}
   \left(\prod_{i=0}^{\ell(\alpha) + \ell(\beta)-1} \frac{1}{\prod_{q=0}^{\ell(\alpha) + \ell(\beta)-1} k_i^{(q)} ! } \mathcal{O}({\deg p_i \in p}) 
 \frac1{1 - z(\sqrt{f} \eta + h + \sqrt{f} \eta^{-1})} \right) \left( \begin{array}{c} \frac{1}{\eta} \\ 1 \end{array}\right)  \,\,\,\,\,\,\,\,\,\, \, \,\,\,\ \\ \nonumber
\times \left(|\alpha| + |\beta| - \sum_{j=1}^{\ell(\alpha) + \ell(\beta)-1} M_j\right) ! \frac{\prod_{q=1}^{\ell(\alpha) + \ell(\beta)}\left(-\Sigma_q\right)^{k_{q-1}^{(\ell(\alpha) + \ell(\beta))}}}{\prod_{q=1}^{\ell(\alpha) + \ell(\beta)} k_{q-1}^{(\ell(\alpha) + \ell(\beta))} ! }   \frac1{1 - z(\sqrt{f} \eta + h + \sqrt{f} \eta^{-1})} .
\end{eqnarray}

\subsection{Evaluation at Low Orders} \label{loworder}

Evaluating the String equaitions (\ref{ContString}) using (\ref{finalform})
at leading order in $\frac1n (j = 2 \nu +1)$ yields
\begin{eqnarray*} 
&& \left( \begin{array}{c} 0 \\ x \end{array} \right) = \left(\begin{array}{c} h_0 \\ f_0 \end{array} \right) + j
t  [\eta^{-1} z^{j-1}]  \frac1{1 - z(\sqrt{f_0}\eta + h_0 + \sqrt{f_0}\eta^{-1})}  \left( \begin{array}{c}\frac{1}\eta \\ \sqrt{f_0} \end{array}\right)\\
&& \left( \begin{array}{c} 0 \\ x \end{array} \right) = \left(\begin{array}{c} h_0 \\ f_0 \end{array} \right) +j
t  [\eta^{0} z^{j-1}] \frac1{1 - z(\sqrt{f_0}\eta + h_0 + \sqrt{f_0}\eta^{-1})}  \left( \begin{array}{c} 1\\ \sqrt{f_0}\eta \end{array}\right)\\
&&x =  f_0 +  (2 \nu + 1)  t \sum_{\mu = 1}^{\nu}  {2\nu  \choose 2\mu-1, \nu - \mu, \nu - \mu + 1} h_0^{2\mu-1} f_0^{\nu - \mu + 1}  \\ 
&& \qquad = f_0 +  (2 \nu + 1)  t B_{11} \\
&& \qquad = f_0 + t \psi_0 = \psi_{V,0}\\
&& 0  =  h_0 +  (2 \nu + 1)  t \sum_{\mu = 0}^{\nu}  {2\nu  \choose 2\mu, \nu - \mu, \nu - \mu } h_0^{2\mu} f_0^{\nu - \mu}\\
&& \qquad  = h_0 +  (2 \nu + 1)  t B_{12}\\
&& \qquad = h_0 + t \phi_0 = \phi_{V,0}
\end{eqnarray*}
which coincides with our earlier calculation in (\ref{STRING} - \ref{A12}).

At $\mathcal{O}\left(\frac1n\right)$ one has contributions from $(\alpha, \beta) = ((1,0), \emptyset)$ and $(\emptyset, (1,0) )$. In preparation for this we spell out equation (\ref{finalform}) for these cases. The underlying form of (\ref{finalform}) here is
\begin{eqnarray} \label{FF1} \nonumber
&& 2 \left[ \left( \begin{array}{cc} - \sqrt{f} \partial_f & \partial_h \\ -  \sqrt{f}\partial_h &  f \partial_f\end{array}\right)
\frac1{1 - z(\sqrt{f} \eta + h + \sqrt{f} \eta^{-1})}  \left( \begin{array}{c} \frac{1}{\eta} \\ 1 \end{array}\right)\right] \frac1{1 - z(\sqrt{f} \eta + h + \sqrt{f} \eta^{-1})} \\ 
&=&   \frac{z(\eta - \eta^{-1})}{\left(1 - z(\sqrt{f} \eta + h + \sqrt{f} \eta^{-1})\right)^3}  \left( \begin{array}{c} \frac{1}{\eta} \\ \sqrt{f} \end{array}\right) 
\end{eqnarray}
Now in the case of $(\alpha, \beta) = (\emptyset, (1,0) )$, (\ref{finalform}) applies the order prefix $[\eta^0 z^{j-2}] $ to (\ref{FF1}) with translate, $\Sigma_0 = 0$,  to get 
\beann
 && \frac{1}2 \left( \begin{array}{c} [\eta^0] - [\eta^2]  \cr \sqrt{f} ([\eta^{-1}] 
- [\eta^1]) \end{array}\right) (j-1)_2 \left(\sqrt{f} \eta + h + \sqrt{f} \eta^{-1}\right)^{j-3} \\
&=&  \frac12 \left( \begin{array}{c}  [s^{j-3}](\Phi_0 - \Phi_2)  \cr \sqrt{f} [s^{j-3}](\Phi_{-1} - \Phi_1) \end{array}\right) 
\eeann
where, here, we have set $z=1$ in $\Phi_n$ from (\ref{extendPHI}). Continuing
\beann
&=&  \frac12 \left( \begin{array}{c}  [s^{j-3}] \frac2{2s\sqrt{f}}\Phi_1  \cr 0 \end{array}\right) 
\eeann
where we have again used the general Bessel identity
\begin{eqnarray} \label{BesselIdent}
I_{n-1}(t) - I_{n+1}(t) &=& \frac{2n}{t} I_n(t),
\end{eqnarray}
so that
\beann
&=&  t \left( \begin{array}{c}  \frac1{2 {f_0}}  \psi_1  \cr 0 \end{array}\right).
\eeann
where in the  last line we have evaluated $f$ at $f_0$.

In the case of $(\alpha, \beta) = ( (1,0), \emptyset)$, (\ref{finalform}) applies the order prefix $[\eta^{-1} z^{j-2}] \times  f^{1/2}$ to (\ref{FF1}), with the same previous translate, to get, by a simlar derivation, 
\beann
 \left(\begin{array}{c} 0 \cr 
  - \frac{t}{2} \psi_1
\end{array} \right).
\eeann
Finally, one also needs to extract the coefficient of $h_1$ in the $(\emptyset, \emptyset)$ term, which is straightforward. Putting this all together one has 
\beann
\left( \begin{array}{c} 0 \\ 0 \end{array} \right) &=& \left( \begin{array}{c} h_1 \\ 0 \end{array} \right) + t  \left[j [\eta^0 z^{j-1}]  \left( \begin{array}{c} 1 \\ \sqrt{f_0} \eta \end{array}\right)\frac{z h_1}{(1 - z(\sqrt{f_0}\eta + h_0 + \sqrt{f_0}\eta^{-1}))^2} 
+ \left( \begin{array}{c}\tilde{P}^{(a)}_{(1,0), \emptyset, j} (h_0, f_0)  \partial_x h_0  + \tilde{P}^{(a)}_{\emptyset, (1,0), j} (h_0, f_0)  \partial_x f_0\cr \tilde{P}^{(b)}_{(1,0), \emptyset, j} (h_0, f_0) \partial_x h_0 + \tilde{P}^{(b)}_{\emptyset, (1,0), j} (h_0, f_0)  \partial_x f_0\end{array} \right) \right]\\
&=& \left( \begin{array}{c} h_1 \\ 0 \end{array} \right) +  t  \left[j [\eta^0 z^{j-2}]  \left( \begin{array}{c}  1 \\ \sqrt{f_0} \eta \end{array}\right)\frac{h_1}{(1 - z(\sqrt{f_0}\eta + h_0 + \sqrt{f_0}\eta^{-1}))^2} 
+ \left( \begin{array}{c}\tilde{P}^{(a)}_{(1,0), \emptyset, j} (h_0, f_0)  \partial_x h_0  + \tilde{P}^{(a)}_{\emptyset, (1,0), j} (h_0, f_0)  \partial_x f_0\cr \tilde{P}^{(b)}_{(1,0), \emptyset, j} (h_0, f_0) \partial_x h_0 + \tilde{P}^{(b)}_{\emptyset, (1,0), J} (h_0, f_0)  \partial_x f_0\end{array} \right) \right]\\
&=& \left( \begin{array}{c} h_1 \\ 0 \end{array} \right) + t   \left( \begin{array}{c} \phi_1\\ \psi_1 \end{array}\right) h_1 
- h_{0 x} 
\left(\begin{array}{c} 0  \cr   \frac{t}2 \psi_1 
\end{array} \right)
+ t f_{0 x} 
\left(\begin{array}{c}  
  \frac1{2 f_0} \psi_1  \cr 0
\end{array} \right)\\
&=& \left( \begin{array}{c} A_{11} \\ A_{21} \end{array}\right) h_1 
+  h_{0 x} 
\left(\begin{array}{c} 0  \cr 
 - \frac12 A_{21}
\end{array} \right)
+  f_{0 x} 
\left(\begin{array}{c}    
   \frac1{2} A_{12}  \cr 0
\end{array} \right)
\eeann
where in the last equality we have made use of the identities (\ref{A11}) and (\ref{A12}). Continuing, we have
\beann
\left(\begin{array}{c}  0 \cr 0
\end{array} \right) &=& \left(\begin{array}{c}   
 A_{11} \left( h_1 + \frac12 \frac{A_{12}}{A_{11}} f_{0x} \right) \cr A_{21} \left( h_1 - \frac12 h_{0x}\right) 
\end{array} \right)\\
&=& \left(\begin{array}{c}   A_{11} \left( h_1 - \frac12 h_{0x} \right) \cr A_{21} \left( h_1 - \frac12 h_{0x}\right)  
\end{array} \right)\\
&=& \left(\begin{array}{c}   \hat{\phi}_1 \left( h_1 - \frac12 h_{0x}\right)\cr 
 \hat{\psi}_1 \left( h_1 - \frac12 h_{0x} \right)
\end{array} \right)
\eeann
which recovers a leading order term in the global identity (\ref{hoddident}).

At $\mathcal{O}\left(\frac1{n^2}\right)$ one has additional contributions from\\
 $(\alpha, \beta) = ((2,0), \emptyset), ((1,1), \emptyset), (\emptyset, (1,1) ), ((1,0), (1,0))$ and $(\emptyset, (2,0) )$. Setting $\epsilon = \frac1n$, one has,

\beann
\left( \begin{array}{c} 0 \\ 0 \end{array} \right) = \left( \begin{array}{c}   h_2 \\ f_1 \end{array} \right)
+  \left[  j t [\eta^0 z^{j-1}] \frac12 \frac{d^2}{d\epsilon^2}\Big|_{\epsilon = 0} \left( \begin{array}{c} 1 \\ \sqrt{f} \eta \end{array}\right)\frac{1}{1 - z(\sqrt{f}\eta + h + \sqrt{f}\eta^{-1})} 
+ t \left( \begin{array}{c}\tilde{P}^{(a)}_{(1,0), \emptyset, j} (h_0, f_0)  \partial_x h_1  \cr \tilde{P}^{(b)}_{(1,0), \emptyset, j} (h_0, f_0) \partial_x h_1 \end{array} \right) \right]\\
+  j t h_1 \left(  h_{0 x} [\eta^{-1} z^{j-2}]  \left(\begin{array}{cc} f^{-1/2} & 0 \cr 0 & f^{1/2}  \end{array} \right) + f_{0 x}[\eta^0 z^{j-2}] \left(\begin{array}{cc}  f^{-1}  & 0 \cr 0 & 1  \end{array} \right)  \right) 
 \times \sqrt{f}  \frac{3  z^2 (\eta - \eta^{-1})}{\left(1 - z(\sqrt{f} \eta + h + \sqrt{f} \eta^{-1})\right)^4}  \left( \begin{array}{c} \frac{\sqrt{f}}{\eta} \\ 1 \end{array}\right) \\
+  t \left( \begin{array}{c}\tilde{P}^{(a)}_{(2,0), \emptyset, j} (h_0, f_0)  h_{0xx} +  \tilde{P}^{(a)}_{(1,1), \emptyset, j} (h_0, f_0)  h^2_{0x}  +  \tilde{P}^{(a)}_{(1,0), (1,0), j} (h_0, f_0)  h_{0x}f_{0x} \cr \tilde{P}^{(b)}_{(2,0), \emptyset, j} (h_0, f_0)  h_{0xx} +  \tilde{P}^{(b)}_{(1,1), \emptyset, j} (h_0, f_0)  h^2_{0x}  +  \tilde{P}^{(b)}_{(1,0), (1,0), j} (h_0, f_0)  h_{0x}f_{0x} \end{array} 
\begin{array}{c} + \tilde{P}^{(a)}_{\emptyset, (1,1), j} (h_0, f_0)  f^2_{0x} +\tilde{P}^{(a)}_{\emptyset, (2,0), j} (h_0, f_0)  f_{0xx} \cr + \tilde{P}^{(b)}_{\emptyset, (1,1), j} (h_0, f_0)  f^2_{0x} +\tilde{P}^{(b)}_{\emptyset, (2,0), j} (h_0, f_0)  f_{0xx}\end{array} \right)
\eeann
\beann
\left( \begin{array}{c} 0 \\ 0 \end{array} \right) = \left( \begin{array}{c}  h_2  \\ f_1 \end{array} \right) + j t  \left( \begin{array}{c}  \phi_1 h_2 + \frac1{f_0} \psi_1 f_1 + \frac12 \phi_2 h_1^2  \cr   \psi_1 h_2 + \phi_1 f_1 + \frac12 \psi_2 h_1^2  \end{array}\right) - jt \left( \begin{array}{c}  0 \cr    \frac12 \psi_1 h_{1x}  \end{array}\right)  \\
+  j t h_1 \left(  h_{0 x} [\eta^{-1} z^{j-2}]    \left(\begin{array}{cc}  1  & 0 \cr 0 & f \end{array} \right)   + f_{0 x}[\eta^0 z^{j-2}]  \left(\begin{array}{cc} f^{-1/2} & 0 \cr 0 & f^{1/2}  \end{array} \right)  \right) 
\times   \frac{3 z^2 (\eta - \eta^{-1})}{\left(1 - z(\sqrt{f} \eta + h + \sqrt{f} \eta^{-1})\right)^4}  \left( \begin{array}{c} \frac{\sqrt{f}}{\eta} \\ 1 \end{array}\right) 
\\
+  t \left( \begin{array}{c}\tilde{P}^{(a)}_{(2,0), \emptyset, j} (h_0, f_0)  h_{0xx} +  \tilde{P}^{(a)}_{(1,1), \emptyset, j} (h_0, f_0)  h^2_{0x}  +  \tilde{P}^{(a)}_{(1,0), (1,0), j} (h_0, f_0)  h_{0x}f_{0x} \cr \tilde{P}^{(b)}_{(2,0), \emptyset, j} (h_0, f_0)  h_{0xx} +  \tilde{P}^{(b)}_{(1,1), \emptyset, j} (h_0, f_0)  h^2_{0x}  +  \tilde{P}^{(b)}_{(1,0), (1,0), j} (h_0, f_0)  h_{0x}f_{0x} \end{array} 
\begin{array}{c} + \tilde{P}^{(a)}_{\emptyset, (1,1),jJ} (h_0, f_0)  f^2_{0x} +  \tilde{P}^{(a)}_{\emptyset, (2,0), j} (h_0, f_0)  f_{0xx} \cr + \tilde{P}^{(b)}_{\emptyset, (1,1), j} (h_0, f_0)  f^2_{0x} +\tilde{P}^{(b)}_{\emptyset, (2,0), j} (h_0, f_0)  f_{0xx}\end{array} \right)
\eeann
\beann
\left( \begin{array}{c} 0 \\ 0 \end{array} \right) = \left( \begin{array}{cc} A_{11} & A_{12} \\ A_{21} & A_{22} \end{array} \right) \left( \begin{array}{c} h_2 \\  f_1\end{array} \right) - j t h_{1x} \left( \begin{array}{c}  0  \\ \frac12 {\psi}_1 \end{array} \right) + \frac12 j t h^2_{1} \left( \begin{array}{c}  \phi_2 \\ \psi_2  \end{array} \right) + j t h_1 \left( \begin{array}{c}    \frac12 \psi_2 \frac{f_{0x}}{f_0}  \\ -  \frac12 \psi_2 h_{0x} \end{array} \right) \,\,\,\,\, \\
+ t \left( \begin{array}{c}\tilde{P}^{(a)}_{(2,0), \emptyset, j} (h_0, f_0)  h_{0xx} +  \tilde{P}^{(a)}_{(1,1), \emptyset, j} (h_0, f_0)  h^2_{0x}  +  \tilde{P}^{(a)}_{(1,0), (1,0), j} (h_0, f_0)  h_{0x}f_{0x} \cr \tilde{P}^{(b)}_{(2,0), \emptyset, j} (h_0, f_0)  h_{0xx} +  \tilde{P}^{(b)}_{(1,1), \emptyset, j} (h_0, f_0)  h^2_{0x}  +  \tilde{P}^{(b)}_{(1,0), (1,0), j} (h_0, f_0)  h_{0x}f_{0x} \end{array} 
\begin{array}{c} + \tilde{P}^{(a)}_{\emptyset, (1,1), j} (h_0, f_0)  f^2_{0x} +\tilde{P}^{(a)}_{\emptyset, (2,0), j} (h_0, f_0)  f_{0xx} \cr + \tilde{P}^{(b)}_{\emptyset, (1,1), j} (h_0, f_0)  f^2_{0x} +\tilde{P}^{(b)}_{\emptyset, (2,0), j} (h_0, f_0)  f_{0xx}\end{array} \right). \,\,\,\,\,
\eeann 
We now evaluate the underlying "Bessel operator" forms  
of (\ref{finalform}) for which $|\alpha| + |\beta| = 2$. In the cases with $\ell(\alpha) + \ell(\beta) = 1$, $(\emptyset, (2,0))$ and $((2,0), \emptyset)$, this form is
\begin{eqnarray} \label{FF2}
&& \frac42 \left[ \left[  (f \partial_f)^2 -  f \partial_h^2  \right]
\frac1{1 - z(\sqrt{f} \eta + h + \sqrt{f} \eta^{-1})}  \left( \begin{array}{c} \frac{1}{\eta} \\ 1 \end{array}\right)\right] \frac1{1 - z(\sqrt{f} \eta + h + \sqrt{f} \eta^{-1})} \,\,\,\,\,\,\,\,\, \\ \nonumber
&=&  \left[ \frac{z^2 (\eta - \eta^{-1})^2 f}{\left(1 - z(\sqrt{f} \eta + h + \sqrt{f} \eta^{-1})\right)^4}  \left( \begin{array}{c} \frac{1}{\eta} \\ 1\end{array}\right) 
+\frac{z/2 \sqrt{f} (\eta + \eta^{-1})}{\left(1 - z(\sqrt{f} \eta + h + \sqrt{f} \eta^{-1})\right)^3}  \left( \begin{array}{c} \frac1{\eta} \\ 1 \end{array}\right) \right] .
\end{eqnarray}

In the cases with $\ell(\alpha) + \ell(\beta) = 2$, $(\alpha, \beta) = ((1,1), \emptyset), (\emptyset, (1,1) ),$ and $((1,0), (1,0))$, the underlying form of (\ref{finalform}) is
\begin{eqnarray} \nonumber 
&& 4 \left[ \left( \begin{array}{cc} -f \partial_f & \sqrt{f}\partial_h \\ -  \sqrt{f}\partial_h &  f \partial_f\end{array}\right)
\frac1{1 - z(\sqrt{f} \eta + h + \sqrt{f} \eta^{-1})}  \right]^2 \left( \begin{array}{c} \frac{1}{\eta} \\ 1 \end{array}\right)\frac1{1 - z(\sqrt{f} \eta + h + \sqrt{f} \eta^{-1})} \\ \label{FF4}
&+& 4\left[  \left[  (f \partial_f)^2 -  f \partial_h^2  \right] \frac1{1 - z(\sqrt{f} \eta + h + \sqrt{f} \eta^{-1})}  \left( \begin{array}{c} \frac{1}{\eta} \\ 1 \end{array}\right) \right]\frac1{(1 - z(\sqrt{f} \eta + h + \sqrt{f} \eta^{-1}))^2} \,\,\,\,\,\,\, \\ \nonumber
&=& \left[ \frac{ 3 z^2 (\eta - \eta^{-1})^2 f }{\left(1 - z(\sqrt{f} \eta + h + \sqrt{f} \eta^{-1})\right)^5}  \left( \begin{array}{c}  \frac1{\eta} \\ 1  \end{array}\right) 
+ \frac{z(\eta + \eta^{-1})\sqrt{f}}{\left(1 - z(\sqrt{f} \eta + h + \sqrt{f} \eta^{-1})\right)^4}  \left( \begin{array}{c} \frac{1}{\eta} \\ 1 \end{array}\right) \right] .
\end{eqnarray}

Inserting these forms into the $\tilde{P}^{(a)}_{\alpha, \beta, j}$ and $\tilde{P}^{(b)}_{\alpha, \beta, j}$ of bi-size $|\alpha| + |\beta| = 2$ in  (\ref{finalform}) and applying the corresponding order evaluation prefactors we can solve for $\left( \begin{array}{c}   h_2 \\ f_1 \end{array} \right)$ in the $\mathcal{O}(1/n^2)$- string equations to arrive at:
\begin{eqnarray} \label{unwound1}
\left( \begin{array}{c} h_2 \\  f_1\end{array} \right) &=& -   \left( \begin{array}{cc} A_{11} & A_{12} \\ A_{21} & A_{22} \end{array}  \right)^{-1}  \left[  - h_{1x} \left( \begin{array}{c}  0 \\ \frac12 {\psi}_1 \end{array} \right) + \frac12  h^2_{1} \left( \begin{array}{c}  \phi_2 \\ \psi_2  \end{array} \right) +  h_1 \left( \begin{array}{c}    \frac12 \psi_2 \frac{f_{0x}}{f_0}  \\ -  \frac12 \psi_2 h_{0x} \end{array} \right) \right] \\ \label{unwound2}
&-& \left( \begin{array}{cc} A_{11} & A_{12} \\ A_{21} & A_{22} \end{array} \right)^{-1} 
\left( \begin{array}{c} [\eta^0] \\  \sqrt{f_0} [\eta^{-1}]\end{array}\right) \\ \nonumber
&\times&\left[\left(h_{0xx} + \eta^{-1} \frac{f_{0xx}}{\sqrt{f_0}}\right) \left( \frac{[z^{J-4}]  (\eta - \eta^{-1})^2 f_0}{\left(1 - z(\sqrt{f} \eta + h + \sqrt{f} \eta^{-1})\right)^4}
 + \frac12 \frac{[z^{J-3}] (\eta + \eta^{-1})\sqrt{f_0}}{\left(1 - z(\sqrt{f} \eta + h + \sqrt{f} \eta^{-1})\right)^3}\right)  \right.\\ \nonumber
&+& \left.  \left( h_{0x}^2 + \frac{2 h_{0x}f_{0x}}{\sqrt{f_0}}\eta^{-1} + \frac{f_{0x}^2}{f_0} \eta^{-2}\right) \left( 3 \frac{[z^{J-5}]  (\eta - \eta^{-1})^2 f_0}{\left(1 - z(\sqrt{f} \eta + h + \sqrt{f} \eta^{-1})\right)^5} +  \frac{[z^{J-4}]  (\eta + \eta^{-1}) \sqrt{f_0}}{\left(1 - z(\sqrt{f} \eta + h + \sqrt{f} \eta^{-1})\right)^4}\right) \right.\\ \label{unwound4}
&-& \left. \left(  \frac{h_{0x}f_{0x}}{\sqrt{f_0}}\eta^{-1} + \frac{f_{0x}^2}{f_0} \eta^{-2}\right)  \frac{[z^{J-4}] (\eta - \eta^{-1}) \sqrt{f_0}}{\left(1 - z(\sqrt{f} \eta + h + \sqrt{f} \eta^{-1})\right)^4}   \right]
\end{eqnarray}
where the two terms in the last line above are coming from $\vec{\sigma}$-shifts.
Before continuing we will make use of umbral relations, based on (\ref{PHI}) and repeated  application (\ref{BesselIdent}) in the setting of (\ref{extendPHI}), to reduce the $\eta$-order evaluation operators by shifting. The following should be viewed as equations between order evaluation operators.

\begin{eqnarray}
\left( \begin{array}{c} [\eta^0] \\ \sqrt{f_0} [\eta^{-1}]\end{array}\right) (\eta - \eta^{-1})^2 f_0 &=& f_0 \left( \begin{array}{c} -2 ([\eta^0] - [\eta^2]) \\ - \sqrt{f_0} ([\eta^1] - [\eta^3] \end{array} \right)\\ \nonumber
&=& -\frac2s \left( \begin{array}{c} \sqrt{f_0} [\eta^1] \\ {f_0} [\eta^{0}]\end{array}\right) + \frac2{s^2} \left( \begin{array}{c} 0 \\ \sqrt{f_0} [\eta^{1}]\end{array}\right)\\
\left( \begin{array}{c} [\eta^0] \\ \sqrt{f_0} [\eta^{-1}]\end{array}\right) (\eta + \eta^{-1})  \sqrt{f_0} &=& \sqrt{f_0} \left( \begin{array}{c} 2 [\eta^1] \\  \sqrt{f_0} ([\eta^0] + [\eta^2] ) \end{array}\right)\\ \nonumber
&=& 2 \left( \begin{array}{c} \sqrt{f_0} [\eta^1] \\ {f_0} [\eta^{0}]\end{array}\right) 
- \frac1s \left( \begin{array}{c} 0 \\ \sqrt{f_0} [\eta^{1}]\end{array}\right)\\
\left( \begin{array}{c} [\eta^0] \\ \sqrt{f_0} [\eta^{-1}]\end{array}\right) \eta^{-1}(\eta - \eta^{-1})^2  \sqrt{f_0} &=& \sqrt{f_0} \left( \begin{array}{c} - ([\eta^1] - [\eta^3] )\\ -2 \sqrt{f_0} ([\eta^{0}] - [\eta^2])\end{array}\right)\\ \nonumber
&=& - \frac2{s} \left( \begin{array}{c}  [\eta^{0}] \\ \sqrt{f_0} [\eta^1] \end{array}\right) + \frac2{s^2} \left( \begin{array}{c}\frac1{\sqrt{f_0}}  [\eta^{1}] \\ 0 \end{array}\right)
\end{eqnarray}

\begin{eqnarray}
 \left( \begin{array}{c} [\eta^0] \\ \sqrt{f_0} [\eta^{-1}]\end{array}\right)\eta^{-1}(\eta + \eta^{-1})  &=& \left( \begin{array}{c} [\eta^0]+ [\eta^2]\\ \sqrt{f_0} ([\eta^{-1}] + [\eta^1])\end{array}\right)\\ \nonumber
&=& 2 \left( \begin{array}{c}  [\eta^0] \cr \sqrt{f_0} [\eta^1] \end{array} \right)
- \frac1{s}  \left( \begin{array}{c} \frac1{\sqrt{f_0}} [\eta^1] \cr 0 \end{array} \right)\\
\left( \begin{array}{c} [\eta^0] \cr \sqrt{f_0} [\eta^{-1}]\end{array}\right) \eta^{-2}(\eta - \eta^{-1})^2  &=& \left( \begin{array}{c}  ([\eta^0] - [\eta^2]) - ([\eta^2] - [\eta^4])\\ - \sqrt{f_0} ([\eta^1] - [\eta^3]) \end{array} \right)\\ \nonumber
&=&  - \frac{2}{s} \left( \begin{array}{c} \frac1{\sqrt{f_0}} [\eta^1] \cr [\eta^0] \end{array} \right)
+ \frac2{s^2 f_0}  \left( \begin{array}{c} 3  [\eta^0] \cr  \sqrt{f_0} [\eta^1] \end{array} \right)
- \frac2{s^3 f^2_0} \left( \begin{array}{c} 3 \sqrt{f_0} [\eta^1] \cr  0 \end{array} \right)\\
\left( \begin{array}{c} [\eta^0] \cr \sqrt{f_0} [\eta^{-1}]\end{array}\right) \eta^{-2}(\eta + \eta^{-1}) \frac1{\sqrt{f_0}} &=& 
\frac1{\sqrt{f_0}} \left( \begin{array}{c} [\eta^1] + [\eta^3] \cr  \sqrt{f_0} ([\eta^0] + [\eta^2] ) \end{array}\right) \\ \nonumber
&=& 2 \left( \begin{array}{c} \frac1{\sqrt{f_0}} [\eta^1] \cr [\eta^0] \end{array} \right)
- \frac1{s f_0}  \left( \begin{array}{c} 2  [\eta^0] \cr \sqrt{f_0} [\eta^1] \end{array} \right)
+ \frac2{s^2 f^2_0} \left( \begin{array}{c} \sqrt{f_0} [\eta^1] \cr  0 \end{array} \right)\\
 \left( \begin{array}{c} [\eta^0] \\ \sqrt{f_0} [\eta^{-1}]\end{array}\right)\eta^{-1}(\eta - \eta^{-1})  &=& \left( \begin{array}{c} [\eta^0]- [\eta^2]\\ \sqrt{f_0} ([\eta^{-1}] - [\eta^1])\end{array}\right)\\ \nonumber
&=& \frac1s \left( \begin{array}{c} \frac1{f_0}\sqrt{f_0}[\eta^1] \\ 0 \end{array}\right)\\
\left( \begin{array}{c} [\eta^0] \cr \sqrt{f_0} [\eta^{-1}]\end{array}\right) \eta^{-2}(\eta - \eta^{-1}) \frac1{\sqrt{f_0}} &=& 
\frac1{\sqrt{f_0}} \left( \begin{array}{c} [\eta^1] - [\eta^3] \cr  \sqrt{f_0} ([\eta^0] - [\eta^2] ) \end{array}\right) \\ \nonumber
&=& \frac1{s f_0} \left( \begin{array}{c}  2 [\eta^{0}] \\ \sqrt{f_0}[\eta^1] \end{array}\right)
- \frac2{s^2 f^2_0} \left( \begin{array}{c} \sqrt{f_0}[\eta^1] \\ 0\end{array}\right)
\end{eqnarray}

Applying these reductions in (\ref{unwound2} - \ref{unwound4}) and using (\ref{PHI}), we arrive at
\beann
\left( \frac16 h_{0xx} [s^{J-4}] + \frac18 h^2_{0x} [s^{J-5}]\right) \left[ -\frac2s \left( \begin{array}{c} \sqrt{f_0}\Phi_1 \\ {f_0} \Phi_0 \end{array}\right) + \frac2{s^2} \left( \begin{array}{c} 0 \\ \sqrt{f_0}\Phi_1\end{array}\right)\right]\\
+ \left( \frac14 h_{0xx} [s^{J-3}] + \frac16 h^2_{0x} [s^{J-4}]\right) \left[ 2 \left( \begin{array}{c} \sqrt{f_0} \Phi_1 \\ {f_0} \Phi_0\end{array}\right) 
- \frac1s \left( \begin{array}{c} 0 \\ \sqrt{f_0} \Phi_1\end{array}\right)\right]\\
+ \left(\frac16 f_{0xx} [s^{J-4}]+ \frac28 h_{0x} f_{0x} [s^{J-5}] \right)\left[ - \frac2s \left( \begin{array}{c}  \Phi_0  \\  \sqrt{f_0}\Phi_1\end{array}\right) + \frac2{s^2} \left( \begin{array}{c} \frac{\sqrt{f_0}}{f_0} \Phi_1 \\ 0\end{array}\right)\right]\\
+ \left(\frac14 f_{0xx} [s^{J-3}]+ \frac26 h_{0x} f_{0x} [s^{J-4}] \right)
\left[ 2 \left( \begin{array}{c}  \Phi_0 \cr \sqrt{f_0} \Phi_1 \end{array} \right)
- \frac1{s f_0}  \left( \begin{array}{c} \sqrt{f_0} \Phi_1 \cr 0 \end{array} \right)\right]\\
+ \left( \frac18 f^2_{0x} [s^{J-5}] \right) \left[   - \frac{2}{s} \left( \begin{array}{c} \frac{\sqrt{f_0}}{f_0} \Phi_1 \cr \Phi_0 \end{array} \right)
+ \frac2{s^2 f_0}  \left( \begin{array}{c} 3  \Phi_0 \cr  \sqrt{f_0}\Phi_1 \end{array} \right)
- \frac2{s^3 f^2_0} \left( \begin{array}{c} 3 \sqrt{f_0} \Phi_1 \cr  0 \end{array} \right)\right]\\
+ \left( \frac16 f^2_{0x} [s^{J-4}] \right) \left[2 \left( \begin{array}{c} \frac{\sqrt{f_0}}{f_0} \Phi_1 \cr \Phi_0 \end{array} \right)
- \frac1{s f_0}  \left( \begin{array}{c} 2  \Phi_0 \cr \sqrt{f_0} \Phi_1 \end{array} \right)
+ \frac2{s^2 f^2_0} \left( \begin{array}{c} \sqrt{f_0} \Phi_1 \cr  0 \end{array} \right) \right]\\
 -\left( \frac16 f_{0x} h_{0x} [s^{J-4}] \right) \left[
 \frac1{s f_0}  \left( \begin{array}{c}  \sqrt{f_0}  \Phi_1 \cr 0 \end{array} \right)\right]\\
-\left( \frac16  f^2_{0x} [s^{J-4}] \right) \left[
 \frac1{s f_0}  \left( \begin{array}{c}  2 \Phi_0  \cr   \sqrt{f_0} \Phi_1  \end{array} \right)
- \frac2{s^2 f^2_0} \left( \begin{array}{c} \sqrt{f_0} \Phi_1 \cr 0 \end{array} \right) \right]
\end{eqnarray*}

\begin{eqnarray*}
= \left( \frac16 h_{0xx} [s^{J-3}] + \frac18 h^2_{0x} [s^{J-4}]\right) \left[ -2 \left( \begin{array}{c} \sqrt{f_0} \Phi_1 \\ {f_0} \Phi_0 \end{array}\right) + \frac2{s} \left( \begin{array}{c} 0 \\ \sqrt{f_0} \Phi_1\end{array}\right)\right]\\
+ \left( \frac14 h_{0xx} [s^{J-3}] + \frac16 h^2_{0x} [s^{J-4}]\right) \left[ 2 \left( \begin{array}{c} \sqrt{f_0} \Phi_1 \\ {f_0} \Phi_0\end{array}\right) 
- \frac1s \left( \begin{array}{c} 0 \\ \sqrt{f_0} \Phi_1\end{array}\right)\right]\\
+ \left(\frac16 f_{0xx} [s^{J-3}]+ \frac28 h_{0x} f_{0x} [s^{J-4}] \right)\left[ - 2 \left( \begin{array}{c}  \Phi_0  \\  \sqrt{f_0} \Phi_1\end{array}\right) + \frac2{s} \left( \begin{array}{c} \frac{\sqrt{f_0}}{f_0}  \Phi_1 \\ 0\end{array}\right)\right]\\
+ \left(\frac14 f_{0xx} [s^{J-3}]+ \frac26 h_{0x} f_{0x} [s^{J-4}] \right)
\left[ 2 \left( \begin{array}{c}  \Phi_0 \cr \sqrt{f_0} \Phi_1 \end{array} \right)
- \frac1{s f_0}  \left( \begin{array}{c} \sqrt{f_0} \Phi_1 \cr 0 \end{array} \right)\right]\\
+ \left( \frac18 f^2_{0x} [s^{J-4}] \right) \left[   - 2 \left( \begin{array}{c} \frac{\sqrt{f_0}}{f_0} \Phi_1 \cr \Phi_0 \end{array} \right)
+ \frac2{s f_0}  \left( \begin{array}{c} 3  \Phi_0 \cr  \sqrt{f_0} \Phi_1 \end{array} \right)
- \frac2{s^2 f^2_0} \left( \begin{array}{c} 3 \sqrt{f_0} \Phi_1 \cr  0 \end{array} \right)\right]\\
+ \left( \frac16 f^2_{0x} [s^{J-4}] \right) \left[2 \left( \begin{array}{c} \frac{\sqrt{f_0}}{f_0} \Phi_1 \cr \Phi_0 \end{array} \right)
- \frac1{s f_0}  \left( \begin{array}{c} 2  \Phi_0 \cr \sqrt{f_0} \Phi_1 \end{array} \right)
+ \frac2{s^2 f^2_0} \left( \begin{array}{c} \sqrt{f_0} \Phi_1 \cr  0 \end{array} \right) \right]\\
 -\left( \frac16 f_{0x} h_{0x} [s^{J-4}] \right) \left[
 \frac1{s f_0}  \left( \begin{array}{c} \sqrt{f_0}  \Phi_1 \cr 0 \end{array} \right)\right]\\
-\left( \frac16 f^2_{0x} [s^{J-4}] \right) \left[
 \frac1{s f_0}  \left( \begin{array}{c}  2 \Phi_0  \cr \sqrt{f_0}  \Phi_1  \end{array} \right)
- \frac2{s^2 f^2_{0}} \left( \begin{array}{c} \sqrt{f_0} \Phi_1 \cr 0 \end{array} \right) \right]
\eeann

\beann
= \left( \frac16 h_{0xx} [s^{J-3}] + \frac1{12} h^2_{0x} [s^{J-4}]\right)   \left( \begin{array}{c} \sqrt{f_0} \Phi_1 \\ {f_0} \Phi_0 \end{array}\right) \\
+ \left( \frac16 f_{0xx} [s^{J-3}] + \frac1{6} h_{0x} f_{0x} [s^{J-4}]\right)   \left( \begin{array}{c} \Phi_0 \\ \sqrt{f_0} \Phi_1 \end{array}\right) \\
+ \left( \frac1{12} h_{0xx} [s^{J-2}] + \frac1{12} h^2_{0x} [s^{J-3}]\right)   \left( \begin{array}{c} 0 \\ \sqrt{f_0} \Phi_1 \end{array}\right) \\
+ \left( \frac1{12} f_{0xx} [s^{J-2}] + \frac1{6} h_{0x} f_{0x} [s^{J-3}]\right)   \left( \begin{array}{c} \frac{\sqrt{f_0}}{f_0}\Phi_1 \\  0 \end{array}\right) \\
+ f^2_{0x} \left( \frac{[s^{J-4}]}{12} \left( \begin{array}{c} \frac{\sqrt{f_0}}{f_0} \Phi_1 \cr \Phi_0 \end{array} \right)   + \frac{[s^{J-3}]}{f_0} \left( \begin{array}{c} \frac1{12} \Phi_0 \cr -  \frac{\sqrt{f_0}}{12} \Phi_1 \end{array} \right) - \frac{[s^{J-2}]}{f^2_0} \left( \begin{array}{c}  \frac{\sqrt{f_0}}{12} \Phi_1 \cr 0 \end{array} \right)  \right)\\
-f_{0x} h_{0x}  \frac{[s^{J-3}]}{6} \left( \begin{array}{c}  \frac{\sqrt{f_0}}{f_0} \Phi_1 \cr 0 \end{array} \right)\\
\eeann

\beann
&=& \left( \frac16 h_{0xx} \left( \begin{array}{c} \psi_2 \\ {f_0} \phi_2 \end{array}\right) + \frac1{12} h^2_{0x} \left( \begin{array}{c} \psi_3 \\ {f_0} \phi_3 \end{array}\right)\right)  \\
&+& \left( \frac16 f_{0xx} \left( \begin{array}{c} \phi_2 \\  \psi_2 \end{array}\right) 
+ \frac1{6} h_{0x} f_{0x} \left( \begin{array}{c} \phi_3 \\  \psi_3 \end{array}\right)\right)  \\
&+& \left( \frac1{12} h_{0xx} \left( \begin{array}{c} 0 \\  \psi_1 \end{array}\right) 
+ \frac1{12} h^2_{0x}  \left( \begin{array}{c} 0 \\ \psi_2  \end{array}\right)\right)  \\
&+& \frac1{12} f_{0xx} \left( \begin{array}{c} \frac1{f_0} \psi_1 \\  0 \end{array}\right) \\
&+& \frac1{12} f^2_{0x} \left[ \left( \begin{array}{c} \frac1{f_0} \psi_3 \\  \phi_3 \end{array}\right) + \frac1{f_0} \left( \begin{array}{c} \phi_2 \\  - \psi_2 \end{array}\right)
- \frac1{f^2_0} \left( \begin{array}{c}  \psi_1 \\  0 \end{array}\right)\right]
\eeann

\subsection{Unwinding Analysis} \label{unwindanal}

To connect back to $h_2$ and $f_1$ we bring (\ref{unwound1}) back in, apply the identity
\beann
\left( \begin{array}{cc} A_{11} & A_{12} \\ A_{21} & A_{22} \end{array} \right)^{-1} &=&
\left( \begin{array}{cc} f_{0x} & h_{0x} \\ f_0 h_{0x} & f_{0x} \end{array} \right)
\eeann
and then also apply the unwinding identity to all terms involving $\phi_3$ and $\psi_3$. 
Thus equations (\ref{unwound1}) - (\ref{unwound4}) become
\beann
- \left( \begin{array}{c} h_2 \\  f_1/f_0\end{array} \right) &=&  \left( \frac2{12} h_{0xx} \left( \begin{array}{c} f_{0x} \psi_2 + f_0 h_{0x} \phi_2 \\  h_{0x}  \psi_2 +  f_{0x} \phi_2 \end{array}\right) + \frac1{12}  h^2_{0x} \left( \begin{array}{c} f_0 \phi_{2x} \\ \psi_{2x} \end{array}\right)\right)  \\
&+& \left( \frac2{12} f_{0xx} \left( \begin{array}{c} f_{0x} \phi_2 + h_{0x} \psi_2 \\  h_{0x}  \phi_2 +\frac{f_{0x}}{f_0} \psi_2 \end{array}\right) 
+ \frac2{12} h_{0x} f_{0x} \left( \begin{array}{c} \psi_{2x} \\   \phi_{2x} \end{array}\right)\right)  \\
&+& \left( \frac1{12} h_{0xx} \psi_1 \left( \begin{array}{c} h_{0x} \\ \frac{f_{0x}}{f_0}  \end{array}\right) 
+ \frac1{12} h^2_{0x} \psi_2 \left( \begin{array}{c} h_{0x} \\  \frac{f_{0x}}{f_0}  \end{array}\right)\right)  \\
&+& \frac1{12} \frac{f_{0xx}}{f_0}  \psi_1 \left( \begin{array}{c}  f_{0x} \\  h_{0x} \end{array}\right) \\
&+& \frac1{12} f^2_{0x} \left[ \frac1{f_0} \left( \begin{array}{c}  f_0 \phi_{2x} \\  \psi_{2x} \end{array}\right) + \frac1{f_0} \left( \begin{array}{c} f_{0x} \phi_2 - h_{0x} \psi_2 \\  h_{0x}  \phi_2 - \frac{f_{0x}}{f_0} \psi_2 \end{array}\right)
- \frac1{f^2_0} \psi_1 \left( \begin{array}{c}  f_{0x} \\   h_{0x} \end{array}\right)\right] \\
&+& \frac{h_{0x}}{8} \left[  \begin{array}{c} \left(\frac{2 f^2_{0x}}{f_0} - h^2_{0x}\right) \psi_2 + f_{0x} h_{0x} \phi_2  \\  \frac{f_{0x}}{f_0} h_{0x} \psi_2 +  h^2_{0x}\phi_2 \end{array} \right]
- \frac2{8}    h_{0xx} \psi_1 \left( \begin{array}{c} h_{0x} \\ \frac{f_{0x}}{f_0}  \end{array}\right) 
\eeann
where the last line comes from line (\ref{unwound1}) using the identity $h_1 = \frac12 h_{0x}$.
Rearranging a bit this becomes

\beann
- \left( \begin{array}{c} h_2 \\  f_1/f_0\end{array} \right) &=&  \frac1{12}\left( 2 h_{0xx} \left( \begin{array}{c} f_{0x} \psi_2 + f_0 h_{0x} \phi_2 \\  h_{0x}  \psi_2 +  f_{0x} \phi_2 \end{array}\right) +   h^2_{0x} \left( \begin{array}{c} f_0 \phi_{2x} \\ \psi_{2x} \end{array}\right)\right)  \\
&+& \frac2{12} \left(  f_{0xx} \left( \begin{array}{c} f_{0x} \phi_2 + h_{0x} \psi_2 \\  h_{0x}  \phi_2 +\frac{f_{0x}}{f_0} \psi_2 \end{array}\right) 
+  h_{0x} f_{0x} \left( \begin{array}{c} \psi_{2x} \\   \phi_{2x} \end{array}\right)\right)  \\
&+& \frac1{12} \frac{f^2_{0x}}{f_0} \left[    \left( \begin{array}{c} f_{0x} \phi_2 - h_{0x} \psi_2 \\  h_{0x}  \phi_2 - \frac{f_{0x}}{f_0} \psi_2 \end{array}\right)
+ \left( \begin{array}{c}  f_0 \phi_{2x} \\  \psi_{2x} \end{array}\right) \right] \\
&+& \left( \frac1{12} h_{0xx} \psi_1 \left( \begin{array}{c} h_{0x} \\ \frac{f_{0x}}{f_0}  \end{array}\right) 
+ \frac1{12} h^2_{0x} \psi_2 \left( \begin{array}{c} h_{0x} \\  \frac{f_{0x}}{f_0}  \end{array}\right)\right)  \\
&+& \frac1{12} \left( \frac{f_{0xx}}{f_0}  - \frac{f^2_{0x}}{f^2_0}\right)\psi_1 \left( \begin{array}{c}  f_{0x} \\  h_{0x} \end{array}\right) \\
&+& \frac{h_{0x}}{8} \left[  \begin{array}{c} \left(\frac{2 f^2_{0x}}{f_0} - h^2_{0x}\right) \psi_2 + f_{0x} h_{0x} \phi_2  \\  \frac{f_{0x}}{f_0} h_{0x} \psi_2 +  h^2_{0x}\phi_2 \end{array} \right]
- \frac1{4} h_{0xx} \psi_1 \left( \begin{array}{c} h_{0x} \\ \frac{f_{0x}}{f_0}  \end{array}\right)
\eeann

Integrating by parts in the second terms of the first three lines above and then making all evident cancellations, we derive
\beann
- \left( \begin{array}{c} h_2 \\  f_1/f_0\end{array} \right) &=&  
\frac1{12} \partial_x \left( \left( f^2_{0x} + f_0 h^2_{0x}\right) \left( \begin{array}{c} \phi_2 \\ \frac1{f_0} \psi_2 \end{array}\right) + 2  h_{0x} f_{0x} \left( \begin{array}{c} \psi_2 \\ \phi_2 \end{array}\right)\right)  \\
&+& \frac1{12} \left[  \begin{array}{c} f_{0x} \left(\frac{ f^2_{0x}}{f_0} - h^2_{0x}\right) \phi_2 - h_{0x} \frac{ f^2_{0x}}{f_0} \psi_2  \\  h_{0x} \frac{f^2_{0x}}{f_0}  \phi_2  \end{array} \right]\\
&+& \left(  \frac1{12} h_{0xx} \psi_1 \left( \begin{array}{c} h_{0x} \\ \frac{f_{0x}}{f_0}  \end{array}\right) 
+ \frac1{12} h^2_{0x} \psi_2 \left( \begin{array}{c} h_{0x} \\  \frac{f_{0x}}{f_0}  \end{array}\right)\right)  \\
&+& \frac1{12} \left( \frac{f_{0xx}}{f_0}  - \frac{f^2_{0x}}{f^2_0}\right)\psi_1 \left( \begin{array}{c}  f_{0x} \\  h_{0x} \end{array}\right) \\
&+& \frac{h_{0x}}{8} \left[  \begin{array}{c} \left(\frac{2 f^2_{0x}}{f_0} - h^2_{0x}\right) \psi_2 + f_{0x} h_{0x} \phi_2  \\  \frac{f_{0x}}{f_0} h_{0x} \psi_2 +  h^2_{0x}\phi_2 \end{array} \right]
- \frac1{4} h_{0xx} \psi_1 \left( \begin{array}{c} h_{0x} \\ \frac{f_{0x}}{f_0}  \end{array}\right)
\eeann

\beann
 &=&  
\frac1{12} \partial_x \left( \left( f^2_{0x} + f_0 h^2_{0x}\right) \left( \begin{array}{c} \phi_2 \\ \frac1{f_0} \psi_2 \end{array}\right) + 2  h_{0x} f_{0x} \left( \begin{array}{c} \psi_2 \\ \phi_2 \end{array}\right)\right)  \\
&+& \frac1{12} \left[  \begin{array}{c} f_{0x} \left(\frac{ f^2_{0x}}{f_0} - h^2_{0x}\right) \phi_2 - h_{0x} \left(\frac{ f^2_{0x}}{f_0} - h^2_{0x}\right)  \psi_2  \\   \frac{h_{0x} f_{0x}}{f_0}  
\left( f_{0x} \phi_2  + h_{0x} \psi_2 \right)\end{array} \right]\\
&+&  \frac1{12} h_{0xx} \psi_1 \left( \begin{array}{c} h_{0x} \\ \frac{f_{0x}}{f_0}  \end{array}\right)  \\
&+& \frac1{12} \partial_x \left( \frac{f_{0x}}{f_0}\right)\psi_1 \left( \begin{array}{c}  f_{0x} \\  h_{0x} \end{array}\right) \\
&+& \frac{h_{0x}}{8} \left[  \begin{array}{c} \left(\frac{2 f^2_{0x}}{f_0} - h^2_{0x}\right) \psi_2 + f_{0x} h_{0x} \phi_2  \\  \frac{f_{0x}}{f_0} h_{0x} \psi_2 +  h^2_{0x}\phi_2 \end{array} \right]
- \frac1{4} h_{0xx} \psi_1 \left( \begin{array}{c} h_{0x} \\ \frac{f_{0x}}{f_0}  \end{array}\right)
\eeann

\beann
 &=&  
\frac1{12} \partial_x \left( \left( f^2_{0x} + f_0 h^2_{0x}\right) \left( \begin{array}{c} \phi_2 \\ \frac1{f_0} \psi_2 \end{array}\right) + 2  h_{0x} f_{0x} \left( \begin{array}{c} \psi_2 \\ \phi_2 \end{array}\right)\right)  \\
&+& \frac1{12} \left[  \begin{array}{c}  \left(\frac{ f^2_{0x}}{f_0} - h^2_{0x}\right) \left( f_{0x} \phi_2 - h_{0x}  \psi_2 \right) \\   \frac{h_{0x} f_{0x}}{f_0}  
\left( f_{0x} \phi_2  + h_{0x} \psi_2 \right)\end{array} \right]\\
&+&  \frac1{12} h_{0xx} \psi_1 \left( \begin{array}{c} h_{0x} \\ \frac{f_{0x}}{f_0}  \end{array}\right)  \\
&+& \frac1{12} \partial_x \left( \frac{f_{0x}}{f_0}\right)\psi_1 \left( \begin{array}{c}  f_{0x} \\  h_{0x} \end{array}\right) \\
&+& \frac{h_{0x}}{8} \left[  \begin{array}{c} \left(\frac{ f^2_{0x}}{f_0} - h^2_{0x}\right) \psi_2 + f_{0x} \left(  \frac{f_{0x}}{f_0} \psi_2 + h_{0x} \phi_2  \right)\\  h_{0x}  \left( \frac{f_{0x}}{f_0} \psi_2 +  h_{0x}\phi_2 \right) \end{array} \right]
- \frac1{4} h_{0xx} \psi_1 \left( \begin{array}{c} h_{0x} \\ \frac{f_{0x}}{f_0}  \end{array}\right)
\eeann

\beann
 &=&  
\frac1{12} \partial_x  \left( \begin{array}{c} f_{0x} \psi_{1x}  + f_0 h_{0x} \phi_{1x} \\ f_{0x} \phi_{1x}  + h_{0x} \psi_{1x} \end{array}\right)   \\
&+& \frac1{12} \left[  \begin{array}{c}  \left(\frac{ f^2_{0x}}{f_0} - h^2_{0x}\right) \left( f_{0x} \phi_2 + h_{0x}  \psi_2 \right)  
- 2 h_{0x}  \left(\frac{ f^2_{0x}}{f_0} - h^2_{0x}\right)    \psi_2  \\   \frac{h_{0x} f_{0x}}{f_0}  
\left( f_{0x} \phi_2  + h_{0x} \psi_2 \right)\end{array} \right]\\
&+&  \frac1{12} h_{0xx} \psi_1 \left( \begin{array}{c} h_{0x} \\ \frac{f_{0x}}{f_0}  \end{array}\right)  \\
&+& \frac1{12} \partial_x \left( \frac{f_{0x}}{f_0}\right)\psi_1 \left( \begin{array}{c}  f_{0x} \\  h_{0x} \end{array}\right) \\
&+& \frac{h_{0x}}{8} \left[  \begin{array}{c} \left(\frac{ f^2_{0x}}{f_0} - h^2_{0x}\right) \psi_2 + f_{0x} \left(  \frac{f_{0x}}{f_0} \psi_2 + h_{0x} \phi_2  \right)\\  h_{0x}  \left( \frac{f_{0x}}{f_0} \psi_2 +  h_{0x}\phi_2 \right) \end{array} \right]
- \frac1{4} h_{0xx} \psi_1 \left( \begin{array}{c} h_{0x} \\ \frac{f_{0x}}{f_0}  \end{array}\right)
\eeann

\beann
 &=&  
\frac1{12} \partial_x  \left( \begin{array}{c} f_{0x} \psi_{1x}  +f_0  h_{0x} \phi_{1x} \\ f_{0x} \phi_{1x}  +  h_{0x} \psi_{1x} \end{array}\right)   \\
&+& \partial_x\left(\frac1{12} \left[  \begin{array}{c}  \left(\frac{ f^2_{0x}}{f_0} - h^2_{0x}\right)  \\   \frac{h_{0x} f_{0x}}{f_0}  
\end{array} \right] \psi_{1} + \frac{h_{0x}}{8} \left[  \begin{array}{c}  f_{0x} \\  h_{0x} \end{array} \right] \phi_{1}\right)\\
&-& \frac1{24} \left( \begin{array}{c} 4 h_{0x} h_{0xx} \psi_1 + 3 \left( h_{0x} f_{0xx} - h_{0xx} f_{0x} \right) \phi_1 \\ 0 \end{array} \right) \\
&-& \frac{h_{0x}}{24} \left[  \begin{array}{c}  \left(\frac{ f^2_{0x}}{f_0} - h^2_{0x}\right)  \\   0  
\end{array} \right] \psi_2
\eeann
Making use of the unwinding  identity at level 0 (and, in the last line, at level 2) we deduce that
\beann
 &=&  
\frac1{12} \partial^2_x  \left( \begin{array}{c} 0 \\ 1 \end{array}\right)   
- \frac1{12} \partial_x  \left( \begin{array}{c} f_{0xx} \psi_{1}  +  \partial_x (h_{0x} f_{0}) \phi_{1} \\ f_{0xx} \phi_{1}  + h_{0xx} \psi_{1} \end{array}\right) \\
&+& \partial_x\left(\frac1{12} \left[  \begin{array}{c}  \left(\frac{ f^2_{0x}}{f_0} - h^2_{0x}\right)  \\   0 
\end{array} \right] \psi_{1} +  h_{0x} \left[  \begin{array}{c} \frac18  f_{0x} \\ \frac1{24}  h_{0x} \end{array} \right] \phi_{1}\right)\\
&-& \left( \begin{array}{c} \frac1{12 f_0} f_{0x} f_{0xx} \psi_1 + \frac18 \left( h_{0x} f_{0xx} + h_{0xx} f_{0x} \right) \phi_1 \\ 0 \end{array} \right) \\
&+& \frac{h_{0x}}{24} \left[  \begin{array}{c}  h_{0x} \psi_{1x} -f_{0x} \phi_{1x}  \\   0  
\end{array} \right]\\
&=&  
- \frac1{12} \partial_x  \left( \begin{array}{c} f_{0xx} \psi_{1}  +\partial_x (h_{0x} f_{0}) \phi_{1} \\ f_{0xx} \phi_{1}  + h_{0xx} \psi_{1} \end{array}\right) \\
&+& \partial_x\left(\frac1{12} \left[  \begin{array}{c}  \left(\frac{ f^2_{0x}}{f_0} - h^2_{0x}\right)  \\   0 
\end{array} \right] \psi_{1} +  h_{0x} \left[  \begin{array}{c} \frac18  f_{0x} \\ \frac1{24}  h_{0x} \end{array} \right] \phi_{1}\right)\\
&-& \left( \begin{array}{c} \frac1{12 f_0} f_{0x} f_{0xx} \psi_1 + \frac18 \left( h_{0x} f_{0xx} + h_{0xx} f_{0x} \right) \phi_1 \\ 0 \end{array} \right) \\
&+& \frac{h_{0x}}{24} \partial_x \left[  \begin{array}{c}  h_{0x} \psi_{1} -f_{0x} \phi_{1}  \\   0  \end{array} \right] 
- \frac{h_{0x}}{24}  \left[  \begin{array}{c}  h_{0xx} \psi_{1} -f_{0xx} \phi_{1}  \\   0  \end{array} \right]
\eeann

\beann
&=&  
- \frac1{12} \partial_x  \left( \begin{array}{c} f_{0xx} \psi_{1}  + \partial_x (h_{0x} f_{0}) \phi_{1} \\ f_{0xx} \phi_{1}  + h_{0xx} \psi_{1} \end{array}\right) \\
&+& \partial_x\left(\frac1{12} \left[  \begin{array}{c}  \left(\frac{ f^2_{0x}}{f_0} - h^2_{0x}\right)  \\   0 
\end{array} \right] \psi_{1} +  h_{0x} \left[  \begin{array}{c} \frac18  f_{0x} \\ \frac1{24}  h_{0x} \end{array} \right] \phi_{1}\right)\\
&-& \left( \begin{array}{c} \frac1{12 f_0} f_{0x} f_{0xx} \psi_1 + \frac18 \left( h_{0x} f_{0xx} + h_{0xx} f_{0x} \right) \phi_1 \\ 0 \end{array} \right) \\
&+& \frac{1}{24} \partial_x\left( h_{0x} \left[  \begin{array}{c}  h_{0x} \psi_{1} -f_{0x} \phi_{1}  \\   0  \end{array} \right] \right)
- \frac{1}{24}  \left[  \begin{array}{c}  \partial_x\left(h^2_{0x}\right) \psi_{1} - \partial_x\left( h_{0x}f_{0x} \right) \phi_{1}  \\   0  \end{array} \right] \\
&=&  
- \frac1{12} \partial_x  \left( \begin{array}{c} f_{0xx} \psi_{1}  + \partial_x (h_{0x} f_{0}) \phi_{1} \\ f_{0xx} \phi_{1}  + h_{0xx} \psi_{1} \end{array}\right) \\
&+& \partial_x\left(\frac1{12} \left[  \begin{array}{c}  \left(\frac{ f^2_{0x}}{f_0} - \frac12 h^2_{0x}\right)  \\   0 
\end{array} \right] \psi_{1} +  h_{0x} \left[  \begin{array}{c} \frac1{12}  f_{0x} \\ \frac1{24}  h_{0x} \end{array} \right] \phi_{1}\right)\\
&-& \frac{1}{24}  \left[  \begin{array}{c}  \left( \frac1{f_0} \partial_x \left( f^2_{0x} \right) + \partial_x\left(h^2_{0x}\right) \right) \psi_{1} + 2 \partial_x\left( h_{0x}f_{0x} \right) \phi_{1}  \\   0  \end{array} \right] 
\eeann

\beann
&=&  
\frac1{12} \partial_x  \left( \begin{array}{c} f_{0x} \left( \frac{f_{0x}}{f_0} \psi_1 + h_{0x} \phi_1 \right) - h_{0x} \left( h_{0x} \psi_{1}  + f_{0x} \phi_{1} \right) + \left( \frac12 h^2_{0x} - f_{0xx} \psi_1 \right) - f_0 h_{0xx} \phi_1 \\ \frac12 h^2_{0x} - f_{0xx} \phi_{1}  - h_{0xx} \psi_{1} \end{array}\right) \\
&-& \frac{1}{24}  \left[  \begin{array}{c}  \left( \frac1{f_0} \partial_x \left( f^2_{0x} \right) + \partial_x\left(h^2_{0x}\right) \right) \psi_{1} + 2 \partial_x\left( h_{0x}f_{0x} \right) \phi_{1}  \\   0  \end{array} \right]\\
&=&  
\frac1{12} \partial_x  \left( \begin{array}{c} f_{0x} \phi_{0x} - h_{0x} \psi_{0x} + \left( \frac12 h^2_{0x} - f_{0xx} \psi_1 \right) - f_0 h_{0xx} \phi_1 \\ \frac12 h^2_{0x} - f_{0xx} \phi_{1}  - h_{0xx} \psi_{1} \end{array}\right) \\
&-& \frac{1}{24}  \left[  \begin{array}{c}  \left( \frac1{f_0} \partial_x \left( f^2_{0x} \right) + \partial_x\left(h^2_{0x}\right) \right) \psi_{1} + 2 \partial_x\left( h_{0x}f_{0x} \right) \phi_{1}  \\   0  \end{array} \right]\\
&=&  
- \frac1{12} \left( \begin{array}{c} h_{0xx} \\ 0\end{array} \right) + \frac1{12} \partial_x  \left( \begin{array}{c}  \left( \frac12 h^2_{0x} - f_{0xx} \psi_1 \right) - f_0 h_{0xx} \phi_1 \\ \frac12 h^2_{0x} - f_{0xx} \phi_{1}  - h_{0xx} \psi_{1} \end{array}\right) \\
&-& \frac{1}{24}  \left[  \begin{array}{c}  \left( \frac1{f_0} \partial_x \left( f^2_{0x} \right) + \partial_x\left(h^2_{0x}\right) \right) \psi_{1} + 2 \partial_x\left( h_{0x}f_{0x} \right) \phi_{1}  \\   0  \end{array} \right].
\eeann
Now, making the substitution
\begin{eqnarray} \label{BaseUnwind}
\left( \begin{array}{c} \phi_1 \\ \psi_1 \end{array} \right) &=& \frac1{h^2_{0x} - f^2_{0x}/f_0}\left( \begin{array}{c} \frac{-f_{0x}}{f_0} \\ h_{0x}\end{array} \right) ,
\end{eqnarray} 
which follows from (\ref{A11}), (\ref{A12}) and (\ref{Ainv}),
one finally  has
\begin{eqnarray*}
 \left( \begin{array}{c} h_2 \\  f_1/f_0\end{array} \right) 
&=&  \frac1{12} \left( \begin{array}{c} h_{0xx} \\ 0\end{array} \right) + \frac1{24} \partial_x \frac1{f^2_{0x} - f_0 h^2_{0x}} \left( \begin{array}{c}  f_0 h^3_{0x} -2 f_0 f_{0xx} h_{0x} + 2 f_0 f_{0x} h_{0xx} \\ -f_{0x} h^2_{0x} + 2 f_{0x} f_{0xx} + 2 f_0 h_{0x} h_{0xx}\end{array}\right) \\
&+& \frac1{12} \left( \begin{array}{c} h_{0xx} \\ 0\end{array} \right)\\
&=&  \frac1{6} \left( \begin{array}{c} h_{0xx} \\ 0\end{array} \right) + \frac1{24} \partial_x  \left( \begin{array}{c}  \frac{f_0 h^3_{0x} -2 f_0 f_{0xx} h_{0x} + 2 f_0 f_{0x} h_{0xx}}{f^2_{0x} - f_0 h^2_{0x}} \\ \partial_x \log \left( f^2_{0x} - f_0 h^2_{0x}\right)\end{array}\right).
\end{eqnarray*}

\section{Fine Structure of the Generating Functions} \label{finestruc}

We begin by summarizing the closed form expressions of some of the basic generating functions we
have just derived: 
\begin{eqnarray} \nonumber
h_1 &=& \frac12 h_{0x}\\ \nonumber
f_1/f_0 &=& \frac1{24}\partial^2_x \log \left( f^2_{0x} - f_0 h^2_{0x}\right) \\ \label{INDBASE}
&=& \frac1{24}\partial_x \left( \frac{-f_{0x} h^2_{0x} + 2 f_{0x} f_{0xx} + 2 f_0 h_{0x} h_{0xx}}{f^2_{0x} - f_0 h^2_{0x}}\right)\\ \nonumber
h_2 &=& \partial_x \left[ \frac16 h_{0x} + \frac1{24}\frac{f_0 h^3_{0x} -2 f_0 f_{0xx} h_{0x} + 2 f_0 f_{0x} h_{0xx}}{f^2_{0x} - f_0 h^2_{0x}} \right].
\end{eqnarray} 
We also add here some additional representations involving these basic generating functions that may prove useful for comparison,
\begin{eqnarray}  \label{t2x}
f_1/f_0 &=& \partial^2_x \left[ E_1 + \frac1{12} \log z_0\right]
\end{eqnarray}
which follows directly from (\ref{boson12}-\ref{Cum}).  This can be turned around to yield several representations of $E_1$, 
\begin{eqnarray} \nonumber
E_1 &=& \frac1{24} \log\left( \frac{f^2_{0x} - f_0 h^2_{0x}}{z^2_0}\right) \\ \nonumber
&=& \frac1{24} \log\left( \frac{4 - y_0}{z_0^2 \widehat{D}} \right)\\ \label{E1}
 &=& \frac1{24} \log\left( \frac{\widehat{B}_{12}^2}{\Pi_- \Pi_+} \right) -\frac1{12} \log(z_0)
\end{eqnarray}
where the last two equalities follow form (\ref{basic1} - \ref{basic2}) and (\ref{FundDisc}) respectively. Similarly, $h_2$ may be re-expressed, using (\ref{BOSON}), as 
\begin{eqnarray} \label{t1x}
h_2 &=&  - \partial_x \left[ \frac16 h_{0x} +  \partial_{t_1} E_1\right].
\end{eqnarray}
(We recall here that $\partial_{t_1} E_1$ is the generating function for 1-legged genus 1 maps.)

\subsection{Leading Orders}
From (\ref{E1}) and our prior analysis we see that the singularities of $E_1$ correspond precisely to degenerations in the characteristic geometry: the vanishing of $\widehat{B}_{12} - y_0^{1/2} \widehat{B}_{11}$ coincides with places where the characteristic speeds stagnate and the vanishing of $\Pi_- \Pi_+$ coincides with the formation of caustics. Except in the case of valence $j = 3$ these places are precisely the locations of the vertical tangents of $\mathcal{C}$ over finite values of $\xi^2$.

However, the singularities corresponding to the zeroes of $\widehat{B}_{12} - y_0^{1/2} \widehat{B}_{11}$ are removed in $h_1, h_2, f_1$.
\begin{prop} \label{basestep}
$h_1/h_0, h_2/h_0$, and $f_1$ are rational functions on $\mathcal{C}$ whose singularities are restricted to poles located at the zeroes of $\Pi_-$ or $\Pi_+$. In fact the global polar parts of these functions are a power of $\Pi_-\Pi_+$. This power will be referred to as the {\it polar order} of the generating function. The polar order of $h_2/h_0$ and $f_1$ is 4. The polar order of $h_1/h_0$ is 1.
\end{prop}
\begin{proof}
It follows from (\ref{h0x}) that
\begin{eqnarray*}
\frac{h_{0x}}{h_0} 
&=& \frac{\widehat{B}_{12}}{2f_0 \sqrt{y_0}} \left( \frac{1}{\Pi_-} 
- \frac{1}{\Pi_+}\right) \\
&=& \frac{(\widehat{B}_{12} - y_0^{1/2} \widehat{B}_{11})}{2x \Pi_- \Pi_+} \frac{\Pi_+ - \Pi_-} {\sqrt{y_0}}
\end{eqnarray*}
If follows from (\ref{PiSymm}) that $\Pi_+(\sqrt{y_0}) - \Pi_-(\sqrt{y_0}) = \Pi_-(-\sqrt{y_0}) - \Pi_+(-\sqrt{y_0})$ which is therefore an odd function of $\sqrt{y_0}$ and so 
$\left(\Pi_+(\sqrt{y_0}) - \Pi_-(\sqrt{y_0})\right)/\sqrt{y_0}$ is a polynomial in $y_0$, showing that this is indeed a rational function on $\mathcal{C}$. Moreover, the singularities of $h_{1}/h_0$ are then seen to be confined to the zeroes of $\Pi_\pm$ and there these are poles of order 1 (i.e., simple). A similar argument applies for $f_{0x}$.

We next turn to our representations for $f_1/f_0$ and $h_2$. To that end we will make use of the following lemma.
\begin{lem} \label{y0x}
$y_{0x} = (j-2) \frac{y_0(4 - y_0)}{f_0 \widehat{D}} = (j-2) \frac{y_0 \widehat{B}_{12}(\widehat{B}_{12} - y_0^{1/2} \widehat{B}_{11})}{x \Pi_- \Pi_+} $
\end{lem}
\begin{proof}
As in (\ref{y0deriv}) it is straightforward to calculate that
\begin{eqnarray*}
y_{0x} &=& \frac{h_0}{f^2_0} \left( 2 f_0, -h_0\right) \cdot \left( \begin{array}{c}  h_{0x} \\ f_{0x} \end{array}\right).
\end{eqnarray*}
Applying (\ref{odvalinv}) this expands to 
\begin{eqnarray*}
y_{0x} &=& \frac1D \frac{h_0}{f^2_0} \left( 2 f_0, -h_0\right) \cdot  \left( j x + (j-2) \left(
\begin{array}{cc}
- f_0 & h_{0} \\ f_{0}h_0 & - f_0
\end{array}
\right) \right) \left(
\begin{array}{c}
h_{0} \\2  f_{0}
\end{array}
\right)\\
&=& \frac{(j-2)}D \frac{h_0}{f^2_0} \left( 2 f_0, -h_0\right) \cdot  \left(
\begin{array}{cc}
- f_0 & h_{0} \\ f_{0}h_0 & - f_0
\end{array}
\right) \left(
\begin{array}{c}
h_{0} \\2  f_{0}
\end{array}
\right)\\
&=& \frac{(j-2)}{\widehat{D}} \frac{h^2_0}{f^4_0} (4 - y_0)\\
&=& \frac{(j-2)}{\widehat{D}} \frac{y_0}{f_0} (4 - y_0)
\end{eqnarray*}
\end{proof}

Now applying the lemma to our previously derived expressions for $f_1/f_0$ and $h_2$ starting from $(\ref{t2x})$ and (\ref{t1x}), one has
\begin{eqnarray*}
f_1/f_0 &=& \partial_x \left[ \partial_x E_1 + \frac1{12}  z_{0x}/z_0\right] 
= - \frac1{24} \partial_x\left( \frac{\Pi_{+x}}{\Pi_+} + \frac{\Pi_{-x}}{\Pi_-} \right)\\
&=& -\frac{(j-2)}{24} \partial_x \left[ \frac{y_0(4 - y_0)}{f_0 \widehat{D}}\frac{\left(\Pi_- \Pi_+\right)_{y_0}}{\Pi_- \Pi_+} \right]
= -\frac{(j-2)}{24} \partial_x \left[ \frac{y_0\left(\widehat{B}_{12} - y_0^{1/2} \widehat{B}_{11}\right)}{x \widehat{B}_{12}}\frac{\left(\Pi_- \Pi_+\right)_{y_0}}{(\Pi_- \Pi_+)^2}\right]  
\end{eqnarray*}

\begin{eqnarray*}
h_2 &=& \partial_x \left[ \frac16 h_{0x} -  \partial_{t_1} E_1\right]
= \partial_x \left[ \frac16 h_{0x} - \partial_{y_0} E_1  y_{0t_1} \right]\\
&=& \partial_x \left[ \frac16 h_{0x} - \partial_{y_0} E_1  \frac{h_0}{f^2_0} \left( 2 f_0, -h_0\right) \cdot \left( \begin{array}{c}  h_{0t_1} \\ f_{0t_1} \end{array}\right) \right]\\
&=& \partial_x \left[ \frac16 h_{0x} +\partial_{y_0} E_1 \frac1D \frac{h_0}{f^2_0} \left( 2 f_0, -h_0\right) \cdot \left( \begin{array}{cc} 0 & 1 \\ f_0 & 0 \end{array} \right)\left( j x + (j-2) \left(
\begin{array}{cc}
- f_0 & h_{0} \\ f_{0}h_0 & - f_0
\end{array}
\right) \right) \left(
\begin{array}{c}
h_{0} \\2  f_{0}
\end{array}
\right) \right]\\
&=& \partial_x \left[ \frac16 h_{0x} +\partial_{y_0} E_1 \frac1D \frac{h_0}{f_0} \left( -h_0, 2\right) \cdot \left( j x + (j-2) \left(
\begin{array}{cc}
- f_0 & h_{0} \\ f_{0}h_0 & - f_0
\end{array}
\right) \right) \left(
\begin{array}{c}
h_{0} \\2  f_{0}
\end{array}
\right) \right]\\
&=& \partial_x \left[ \frac16 h_{0x} +\partial_{y_0} E_1 \frac{h_0}D (4 - y_0) (jx - (j-2) f_0  )  \]\\
&=& \partial_x \left[ \frac16 h_{0x} +\frac1{24} \frac{h_0}{w^2}\partial_{y_0}\log\left( \frac{(4 - y_0)}{\widetilde{D}}\right) \frac{(4 - y_0)}{\widetilde{D}} (jx - (j-2) f_0  )  \]
\end{eqnarray*}

\begin{eqnarray*}
&=& \partial_x \left\{ h_0\left[ \frac16 \frac{h_{0x}}{h_0} +\frac1{24} \frac{1}{w^2}\partial_{y_0}\left( \frac{(4 - y_0)}{f^2_0\widehat{D}}\right)  (jx - (j-2) f_0  )  \] \right\}\\
&=& \partial_x \left\{ h_0 \left[ \frac1{12} \frac{(\widehat{B}_{12} - y_0^{1/2} \widehat{B}_{11})}{x \Pi_- \Pi_+} \frac{\Pi_+ - \Pi_-} {\sqrt{y_0}} +\frac1{24} \frac{1}{x^2}\partial_{y_0}\left( \frac{\left(\widehat{B}_{12} - y_0^{1/2} \widehat{B}_{11}\right)^2}{\Pi_+ \Pi_-}\right)  (jx - (j-2) f_0  )  \] \right\} \\ 
&=& \partial_x \left\{\sqrt{\frac{y_0 \widehat{B}_{12}}{\widehat{B}_{12} - y_0^{1/2} \widehat{B}_{11}}} \left[ \frac1{12} \frac{(\widehat{B}_{12} - y_0^{1/2} \widehat{B}_{11})}{x \Pi_- \Pi_+} \frac{\Pi_+ - \Pi_-} {\sqrt{y_0}} \right. \right. \\
&+& \left. \left. \frac1{24} \frac{1}{x^2}\partial_{y_0}\left( \frac{\left(\widehat{B}_{12} - y_0^{1/2} \widehat{B}_{11}\right)^2}{\Pi_+ \Pi_-}\right)  \left(jx - (j-2) \frac{ \widehat{B}_{12}}{\widehat{B}_{12} - y_0^{1/2} \widehat{B}_{11}}  \right)  \]\right\}
\end{eqnarray*}

Let us now check the singularity confinement claims, beginning with the next to last line in the expression for $h_2$. The terms inside the square brackets, based on our prior arguments, a rational function of $y_0$ with poles confined to the zeroes of $\Pi_- \Pi_+$. Hence, the full expression divided by $h_0$ has the form
\begin{eqnarray*}
h_2/h_0 &=& \frac{h_{0x}}{h_0}\left[ \dots\right] + \partial_x \left[ \dots\right]\\
&=& \frac{h_{0x}}{h_0}\left[ \dots\right] + \partial_{y_0} \left[ \dots\right] y_{0x}.
\end{eqnarray*}
Since we already saw, at the start of the proof, that $h_{0x}/h_0$ is a rational function of $y_0$
with poles constrained to $\Pi_- \Pi_+$, it follows that the first term above is a product of two such functions. Similarly for the second term since $y_{0x}$ is such a rational function by Lemma \ref{y0x}.

The case of $f_1/f_0$ is similar but one also has singularities at the zeroes of $\widehat{B}_{12}$ which are off at infinity in $\xi^2$. However, we will see shortly that these latter poles are not present in $f_1$ by itself. Hence, the singularities of $f_1$ and $h_2/h_0$ are confined to the the zeroes of $\Pi_-$ or $\Pi_+$, and that they are poles there. 

 We may also determine the order of the poles along $\Pi_\pm = 0$. Consider the case of $f_1/f_0$. Carrying out the $x$-derivtive in the fnal expression yields another expression of the form
\begin{eqnarray*}
-\left[ \frac{y_0\left(\widehat{B}_{12} - y_0^{1/2} \widehat{B}_{11}\right)}{x^2 \widehat{B}_{12}}\frac{\left(\Pi_- \Pi_+\right)_{y_0}}{(\Pi_- \Pi_+)^2}\right]
&+& \left[ \frac{y_0\left(\widehat{B}_{12} - y_0^{1/2} \widehat{B}_{11}\right)}{x \widehat{B}_{12}}\frac{\left(\Pi_- \Pi_+\right)_{y_0}}{(\Pi_- \Pi_+)^2}\right]_{y_0} y_{0x}\\
= -\left[ \frac{y_0\left(\widehat{B}_{12} - y_0^{1/2} \widehat{B}_{11}\right)}{x^2 \widehat{B}_{12}}\frac{\left(\Pi_- \Pi_+\right)_{y_0}}{(\Pi_- \Pi_+)^2}\right]
&+& \left[ \frac{y_0\left(\widehat{B}_{12} - y_0^{1/2} \widehat{B}_{11}\right)}{x \widehat{B}_{12}}\frac{\left(\Pi_- \Pi_+\right)_{y_0}}{(\Pi_- \Pi_+)^2}\right]_{y_0}  (j-2) \frac{y_0 \widehat{B}_{12}(\widehat{B}_{12} - y_0^{1/2} \widehat{B}_{11})}{x \Pi_- \Pi_+}
\end{eqnarray*}
Near the poles the second term is clearly the dominant one and asymptotically as one approaches the singularities it grows like $\mathcal{O}\left( \frac1{(\Pi_- \Pi_+)^4}\right)$. Hence the order of the poles is uniformly 4. We also note that the terms inside brackets all have a common factor of the form 
$\frac{\widehat{B}_{12} - y_0^{1/2} \widehat{B}_{11}}{x \widehat{B}_{12}}$ which is precisely equal to $1/f_0$. Hence, multiplying by $f_0$ and taking into account the factors coming from $y_{0x}$ in the above $x$-derrivative, one sees that $f_1$ has no poles other than those corrseponding to the zeroes of $\Pi_\pm$.
A completely similar argument shows that the order of the poles for $h_2/h_0$ is also 4.  We note that this does not rule out the possibility of multiplicity of poles or cancellation of poles by zeroes in the numerator at the complex zeroes of $\Pi_- \Pi_+$. However, at the real zeroes our previous analysis shows that this does not happen.
\end{proof}
\medskip

\begin{cor} \label{hyperellip}
$h_1$ and $h_2$ each have the form of a rational function of $y_0$, with poles localized in the zero set of $\Pi_- \Pi_+$, times $h_0$. The latter is not 
uniquely defined on  $\mathcal{C}$; however, it is well-defined on $\widehat{\mathcal{C}}$ and,  hence, so are $h_1$ and $h_2$.  The polar order, defined in Proposition \ref{basestep} of $h_1$ and $h_2$ on $\widehat{\mathcal{C}}$ is the same as that of $h_1/h_0$ and $h_2/h_0$ on 
 $\mathcal{C}$.
\begin{proof}
Observing that $h_1 = \frac12 h_{0x}$ it is clear, from the analysis at the start of the proof of the proposition, that $h_1/h_0$ is a rational function on 
$\mathcal{C}$. Similarly, the analysis at the end of that proof establishes that $h_2/h_0$ is rational on $\mathcal{C}$. , 
So the matter comes down to showing that $h_0$ is definable on $\widehat{\mathcal{C}}$. It follows from (\ref{spectralcurve}), using (\ref{s2nu}), that 
\beann
h_0  = x^{1/2-\nu} \xi  {\phi_0}.
\eeann
Since $\phi_0$ is a polynomial function of $y_0$ and $f_0$,  it is a rational function on $\mathcal{C}$ and so its pullback is rational on  $\widehat{\mathcal{C}}$. It follows that $h_0$ is well-defined on $\widehat{\mathcal{C}}$ since $\xi$ is a coordinate function on this hyperelliptic curve. 
\end{proof}
\end{cor}

\subsection{Higher Orders}
To extend this analysis to higher order generating functions, we will make use of the following lemma which is essentially a consequence of the Fa{\'a} di Bruno formula \cite{FS}.
\begin{lem} \label{FdB}
Let $G(h,f)$ be a polynomial in the generating function series 
\begin{eqnarray*}
h &=& \sum_{g=0}^\infty h_g n^{-g}\\
f &=& \sum_{g=0}^\infty f_g n^{-2g};
\end{eqnarray*}
then  
\begin{eqnarray} \label{FdBeq}
[n^{-g}] G(h,f) &=& \sum_{|k^{(a)}| + 2 |k^{(b)}| = g} 
\prod_i h_{k^{(a)}_i} \prod_i f_{k^{(b)}_i}  \frac{\partial_{h_0}^{\ell(k^{(a)})}}{\ell(k^{(a)})!} \frac{\partial_{f_0}^{\ell(k^{(b)})}}{\ell(k^{(b)})!} G(h_0,f_0)
\end{eqnarray}
where the sum is over integer partition pairs $(k^{(a)},k^{(b)})$, $k^{(a)} = (k^{(a)}_1, k^{(a)}_2, \dots)$, $k^{(b)} = (k^{(b)}_1, k^{(b)}_2, \dots)$ satisfying the summand constraint.
\end{lem}

\begin{thm} \label{uz}
The generating functions, $h_g(\xi)$ and $f_g(\xi)$ for $g > 0$, are rational functions on $\widehat{\mathcal{C}}$ (hyperelliptic functions) whose singularities are poles confined to the locus where $\Pi_+$ or $\Pi_-$ vanish.   The polar expansions at these zeroes are in fact locally rational functions of $y_0$ and so locally project to the coordinate ring,  ${\bf \tilde{S}}$, of $\mathcal{C}$.     The $f_g(\xi)$ are, in fact, rational functions on $\mathcal{C}$ and so are globally elements of ${\bf \tilde{S}}$ localized along $\Pi$.    
\end{thm}
\begin{proof}
Lemma \ref{FdB} may be applied directly to the polynomials $\tilde{P}^{(a)}_{\alpha, \beta, j} (h, f)$ and $\tilde{P}^{(b)}_{\alpha, \beta, j} (h, f)$ appearing in the string equations (\ref{ContString}) to conclude that the terms involving $f_g$ and $h_{2g}$ at order $[n^{-2g}]$ in those polynomials  come down to $\tilde{P}_{\emptyset, \emptyset, j} (h, f)$ and 
$\tilde{P}_{\emptyset, \emptyset, j} (h, f)$. Otherwise derivative factors, of the form 
$\partial_x^\alpha h$ or $\partial_x^\beta f$ carrying a non-zero weight in $n$,  enter which prevent $f_k$ or $h_k$ from attaining the maximal values of their subscripts at this order.  Evaluating on just these string coefficients then yields

\beann
 [n^{-2g}]
\left( \begin{array}{c}\tilde{P}^{(a)}_{\emptyset, \emptyset, j} (h, f) \\ \tilde{P}^{(b)}_{\emptyset,\emptyset, j} (h, f) \end{array} \right) 
&=&  [n^{-2g}]  \left( \begin{array}{c}   [\eta^{0} z^{j-1}] \cr [\eta^{-1} z^{j-1}] f^{1/2}  \end{array} \right) \frac1{1 - z(\sqrt{f} \eta + h + \sqrt{f} \eta^{-1})} \\
&=&  [n^{-2g}]  \left( \begin{array}{c}   [\eta^{0}] \cr [\eta^{-1}] f^{1/2}  \end{array} \right)
(\sqrt{f} \eta + h + \sqrt{f} \eta^{-1})^{j-1}\\
&=& [n^{-2g}]  \left( \begin{array}{c}   \phi_0 \cr \psi_0  \end{array} \right)\\
&=& \sum_{|k^{(a)}| + 2 |k^{(b)}| = 2g} 
\prod_i h_{k^{(a)}_i} \prod_i f_{k^{(b)}_i}  \frac{\partial_{h_0}^{\ell(k^{(a)})}}{\ell(k^{(a)})!} \frac{\partial_{f_0}^{\ell(k^{(b)})}}{\ell(k^{(b)})!} \left( \begin{array}{c}   \phi_0(h_0,f_0) \cr \psi_0(h_0,f_0)  \end{array} \right)
\eeann
From this the terms involving just $h_{2g}$ or $f_g$ are immediately seen to be
\beann
\left( \begin{array}{c} h_{2g}  \partial_{h_0}\phi_0(h_0,f_0) + 
f_{g}  \partial_{f_0}\phi_0(h_0,f_0)\cr h_{2g}  \partial_{h_0}\psi_0(h_0,f_0) + 
f_{g} \partial_{f_0}\psi_0(h_0,f_0)  \end{array} \right) &=& 
(2\nu+1) \left( \begin{array}{c} h_{2g}  \partial_{h_0} B_{12}+ 
f_{g}  \partial_{f_0} B_{12}\cr h_{2g}  \partial_{h_0} B_{11} + 
f_{g} \partial_{f_0} B_{11} \end{array} \right)
\eeann
The RHS follows from the observation, in Proposition \ref{prop01}, that $\phi_0 = (2\nu+1) B_{12}, \psi_0 = (2\nu+1) B_{11}$.  Now inserting this into the equations (\ref{ContString}) one has that the $\mathcal{O}\left( n^{-2g}\right)$-level equations are of the form

\beann
0 &=& \left( \begin{array}{c} h_{2g} \cr f_g \end{array}\right) + (2\nu+1) t 
\left( \begin{array}{c} h_{2g}  \partial_{h_0} B_{12}+ 
f_{g}  \partial_{f_0} B_{12}\cr h_{2g}  \partial_{h_0} B_{11} + 
f_{g} \partial_{f_0} B_{11} \end{array} \right) + \mbox{lower genus terms}\\
&=& 
\left( \begin{array}{cc} 1 +  (2\nu+1) t  \partial_{h_0} B_{12} &
(2\nu+1) t \partial_{f_0} B_{12}\cr (2\nu+1) t  \partial_{h_0} B_{11} & 
1 + (2\nu+1) t  \partial_{f_0} B_{11} \end{array} \right) \left( \begin{array}{c} h_{2g} \cr f_g \end{array}\right)  + \mbox{lower genus terms}\\
&=& \left( \begin{array}{cc} A_{11} & A_{12}\cr A_{21} & 
A_{22} \end{array} \right) \left( \begin{array}{c} h_{2g} \cr f_g \end{array}\right) + \mbox{lower genus terms} \\
&=& \left( \begin{array}{cc}   \phi_1 & f_0^{-1} 
\psi_1 \cr   \psi_1 &
  \phi_1  \end{array} \right) \left( \begin{array}{c} h_{2g} \cr f_g \end{array}\right)
+ \mbox{lower genus terms}.
\eeann
where in the third equality we have used the identities (\ref{hess1} - \ref{symm3}) and the last equality may be established by direct comparison  of (\ref{A11} - \ref{A12}) with the last equation in Proposition \ref{prop01} and use of the identity (\ref{symm4}). The phrase {\it lower genus terms} here refers to products of terms from $\partial_x^\alpha h$ or $\partial_x^\beta f$ with terms of the form of the RHS of (\ref{FdBeq}) for $G$ a general string coefficient.

It then follows, from applying  (\ref{Ainv}) in the next to last equality, that we may isolate these highest genus terms as
\begin{eqnarray} \label{evRec}
\left( \begin{array}{c} h_{2g} \cr f_g \end{array}\right) &=& \left(\begin{array}{cc}  f_{0x} & h_{0x} \\ f_0 h_{0x} & f_{0x} \end{array} \right) (\mbox{lower genus terms}).
\end{eqnarray}
This can also be seen from the last equality using the unwinding identity (Proposition \ref{prop05}).

As a consequence of (\ref{hoddident}) one also has that the odd index functions $h_{2g+1}$ may be expressed in terms of lower genus generating functions:
\begin{eqnarray} \label{oddRec}
h_{2g+1} &=& - \sum_{m=1}^{2g+1} \partial_x^m h_{2g+1 - m}.
\end{eqnarray} 

It follows from all of the above that all generating functions may be recursively calculated in terms of lower genus generating functions and their $x$-derivatives. The theorem then follows by  induction. The base step is given by Proposition \ref{basestep}. By induction we may assume that, at any given level, all lower genus generating functions satisfy the statements of the theorem. 

Now by strong induction from the base step one sees that by the recursions, (\ref{evRec}) and (\ref{oddRec}), and the fact shown at the start that the elements of the prefactor matrix in (\ref{evRec}) has the form stated in the theorem, it follows that $h_g$ and $f_g$ are ultimately expressed as rational functions on   $\widehat{\mathcal{C}}$. This completes the induction. 
To see that the $f_g$ are in fact rational functions on ${\mathcal{C}}$, we recall from (\ref{fcount}) that $f_g$ is a generating function for enumerating two-legged, $j$-valent. $g$-maps. It is a straightforward consequenc of Euler's genus formula that such maps must have an even number of vertices. It follows that $f_g$ is a function of just $\xi^2$. but as we have seen from (\ref{curve}), $\xi^2$ is purely a function of $y_0$; hence, so is $f_g$. We note that this argument does not work for $h_g$ since this generating function enumerates maps that are just one-legged and, indeed, we have seen in Corollary \ref{hyperellip} that $h_1$ and $h_2$ do not descend to rational functions on $\mathcal{C}$.
\end{proof}
\medskip

\begin{thm} \label{Dmax}
The maximal pole order in $\Pi = \Pi_- \Pi_+$ of $h_{2g}$ or $f_g$ is $5g-1$. The maximal pole order of $h_{2g+1}$ is $5g+1$. 
\end{thm}

\begin{proof} We first show that the maximal pole order is bounded above by these numbers. This can be established inductively. The base step is immediate from Corollary \ref{hyperellip}. Recall that $k^{(a)}$ denotes the partition of $h$-genus terms appearing in a given term of the forcing for the $f_g$ recursion and $k^{(b)}$ denotes the partition of $f$-genus terms in the same forcing term. Further let $m$ denote the total $x$-derivative degree in that term. Then one has from 
(\ref{FdBeq}) the following constraint:
\begin{eqnarray} \label{constraint}
2g = |k^{(a)}| + 2|k^{(b)}| + m.
\end{eqnarray}
As a consequrence of  Theorem \ref{uz}, the polar parts of $f_g$ project to elements in the coordinate ring ${\bf \tilde{S}}$ localized along $\Pi = \Pi_-\Pi_+$ and, as such,  have a well defined polar expansion in inverse powers of $\Pi = \Pi_-\Pi_+$  with coefficents lying in the coordinate ring itself. 
Define a {\it term} of the polar expansion of $f_g$ to be a monomial divided by a power of $\Pi$. Furthermore define the { $\Pi$-weight}  of that term to be this power.  This expansion is built by multiplying together $x$-derivatives of the $h_{2k+1}, h_{2k}$ and $f_k$ for $k<g$. It follows by induction that the { $\Pi$-weight} of a term is given by
\begin{eqnarray} \label{Dweight}
\frac52 \left(2g-m\right)  -4 \ell(k_{od}^{(a)}) - \ell(k_{ev}^{(a)}) - \ell(k^{(b)}) + 2m + 1
\end{eqnarray}
where $k_{od}^{(a)}$ is the partition built from the odd parts of $k^{(a)}$  and $k_{ev }^{(a)}$  is the partition built from the even parts of $k^{(a)}$; the $2m$ comes from the fact that an $x$-derivative raises the { $\Pi$-weight} of a term by 2; and $1$ is added for the factor of $\Pi$ that comes from multiplication by the matrix involving $h_{0x}$ and $f_{0x}$. The partiton length terms appear here because the inductive $\Pi$-weight is not a perfect multiple of the genus. Hence we pick up a $-1$ for each part in $k_{ev}^{(a)}$ or $k^{(b)}$ and a $-4$
for each part in $k_{od}^{(a)}$ since $5(k-1) +1 = 5k-4$. 

Substituting the constraint (\ref{constraint}) into (\ref{Dweight}) one has that  the maximum pole order is the maximal value of 
\begin{eqnarray} \label{maximizer}
4g +\frac12 |k^{(a)}| + |k^{(b)}| -4 \ell(k_{od}^{(a)}) -\ell(k_{ev}^{(a)}) - \ell(k^{(b)}) + 1
\end{eqnarray}
that can be taken on the partitions $(k^{(a)}, k^{(b)})$. Since partition lengths are always non-negative, this form makes it clear that in order to maximize one should minimize $4 l(k_{od}^{(a)}) + \ell(k_{ev}^{(a)}) + \ell(k^{(b)})$. Since $(k^{(a)}, k^{(b)})$ cannot be vacuous, we first consider the case that this length sum = 1. In this case one must clearly have $k_{od}^{(a)}=0$ and then, mutually exclusively, either $\ell(k_{ev}^{(a)}) =1$ or $\ell(k^{(b)}) = 1$. In either of those cases the maximum possible value of $\frac12 |k^{(a)}| + |k^{(b)}|$ is $g-1$. It cannot be $g$: since there is only one part overall this would have to correspond to a factor of either $h_{2g}$ or $f_g$ which is not allowed by the induction. So in this case we have an upper bound of 
$4g +g-1 -1 +1 = 5g-1$. If this were realized the constraint (\ref{constraint}) would require that $m=2$. But this is realized by terms proportional to $ h_{2g-2, xx}$ or $f_{g-1,xx}$. 

Next we consider the case that the length sum is $\geq 2$. In this case it is clear that (\ref{maximizer}) is bounded above by $4g +g -2 +1 = 5g-1$ which occurs when the legth sum = 2 and also $\frac12 |k^{(a)}| + |k^{(b)}|=g$. These extrema are simultaneously realized when, and only when, one has terms proportional to $h_{2k}h_{2g-2k}, f_{k}f_{g-k}$ or $h_{2k}f_{g-k}$ for $0< k < g$.

The description of the maximal pole order terms for $h_{2g}$ is exactly the same as that just presented for $f_g$. That of $h_{2g+1}$ can be deduced from this by application of the Bernoulli identity (\ref{oddRec})
and yields the maximum order of $5g+1$. To complete the induction one must still show that these terms, within the right hand side of (\ref{evRec}) or (\ref{oddRec}), do not somehow combine to cancel at the maximal order as one approaches the singularity locus. We show this in the next theorem.
\end{proof}

\subsection{Nondegeneracy at Real Turning Points} \label{sec:Laurent}

We have just shown that
\begin{align} \label{rs1}
 f_{g} = \frac{P^{(g)}(y_0)}{x^{2g-1} \Pi^{5g-1}},
\end{align} 
for some polynomials $P^{(g)}$ in $y_0$.
However, we have not yet shown that the exponent $5g-1$ in the denominator is the smallest possible, i.e.\ that $\Pi$ does not divide any of the polynomials $P_g$. In Section \ref{gaussleaf} 
we established that the two finite real turning points of the spectral curve $\mathcal{C}$ are simple
and, at these critical points, the generating functions $h_0$ and $f_{0}$ limit to finite non-zero values (see Theorem \ref{genthm} (b)). Let us denote the turning point of the right by $(\xi_c^2, y_{0c})$.

As $\xi^2\rightarrow \xi^2_c$ (or, equivalently, as $\xi \to \xi_c$), the variables $y_0$
and $f_0$ approach positive real values. In the case of $y_0$ this follows from Theorem \ref{genthm} where we showed that the critical point occurs at the turning point on the right where $y_0 > 0$. The global expression for $f_0$ given in (\ref{zzero}) has a numerator and denominator which are both positive for $y_0 > 0$ and $\xi^2 >0$ as was also shown in Theorem \ref{genthm}.

By simplicity of the turning point, $\Pi \sim C(\xi-\xi_c)^{1/2}$ for some $C\neq 0$ as $\xi \rightarrow \xi_c$.
Thus the maximal pole order $5g-1$ in (\ref{rs1}) is equivalent to the asymptotic behavior of $f_g$ as $\xi \rightarrow \xi_c$.
Define coefficients $\zeta_g$ and $\gamma$ by
\begin{alignat}{2} \label{turningfg}
f_g \sim&  \zeta_g  \tau^{1-5g }(1+O(\tau))&\qquad \text{as }\xi\rightarrow & \xi_c,\qquad
\tau=(\xi - \xi_c)^{1/2},\\
f_{0x} \sim & \gamma \tau^{-1}+O(1)&\qquad \text{as }\xi \rightarrow & \xi_c.
\label{qwqw1}
\end{alignat}
In the following proposition and below, the notation $'$ will mean $\partial_{\sqrt{y_0}}$.
A subscript $c$, as in $y_{0c}$, indicates a value at the critical time $\xi_c$.
Also in this section we assume that $x$ has been set to $1$ after any $x$-derivatives have been taken.

\begin{prop} \label{p872}
The coefficients of the most singular terms for $f_g$ satisfy the recurrence
\begin{align} \label{rec9023}
0= & \zeta_{g+1} + C_1     C_2 z_{c}^{2g-1} (25g^2-1)   \zeta_{g } +  6C_1  \sum_{m=1}^{g} \zeta_{m} \zeta_{g+1-m} \qquad  (g\geq 1),
\end{align}
where 
$
C_1 =  j\xi_c z_{0c}^{j/2-1} \gamma  \left( \frac{1}{3}S''_c +\frac{1}{3} \R''_c-\frac{1}{12}\R'_c  \right)
$
and
$
 C_2 =   \frac{1}{4}(j/2-1)^2 \xi_{c}^2.
$
It follows that $\zeta_g>0$ for all $g\geq 1$.
\end{prop}

\begin{proof}
The essential idea here is straightforward: the string equations give recursive formulas for $f_g$ and $h_{2g}$; using Theorem \ref{uz} we will identify all terms that are not $o(\tau^{5g-2})$ as $\xi \rightarrow \xi_c$.
Along with (\ref{rs1}) one has
\beann
 \frac{h_{2g}}{h_0} &=& \frac{Q_{2g}}{x^{2g} \Pi^{5g-1}}, \\ 
\frac{h_{2g+1}}{h_0} &=& \frac{Q_{2g+1}}{x^{2g+1} \Pi^{5g+1}}.
\eeann
where the $Q$ are polynomials in $y_0$.
Using Lemma \ref{y0x} it is easy to see that
\begin{align}
&\prod_{i=1}^I f_{g_{i},x^{(m_i)}}\prod_{j=1}^J h_{g'_{j},x^{(m'_j)}} =O(\tau^{p}),\qquad \text{where}\label{form1843} \\
&p=\sum_{i=1}^I 5 g_i +2m_i -\one_{ g_i\neq 0 \text{ or }m_i\neq 0} 
+\sum_{j=1}^J \frac{5}{2}g'_j +2m'_j -\one_{ g'_j\neq 0 \text{ or }m'_j\neq 0} -\one_{g'_j \text{ odd}}.
\end{align}
At order $n^{-2g}$, the string equations have terms of the form (\ref{form1843}) with the constraint that $\sum 2g_i +g'_i +m_i+m'_i=2g$.
Thus the forcing terms we must retain are
\begin{align}
f_{0}^q h_{0}^{q'} f_{g-1,xx},\qquad 
f_{0}^q h_{0}^{q'} f_{g'}f_{g-g'},\qquad
f_{0}^q h_{0}^{q'} h_{2g-2,xx},\qquad 
f_{0}^q h_{0}^{q'} h_{2g'}h_{2g-2g'}
\end{align}
for some $q,q'\geq 0$ and $g'\geq 1$.

We set $S=S_{2\nu}(\sqrt{y_0})$ and 
$\R$ a slightly different than $R_{2\nu}(\sqrt{y_0})$:
\begin{align}
S=&I_0(2\partial_{\sqrt{y}}) y^\nu = [\eta^0](\eta +\sqrt{y}+\eta^{-1})^{j-1}= \sum \binom{2\nu}{2\mu,\nu-\mu,\nu-\mu} y^{\mu} \\
 \R=&I_1(2\partial_{\sqrt{y}}) y^\nu = [\eta^1](\eta +\sqrt{y}+\eta^{-1})^{j-1}=\sqrt{y} \sum \binom{2\nu}{2\mu-1,\nu-\mu,\nu-\mu+1} y^{\mu};
\end{align}
i.e., these are essentially the scaled string polynomials, $\hat{\phi}_0, \sqrt{y_0}\hat{\psi}_0$ respectively.

In the following display, $\sim$ means that the difference of the left and right hand sides is order $O(\tau^{3-5g})$.
Also let $\wt{f}_g= f_{0}^{2g-1}f_g$ and $\wt{h}_g=f_{0}^{g-1/2}h_g$.
By straightforward but tedious calculations we find that the equation at order $n^{-2g}$ from the relation $0=V'(L)_{n,n}$ gives
\begin{align} \label{izdbf}
0 \sim & \wt{h}_{2g} +jt f_{0}^{j/2-1} \bigg\{ S' \wt{h}_{2g} + \R' \wt{f}_g  + \frac{1}{6}\R'' \wt{h}_{2g-2}'' + \left( \frac{1}{6}S''+ \frac{1}{12} \R'\right) \wt{f}_{g-1}'' \\
& \hspace{100pt} +\sum_{g'=1}^{g-1} \frac{1}{2} S''  \wt{h}_{2g'} \wt{h}_{2g-2g'} + \R'' \wt{h}_{2g'}\wt{f}_{g-g'} +\frac{1}{2}(-\R' + S'')\wt{f}_{g'} \wt{f}_{g-g'} \bigg\} 
\end{align}
It turns out that there is a simple relation between  asymptotic behaviors of $f_g$ and $h_{2g}$ as $\xi \rightarrow \xi_c$ that will let us simplify formula (\ref{izdbf}).
It follows form (\ref{basic1}) that 
\begin{align}f_{0,w}^2-f_0 h_{0,w}^2 = \frac{4 - y_0}{\widehat{D}}.
\end{align}
As $\xi \rightarrow \xi_c$ both terms on the left hand side are order $\tau^{-2}$ but the right hand side is only  order $\tau^{-1}$ and thus can be ingored.
It follows that
\begin{align}
f_{0,x} \sim \sqrt{f_0}h_{0,x}(1+O(\tau))\qquad \text{as } \xi \rightarrow \xi_c.\label{dofbn}
\end{align}
The string equations at order $n^{-2g}$ have the form
\begin{align}
\vect{u_{2g}}{z_g} =\matr{u_{0x}}{z_{0x}}{z_{0x}}{z_0u_{0x}}\vect{F_1}{F_2}, \label{dofbn2}
\end{align}
where $F_1,F_2$ are the forcing terms.
It follows from equations (\ref{dofbn},\ref{dofbn2}) that
\begin{align}
\wt{f}_{g,x^{(m)}} \sim \wt{h}_{2g,x^{(m)}}(1+O(\tau))\qquad \text{as } \xi \rightarrow \xi_c.
\end{align} 
Using this to simplify equation (\ref{izdbf}) we have
\begin{align}
0\sim & \wt{f}_{g} +jt f_{0}^{j/2-1} \bigg\{ \left(S'   + \R'  \right) \wt{f}_g  + \left( \frac{1}{6}\R'' + \frac{1}{6} S''+ \frac{1}{12} \R'\right) \wt{f}_{g-1,xx} \\
& \hspace{140pt} +\sum_{g'=1}^{g-1} \left(   S''   + \R''  -\frac{1}{2} \R'  \right)\wt{f}_{g'} \wt{f}_{g-g'} \bigg\} \label{qaz1}.
\end{align}
Although the $\wt{f}_g$ terms are order $\tau^{1-5g}$, a consequence of the leading order string equations is that $1+jt f_{0}^{j/2-1}(S' +\R')=O(\tau)$; thus all the terms on the right hand side above are order $O(\tau^{2-5g})$.

In a similar way the order $n^{-2g}$ equation of other string equation $w=V'(L)_{n,n-1}$ gives
\begin{align}
0 \sim & 
\wt{f}_g + j t f_{0}^{j/2-1} \bigg\{   \left(\R'   + S' \right) \wt{f}_g + \left( \frac{1}{6} S'' -\frac{1}{6}\R'    + \frac{1}{6} \R'' \right)  \wt{f}_{g-1,xx}  \\
& \hspace{140pt} + \sum_{g'=1}^{g-1} \left( \R''   +S''   + \right)\wt{f}_{g'} \wt{f}_{g-g'}   \bigg\} \label{qaz2}.
\end{align}
Taking a linear combination $f_0 h_{0,x} (\ref{qaz1}) + f_{0,x}  (\ref{qaz2})$ gives
\begin{align}
0\sim & 
\wt{f}_g + jt f_{0}^{j/2-1} \wt{f}_{0,x} \left( \frac{1}{3} S'' +\frac{1}{3} \R''-\frac{1}{12}\R'  \right)  \bigg\{   \wt{f}_{g-1,xx}  +  6 \sum_{g_1 +g_2 =g} \wt{f}_{g_1} \wt{f}_{g_2} \bigg\}.  \label{w3}
\end{align}
The reduction of the $\wt{f}_g$ terms in the above display follows from the identity $f_{0,x}S' +f_{0}^{1/2}h_{0,x}\R'=1$.
A direct calculation shows that
\begin{align}
\wt{f}_{g-1,xx}
\sim & f_{0}^{2g-1}( 5g-6) (5g-4) \frac{(j/2-1)^2 \xi_{c}^2}{4}  \frac{\zeta_{g-1}}{ \tau^{5g-2}}(1+O(\tau)) \qquad \text{as }\xi\rightarrow \xi_c.
\end{align}
Using this formula in (\ref{w3}) gives the recurrence in the proposition statement.

It remains only to show that the coefficients $\zeta_g$ are strictly positive.
We claim that $C_1(j)$ is strictly negative. 
A direct calculation shows that
\begin{align}
f_{0}^{5/2-j/2}\partial_{f_{0}}^2 [\eta^0](\eta +h_0 +f_0 \eta^{-1})^{j-1} =& S''-\wt{R}' \label{q150}\\
f_{0}^{2-j/2}\partial_{f_{0}}\partial_{h_0} [\eta^0](\eta +h_0 +f_0 \eta^{-1})^{j-1} =&  \wt{R}'' \label{q250}\\
f_{0}^{3/2-j/2}\partial_{h_{0}}^2 [\eta^0](\eta +h_0 +f_0 \eta^{-1})^{j-1} =& S''.\label{q350}
\end{align}
The left hand sides of the three expressions above are obviously positive, and
\begin{align}
C_1(j) =& j\xi_c z_{0c}^{j/2-1}\gamma\left( \frac{1}{12}
(\text{eq. }\ref{q150})+\frac{1}{3} (\text{eq. }\ref{q250}) +\frac{1}{4} (\text{eq. }\ref{q350})\right).
\end{align}
The claim that $C_1(j)<0$ now follows since $\xi_c<0$ and $z_{0c},\gamma>0$.
The positivity of the $\zeta_g$'s is now clear.
\end{proof}

\begin{cor} \label{c873}
Let $\kappa^{(j; 2:1)}_{g}(2k)$ be the number of two-legged maps of genus $g$ with $2k$ vertices of valence $j$.
Then
\begin{align}
\frac{\kappa^{(j; 2:1)}_{g}(2k)}{(2k)!} \sim \frac{\zeta_g}{ \Gamma(\frac{5g-1}{2})} t_{c}^{1/2-5g/2-2k} (2k)^{(5g-3)/2} \qquad \text{as }k\rightarrow \infty.
\end{align}
\end{cor}

\begin{proof}
As we saw in (\ref{fcount}),
\begin{align}
f_{g}(t_j,w=1)=  \sum_{k\geq 0} \kappa^{(j; 2:1)}_{g}(2k) \frac{t_{j}^{2k}}{(2k)!}.
\end{align}
Corollary 2 of \cite{FO90} states that if a complex function $f(z)$ is analytic on a suitable domain then $f(z)\sim (1-z)^\alpha$ implies $[z^k]f(z) \sim [z^k](1-z)^\alpha$.
Our assertion follows directly from this observation, (\ref{turningfg}), and Proposition \ref{p872}.
\end{proof}
\medskip

\begin{rem}
The results in Theorem \ref{Dmax} precisely parallel those for the even valence case that were derived in Section 3 of \cite{Er09}.  
\end{rem}

\section{Form of the Fundamental Map Enumeration Generating Functions} \label{closedforms}

We are now in a position to realize our original goal: the derivation of compact closed form expressions for the map generating functions $E_g$ and to understand these forms in terms of the geometry of of the spectral curve $\widehat{\mathcal{C}}$. A first step has already been taken in the derivation of the genus 1 generating function in (\ref{E1}) whose result we recall here:
\begin{eqnarray*}
E_1 &=& \frac1{24} \log\left( \frac{f^2_{0x} - f_0 h^2_{0x}}{z^2_0}\right) 
= \frac1{24} \log\left( \frac{4 - y_0}{z_0^2 \widehat{D}} \right)\\ 
 &=& \frac1{24} \log\left( \frac{\widehat{B}_{12}^2}{\Pi_- \Pi_+} \right) - \frac1{12}\log{z_0}= e_1.
\end{eqnarray*}

We note that this formula is completely consistent with the expression for $e_1$ in the even valence case given in Theorem \ref{thm51}. Indeed, setting $j=2\nu$ and $u_0 = 0$ in $\widehat{D}$ this reduces to $\widehat{D} = 4(\nu - (\nu-1)z_0)$. Then substituting this into the second equation for $E_1$ above, with $y_0 = 0$, this collapses to precisely our expression (\ref{e1ev}) in the even valence case. Thus the above expression for $E_1$ is valid for arbitrary valence. 

From the viewpoint of Riemann surfaces, genus 1 is the case of flat curvaturre and the simplicitly of our expression in direct terms of the discriminant, $D = d_+ d_-$ of the spectral curve is consistent with that. In the remainder of this section we will derive the explicit form of the genus 0 generating function, corresponidng to positive curvature and a general description of the generating functions with $g >1$, corresponding to the case negative curvature or surfaces of general type. We observe that $d_\pm$ is directly related to $\lambda_\pm$ through the Riemann invariants. This is the key point. $E_0$ is completely defined in terms of the $\lambda_\pm$. $E_1$ is completely defined in terms of the $d_\pm$ which, in turn, completely determine the spectral curve $\widehat{\mathcal{C}}$ and the branch points of its projection onto the $\xi$ plane. After this all the $h_g, f_g$ are rational functions with polar locus along $d_\pm$. 

We recall from (\ref{boson12}) that one has the general differential equations for the $E_g$:
\begin{align} \label{Hirota2}
\frac{\partial^2}{\partial w^2} E_g(t,w)|_{w = x} & =  -\sum_{\ell = 1}^g
\frac{2}{(2\ell + 2)!} \frac{\partial^{2\ell + 2}}{\partial w^{2\ell + 2}} E_{g - \ell}(t,w)|_{w=x}\\ 
& \nonumber + \mbox{the}\,\,\,  n^{-2g} \,\,\, \, \mbox{terms of}\,\,\, \log\left( 1 + \sum_{m=1}^\infty \frac1{n^{2m}} \frac{f_m}{f_0}\right)\end{align}
where $E_h(t,w) = w^{2 - 2h} e_h(w^{j/2 - 1} t)$. 

\subsection{Derivation of $E_0$} \label{DerE0}

 We will show that
\begin{align}		E_0 =& \frac{1}{2} x^2 \log z_0 + \frac{3(j-2) x^2}{4j}-\frac{(j-2)(j+1)x}{j(j+2)}\left(\q\right)+\frac{(j-2)^2}{2j(j+2)}\left(\frac{1}{2} f_{0}^2 + h_{0}^2 f_0 \right).		\label{e0}	\end{align}
This extends the formula 
\begin{align*}		e_0 =& \frac{1}{2} x^2\log z_0 +\frac{(\nu-1)^2}{4\nu(\nu+1)}\left(f_0-x\right)\left(f_0-\frac{3(\nu+1)w}{\nu-1}\right)	\end{align*}
from the even valence case given in Theorem \ref{thm51} by setting $j = 2\nu$ and $h_0 \equiv 0$
in the above expression for $E_0$. Hence this is again a valid formula for arbitrary regular valence.
\smallskip

We will make use of the fundamental equations (\ref{odvalinv}). We recast those ODEs for use here:
 
\begin{align}		-h_0 +jxh_{0,x} =& (j-2)\left(h_0 f_{0,x} + h_{0,x}f_0\right)		\label{ut0}		\\
					-2f_0 + jxf_{0,x} =&(j-2)f_0\left(h_0h_{0,x}+f_{0,x}\right).		\label{ut1}		\end{align}
					
\subsubsection{$w$-Integral formula for $E_0$}
We will compute $E_0$ by re-expressing (\ref{Hirota2}) as a double $w$-integral formula and then applying integration by parts.  In doing this we write $\log z_0 = \log f_0 - \log w$ so that we can work with $w$ as an independent variable, not constrained to be evaluated at $w = x$. Thus,
\begin{align}		E_0	=&	\text{Const.(j, {\it w})}+\int_{0}^w \int_{0}^{w_1} \log f_0(w_2) - \log w_2\, dw_2 dw_1 \nonumber		\\
					=&	\text{Const.(j, {\it w})} +w\int_{0}^w \log f_0(w_2) \, dw_2 -\int_{0}^w w_2 \log f_0(w_2) \,dw_2 	\nonumber	\\
                                       - & \left( \frac12 {w^2} \log w - \frac34 w^2\right)
         \nonumber                       \\
				=&	\text{Const.(j, {\it w})} +w\cdot\text{Integral 1} - \text{Integral 2} - \frac12 {w^2} \log w + \frac34 w^2 ,		\label{e0int}		\end{align}
where Const.(j, {\it w}) is at worst a linear function of $w$.

\subsubsection{Computation of integral 1}
For integrals over the interval $[0,w]$ we won't write the limits of integration.  Dividing equation (\ref{ut1}) by $f_0$ we have
\begin{align}		jw\partial_w \log\, f_0 =& 2+(j-2)\left(h_0h_{0,w}+f_{0,w}\right).		\label{utlog}		\end{align}
This identity facilitates computation of integral 1:
\begin{align}		\int \log f_0 \, dw =&	w\log\, f_0 -\int w \partial_w \log\, f_0 \,dw		\nonumber \\
					=&	w\log\, f_0 -\frac{1}{j}\int 2+(j-2)\left(h_0h_{0,w}+f_{0,w}\right) \,dw		\nonumber \\
					=&	w\log\, f_0 -\frac{2w}{j}-\frac{j-2}{j}\left(f_0 +\frac{1}{2} h_{0}^2\right).		\label{int1}		\end{align}

\subsubsection{Computation of integral 2}
Integral 2 requires slightly more effort.  Multiplying equation (\ref{ut0}) by $h_0$ we get
\begin{align}		0	=&	h_{0}^2 - jw h_{0,w}h_0 + (j-2)\left(h_{0}^2 f_{0,w} +h_{0,w}h_0f_0\right)		\nonumber		\\
					=&	\int h_{0}^2 \, dw -j\int w h_{0,w}h_0 \, dw + (j-2)\int h_{0}^2 f_{0,w} +h_{0,w}h_0f_0 \, dw		\nonumber \\
					=&	\int h_{0}^2 \, dw -j\left(\frac{1}{2} wh_{0}^2 -\int \frac{1}{2} h_{0}^2 \, dw\right)+ (j-2)\left(h_{0}^2f_0- \int h_{0,w}h_0f_0 \, dw	\right)	\nonumber		\end{align}
Therefore we have the identity
\begin{align}		\frac{1}{2} jwh_{0}^2-(j-2)h_{0}^2f_0 =&	\int (j+2)\frac{1}{2} h_{0}^2 - (j-2)h_{0,w}h_0f_0 \, dw		\label{ident1}		\end{align}
We will combine this with another identity; start with (\ref{ut1}) and integrate as before.
\begin{align}		0=&	2f_0 - jw f_{0,w} + (j-2)\left(h_{0,w}h_0f_0 +f_{0,w}f_0 \right)		\nonumber		\\
					=&	2\int f_0 \, dw -j \int w f_{0,w} +(j-2)\int h_{0,w}h_0f_0 +f_{0,w}f_0\, dw		\nonumber \\
					=&	2\int f_0 \, dw -j \left(wf_0- \int  f_{0}\, dw\right) +(j-2)\int h_{0,w}h_0f_0 \, dw +\frac{1}{2} (j-2)f_{0}^2.	\nonumber		\end{align}
Therefore
\begin{align}		jw f_0 -\frac{1}{2}(j-2)f_{0}^2 =& \int (j+2)f_0 +(j-2)h_{0,w}h_0f_0 \, dw		\label{ident2}		\end{align}
Combining identities (\ref{ident1}) and (\ref{ident2}) we have
\begin{align}		(j+2)\int f_0 +\frac{1}{2} h_{0}^2 \, dw =& jw\left(f_0+\frac{1}{2} h_{0}^2\right)- (j-2)\left(\frac{1}{2} f_{0}^2 +f_{0}h_{0}^2\right).		\label{ident3}		\end{align}
We can now compute integral 2.  We first use (\ref{utlog}), and then apply the integral formula (\ref{ident3})
\begin{align}		\int w \log f_0 \, dw =& \frac{1}{2} w^2 \log f_0 - \int \frac{1}{2} w^2 \partial_w \log f_0 \, dw		\nonumber		\\
					=&	\frac{1}{2} w^2 \log f_0 - \int \frac{w}{2j}\left(2+(j-2)\left(h_0h_{0,w}+f_{0,w}\right)\right) \, dw		\nonumber		\\
					=&	\frac{1}{2} w^2 \log f_0 - \frac{w^2}{2j} -\frac{j-2}{2j}\left[w\left(f_0 +\frac{1}{2} h_{0}^2\right)-\int\q \, dw \right]		\nonumber \\ \nonumber
					=&	\frac{1}{2} w^2 \log f_0 - \frac{w^2}{2j} - \frac{(j-2)w}{2j}\left(\q\right)\\ +& \frac{j-2}{2j(j+2)}\left[jw\left(\q\right)-(j-2)\left(\frac{1}{2} f_{0}^2 +f_0 h_{0}^2\right)\right].		\label{int2}		
					\end{align}
Combining our formulae for (\ref{int1}) and (\ref{int2}) as per (\ref{e0int}), we have the asserted representation of $E_0$ up to the linear function Const.(j, {\it w}). However, applying the constraint that there are no zero vertex maps, one sees that this linear function must in fact be identically zero.

\subsubsection{$E_0$ in terms of Characteristic Geometry}
The previous formula can also be expressed directly in terms of the characteristic geometry (i.e., in terms of just $y_0$) as
\begin{thm} \label{e0y0}
\begin{eqnarray*} 
e_0 = w^{-2} E_0 &=& \frac12 \log \left(z_0\right) + \frac{j-2}{4 j (j+2)} \left[(j-2) z_0^2 \left(1 + 2 y_0\right) -2 (j+1) z_0 \left(2 +  y_0\right) + 3 (j+2)\right]\\
&=&\frac12 \log \left(\frac{\widehat{B}_{12}}{\widehat{B}_{12} - y_0^{1/2}\widehat{B}_{11}}\right) \\
&+& \frac{j-2}{4 j (j+2)} \left[(j-2) \left(\frac{\widehat{B}_{12}}{\widehat{B}_{12}- y_0^{1/2}\widehat{B}_{11}}\right)^2 \left(1 + 2 y_0\right) -2 (j+1) \left(\frac{\widehat{B}_{12}}{\widehat{B}_{12}- y_0^{1/2}\widehat{B}_{11}}\right) \left(2 +  y_0\right) + 3(j+2) \right].
\end{eqnarray*}
In particular this shows that $e_0$ is completely determined by the $\nu$ zeros of the characteristic speed $\widehat{\lambda}_- \left(= - (\widehat{B}_{12}- y_0^{1/2}\widehat{B}_{11})\right)$ and the values of $\widehat{B}_{12}(y_0)$ at those zeros. 
\end{thm}

\subsubsection{$E_0$ for Mixed Valence}. \label{sec:MV}

For general potentials of the form (\ref{eq:genpot}) the generating function $e_0$ can be expressed directly in terms of the equilibrium measure (\ref{dmu}) 
(see section 3 of \cite{EM03}). Consequently the resulting expression is similar to that of (\ref{e0}) in that it only depends on $h_0$ and $f_0$ with the important difference that now those functions depend on all the multiple parameters of the potential. We will not need the explicit expression for this paper; however, what is important for us is that the derivatives $d e_0/ d \xi_k$, as evident from \cite{EM03}, are rational functions of $h_0$ and $f_0$.

\subsection{Universal Expressions for $E_g$} \label{UnivEg}

In previous subsections we explicitly showed that $E_1$ could be realized as closed form valence-independent expression in just $f_0$, $h_0$ and their $x$-derivatives.. Now we show that this extends to $E_g, g>1$ in the form stated in Theorem \ref{MAINTHEOREM} (b).
Our proof of this will be based on an analogous result for the generating functions $f_g$ and $h_g$. 
\begin{prop} \label{Univfh}
For $g>0$,
\beann 
f_g &=& \frac{P_{f_g}(h_0, f_0, h_{0x}, f_{0x}, h_{0xx}, \dots)}{(f^2_{0x} - f_0 h^2_{0x})^{8g-3}} \\
h_{2g} &=& \frac{P_{h_{2g}}(h_0, f_0, h_{0x}, f_{0x}, h_{0xx}, \dots)}{(f^2_{0x} - f_0 h^2_{0x})^{8g-3}} \\
h_{2g+1} &=& \frac{P_{h_{2g+1}}(h_0, f_0, h_{0x}, f_{0x}, h_{0xx}, \dots)}{(f^2_{0x} - f_0 h^2_{0x})^{8g-2}}
\eeann
where $P_{f_g}$ and $P_{h_{2g}}$ are valence independent polynomials. $f_g$ and $h_{2g}$ have  polynomial weight $1$ and $\frac12$ respectively, and differential weight $2g$. $P_{h_{2g+1}}$ has polynomial weight $\frac12$ and differential weight $2g+1$.
\end{prop}

Before turning to the proof of this we require a lemma. 
\begin{lem} \label{phipsi}
The string  polynomials  $\phi_m$ and $\psi_m$ have differential weight = -1 and polar order in $(f^2_{0x} - f_0 h^2_{0x})$ at most $2 m -1$. Moreover, these string polynomials have expressions which are independent of valence.
\end{lem}
\begin{proof}
Making use of the unwinding identity of Section \ref{Unwinding} one deduces by recursion that
\beann
\left(\begin{array}{c} \phi_m \\ \psi_m \end{array} \right) &=& \left[ \left(\begin{array}{cc} -f_0 h_{0x} & f_{0x} \\ f_0 f_{0x} & -f_0 h_{0x}\end{array} \right) \frac1{f^2_{0x} - f_0 h^2_{0x}} \partial_x\right]^{m-1} \left(\begin{array}{c} \phi_1 \\ \psi_1 \end{array} \right) \\
&=& \left[ \left(\begin{array}{cc} -f_0 h_{0x} & f_{0x} \\ f_0 f_{0x} & -f_0 h_{0x}\end{array} \right) \frac1{f^2_{0x} - f_0 h^2_{0x}} \partial_x\right]^{m-1} 
\frac1{f^2_{0x} - f_0 h^2_{0x}} \left(\begin{array}{c} f_{0x} \\ - f_0 h_{0x} \end{array} \right)
\eeann
where the second equation  follows from (\ref{BaseUnwind}). The differential weight and pole order evaluations follow  directly from this. Moreover, this expression is manifestly valence independent.
\end{proof}
We now turn to the proof of the proposition.
\begin{proof}
The proof of this is an induction based on the recursion (\ref{evRec}). The formulas in (\ref{INDBASE}) provide the base steps for this induction at $f_1, h_2$ and $h_1$. 

Now we consider $f_g$ and $h_{2g}$, supposing the proposition holds for all $f_k$ with $k < g$ and $h_k$ 
with $k < 2g$. Now, as demonstrated in section \ref{toprec} we may extract the order $n^{-2g}$ terms from the continuum string equations (\ref{ContString}) by applying (\ref{FdBeq}) to (\ref{StringCoeff}). One may then extract $f_g$ and $h_{2g}$ from all the terms at this order as observed in Theorem \ref{uz}.  By the preceding, the terms denoted by {\it (lower genus terms)} in (\ref{evRec}) are all sums of monomial terms having the general form 
\begin{eqnarray} \label{IndMon}
h_{2g_1} \cdots h_{2g_k} f_{g_{k+1}} \cdots f_{g_{K}} h^{(m'_1)}_{2g'_1} \cdots h^{(m'_{k'})}_{2g'_{k'}} f^{(m'_{k'+1})}_{g'_{k'+1}} \cdots f^{(m'_{K'})}_{g'_{K'}} h^{(m_1)}_0 \cdots h^{(m_\ell)}_0  f^{(m_{\ell +1})}_0 \cdots f^{(m_L)}_0 (\phi_m \mbox{or}\,\,\, \psi_m)
\end{eqnarray} 
where $0 < g_i, g_{i'} < g$ and in the first (top) component these monomials have polynomial weight in $\frac12 + \mathbb{Z}$ while in the second (bottom) component the monomials have polynomial weight in $\mathbb{Z}$. These last requirements follow from the external weight conditions given in (\ref{a-shift}), (\ref{b-shift}) respectively: differentiation can only alter the external weight by an integer amount, so the polynomial weight of (\ref{IndMon}) is an integer (half-integer) iff they are terms in the expansion of $f_g$ ($h_{2g}$). Observe also that the prefactor matrix in (\ref{evRec}) preserves this polynomial weight dichotomy between the top and bottom components. The reader will also note that no factors involving $h_{2k+1}$ appear in (\ref{IndMon}). We may assume this since, if any such factor did appear,  the relation  (\ref{oddRec}) may be used to replace it; and then if that yields terms containing a factor of type $h_{2k+1}$ but for a smaller value of $k$ the relation may be applied again.  So, by repeated use of this relation we reduce to just considering terms of the form (\ref{IndMon}).

The final and main constraint that now comes in is the balance of asymptotic orders in $n$:
\begin{eqnarray} \label{CONSTRAINT}
2g &=& \sum 2 g_i + 2 g'_{i'} + m'_{i'} + m_{j}.
\end{eqnarray}

Since, by Lemma \ref{phipsi}, $\phi_m$ and $\psi_m$ each have differential weight $-1$ it follows from this last constraint and induction that  (\ref{IndMon}) has differential weight $2g -1$. Since the matrix prefactor in  (\ref{evRec}) has differential weight 1, it follows that $h_{2g}$ and $f_{2g}$ each have differential weight $2g$ in general. It then readily follows from (\ref{oddRec}) and induction that 
$h_{2g+1}$ has differential weight $2g+1$. Finally it follows from the fact that differentiation lowers external weight of a factor by 1, that the total polynomial weight 
equals the sum of the external weight and differential weight. The external weights of  $h_{2g}$ and $f_{g}$ are, respectively $1/2 - 2g$ and $1- 2g$; hence, their polynomial weights are respectively $1/2$ and 1. Similarly, the total polynomial weight of $h_{2g+1}$ is $-1/2 -2g + 2g +1 = 1/2$.

We next handle the order of the denominators in $f^2_{0x} - f_0 h^2_{0x}$. Examining the general formula for the string coefficients given by (\ref{finalform}) one sees that the maximal value of the index $m$ for the string polynomials appearing in (\ref{IndMon}) is given in terms of the maximal power of $(1 - z(\sqrt{f} \eta + h + \sqrt{f} \eta^{-1}))^{-1}$ appearing in (\ref{finalform}). By inspection this maximal power can be seen to equal $|\alpha| + |\beta| + \ell(\alpha) + \ell(\beta) + 1$. But then relating this to the Continuum String equations (\ref{ContString}) and (\ref{IndMon}) one observes that
\beann
|\alpha| + |\beta| &=& \sum_{j=1}^L m_j + \sum_{i' =1}^{K'} m'_{i'} \\
\ell(\alpha) + \ell(\beta)  &=& K + K' + L
\eeann
Since our indexing starts at $-1$, it follows that the value of $m$ appearing in (\ref{IndMon}) is given by
\beann
m &=& K + K' + L + \sum_{j=1}^L m_j + \sum_{i' =1}^{K'} m'_{i'} - 1.
\eeann
Now by Lemma \ref{phipsi} we know that the maximal order of $f^2_{0x} - f_0 h^2_{0x}$ appearing in a string polynomial with index $m$ is $2m-1$, so to bound the denominator order in (\ref{IndMon}) we need, by induction, to maximize 
\begin{eqnarray} \label{orderbound}
2\left( K + K' + L + \sum_{j=1}^L m_j + \sum_{i' =1}^{K'} m'_{i'} - 1\right) - 1 + \sum_{i=1}^K (8 g_i - 3) + \sum_{i'=1}^{K'} (8 g'_{i'} - 3)
\end{eqnarray}
subject to the constraint (\ref{CONSTRAINT}). Using this constraint one readily sees that maximizing (\ref{orderbound}) reduces (by setting $K = K' = 0$) to maximizing 
\beann
8g - 3 + 2\left( L - \sum_{j=1}^L m_j\right)
\eeann
whose maximum value, $8g-3$, is realized by monomials whose initial segment is $h_{0x}^\ell f_{0x}^{L - \ell}$.  The order statement for $h_{2g+1}$ now follows by (\ref{oddRec}).

We note that nothing in the above arguments depended on the valence.
\end{proof}

We next extend this analysis to prove part (b) of Theorem \ref{MAINTHEOREM}:

\begin{prop} \label{valindep}
For $g > 1$ one has,
\beann
E_g &=& \partial^{-2}_x \frac{P_g(h_0, f_0, h_{0x}, f_{0x}, h_{0xx}, \dots)}{f_0^g(f^2_{0x} - f_0 h^2_{0x})^{8g - 3}}
\eeann
where $P_g$ is a valence independent polynomial. $E_g$  has polynomial weight $0$ and differential weight $2g-2$. Moreover, the symbol $\partial^{-2}_x$ here signifies that $E_g$ is an exact second antiderivative of the given expression which has the form of a rational funciton of $f_0, h_0$ and their x-derivatives, and so is again universal. 
\end{prop}
\begin{proof}
Starting with (\ref{Hirota2}) we observe that the order
$n^{-2g}$ terms of $\log\left( 1 + \sum_{m=1}^\infty \frac1{n^{2m}} \frac{f_m}{f_0}\right) $ consist of monomials of the form
\beann
f_0^{-M} f_{k_1} \cdots f_{k_M}\\
0 < k_j \leq g \\
\sum k_j = g ,
\eeann
which by Proposition \ref{Univfh} is valence independent and has polynomial weight 0 and differential weight $2g$. Moreover this expression is a rational function in $f_0, h_0$ and their $x$-derivatives with denominator maximally of the form $f_0^M (f^2_{0x} - f_0 h^2_{0x})^{8g - 3M}$. Hence, the sum of all order $n^{-2g}$ terms has a common denominator at worst of the form
$f_0^g(f^2_{0x} - f_0 h^2_{0x})^{8g - 3}$.  

Turning next to the differential terms in (\ref{Hirota2}), we observe, inductively, that as a consequence of (\ref{jet1}), differentiating $E_{g-\ell}$ wrt $x$ does not increase the order of the denominator in $f_0$. It is also clear from (\ref{jet1Alt} - \ref{basic1}) that the stated differentiations in  (\ref{Hirota2}) do not raise the order of $E_{g-\ell}$ in $f^2_{0x} - f_0 h^2_{0x}$ above the bound $8g-3$. Hence, the stated order bounds in the denominator are preserved by induction. Additionally, since differentiation cannot alter the polynomial weight, this remains 0 while the differential weight increases by 1 with every differentiation. Hence the differential weight of all these terms is $2g$. The application of $\partial^{-2}_x$  reduces this to $2g-2$. This completes the inductive proof of the formula for $E_g$. The last statement of the proposition follows from Theorem \ref{Eg}.
\end{proof}

\subsection{Rational Expressions for $E_g$,  $g \geq 2$} \label{RatEg}

We now turn to the alternate representation of the $E_g$, stated in Theorem \ref{MAINTHEOREM} (c),  as rational functions of just $f_0$ and $h_0$ (without any derivatives) but depending on valence.  We begin with the case of regular valence and then turn to the mixed valence case. This will be expedited though a convenient modification of (\ref{Hirota2}):

\begin{eqnarray} \label{eggDE}
\frac{\partial^2}{\partial x^2} \widehat{E}_g &=& \frak{C}_g(f_0, f_1, \dots, f_g)\\ \label{hateg}
\widehat{E}_g  &=&  E_g + \sum_{\ell = 1}^g
\frac{2}{(2\ell + 2)!} \frac{\partial^{2\ell}}{\partial x^{2\ell}} E_{g - \ell}(s,x)\\ \label{cg}
\frak{C}_g(f_0, f_1, \dots, f_g) &=& \mbox{the}\,\,\,  \mathcal{O}\left(n^{-2g}\right) \,\,\, \, \mbox{terms of}\,\,\, \log\left( 1 + \sum_{m=1}^\infty \frac1{n^{2m}} \frac{f_m}{f_0}\right)
\end{eqnarray}

\subsubsection{Regular Valence}

\begin{thm} \label{Eg}
For $g>1$, and regular odd valence $j$, $e_g$ is a rational function on $\mathcal{C}$ whose singularities are poles confined to the zeroes of $\Pi_+$ or $\Pi_-$. Locally in the vicinity of the real turning points the maximal order of these poles is $5g-5$.  \end{thm}

\begin{proof}
We begin by establishing the rationality of  $e_g$. This will again be based on an induction in the genus $g$ with the base step given by our explicit expressions for $E_0$ and $E_1$ given earlier in this section. Though these expressions involve logarithms, their higher $x$-derivatives ( which are what appear in the inhomogeneous terms of the induction) are rational. We note that, because of this inductive structure, it follows from relation (\ref{hateg}) that $e_g$ is rational on
$\mathcal{C}$ if and only if $\hat{e}_g$ is ($\hat{e}_g(\xi) = x^{2g-2}\widehat{E}_g$ where $\xi = x^{j/2 - 1}t$).

Expanding the LHS of (\ref{eggDE}) on the self-similar form of $\widehat{E}_g$ and then evaluating at $x=1$, (\ref{eggDE}) becomes an inhomogeneous second order ODE of Cauchy-Euler type:
\begin{eqnarray*} 
\left(\frac{j}{2} - 1\right)^2 \xi^2 \hat{e}''_g + \left(\frac{j}{2} - 1\right) \left(\frac{j}{2} +2 -4g\right) \xi \hat{e}'_g + (2-2g) (1-2g) \hat{e}_g = \frak{C}_g(f_0, f_1, \dots, f_g)
\end{eqnarray*}
where $ ' = \frac{d}{d\xi}$. Changing variables to $\eta = \xi^2$,  this equation becomes
\begin{eqnarray} \label{CEODE}
\left(j - 2\right)^2 \eta^2 \hat{e}''_g + \left(j-2\right) \left[j - 2  - (2g-1) - (2g-2)\right] \eta \hat{e}'_g + (2g-2) (2g-1) \hat{e}_g = \frak{C}_g(f_0, f_1, \dots, f_g)
\end{eqnarray}
where now $ ' = \frac{d}{d\eta}$.

Applying variation of parameters and using the substitution 
\begin{eqnarray} \label{etaform}
\eta(y_0) = \frac{1}{j^2} y_0 \frac{(\widehat{B}_{12} - y_0^{1/2}\widehat{B}_{11})^{j-2}}{\widehat{B}_{12}^j},
\end{eqnarray}
from (\ref{curve}), yields the following integral representation of the solution
\begin{eqnarray} \nonumber
\hat{e}_g(y_0) = \frac1{(2-2g)(1-2g)} \frak{C}_g &+& \frac2{2-2g} (\eta(y_0))^\frac{2-2g}{j-2} \int_0^{y_0} (\eta(y))^\frac{2g-2}{j-2}  {\frak{C}_g}^\bullet dy \\ \nonumber
&-& \frac2{1-2g} (\eta(y_0))^\frac{1-2g}{j-2} \int_0^{y_0} (\eta(y))^\frac{2g-1}{j-2}  {\frak{C}_g}^\bullet dy\\ \label{VarParam}
&+& K_1 \eta(y_0)^\frac{2g-2}{j-2}  + K_2 \eta(y_0)^\frac{2g-1}{j-2} 
\end{eqnarray}
where $^\bullet = \frac{d}{dy}$. (A derivation of this representation is given, for the similar case of regular, even valence, in \cite{EMP08} just following equation (5.14) of that reference.)
Note that the fundamental solutions of (\ref{CEODE}) take the form
\begin{eqnarray*}
\eta(y_0)^\frac{2g-2}{j-2}  &=&   \left(\frac{1}{j^2} \frac{y_0}{\widehat{B}_{12}^2} \right)^\frac{2g-2}{j-2} \left( \frac{\widehat{B}_{12} - y_0^{1/2}\widehat{B}_{11}}{\widehat{B}_{12}}\right)^{2g-2},    \\
\eta(y_0)^\frac{2g-1}{j-2} &=& \left(\frac{1}{j^2} \frac{y_0}{\widehat{B}_{12}^2} \right)^\frac{2g-1}{j-2} \left( \frac{\widehat{B}_{12} - y_0^{1/2}\widehat{B}_{11}}{\widehat{B}_{12}}\right)^{2g-1}.
\end{eqnarray*}

It follows from Theorem \ref{uz} that $\frak{C}_g$ (and so $\frak{C}^\bullet_g$) is rational on $\mathcal{C}$ with poles confined to the zeroes of $\Pi_- \Pi_+$. Additionally, it follows from Theorem \ref{genthm} (a) that $\widehat{B}_{12}$ appearing in (\ref{etaform}) does not vanish anywhere
on $\mathcal{C}$. The polynomial in the numerator of that expression vanishes at the points of tangency of $\mathcal{C}$ with the $y_0$-axis. These produce poles in the pre-factors to the integrals appearing in (\ref{VarParam}). However, bringing those pre-factors under the integral reveals that the integrand in the vicinity of those inflection points is in fact analytic since those points of tangency do not coincide with the zeroes of $\Pi_- \Pi_+$. The remaining factor of $y_0$ appearing in the fundamental solutions does vanish on $\mathcal{C}$ at the point $(\xi^2, y_0) = (0,0)$. However, in the neighborhood of this point we know from the original definition of the $e_h$, as generating functions for enumeration, that  $\hat{e}_g(y_0$ is analytic. Finally we consider the point at infinity on $\mathcal{C}$ which occurs at $(\xi^2, y_0) = (0,\infty)$. It is straightforward to see that as $y_0 \to \infty$, $\eta(y_0) \sim y_0^{j-2}$.  As with the inflection points just discussed, it follows that $\hat{e}_g$ is analytic at $\infty$. So it is indeed the case that the global singularities of $\hat{e}_g(y_0)$ are confined to the zeroes of $\Pi_- \Pi_+$. 

While $\frak{C}^\bullet_g$ clearly has no simple pole, it is a priori possible that its product with one of the fundamental solutions could result in the formation oof simple pole making the integrand fail to be exact. To rule out this possibility we first make use of the {\it loop equations} \cite{EM7} to show that $\xi de_{g}/d\xi$ must be a rational function on $\mathcal{C}$. These equations may be briefly presented as follows.  Recalling the one point function, $\rho_n^{(1)}$, from (\ref{eq:onepoint}) a mild extension \cite{EM7} of Theorem \ref{Workhorse} yields the asymptotic expansion
\begin{eqnarray*}
\int_{-\infty}^\infty \frac{\rho_n^{(1)}(\lambda)}{z - \lambda} d\lambda = \sum_{g = 0}^\infty n^{-2g} R_g(z), 
\end{eqnarray*}
where $R_g(z) = \sum_{j=0}^J \frac{1}{z^{j+1}} \frac{de_{g}(\bf{\xi})}{d\xi_j}$. The loop equations may then be stated as
\begin{eqnarray} \label{loops}
\frac1{2 \pi i} \oint_C \frac{V'(\lambda) - 2 R_0(\lambda)}{\lambda - z} R_g(\lambda) d\lambda &=& \sum_{j=0}^J \frac{1}{z^{j+1}} \frac{dR_{g-1}(z)}{dt_j} 
+ \sum_{g'=1}^{g-1} R_{g'}(z) R_{g-g'}(z),
\end{eqnarray}
where  
$C$ denotes  a counterclockwise contour in the complex $\lambda$-plane surrounding the slit $[r_-, r_+]$.  The LHS iis evaluated by the  residue formula. 
We are interested in picking off the coefficient of $z^{-2J -2}$. It is easy to see that this comes entirely from the residue at $\lambda = z$
Thus, from (\ref{loops}), one sees that
\begin{eqnarray*}
\left[\frac1{2 \pi i} \oint_C \frac{-2 R_0(\lambda)R_g(\lambda)}{\lambda - z}  d\lambda\right]_{-2J - 2} &=& 
 \frac{d^2e_{g-1}}{d\xi_J^2}  + \sum_{g'=1}^{g-1} \frac{de_{g'}}{d\xi_J}  \frac{de_{g-g'}}{d\xi_J} \\
\left [-2 R_0(z) R_g(z)\right]_{-2J - 2} &=&   \frac{d^2e_{g-1}}{d\xi_J^2}  + \sum_{g'=1}^{g-1} \frac{de_{g'}}{d\xi_J}  \frac{de_{g-g'}}{d\xi_J} \\
 \frac{de_0}{d\xi_J}  \frac{de_g}{d\xi_J} &=&   \frac{d^2e_{g-1}}{d\xi_J^2}  + \sum_{g'=1}^{g-1} \frac{de_{g'}}{d\xi_J}  \frac{de_{g-g'}}{d\xi_J} \\
 \xi \frac{de_0}{d\xi}  \xi \frac{de_g}{d\xi} &=&  \xi^2  \frac{d^2e_{g-1}}{d\xi^2}  + \sum_{g'=1}^{g-1} \xi \frac{de_{g'}}{d\xi} \xi \frac{de_{g-g'}}{d\xi}
\end{eqnarray*}
where in the last line we have multiplied both sides of the equation by $\xi_J^2$ and then suppressed the explicit $J$ dependence since we are evaluating at our original $j$-regular potential. By induction the RHS is globally rational as is  $\xi de_0/d\xi$. Hence, by the loop equation, so is $\xi de_g/d\xi (= \eta de_g/d\eta )$. 
It follows that $\eta^2 d^2e_g/d\eta ^2$ is rational as well, since $\eta$ is a coordinate on $\mathcal{C}$. Then, from our observation at the outset, $\eta d\hat{e}_g/d\eta $ and $\eta^2 d^2\hat{e}_g/d\eta ^2$ are rational. Finally applying these last observations to (\ref{CEODE}) we have that $\hat{e}_g$ is globally rational on $\mathcal{C}$. Combining this with our prior analysis we may finally conclude that ${e}_g$  is a rational function on $\mathcal{C}$ whose singularities are poles confined to the zeroes of $\Pi_+$ or $\Pi_-$.
\smallskip

We next turn to determining the precise order of the poles of ${e}_g$ at the real turning points.${e}_g$. It follows from Theorems \ref{uz} and \ref{Dmax} that $\frak{C}_g(f_0, f_1, \dots, f_g)$ as given by (\ref{cg}) is a rational function whose singularities are poles confined to the locus where $\Pi(y_0) = 0$ with maximal pole order equal to $5g-1$. Observing this in (\ref{eggDE}) and noting that anti-differentiation twice decreases the $\Pi$-weight by 2 and hence the pole order by 4, it follows that the maximal pole order of $\widehat{E}_g$ is $5g-5$. 

It is still possible that the maximal pole order of $E_g$ is greater than  $5g-5$ but that this is cancelled by one of the other terms in the definition (\ref{hateg}) of $\widehat{E}_g$. This can be ruled out by strong induction on $k < g$. For $0< \ell < g-1$ one would then have that the maximal pole order of $\frac{\partial^{2\ell}}{\partial w^{2\ell}} E_{g - \ell}(s,x)$ must be equal to $5(g-\ell) -5 + 4\ell = 5g - \ell -5 < 5g-5$. So it will suffice to establish this premise for the base cases of $\ell = g-1, g$ with $g \geq 2$; i.e., for $x$-derivatives of $E_0$ and $E_1$. For $E_0$ this amounts to considering 
\begin{eqnarray*}
\frac{\partial^{2\ell}}{\partial x^{2\ell}} E_0 &=& \frac{\partial^{2\ell-2}}{\partial x^{2\ell-2}} \log(f_0)\\
&=& \frac{\partial^{2\ell-3}}{\partial x^{2\ell-3}} \frac{f_{0x}}{f_0},
\end{eqnarray*}
which, as a consequence of (\ref{jet1}), has maximal pole order  $\leq 4\ell -4 \leq 4g - 12$. Hence this presents no problem. 

Similarly for $E_1$ one has, by direct calculation,
\begin{eqnarray*}
\frac{\partial^{2\ell}}{\partial x^{2\ell}} E_1 &=& \frac{\partial^{2\ell}}{\partial x^{2\ell}} \frac1{24} \log\left(\frac{4 f_0 - h_0^2}{f_0 D}\right)\\
&=& \frac{\partial^{2\ell - 1}}{\partial x^{2\ell -1}} 
\frac1{24}\left[ \frac {-2j\,\left(jx-\, \left( j-2 \right) f_0\right) - \left( j-2
 \right)  h_0^2 }{D} \right. \\
&+& \left. \frac { \left( j-2 \right) ^2 f_0 h_0^2 \left( 2\,jx-2\,
 \left( j-2 \right) f_0+ \left( j-2 \right)  h_0^2 \right)}{D^2}\right. \\
&+& \left. \frac{4\, \left( j-2 \right) ^{2}f_0 h_0^2 \left( jx-
 \left( j-2 \right) f_0 \right) +4\, \left( j-2
 \right) f_0 \left( jw^2- \left( j-2 \right) f_0^2+ \left( j-2 \right)^2 f_0 h_0^2 \right)}{D^2}\right]
\end{eqnarray*}
whose maximal pole order is $\leq 4\ell+2 \leq 4g - 6$. So this does not challenge the maximality of $5g-5$ either. 
\smallskip

Finally, to see that this maximal pole order is actually achieved, we argue that the coefficient of this order in the Laurrent expansion of $e_g$ is proportional to the coefficient in the expansion of $f_g$ at order $5g-1$, $\zeta_g$, which was shown to be positive from Theorem \ref{Dmax} and Proposition \ref{p872}. To establish this identification of coefficients one simply needs to study the structure of the $\mathcal{O}(n^{-2g})$-cumulant in (\ref{cg}). But this is explicitly given by

\begin{eqnarray*}
\frak{C}_g &=& \sum_{m=1}^g (-1)^{m-1} \sum_{\begin{array}{c} \lambda \in \Lambda^g_{m}\end{array}} 
\frac{(m-1)!}{\prod r_j(\lambda)!} \prod_{j=1}^g \left(\frac{ f_j}{f_0}\right)^{r_j(\lambda)}
\end{eqnarray*}
where $\Lambda^g_{m}$ denotes the subset of partitions of $g$ having length $m$ and the $r_j(\lambda)$ are the occupation numbers of the partition $\lambda$.
It follows from Theorem \ref{Dmax} that
\begin{eqnarray*}
\frac{ f_j}{f_0} = \frac{\zeta_j}{z_{0c}(y-y_{0c})^{5j-1}}(1 + o(1))
\end{eqnarray*}
where $\zeta_j > 0$ and satisfies the recursion (\ref{rec9023}). Recalling from elementary partition theory that $\sum r_j = m$ and $\sum j r_j =g$ it follows that that the $m^{th}$ term in the above sum for $\frak{C}_g $ is $\mathcal{O}((y-y_{0c})^{5g-m})$. Hence,
$$
\frak{C}_g = \frac{\zeta_g}{z_{0c}(y-y_{0c})^{5g-1}}(1 + o(1))
$$
and so by the prior arguments 
\begin{eqnarray} \label{turningeg}
e_g = \frac{c_g}{(y-y_{0c})^{5g-5}}(1 + o(1))
\end{eqnarray}
where $c_g$ is proportional to $\zeta_g$. 
\end{proof}

\begin{rem}
The arguments presented in the above proof for the rationality of $e_g$ extend, mutatis mutandis, to the case of regular even valence and serve to amplify the somewhat terse argument for that case given in \cite{Er09}.
 \end{rem}
 
 \begin{cor} \label{cor:egasymp} The large size asymptotics for regular j-valent (j odd) g-maps is given by 
 \begin{eqnarray} 
\frac{\kappa^{(j)}_{g}(2k)}{(2k)!} &\sim& \frac{c_g}{ \Gamma(\frac{5g-5}{2})} \xi_{c}^{1/2 - 5g/2 -2k} (2k)^{(5g-7)/2} \qquad \text{as }k\rightarrow \infty,
\end{eqnarray}
where $t_c$ is the radius of convergence for the Taylor-Maclaurin expansion of $e_g(t)$. 
 \end{cor}
This follows from (\ref{turningeg}) in precisely the same way as Corollary \ref{c873} follows from (\ref{turningfg}) and Proposition \ref{p872}.

\subsubsection{Mixed Valence}

\begin{thm} \label{Eg2}
For $g>1$, $e_g$ is a rational function on $\mathcal{S}$ whose poles are confined to the zero-locus of $\psi^2_{V,1} - f_0 \phi^2_{V,1}$, where $\psi_{V,1}, \,  \phi_{V,1}$ are the string polynomials for the mixed valence potnetial $V$. \end{thm}

\begin{proof} To establish the rationality of the $e_g$ for general potentials we will make use again of the loop equations (\ref{loops}). As in the regular case we use this to first inductively establish the rationality. Also as before we focus, in the residue sum on the LHS of these equations, on the terms that are Laurent in $z$. In particular we consider the equations that come from the balance of coefficients on both sides of (\ref{loops}) for $z^{-2}, z^{-3}, \dots, z^{-J-2}$. This yields a lower triangular system of Toeplitz type:
\begin{eqnarray} \label{loops2}
\left(\begin{array}{ccccc}
\dfrac{d e_0}{d\xi_0} & 0 & \dots & &   \\
\dfrac{d e_0}{d\xi_1} & \dfrac{d e_0}{d\xi_0} & 0 & \dots &\\
\dfrac{d e_0}{d\xi_2} & \dfrac{d e_1}{d\xi_1} & \dfrac{d e_0}{d\xi_0} & 0 & \dots\\
\vdots & \ddots & \ddots & \ddots & 0   \\
\dfrac{d e_0}{d\xi_J} & \dots & \dfrac{d e_0}{d\xi_2}& \dfrac{d e_0}{d\xi_1} & \dfrac{d e_0}{d\xi_0} 
\end{array}\right) 
\left(\begin{array}{c}
\dfrac{d e_g}{d\xi_0} \\ \dfrac{d e_g}{d\xi_1} \\ \dfrac{d e_g}{d\xi_2} \\ \vdots \\ \dfrac{d e_g}{d\xi_J}
\end{array}\right) 
&=& \text{vector of terms involving derivatives of}\,\,\, e_h, \,\,\,\, h <g.
\end{eqnarray} 
The diagonal elements, $\frac{d e_0}{d\xi_0}$, are all equal to the total mass of the equilibrium measure, (\ref{dmu}), and so are positive (the elements of the $j^{th}$ sub-diagonal all equal the $j^{th}$ moment of the equilibrium measure). Thus the system is invertible, showing, by induction, that the derivative $d e_g/ d\xi_j$ are rational functions of $h_0$ and $f_0$. The base step is given by the observation from section \ref{sec:MV} that the just mentioned derivatives 
of $e_0$ are rational.

Again, as in the regular case, we expand $\partial^2_x E_g$ in (\ref{eggDE}) which, for mixed valence, takes the form
\begin{eqnarray*}
(1-2g)(2-2g) e_g + (3-4g) \sum_{j=1}^J \left( \frac{j}{2} - 1\right) \xi_j \frac{d}{d\xi_j} e_g + \sum_{j,k =1}^J \left( \frac{j}{2} - 1\right) \left( \frac{k}{2} - 1\right) \xi_j \frac{d}{d\xi_j}   \xi_k \frac{d}{d\xi_k} e_g &=& \text{rational function}.
\end{eqnarray*}
It then follows that $e_g$ must be a rational function on $\mathcal{S}$. The statement on the polar locus follows from the same construction that established 
(\ref{A11} - \ref{A12}) and (\ref{basic1}). 
\end{proof}

\subsection{Special Cases and a Conjecture}
It follows from Theorem \ref{Eg} that $E_g$, for $g \geq 2$  are rational functions on $\mathcal{C}$ with poles contained in the locus $\Pi(y_0)=0$. Moreover, as was shown in the proof of this theorem, $E_g$ is a function of  just $y_0$ since it is built up inductively from the cumulants, $\frak{C}_g$ which are functions of the $f_k$ that in turn are just functions of  $y_0$. Hence, the $E_g$ are rational functions on $\mathcal{C}$ with poles restricted to the zero-set of $\Pi$ on $\mathcal{C}$. Let us denote this locus by $[\Pi]$ which is a divisor on $\mathcal{C}$. $E_g$ is then an element of a sub-linear system  of the divisor $(5g-5) [D]$ on $\mathcal{C}$.  The dimension of this vector space plays a role analogous to that of the the degrees of freedom represented by the coefficients $c^{(g)}_k$ appearing in (\ref{leadcoeff}). Similarly $z_g$ is an element of the linear system with divisor $(5g-1) [D]$ and $u_{2g}$ is an element of the pullback of this system to $\widehat{\mathcal{C}}$. Also, $u_{2g+1}$ is an element of the system associated to the pullback of the divisor $(5g+1) [D]$. A direct comparison can be made in the trivalent case and so we take a moment to do this in order to better understand how things should work. 

Recall that in the trivalent case $\mathcal{C}$, as given by (\ref{3curve}), is a degree 3 covering of the $\xi^2$ projective line branched simply at the roots of $y^2_0 -8 y_0 +4$, $y_0 = 2(2 \mp \sqrt{3})$. Noting that in this case there is a simple change of variables from $z_0$ to $y_0$ given by
\begin{eqnarray} \label{ztoy}
z_0 &=& \frac{2+y_0}{2-y_0}
\end{eqnarray}
we use this to transform the generating functions given in \cite{EP11} to their forms on $\mathcal{C}$:
\begin{eqnarray}
u_1 &=& \frac{y_0 - 2}{4 -8y_0 + y_0^2} \\
\frac{u_2}{u_0} &=& -1/2\,{\frac { \left( -2 + y_0 \right) ^{3} \left( 20-16\,y_0+5\,{y_0}^{2}
 \right) y_0}{ \left( 4-8\,y_0+{y_0}^{2} \right) ^{4}}}\\
\frac{z_1}{z_0} &=& 2\,{\frac { \left( -2+y_0 \right) ^{2} \left( 20-16\,y_0+5\,{y_0}^{2}
 \right) {y_0}^{2}}{ \left( 4-8\,y_0+{y_0}^{2} \right) ^{4}}}.
\end{eqnarray}
Thus $u_1$ is partially characterized by saying that it is a rational function on $\mathcal{C}$ with poles in $[D]$ that vanishes simply at $y_0 = -2, \infty$. This determines $u_1$ up to one degree of freedom.  
Similarly $z_1/z_0$ (respectively $u_2/u_0$) is partially characterized by saying that it is a rational function on $\mathcal{C}$ with poles in $4[D]$ that vanishes doubly (i.e., to second order) at $y_0 = 0, 2, \infty$ (resp. simply at $y_0 = 0$, doubly at $y_0 = \infty$ and triply at $y_0 = 2$). This determines $z_1/z_0$ (respectively $u_2/u_0$) up to three degrees of freedom. We note that this coincides with the three degrees of freedom needed to determine $z_1/z_0$ in the even valence cases. 

Based on \cite{EP11}, one can show that
\begin{eqnarray} \label{e2j3}
e_2 &=& 1/30\,{\frac { \left( 2800-4240\,y_0+2712\,{y_0}^{2}-1060\,{y_0}^{3}+175\,{y_0
}^{4} \right) {y_0}^{3}}{ \left( 4-8\,y_0+{y_0}^{2} \right) ^{5}}}
\end{eqnarray}
which is partially characterized by saying that it is a rational function on $\mathcal{C}$ with poles in $5[D]$ that vanishes triply at $y_0 = 0, \infty$. This determines $e_2$ up to five degrees of freedom which  coincides with the number of degrees of freedom for $e_2$ when the valence is even.

The perspective illustrated by this extends naturally to higher odd valence. Based on computer-aided evaluations of a number of other low valence cases, we close this section with the following conjectured  characterization of the divisor structure of higher genus generating functions:
\begin{conj} 
In the following $y_0^*$ denotes the value at which the real spectral curve crosses the $y_0$-axis. For $\nu > 1$,
\begin{itemize}
\item[a)] $u_{2g+1}$ is a rational function on $\widehat{\mathcal{C}}$ with poles in the pullback of $(5g+1)[D]$ that vanishes to order $2g+1$ at $y_0 = y_0^*$, to order $r = \max\left( 0, \lfloor \frac{2g-1/2}{\nu - 1/2} \rfloor\right)$ at $y_0 = 0$ and to order $2g+1$ at $y_0 = \infty$;
\item[b)] $u_{2g}$ is a rational function on $\widehat{\mathcal{C}}$ with poles in the pullback of $(5g-1)[D]$ that vanishes to order $2g$ at $y_0 = y_0^*$, to order $r = \max\left( 0, \lfloor \frac{2g-3/2}{\nu - 1/2} \rfloor\right)$ at $y_0 = 0$ and to order $2g$ at $y_0 = \infty$;
\item[c)] $z_{g}$ is a rational function on $\mathcal{C}$ with poles in $(5g-1)[D]$ that vanishes to order $2g$ at $y_0 = y_0^*$, to order $r = \max\left( 0, \lceil \frac{2g-3/2}{\nu - 1/2} \rceil\right)$ at $y_0 = 0$ and to order $2g$ at $y_0 = \infty$;
\item[d)] $e_{g}$, for $g \geq 2$, is a rational function on $\mathcal{C}$ with poles in $(5g-5)[D]$ that vanishes to order $r = \max\left( 0, \lceil \frac{g-1/2}{\nu - 1/2} \rceil\right)$ at $y_0 = 0$ and to order $2g-1$ at $y_0 = \infty$.
\end{itemize}
\end{conj}
\subsection{Integral Formulas for Map Enumeration} \label{integform}

We can now return to one of our initial primary motivations: the systematic concise enumeration of maps with fixed discrete characteristics. Starting with the case of $g=0$, or planar maps and with $j$ odd, the Taylor coefficients in (\ref{TMexp}) can be presented in terms of contour integration by
\beann
\frac{\kappa^{(j)}_{0}(2m)}{(2m)!} &=& \frac1{2\pi i}\oint e_0(\xi) \frac{2 d\xi}{\xi^{2m+1}}.
\eeann
The factor of 2 here stems from the fact that the contour here is taken to be around the preimages on $\widehat{\mathcal{C}}$ of a small circle about the origin in the $\xi^2$-plane.

Note that it follows from Euler's genus formula that in all cases of a single odd valence, the number of vertices for a map must be even; hence, only the even Taylor coefficients need be considered. Integrating by parts, this conotur integral may be re-expressed as
\beann
\frac{\kappa^{(j)}_{0}(2m)}{(2m)!} &=& \frac1{2\pi i}\oint \frac1{2m}\frac{de_0}{dy_0} 2 \xi^{-2m} dy_0\\
&=& \frac1{2\pi i}\oint \frac1{m}\frac{de_0}{dy_0} \frac{j^{2m}}{y_0^m}\frac{\widehat{B}_{12}^{jm}}{(\widehat{B}_{12} - y_0^{1/2}\widehat{B}_{11})^{(j-2)m}} dy_0
\eeann
where we have used (\ref{curve}) to derive the second equation. The integrand in the second equation is a differential on the rational curve $\mathcal{C}$ expressed entirely in terms of the global uniformizing parameter $y_0$ on that curve. In particular note that as a consequence of the representation in Theorem \ref{e0y0}, $e_0$ is a rational function of $y_0$ and hence so is $d e_0 / d y_0$. A more explicit form of this derivative is given by
\beann
\frac{de_0}{dy_0} &=& C \frac{\Pi(y_0) p_{j-1}(y_0)}{\widehat{B}_{12} (\widehat{B}_{12} - y_0^{1/2}\widehat{B}_{11})^3}
\eeann
where $p_{j-1}(y_0)$ is a monic polynomial of degree $j-1$ (the same degree as $\Pi$). Furthermore, $C$ is a constant such that $\lim_{y_0 \to \infty} \frac{de_0}{dy_0} = \frac1{2j}$. Some examples are 
\beann
\frac{de_0}{dy_0} &=& 1/6\,{\frac { \left( {y}^{2}-8\,y+4
 \right) \left( {y}^{2}+4\,y-4 \right)   }{  \left( 2+y \right) \left( -2+y \right) ^{3}}}; \,\,\,\,\, j=3\,\,\,\,\,\,\, (C = -\frac16)\\
\frac{de_0}{dy_0} &=& 1/10\,{\frac { \left( {y}^{4}-8\,{y}^{3}-20\,{y}^{2}-32\,y+4 \right) 
 \left( {y}^{4}+20\,{y}^{3}+68\,{y}^{2}+40\,y-12 \right) }{ \left( {y}^{2}+12\,y+6 \right) \left( {y}^{2}-2 \right) ^{3}  }}; \,\,\,\,\, j = 5\,\,\,\,\,\,\, (C = - \frac95).
\eeann
Inserting this expression for the derivative into the Taylor coefficient formula yields
\begin{eqnarray} \label{residue}
\frac{\kappa^{(j)}_{0}(2m)}{(2m)!} &=& \frac1{2\pi i} \frac{C j^{2m}}{m}\oint \frac{\Pi(y_0) p_{j-1}(y_0) \widehat{B}_{12}^{jm - 1}}{ (\widehat{B}_{12} - y_0^{1/2}\widehat{B}_{11})^{(j-2)m + 3}} \frac{dy_0}{y_0^m}.
\end{eqnarray}
From here it is straightforward to calculate the residues of this rational differential. Some examples are
\beann
\left[\frac{\kappa^{(3)}_{0}(2)}{3 \cdot 2!}, \frac{\kappa^{(3)}_{0}(4)}{3 \cdot 4!}, \frac{\kappa^{(3)}_{0}(6)}{3 \cdot 6!}, \frac{\kappa^{(3)}_{0}(8)}{3 \cdot 8!}, \frac{\kappa^{(3)}_{0}(10)}{3 \cdot 10!}, \dots\right] 
&=& \left[2, 72, 4536, 373248, \frac{180138816}{5}, \dots \right]\\
\left[\frac{\kappa^{(5)}_{0}(2)}{5 \cdot 2!}, \frac{\kappa^{(5)}_{0}(4)}{5 \cdot 4!}, \frac{\kappa^{(5)}_{0}(6)}{5 \cdot 6!}, \frac{\kappa^{(5)}_{0}(8)}{5 \cdot 8!}, \frac{\kappa^{(5)}_{0}(10)}{5 \cdot 10!}, \dots\right] 
&=& \left[ 18, 54000, 345060000, 3098250000000, 33814409850000000, \dots\right].
\eeann
(The additional division by 3 and 5 respectively is to take into account the removal of edge labelling in order to get a purely geometric (unlabelled) count as discussed after equation (\ref{GenFunc}).)

In the trivalent case ($j=3$) this sequence has been calculated by other means and a general formula found in \cite{EP11} and \cite{BD10}:
$\kappa^{(3)}_{0}(2m)/3 \cdot (2m)! = \frac{3^{2m} 2^{3m}}{3m} \frac{\Gamma(3m/2)}{\Gamma(m/2) \Gamma(3 + m)})$. The existence of a compact representation, in terms of known functions, of these residues in the trivalent case is not surprising given the hypergeometric character of the integral in this case. However, for higher valence one has more than three poles in the differential so one would expect a higher order function theory to be involved, albeit one that still is related to the configuration of poles in the differential. One may note further that this configuration is that of the inflection points of the curve $\mathcal{C}$ along the the $y_0$-axis; i.e., to the stagnation points of the characteristic flow.
\medskip

We turn next to the genus 1 generting functions whose Taylor coefficients are given by
\beann
\frac{\kappa^{(j)}_{1}(2m)}{(2m)!} &=& \frac1{2\pi i}\oint e_1(\xi) \frac{2 d\xi}{\xi^{2m+1}}\\
&=& \frac1{2\pi i}\oint \frac1{m}\frac{de_1}{dy_0} \xi^{-2m} dy_0\\
&=& \frac1{2\pi i}\oint \frac1{m}\frac{de_1}{dy_0} \frac{j^{2m}}{y_0^m}\frac{\widehat{B}_{12}^{jm}}{(\widehat{B}_{12} - y_0^{1/2}\widehat{B}_{11})^{(j-2)m}} dy_0.
\eeann
With $e_1$ given by (\ref{E1}), its derivative is given by 
\beann
\frac{de_1}{dy_0} &=& \frac1{24} \frac{2 (\partial \widehat{B}_{12}/ \partial y_0) \Pi  - \widehat{B}_{12} (\partial \Pi / \partial y_0 ) }{\widehat{B}_{12} \Pi}.
\eeann
Some examples are
\beann
\frac{de_1}{dy_0} &=& -1/2\,{\frac {-2+y}{ \left( 2+y \right)  \left( {y}^{2}-8\,y+4\right) }}; \,\,\,\,\, j=3\\
\frac{de_1}{dy_0} &=& -2/3\,{\frac {2\,{y}^{4}-2\,{y}^{3}-3\,{y}^{2}+8\,y-18}{ \left( {y}^{2
}+12\,y+6 \right)  \left( {y}^{4}-8\,{y}^{3}-20\,{y}^{2}-32\,y+4 \right) }}; \,\,\,\,\, j=5.
\eeann
This yields a general formula for the genus 1 Taylor coefficients:
\begin{eqnarray} \label{residue1}
\frac{\kappa^{(j)}_{1}(2m)}{(2m)!} &=& \frac1{2\pi i} \frac{j^{2m}}{24 m} \oint \frac{(2 (\partial \widehat{B}_{12}/ \partial y_0) \Pi  - \widehat{B}_{12} (\partial \Pi / \partial y_0 ) ) \widehat{B}_{12}^{jm-1}}{\Pi (\widehat{B}_{12} - y_0^{1/2}\widehat{B}_{11})^{(j-2)m}} \frac{dy_0}{y_0^m}.
\end{eqnarray}
The residue sequences for our examples start out as
\beann
\left[\frac{\kappa^{(3)}_{1}(2)}{3 \cdot 2!}, \frac{\kappa^{(3)}_{1}(4)}{3 \cdot 4!}, \frac{\kappa^{(3)}_{1}(6)}{3 \cdot 6!}, \frac{\kappa^{(3)}_{1}(8)}{3 \cdot 8!}, \frac{\kappa^{(3)}_{1}(10)}{3 \cdot 10!}, \dots\right] &=& \left[\frac32, 135, 16524, 2291976, \frac {1701555984}{5}, \dots \right]\\
\left[\frac{\kappa^{(5)}_{1}(2)}{5 \cdot 2!}, \frac{\kappa^{(5)}_{1}(4)}{5 \cdot 4!}, \frac{\kappa^{(5)}_{1}(6)}{5 \cdot 6!}, \frac{\kappa^{(5)}_{1}(8)}{5 \cdot 8!}, \frac{\kappa^{(5)}_{1}(10)}{5 \cdot 10!}, \dots\right] 
&=& \left[ 90, 1035000, 15746400000, 268824825000000, 4889505205800000000, \dots\right].
\eeann
\medskip

Finally we turn to the case of $g > 1$. We have seen that these generating functions have the form
\beann
e_g &=& \frac{P^{(g)}(y_0)}{\Pi^{5g-5}}
\eeann
for a polynomial $P^{(g)}(y_0)$. It follows that its derivative has the form
\beann
\partial e_g / \partial y_0 &=& \frac{M^{(g)}(y_0)}{\Pi^{5g-4}}
\eeann
for another polynomial $M^{(g)}(y_0)$. Proceeding as before we have the integral representation
\begin{eqnarray} \label{residueg}
\frac{\kappa^{(j)}_{g}(2m)}{(2m)!} &=& \frac1{2\pi i} \frac{j^{2m}}{m} \oint \frac{M^{(g)}(y_0)}{\Pi^{5g-4}} \frac{\widehat{B}_{12}^{jm}}{(\widehat{B}_{12} - y_0^{1/2}\widehat{B}_{11})^{(j-2)m}} \frac{dy_0}{y_0^m}.
\end{eqnarray}
We take here as an example the case of $g=2, j=3$, using (\ref{e2j3}),
\beann
\partial e_2 / \partial y_0 &=& -1/2\,{\frac {{y}^{2} \left( -2+y \right)  \left( 2+y \right)  \left( 
35\,{y}^{4}-96\,{y}^{3}+152\,{y}^{2}-384\,y+560 \right) }{ \left( {y}^
{2}-8\,y+4 \right) ^{6}}}; \,\,\,\,\, j=3
\eeann
whose initial residue sequence is
\beann
\left[\kappa^{(3)}_{2}(2)/3 \cdot (2)!, \kappa^{(3)}_{2}(4)/3 \cdot (4)!, \kappa^{(3)}_{2}(6)/3 \cdot (6)!, \kappa^{(3)}_{2}(8)/3 \cdot (8)!, \kappa^{(3)}_{2}(10)/3 \cdot (10)!, \kappa^{(3)}_{2}(12)/3 \cdot (12)!,\dots\right] \\
= \left[0, 0, \frac{2835}{2}, 739206, \frac{1301676156}{5}, 77075478720, \dots \right].
\eeann

\subsection{Asymptotic Enumeration} \label{asympenum}

As mentioned in section \ref{approaches}, more recent purely combinatorial approaches to map enumeration build on some remarkable combinatorial bijections. However building such combinatorial bridges requires conditioning the map class in some way whether by rooting or coloring or some other means. By constrast the enumerations presented in this paper are direct and free from such conditioning. The other aspect to note is that those combinatorial studies focus almost entirely on asymptotic evaluations; i.e., asymptotic descriptions of the map counts as the size of the map becomes very large. Our enumerations, on the other hand, can be equally well carried out in the small and moderate size ranges as we have illustrated in section \ref{integform}. But, by way of comparison, let us now focus on how our approach compares in the setting of asymptotic enumerations.

Many of the combinatorial approaches, such as  \cite{JV, BGR, CMS, Chap},  comment on the universal character of the leading coefficients $c_g$ in (\ref{KAPPA-ASYMP}) for various restricted settings. This is indeed related to universality in the sense we have been discussing in this paper; however, it needs to be qualified.  The more precise statement is that there are {\it basins of attraction} corresponding to the stable neighborhoods in $\vec{t}$ of regular maps with fixed genus in the sense that we have made precise in section \ref{stability}. Moreover, the dependence of the $c_g$ on valence  $j$  is in 1:1 correspondence with the {\it ladder} of equations in the integrable hierarchy  of the Painlev\'e I equation.
So the universality here is inherent in this {\it full} integrable hierarchy rather than in just one particular integrable equation. Indeed we showed in \cite{Er09} that for $j$ even the recursion (\ref{RECURRENCE}) with $C_k$ corresponding to level $j$ also is a recursion for the coefficients in an asymptotic expansion of an analogue of the tritronqu\'ee solution for the $j^{th}$ equation in the Painlev\'e I hierarchy. Frroom this paper we are able to extend this to the case of valence 3 and formulate a precise conjecture for how this should extend to arbitrary odd valence (see Appendix \ref{POne}).

Prior works on enumerative asymptotics did not make these distinctions concerning $c_g$; however, from our perspective the classes of maps considered in those works appear to fall primarily into the basin of attraction of valence 3 or valence 4 ( dually, triangulations or quadrangulations) which correspond to the original Painlev\'e I transcendent. Looking at their associated critical points, $t_c$, one may observe that these are closer to the origin than the critical parameters for higher valence. Looking  at (\ref{KAPPA-ASYMP}) one sees, therefore, that the lower valence expressions are asymptotically dominant. So from this point of view it may not be so surprising that these earlier results seem to have condensed on the lowest valence asymptotics. Of course in the case of mixed valences that are not just perturbations of a regular case, there may be more complicated phase transiitons. However, our point is that since our basic results are global and exact, we are in a position to derive complete asymptotic expansions in which one may compare asymptotic contributions coming from multiple critical points. We will not delve further into this observation in this paper, but defer it to later work.

\appendix

\section{Review of the structure of even valence maps} \label{evenvalence}
\renewcommand{\theequation}{A.\arabic{equation}}
In \cite{Er09}, \cite{CLRM} the following two theorems were respectively established, where $z_g$, inclucing $z_0$, are defined in (\ref{b-shift}). In the case of even valence, the $u_k$ and hence $h_k$ are idendentically zero.
\begin{thm} \label{result} \cite{Er09} For $j=2\nu$ one has

\begin{eqnarray} \nonumber
 z_g (z_0) &=& \frac{z_0  (z_0 -1) P_{3g-2}(z_0)}{(\nu -(\nu -1)z_0)^{5g-1}}\\
 \label{rational} &=& z_0 \left\{ \frac{a_0^{(g)}(\nu)}{(\nu - (\nu-1)z_0)^{2g}} + \frac{a_1^{(g)}(\nu)}{(\nu - (\nu-1)z_0)^{2g+1}}+ \cdots + \frac{a_{3g-1}^{(g)}(\nu)}{(\nu - (\nu-1)z_0)^{5g-1}}\right\} \,\,\,\,\,\,\,\,\,\,\,\,\,\,\,\,\,\,\,\,
\end{eqnarray}
where $P_{3g-2}$ is a polynomial of degree $3g-2$ in $z_0$ whose coefficients are rational functions of $\nu$ over the rational numbers $\mathbb{Q}$, $z_0(\xi) = \sum_{\ell \geq 0} c_\nu^\ell \zeta_\ell \xi^\ell $ where here  \cite{EMP08}
$$
\zeta_\ell = \frac1\ell{\nu \ell \choose \ell - 1}
$$
(so that $z_g$ is naturally a function of $\xi$) and where $a_{3g-1}^{(g)}(\nu)$ is positive and satisfies a recursion in $g$ entirely analogous to that for the coefficients of the asymptotic expansion at infinity of the the Painlev\'e I equation (for $\nu = 2$ the two recursions are equivalent).  
\end{thm}
From this, the coefficient asymptotics for $z_g$ are seen to be 

\begin{eqnarray} \label{zg-asymp}
\frac{\kappa^{(j; 2:1)}_{g}(m)}{m!} &=& \frac{\nu}{\nu-1} \frac{a_{3g-1}^{(g)}(\nu) }{ \left( \sqrt{2\nu} (\nu - 1)\right)^{5g-1}} \frac{m^{\frac{5g-3}{2}}}{\Gamma\left(\frac{5g-1}{2} \right)}
s_c^{-m} (1 + o(1))\\ \nonumber
s_c &=& \frac{(\nu - 1)^{\nu -1}}{c_\nu \nu^\nu},\\ \nonumber
c_\nu &=& (\nu + 1) {2\nu \choose \nu+1},
\end{eqnarray}
as $m \to \infty$. (Here $\kappa^{(j; 2:1)}_{g}$ denotes counts of $j$-valent $g$-maps with 2 legs (1-valent vertices).)

\begin{thm} \label{thm51} \cite{CLRM}
For $j=2\nu$ one has
\begin{eqnarray} \nonumber
e_0(z_0) &=& \frac12 \log z_0 + \frac{(\nu-1)^2}{4\nu(\nu+1)}\left(z_0 - 1\right)\left(z_0 - \frac{3(\nu + 1)}{\nu - 1} \right)\\  \label{e1ev}
e_1(z_0) &=& -\frac1{12} \log(\nu - (\nu - 1) z_0),
\end{eqnarray}
which may be found in \cite{EMP08} and are also explicitly derivable from (\ref{E1}) and (\ref{e0}).

For $j=2\nu$ and $g \geq 2$,
\begin{eqnarray}
\label{note} e_g(z_0) &=& C^{(g)} + \frac{c_0^{(g)}(\nu)}{(\nu - (\nu-1)z_0)^{2g-2}} + \cdots + \frac{c_{3g-3}^{(g)}(\nu)}{(\nu - (\nu-1)z_0)^{5g-5}}\\
\label{note2}&=&  \frac{(z_0 - 1)^r Q_{5g-5-r}(z_0)}{(\nu - (\nu - 1)z_0)^{5g-5}}\, ,\\
 \label{note3} r &=& \max\left\{ 1, \left\lfloor\frac{2g-1}{\nu-1}\right\rfloor\right\}\, ,
\end{eqnarray}
for all $\nu \geq 2$.  The top coefficient and the constant term are respectively given by
\begin{eqnarray} \label{leadcoeff} 
&&\\ \nonumber 
&& c_{3g-3}^{(g)}(\nu) = \frac1{(5g-5)(5g-3) \nu^2} a_{3g-1}^{(g)}(\nu) > 0\, ,\\ \label{Constant} 
&&\\  \nonumber
 && C^{(g)}  =   -2  (2g-3)! \left[\frac{1}{(2g+2)!} - \frac{1}{(2g)! 12} + \frac{(1 - \delta_{2,g})}{(2g-1)!}\sum_{k=2}^{g-1} \frac{(2-2k)_{2g-2k+2}}{(2g-2k+2)!} C^{(k)}\right] \qquad
\end{eqnarray}
 and $(r)_\ell = r(r-1)\dots(r-\ell+1)$. 
\end{thm}

From this, the coefficient asymptotics for $e_g$, $g \geq 2$, are seen to be   
\begin{eqnarray} \label{eg-asymp}
\frac{\kappa^{(2\nu)}_{g}(m)}{m!}&=& \frac{c^{(g)}_{3g-3}(\nu)}{(\sqrt{2\nu}(\nu-1))^{5g-5}} \frac{m^{\frac{5g-7}{2}}}{\Gamma\left( \frac{5g-5}{2}\right)} s_c^{-m} (1+ o(1))
\end{eqnarray}
which is consistent with asymptotics found by Tutte and his school by purely combinatorial means (see \cite{Gao}). The asymptotics in 
(\ref{zg-asymp}), (\ref{eg-asymp}) follow directly, respectively, from the Laurent polynomial representations (\ref{rational}), (\ref{note}) by a standard result of analytical combinatorics \cite{FO90}.

\section{Zeroes of the Appell Polynomials} \label{Appellzeroes}
\renewcommand{\theequation}{B.\arabic{equation}}
\begin{prop} 
$S_n(\zeta)$, viewed as a function of a complex $\zeta$ variable, has exactly $n$ zeroes (which must therefore be simple) along the imaginary axis in the complex $\zeta$ plane which are symmetric about the origin and the zeroes of $S_{n-1}(\zeta)$ interlace those of $S_n(\zeta)$.  $\widehat{B}_{12}(\nu)$ is inductively given by 
\begin{eqnarray} 
\widehat{B}_{12}(\nu) &=& \partial^{-2} S_{2 \nu -2}(\sqrt{y_0}) + C_{2\nu}
\end{eqnarray}
where $C_{2\nu} = {2\nu \choose \nu}$ are the coefficients of the generating function $(1 - w \partial_w) C(w)$  ($C(w)$ here is the generating function for the standard Catalan numbers)  which in turn is related to the Fourier transform of the Bessel function of the first kind as in (\ref{FourierBessel}). Consequently, one further has that $\widehat{B}_{12}(\nu)(y_0)$ has $\nu$ negative real zeroes as a function of $y_0$.
\end{prop}
\begin{proof}
In showing this we will make essential use of (\ref{genfcn}) and (\ref{diffrecur}). We first observe that since $S_1(\zeta) = \zeta$ this statement (and in fact the Lemma) are clearly true when $n = 1$.  Continuing, we have $\partial S_2 = 2 S_1 = 2\zeta$. Hence $S_2 =\zeta^2 + C_2$. But $C_2$ is determined from the constant term in our general expression, (\ref{s2n}), for $S_{2\nu}(\zeta)$ which in general is $C_{2\nu} = {2\nu \choose \nu}$ (which is also characterized by (\ref{FourierBessel})). In this case that gives $C_2 = 2$; i.e., $S_2(\zeta) = \zeta^2 + 2$, as we already knew and which is consistent with what we saw in the valence 3 case for $\widehat{B}_{12} (= y_0 + 2)$. Then, for $n = 2$, one has $\partial S_3 = 3 S_2 = 3\zeta^2 + 6$. Since $S_{2\nu + 1}$ is odd, one must have $C_{2\nu + 1} \equiv 0$ by (\ref{s2n+1}). Hence $S_3(\zeta) = \zeta^3 + 6\zeta$. Continuing one more step one finds $\partial S_4(\zeta) = 4 S_3(\zeta)$ so that $S_4(\zeta) = \zeta^4 + 12 \zeta^2 + C_4 = \zeta^4 + 12 \zeta^2 + 6$ which is consistent with what we saw in the valence 5 case for $\widehat{B}_{12} (= y_0^2 + 12 y_0 +6)$ which has the negative real roots, $-6 \pm \sqrt{30}$, so that the four roots of  $S_4(\zeta)$ are indeed pure imaginary and symmetric about zero. 

To establish the proposition in general we consider $S_{2 \nu}$ asymptotically for large $\nu$.
We make use of the forward representations, (\ref{sys1}) and (\ref{s2nu}), which together imply
\begin{eqnarray*}
S_{2n}(i\sqrt{y}) &=& [z^0] (z + i \sqrt{y} + z^{-1})^{2n} \\
&=& \sum_{k = 0}^{n} {2n \choose 2k, n - k, n - k} (-y)^{k}
\end{eqnarray*}
by setting $h_0 = i \sqrt{y}$ and $\sqrt{f_0} = 1$. The representation in the first line above may be expressed as a conotur integral
\begin{eqnarray*}
S_{2n}(i\sqrt{y}) &=& \frac1{2 \pi i} \oint_C \exp(n g(z)) \frac{dz}{z}\\
g(z) &=& 2 \log(z + i \sqrt{y} + z^{-1}).
\end{eqnarray*}
where we take the conotur $C$ to be the unit circle. We are considering this representation for $0 < y < \infty$. It is straightforward to check that the singularities of $g$ lie outside the unit circle for $y$ in this range. It is also straightforward to check that the critical points of $g$ are located at $z = \pm 1$. We will use this contour representation and steepest descent to study the structure of $S_{2n}$ for large $n$. It is evident that the magnitude of the integrand along the unit circle but away from a small neiighborhood of the critical points is exponentially suppressed, in $n$, in comparison to its values on this circle in the neighborhoods of the critical points. So up to exponentially small errors for large $n$ it suffices consider the integral locally.For initial simplicity we Taylor expand $g(z)$ near $z = 1$:
\begin{eqnarray} \label{adtraj}
g(z) - 2 \log(2 + i \sqrt{y}) &=& \frac2{2 + i \sqrt{y}} (z-1)^2 - \frac2{2 + i \sqrt{y}} (z-1)^3 + \mathcal{O}(|z - 1|^4).
\end{eqnarray}
The steepest descent/ascent trajectories emerging from $z=1$ are then locally given by
\begin{eqnarray*}
0  &=& \Im \left( g(1 + w) - 2\log(2 + i \sqrt{y})\right)\\
&=& \Im \left( \frac{2 w^2}{2 + i \sqrt{y}} (1 - w + \mathcal{O}(|w|^2)) \right)
\end{eqnarray*}
where $w = z - 1$. Setting $w = \alpha + i \beta$, this becomes, for $\alpha \ne 0$
\begin{eqnarray*}
\alpha^2 \left( \left(\frac{\beta}{\alpha} \right)^2 + \frac4{\sqrt{y}} \frac{\beta}{\alpha} - 1 \right) 
&=& 0
\end{eqnarray*}
whose 4 descent/ascent branches therefore emerge at $\pm$ the angles, $\theta$ determined by 
\begin{eqnarray*} 
\tan \theta &=& \frac{\beta}{\alpha} \\
&=& - \frac2{\sqrt{y}} \left( 1 \mp \sqrt{1 + y/4}\right).
\end{eqnarray*}
The descent trajectories  correspond to the choice of the plus sign on the radical. Along these trajectories the values of (\ref{adtraj}) are real and, with a bit of calculation, are locally seen to be given by
\begin{eqnarray} \label{drate}
-\frac{8 \alpha^2}{y}\left( 1 + \sqrt{1 + y/4} \right) + \mathcal{O}(\alpha^3).
\end{eqnarray}
A completely similar analysis near $z = -1$ leads to steepest descent/ascent directions given by
\begin{eqnarray*} 
\tan \theta &=& \frac{\beta}{\alpha} \\
&=&  \frac2{\sqrt{y}} \left( 1 \pm \sqrt{1 + y/4}\right)
\end{eqnarray*}
with the descent directions again being given by the choice of the positive sign on the radical.
The corresponding local decay rate is again given by (\ref{drate}).
One may now deform the original countour $C$ in a vicinity of the critical points by 
re-connecting the unit circle nearby to two incremental descent branches near each critical point. By Cauchy's Theorem this leaves the values of the integral, for each $y$ unchanged. Restricitng attention again to just the neighborhood of $z=1$, Taking $w = \frac{2\alpha}{\sqrt{y}} \left( 1 + \sqrt{1 + y/4}\right) (1 + \mathcal{O}(\alpha^2))$, the integral along the incremental steepest descent branches is apporximated as
\begin{eqnarray*}
S_{2n}(i\sqrt{y}) &=& \frac{(2+ i \sqrt{y})^{2n}}{2 \pi} 
\frac2{\sqrt{y}} \left( 1 + \sqrt{1 + y/4} \right)  \int_{-A}^A \exp\left( -\frac{8\alpha^2 n}{x} \left( 1 + \sqrt{1 + y/4} \right) + n\mathcal{O}(|\alpha|^3)\right) (1 + \mathcal{O}(\alpha^2)) d\alpha
\end{eqnarray*} 
the implicit constants in the error terms depend on $y$. We expect that by choosing $A$ appropriately in relation to $y$, these error terms can be taken to be uniform in $y$. (This has been checked numerically in special cases but not proven in general.)

This thus becomes an integral of Laplace type and so may be asymptotically evaluated as
\begin{eqnarray*}
 \frac{(2+ i \sqrt{y})^{2n}}{2 \sqrt{2\pi n}} \sqrt{ 1 + \sqrt{1 + y/4}} \left(1 + \mathcal{O}\left( \frac1{\sqrt{n}}\right)\right)
\end{eqnarray*}
The evaluation near $z = -1$ is entirely similar and leads to the complex conjugate of the above contribution. This, finally, yields
\begin{eqnarray*}
S_{2n}(i\sqrt{|y|}) &=& \frac{\Re (2+ i \sqrt{|y|})^{2n}}{\sqrt{2\pi n}} \sqrt{ 1 + \sqrt{1 + |y|/4}} \left(1 + \mathcal{O}\left( \frac1{\sqrt{n}}\right)\right)\\
&=& \frac{(4 + |y|)^{n + 1/4} \sqrt{ 1 + \frac2{\sqrt{4 + |y|}}}}{\sqrt{4\pi n}} 
\cos\left(2n \tan^{-1} \left( \frac{\sqrt{|y|}}{2}\right)\right) \left(1 + \mathcal{O}\left( \frac1{\sqrt{n}}\right)\right)
\end{eqnarray*}
where, by symmetry considerations for $S_{2n}$, we are able to replace $y$ by $|y|$. 
The leading order expression has $2n$ nodes (zeroes) determined by
\begin{eqnarray*}
\sqrt{|y|} &=& 2 \tan\left( \frac{(2k-1) \pi}{4 n} \right), \,\,\,\,\, k = 1, \dots n.
\end{eqnarray*}
Equivalently there are $2n-1$ {\it zones} between these nodes where the graph of this leading term alternates between being positive and negative. For $n$ sufficiently large, the higher order corrections will not modify this last characterization. Hence, $S_{2n}(\zeta)$ has $2n$ zeroes, symmetrically distributed about the origin along the imaginary $\zeta$-axis, for $n$ suficiently large. By (\ref{diffrecur}), $S_{2n-1}(\zeta) = n \partial S_{2n}(\zeta)$ which consequently has $2n-1$ zerow
interlacing those of  $S_{2n}(\zeta)$. Continuing this differentiation, one sees that the proposiiton holds in general since the original $n$ may be taken aribitrarily large.
\end{proof}

\begin{rem}
It is interseting to note here that the critical points in the above steepest descent argument are precisely the Riemann invariants (\ref{char}) evaluated along the negative $y_0$-axis.
\end{rem}

\section{Proof of Theorem \ref{genthm}} \label{GenThm}
\renewcommand{\theequation}{C.\arabic{equation}}
\begin{proof} Part (a) follows directly from the arguments given in the proof of Corollary \ref{hypcurve} and (\ref{spectralcurve}) which shows that the negative zeros of $\xi^2$ along the curve $\mathcal{C}$ have multiplicity $2 \nu -1$ while the poles have multiplicty $2 \nu +1$.

We will prove part (b) in stages. As was seen in the proof of Corollary \ref{hypcurve}, $\xi^2$ is monotone increasing for $y_0 < 0$ as $y_0$ increases. By (\ref{spectralcurve}) and Corollary \ref{hypcurve}, $\xi^2$ cannot vanish; hence, the first place it could vanish is when $y_0$ reaches 0 and, again by (\ref{spectralcurve}) it does vanish there. From the formulas in Proposition
\ref{baseprop} one may explicitly evaluate (\ref{crit2}) at the origin to find that
\begin{eqnarray*}
\frac{d\xi^2}{dy_0}|_{(\xi^2, y_0) = (0,0)} = \left((2\nu+1) {2\nu \choose \nu}\right)^{-2} > 0.
\end{eqnarray*}
Hence the curve must continue to increase monotonically until the smallest positive value of $y_0$ at which  (\ref{crit2}) vanishes. By examination of  (\ref{crit2}) this happens for the smallest positive value of $y_0$ at which either $z_0$ vanishes or $\frac{\widehat{D}}{(4 - y_0)}$ vanishes. But, by (\ref{zzero}), $z_0$ vanishes if and only if $\widehat{B}_{12}$ vaishes and we have already seen that this polynomial only vanishes for negative values of $y_0$. Hence the first turning point occurs at
the smallest positive zero of $\widehat{D}/(4 - y_0)= - (\widehat{d}_+ /  \widehat{r}_+) \cdot (\widehat{d}_- /  \widehat{r}_-)$.

As discussed in section \ref{chargeom}, we expect the turning points on the right to be among the roots of the caustic $\widehat{d}_- /  r_+$ (which is the envelope of the $r_+$-characteristics) and the turning points on the left to be among the roots of the caustic $\widehat{d}_+ /  r_-$ (which is the envelope of the $r_-$-characteristics). We confirm this by observing that from (\ref{RinvtToda3}) one has that the equation for the $r_+$-characteristics takes the form
\begin{eqnarray}
\xi\partial_\xi r_+ &=& \frac{(1 - 1/z_0) + \sqrt{y_0} }{\widehat{d}_-}  r_+ \\ \nonumber
\frac{d}{d\xi^2} r_+ &=& \frac12 \frac{(1 - 1/z_0) + \sqrt{y_0} }{\widehat{d}_-}  r_+ \\
                               &=& \frac12 \sqrt{y_0}  \frac{ \widehat{B}_{11}/\widehat{B}_{12} + 1}{\widehat{d}_-}  r_+ \,,
\end{eqnarray}
where $\widehat{B}_{11}/\widehat{B}_{12} + 1$ is manifestly positve for all positive values of $y_0$. Hence the first turning point on the right occurs at the zero of $\widehat{d}_-/r_+$ having the smallest $y_0$ coordinate.

To see that this caustic is in fact a {\it simple} turning point (as opposed to an inflection or higher order point) consider the second derivative of $\xi^2$ evaluated at the caustic:
\begin{eqnarray*}
\frac{d^2 \xi^2}{d y_0^2}\Big|_{\mbox{caustic}}  &=& - z_0 \xi^2 \frac{\widehat{d}_+}{y_0 (y_0 - 4)} \frac{d}{d y_0} \widehat{d}_- \Big|_{\mbox{caustic}}.
\end{eqnarray*}
Since $z_0 = \widehat{B}_{12} /(\widehat{B}_{12} - y_0^{1/2}\widehat{B}_{11})$ and $y_0 < y_0^*$ (since the caustic occurs for $\xi^2 > 0$ and $y_0$ reaches $y_0^*$ only when $\xi^2 = 0$) it follows that $z_0 >0$ at the caustic.  Both $\xi^2$ and $y_0$ are positive there as well. We may further assert that $y_0 < 1$ at this caustic. To see this, observe that 
$\widehat{d}_- = 0$ may be expressed as
\begin{eqnarray*}
j \left(\frac{1}{z_0}\right) &=& (j-2) (\sqrt{y_0} + 1).
\end{eqnarray*}
As we have just seen, $z_0 >0$; moreover, since $1/z_0 = 1 - y_0^{1/2}\widehat{B}_{11}/\widehat{B}_{12}$ 
it follows that the LHS of this equation is strictly less than $j$, while, if $y_0 > 1$, then the RHS is strictly greater than $2j -4$. But  this can only hold if $j < 4$. The theorem has already been established in the case of $j = 3$, so we may therefore assume that $y_0 < 1$.  Hence the $y_0$ coordinate of the caustic is less than 4.  We further observe from (\ref{red2}) that 
\begin{eqnarray} \label{crossover}
\widehat{d}_+ - \widehat{d}_- &=&  2 (j-2) \sqrt{y_0}. 
\end{eqnarray}
Since $y_0$ is increasing from 0,  with $\widehat{d}_+ = \widehat{d}_- = 2$ at $(\xi^2, y_0) = (0,0)$, it follows that $\widehat{d}_+$ is positive at the caustic.

It remains to study the $y_0$-derivative of $\widehat{d}_-$ at the caustic. From (\ref{crit5}) and (\ref{crossover})), one sees that
\begin{eqnarray}
\frac{d}{d y_0} \widehat{d}_- &=& \frac{d}{d y_0} \left(j \frac1{z_0} - (j-2) \sqrt{y_0}\right)\\ \nonumber
&=& \frac{j}{8 y_0} \left[\left(1 - \frac1{z_0}\right)\left(\frac{\widehat{d}_+}{\sqrt{y_0} - 2} - \frac{\widehat{d}_-}{\sqrt{y_0} + 2}\right) + \sqrt{y_0} \left(\frac{\widehat{d}_+}{\sqrt{y_0} - 2} + \frac{\widehat{d}_-}{\sqrt{y_0} + 2}\right)\right] - \frac{j-2}{2 \sqrt{y_0}}\\ \label{dminuscaustic}
\frac{d}{d y_0} \widehat{d}_- \Big|_{\mbox{caustic}} &=& - \frac{j}{8 y_0} \left[\frac{1 - \frac1{z_0} + \sqrt{y_0}}{2 - \sqrt{y_0}}\cdot 2(j-2)\sqrt{y_0}\right] - \frac{j-2}{2 \sqrt{y_0}}\\ \nonumber
&<& 0
\end{eqnarray}
where in the last line we have twice used the fact that $\widehat{d}_- = 0$ at the caustic. Putting this all together one finally sees that 
\begin{eqnarray*}
\frac{d^2 \xi^2}{d y_0^2}\Big|_{\mbox{caustic}}  &<& 0 
\end{eqnarray*}
from which it follows that this first caustic is indeed a turning point on the right (i.e., a maximum of $\xi^2$ with respect to $y_0$). 

We next claim that as $y_0$ increases past this first turning point, $\widehat{d}_-$ cannot vanish again as long as $\xi^2$ remains positive. For if $\widehat{d}_- = 0$ then, by (\ref{crossover}), $\widehat{d}_+ > 0$ and so (\ref{dminuscaustic}) is again negative for the same reasons it was negative at the first turning point. But that means that this point must again be a maximum of $\xi^2$ as a function of $y_0$. But then one would have two consecutive maxima which is not possible for a smooth curve. Consequently $\widehat{d}_-$ decreases monotonically (and so does $\xi^2$) at least until $\xi^2$ vanishes (i.e., $\mathcal{C}$ crosses the positive $y_0$-axis at $y_0^*$) or $\widehat{d}_+$ vanishes. 

But we also claim that $\widehat{d}_+$ cannot vanish as long as $\xi^2$ remains positive. For, if there were such a point, then at it one would have
\begin{eqnarray*}
\frac{d}{d y_0} \left( 1 - \frac{y_0^{1/2} \widehat{B}_{11}}{\widehat{B}_{12}}\right) &=& \frac{d}{d y_0} \left( \frac1{z_0} \right)\\
&=& - \frac{j}{8 y_0} \left[\frac{1 - \frac1{z_0} - \sqrt{y_0}}{2 + \sqrt{y_0}}\cdot 2(j-2)\sqrt{y_0}\right]\\
&=& - \frac{j-2}{2 \sqrt{y_0}} \left[\frac{1 -  \sqrt{y_0}}{2 + \sqrt{y_0}}\right]\\
&<& 0
\end{eqnarray*}
where the evaluation of the derivative is similar to the derivation of (\ref{dminuscaustic}) and the final inequality holds because $\xi^2 > 0$ and $\widehat{d}_+ = 0$ together imply that $\sqrt{y_0} < 1$. Applying this observation to  (\ref{curve}) in the form
\begin{eqnarray*}
\xi^2 &=& \frac1{(2\nu+1)^2} \frac{y_0}{\widehat{B}_{12}^2}\left(1 - \frac{y_0^{1/2} \widehat{B}_{11}}{\widehat{B}_{12}} \right)^{2\nu-1}
\end{eqnarray*}
one sees that at such a point $\xi^2$ must be monotonically decreasing as a function of $y_0$ since 
$\widehat{B}_{12}^2$ grows faster than linearly with respect to $y_0$. But this decay is inconsistent with the formation of a caustic which would necessarily arise if $\widehat{d}_+$ were to vanish. This contradiction establishes our claim.

So now one sees that after turning once on the right, $\xi^2$ decreases until it reaches the inflection point at $y_0^*$ along the $y_0$-axis.

As $\xi^2$ continues to decrease through negative values for $y_0 > y_0^*$ we note that combining (\ref{spectralcurve}) with (\ref{zzero}) shows that $z_0$ changes sign from positive to negative as soon as $\xi^2$ does. Hence our earlier argument shows once again that $\widehat{d}_-$ cannot vanish a second time before before $\widehat{d}_+$ does. However for some value of $y_0 > y_0*$, $\widehat{d}_+$ must vanish since $\xi^2 \to 0$ as $y_0 \to +\infty$ and so must have at least one minimum as a function of increasing $y_0$ (which also  corresponds to the formation of a caustic along the envelope of $r_-$ characteristics). This point will necessarily have negative $\xi^2$-coordinate. Once again one may examine the second derivative of $\xi^2$ at this caustic to find that
\begin{eqnarray*}
\frac{d^2 \xi^2}{d y_0^2}\Big|_{\mbox{caustic}}  &=& z_0 \xi^2 \frac{\widehat{d}_-}{y_0 (4 - y_0)} \frac{d}{d y_0} \widehat{d}_+ \Big|_{\mbox{caustic}}.
\end{eqnarray*}
Again, combining (\ref{spectralcurve}) with (\ref{zzero}), one sees that $z_0 \xi^2/y_0 > 0$. Paralleling the calculations for (\ref{dminuscaustic}) one has here that
\begin{eqnarray*}
\frac{d}{d y_0} \widehat{d}_+ \Big|_{\mbox{caustic}} &=& - \frac{j}{8 y_0} \left[\frac{\sqrt{y_0} - (1 - \frac1{z_0}) }{2 + \sqrt{y_0}}\cdot 2(j-2)\sqrt{y_0}\right] +
\frac{j-2}{2 \sqrt{y_0}}\\
&=& - \frac{j-2}{2 \sqrt{y_0}} \left[\frac{1 - \sqrt{y_0}}{2 + \sqrt{y_0}}\right] + \frac{j-2}{2 \sqrt{y_0}} \\
&=&  \frac{j-2}{2 \sqrt{y_0}} \frac{1 + 2 \sqrt{y_0}}{2 + \sqrt{y_0}}\\
&>& 0.
\end{eqnarray*}
By (\ref{crossover}) one also has $\widehat{d}_- < 0$ at the caustic. Finally we determine the sign of  $4 - y_0$. Since $\widehat{d}_+$ vanishes at the caustic by definition and, by the previous calculation, it is increasing as it crosses the caustic, it follows that both $\widehat{d}_+$ is negative just below the caustic and, by the previous sentence, so is $\widehat{d}_-$. Hence, $\widehat{D} > 0$ just below the caustic. Now considering all these inqualities in the context of (\ref{crit2}) whose LHS is negative just below the caustic, it follows that one must have $y_0 > 4$ just below the caustic, and hence at the caustic as well. So, 
$4 - y_0 < 0$ there. Putting all this together one concludes that  the second derivative of the curve at this caustic is positive and therefore this caustic is indeed a second turning point which is a minimum of $\xi^2$ with respect to $y_0$. 

For $y_0$ larger than the $y_0$-coordinate of this second turning point, the  $y_0$-coordinate of the turning point continues to be greater than 4. As long as $\xi^2$ remains negative there can be no further turning points since, by (\ref{crossover}), for this to happen $\widehat{d}_+$ would need to vanish first thus producing two consecutive minima for $\xi^2$ which is not possible for a smooth curve. 

Now, since $\widehat{B}_{12} - y_0^{1/2}\widehat{B}_{11}$ cannot vanish further for $y_0 > y_0^*$ the curve cannot cross the $y_0$-axis again, but must asymptote to it as $y_0 \to + \infty$. Similarly, for $y_0$ less than the most negative horizontal asymptote,  $\widehat{B}_{12} - y_0^{1/2} \widehat{B}_{11}$ cannot vanish further and the curve must asymptote to the $y_0$-axis as $y_0 \to - \infty$.

Finally, turning to part (c)  we examine more directly the zeroes of the discriminant $A_{11}^2 - f_0 A_{12}^2$ which, we have been referring to as caustics. As we have seen, the discriminant has two natural factors, $\frac{\widehat{d}_\pm}{2 \mp \sqrt{y_0} }$ where
\begin{eqnarray} \label{dmp}
\widehat{d}_\mp &=& 2- y_0^{1/2} \left(\frac{ j\widehat{B}_{11}}{\widehat{B}_{12}} \pm (j-2)\right)
\end{eqnarray}
as a function of only $\sqrt{y_0}$. Since $\widehat{B}_{12}$ is nonvanishing for positive $y_0$, it follows that the zeroes of the discriminant coincide with the zeroes of 
\begin{eqnarray} \label{PImp}
\Pi_{\mp}(\sqrt{y_0}) &=& \frac{2\widehat{B}_{12} -j y_0^{1/2} \widehat{B}_{11} \mp (j-2) y_0^{1/2} \widehat{B}_{12}}{2 \pm \sqrt{y_0}}.
\end{eqnarray}
which are both polynomials of degree $2\nu$ in $\sqrt{y_0}$ since the denominator divides the numerator as will be deomonstrated shortly. Moreover, it is immediate from this representation that
\begin{eqnarray} \label{PiSymm}
\Pi_{\mp}(-\sqrt{y_0}) &=& \Pi_{\pm}(\sqrt{y_0})
\end{eqnarray}
so that $\Pi_{-}$ may be gotten from $\Pi_{+}$ by "flipping" the signs on the terms which are odd in $\sqrt{y_0}$. It follows that if $\pm \sqrt{y_0}$ is a zero of $\Pi_{+}$ then $\mp \sqrt{y_0}$ is, respectively,  a zero of $\Pi_{-}$ and furthermore the squares, $y_0$, of these zeroes yield the full set of $2 \nu$ zeroes, in $y_0$ of  $\frac{\widehat{D}}{4 - y_0}$. 
\medskip

We complete the proof by now showing that  $\Pi_{\mp}$ is indeed a polynomial in $\sqrt{y_0}$ as opposed to being rational with a pole at $\sqrt{y_0} = \pm 2$ respectively. Observe that
\begin{eqnarray*}
\widehat{B}_{12}|_{\sqrt{y_0} = \mp 2} &=& [y_0^0] \left(y_0 + 2 + y_0^{-1}\right)^{2\nu}\\
&=& [y_0^0] \left(y_0 ^{1/2}+ y_0^{-1/2}\right)^{4\nu}\\
&=& {4\nu \choose 2\nu}.\\
\widehat{B}_{11}|_{\sqrt{y_0} = \mp 2} &=& \mp [y_0^{-1}] \left(y_0 + 2 + y_0^{-1}\right)^{2\nu}\\
&=& \mp [y_0^{-1}] \left(y_0 ^{1/2}+ y_0^{-1/2}\right)^{4\nu}\\
&=& \mp {4\nu \choose 2\nu + 1}.
\end{eqnarray*}
It follows that
\begin{eqnarray} \label{dmp2}
\widehat{d}_\mp \Big|_{y_0^{1/2} = \mp 2}  &=& \frac2{\widehat{B}_{12}}\left(2\nu \widehat{B}_{12} \pm (2\nu +1) \widehat{B}_{11} \right)\Big|_{y_0^{1/2} = \mp 2}\\ \nonumber
&=& \frac2{{4\nu \choose 2\nu}}\left(2\nu {4\nu \choose 2\nu} - (2\nu +1) {4\nu \choose 2\nu + 1} \right) \\ \nonumber
&=& 0,
\end{eqnarray}
so that the pole at $\sqrt{y_0} = \mp 2$ is cancelled respectively by a corresponding zero of $\widehat{d}_\mp$. We observe that $\frac{y_0 ^{1/2}+ y_0^{-1/2}}{2}$ is the Joukowski transform which conformally maps the exterior and interior of the unit disc in the $y_0 ^{1/2}$-plane onto the complex plane slit along the interval $[-1, 1]$. From the viewpoint of characteristic geometry, $y_0 ^{1/2} = \pm 2$ correpsonds to one or the other of the characteristic directions vanishiing, which, spectrally, means that the asymptotic support of the spectrum has become either entirely positive or entirely negative: the index of the asymptotic matrix becomes extremal.

We may now also observe that
\begin{eqnarray*}
\frac1{z_0}\Big|_{y_0^{1/2} = \pm 2} &=& 1 - \frac{y_0^{1/2}\widehat{B}_{11}}{\widehat{B}_{12}}\Big|_{y_0^{1/2} = \pm 2}\\
&=& 1 - 2 \frac{{4\nu \choose 2\nu + 1}}{{4\nu \choose 2\nu}}\\
&=& 1 - \frac{4 \nu}{2\nu + 1} = - \frac{2 \nu -1}{2 \nu +1} < 0.
\end{eqnarray*}
(Figure \ref{simplezero} (b) plots a number of these values.) 

We also note that on $\mathcal{C}$, when $y_0 = 4$,
\begin{eqnarray} \label{xisqrd}
\xi^2\Big|_{y_0 = 4} &=& \frac4{(2\nu + 1)^2 \widehat{B}_{12}^2}\left(1 - \frac{y_0^{1/2}\widehat{B}_{11}}{\widehat{B}_{12}}\right)^{2\nu-1}\Big|_{y_0 = 4}\\ \nonumber
&=& \frac4{(2\nu + 1)^2 {4\nu \choose 2\nu}^2}\left(-\frac{2 \nu -1}{2 \nu + 1}\right)^{2\nu-1}\\ \nonumber
&=& - \frac{4 (2\nu)!^4}{{(4\nu) !}^2}\left(\frac{(2 \nu -1)^{2\nu-1}}{(2 \nu + 1)^{2\nu+1}}\right)\\ \nonumber
&<& 0,
\end{eqnarray}
consistent with our earlier observation that the $y_0$ coordinate of the first caustic is less than $4$.

In addition, one may observe from (\ref{dmp}) that

\begin{eqnarray} \label{dmp3}
\frac{d}{dy_0} \widehat{d}_\mp &=& -(2\nu+1) \frac{d}{dy_0} \left(\frac{y_0^{1/2} \widehat{B}_{11}}{\widehat{B}_{12}}\right) \mp \frac{2 \nu - 1}{2 y_0^{1/2}}\\
&=&-\frac{(2\nu+1) (2\nu+2)W\left( \widehat{B}_{12}, \widehat{\mathcal{F}}_1\right)}{\widehat{B}_{12}^2} \mp \frac{2\nu-1}{ 2 y_0^{1/2}}\\  \label{dmp4}
\frac{d}{dy_0} \widehat{d}_\mp\Big|_{{\sqrt{y_0} = \mp 2}} &=& -\frac14 \frac{2\nu -1}{4 \nu - 1}
\end{eqnarray}
where $W$ denotes the Wronskian. In the second line we have made use of the identity (\ref{B-identity})
and in the third line the evaluations
\begin{eqnarray*}
\widehat{\mathcal{F}}_1\Big|_{{\sqrt{y_0} = \mp 2}} &=& [y_0^1] \left(y_0 + 2 + y_0^{-1}\right)^{2\nu+1}\\
&=& \frac1{2 \nu + 1} {4 \nu +2 \choose 2 \nu}\\
\partial_{y_0} \widehat{\mathcal{F}}_1\Big|_{{\sqrt{y_0} = \mp 2}} &=&  \left[\frac1{2 y_0^{1/2}} \widehat{B}_{11}\right]_{{\sqrt{y_0} = \mp 2}}\\
&=& \frac1{4}{4\nu \choose 2\nu+1}\\
\partial_{y_0} \widehat{B}_{12}\Big|_{{\sqrt{y_0} = \mp 2}} &=& [y_0^0] \left(y_0 + 2 + y_0^{-1}\right)^{2\nu - 1}\\
&=&\frac{\nu}{2}{4\nu -2\choose 2\nu-1}.
\end{eqnarray*}
Figure \ref{simplezero} (a)  illustrates (\ref{dmp4}).
It follows that the vanishing of (\ref{dmp2}) is simple and $\frac{\widehat{D}}{y_0 - 4}$ limits to a finite, non-zero value as $y_0 \to 4$. 
\end{proof}
\medskip

\begin{figure}[h] 
\begin{center}
\resizebox{2in}{!}{\includegraphics{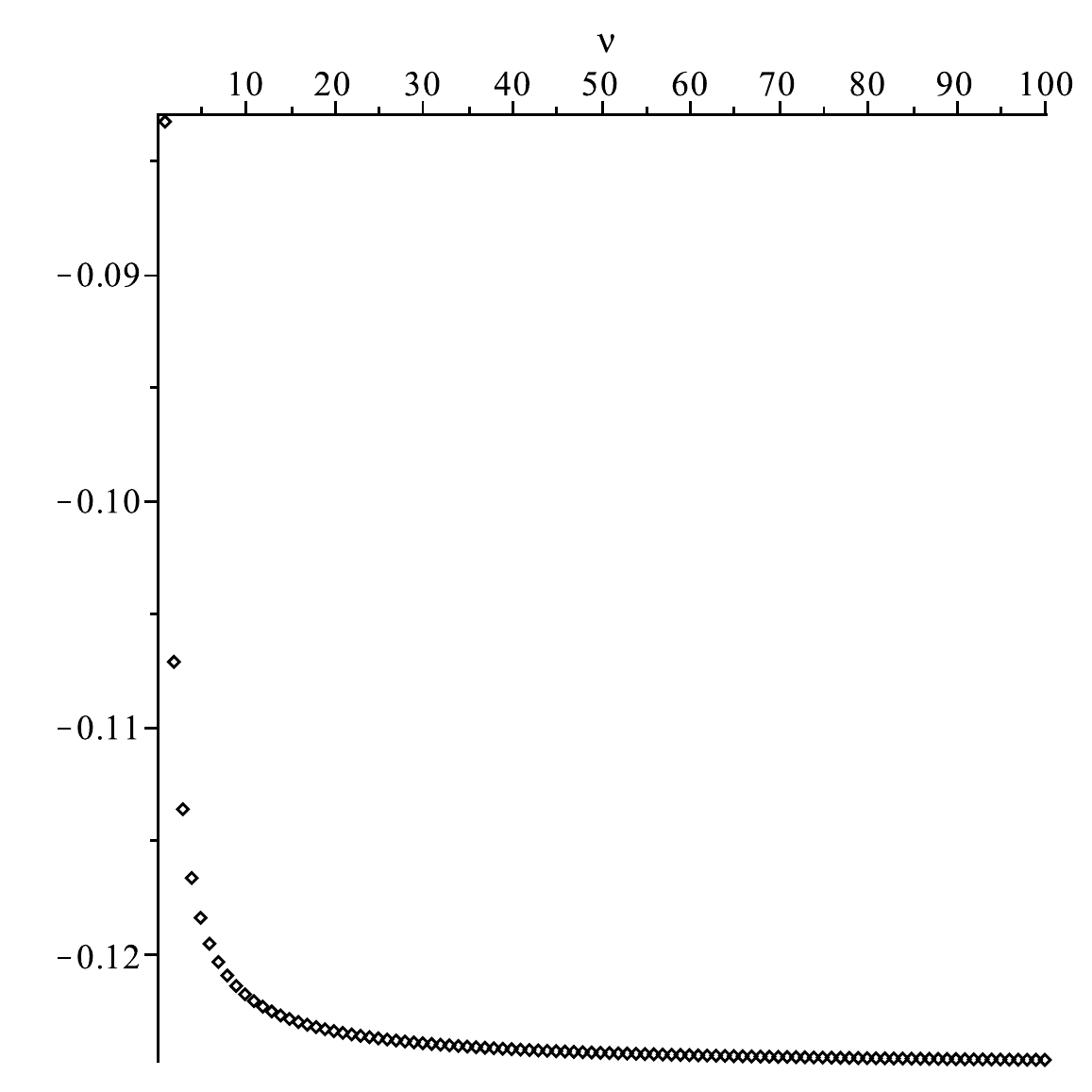}}
\resizebox{2in}{!}{\includegraphics{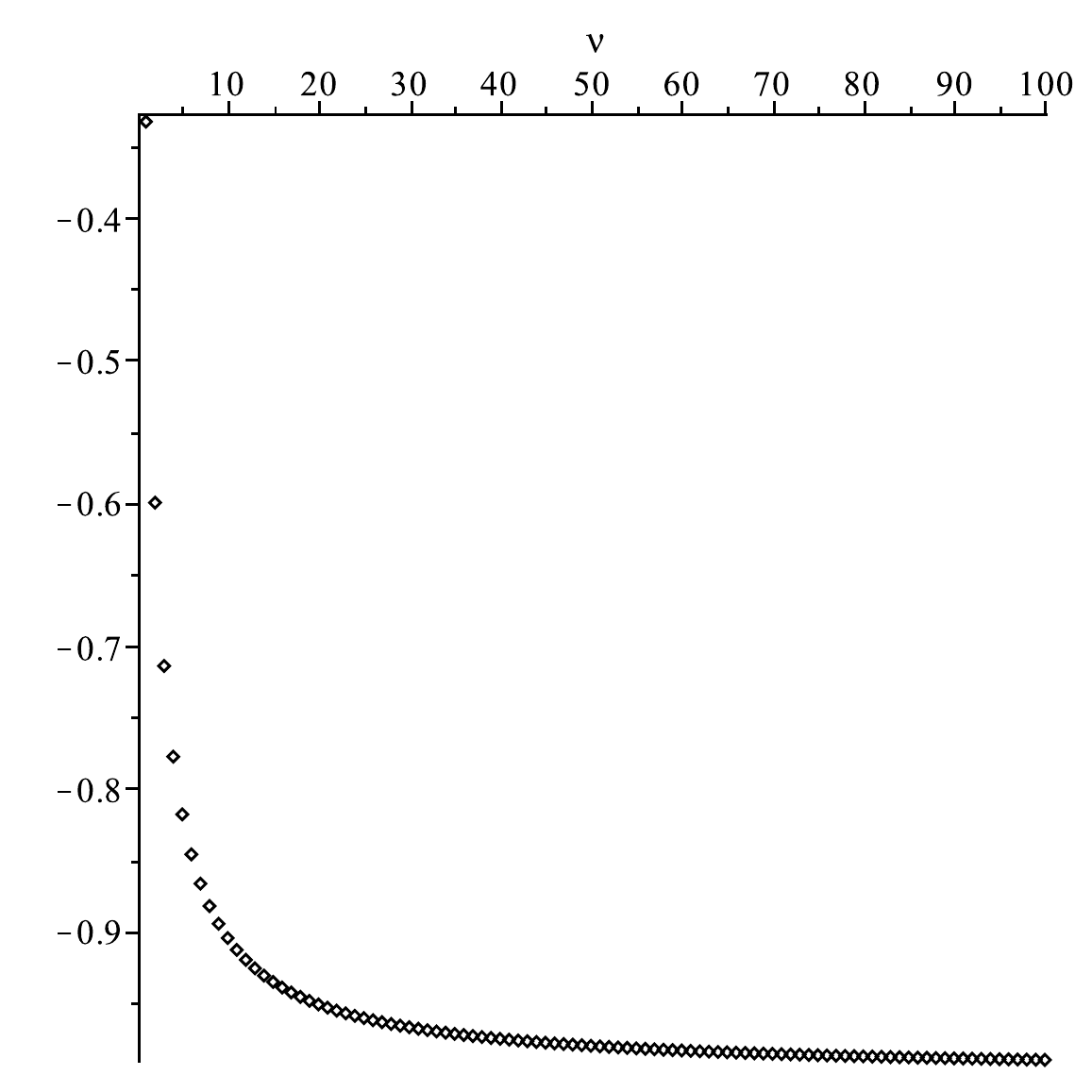}}
\end{center}
\caption{\label{simplezero}  a) $\frac{d}{dy_0} \widehat{d}_\mp \Big|_{y_0^{1/2}=\mp 2}$ \,\,\,\, 
b) $1-\frac{y_0^{1/2}\widehat{B}_{11}}{\widehat{B}_{12}}\Big|_{y_0 = 4}$ }
\end{figure}

\section{Some Additional Observations}

\subsection{Odd Valence Painlev\'e Transcendents} \label{POne}
In Proposition \ref{p872} we established a recurrence relation for the dominant coefficient of $f_g$ near the critical turning point $\xi_c$. This is the analogue of a recursion formula for the even valence case derived in \cite{Er09} that was used to establish a fundamental connection between the asymptotics of recursion coefficients near the singularity and  Painlev\'e I transcendents. 
We conjecture an extension of this result to the odd valence case based on the analysis carried out in the proof of  Proposition \ref{p872}:
\begin{conj} \label{c874}
Define $a=(-8\sqrt{6}C_1 C_2)^{2/5}$ and $c=12C_1/(z_c \sqrt{a})$.
In the double scaling limit where
\begin{align}
t_N(\lambda) =t_c +a \lambda N^{-4/5}
\end{align}
the recurrence coefficients have the following asymptotics
\begin{align}
b^{2}_N(t_N(\lambda)) \sim z_c + c^{-1} N^{-2/5} Y(\lambda),
\end{align}
where $Y(\lambda)$ is a particular solution to an equation in the Painlev\'e I hierarchy.
\end{conj}
In the trivalent case ($j=3$) we have compared the statement in this conjecture to a known result due to \cite{BD16} and found it to be in exact agreement.

\subsection{The Edge Toda equations} 
We conclude with the description of another system of continuum equations emerging from our anlysis which is, in fact, a novel result of possible future interest. It can be derived from  a combination of the Toda and String equutions. They can also be obtained from the $k=0$ Virasoro constraints but we do not pursue that in this paper. 

We will show that the recurrence coefficients for orthogonal polynomials satisfy the following equations, which we refer to as the edge Toda equations.  
\begin{align}		
N^{-1} \left(1 +   jt  \partial_{t} \right)a_n =& a_{n+1}b^{2}_{n+1}-a_nb^{2}_n +a_n b^{2}_{n+1} -a_{n-1}b^{2}_n  \label{EdgeTodaA} \\
N^{-1}  \left(2+  jt  \partial_t \right)b^{2}_{n} =& b^{2}_n\left(a_{n}^2-a_{n-1}^2+b^{2}_{n+1}-b^{2}_{n-1}\right).	\label{EdgeTodaB}		
\end{align}
We call equations  (\ref{EdgeTodaA}) and (\ref{EdgeTodaB}) {\it edge Toda equations} because they can be alternately derived by introducing an auxiliary variable in the matrix model which counts the edges of maps and differentiating with respect to this variable as in the derivation of the Toda lattice equations sketched earlier in Section \ref{sec:Toda}.  

We first prove (\ref{EdgeTodaA}): starting with the Toda equation (\ref{MotzkinTodaA}), we condition the term $L^{j}_{n,n-1}$ on the first Motzkin path step, and the term $L^{j}_{n+1,n}$ on its last Motzkin path step.
\begin{align*}
\frac{1}{N}  j t_j \partial_{t_j} a_n	
=& j t_j \left[ L^{j}_{-},L\right]_{n,n} \\
=&  j t_j \left(L^{j}_{n,n-1}-L^{j}_{n+1,n}\right) \\
=& j t_j \left(
\begin{array}{c} b^{2}_n L^{j-1}_{n-1,n-1}+a_nL^{j-1}_{n,n-1}+L^{j-1}_{n+1,n-1}\\ -L^{j-1}_{n+1,n+1}b^{2}_{n+1}-L^{j-1}_{n+1,n}a_{n}-L^{j-1}_{n+1,n-1}\end{array}\right) \\
=&   j t_j \left(b^{2}_n L^{j-1}_{n-1,n-1}+a_nL^{j-1}_{n,n-1}-L^{j-1}_{n+1,n+1}b^{2}_{n+1}-L^{j-1}_{n+1,n}a_{n} \right)	
\end{align*}
We now substitute index shifts (i.e. $n\rightarrow n\pm 1$) of the string equations.  We find that
\begin{align*}
\frac{1}{N}  jt_j \partial_{t_j} a_n =& b^{2}_n (-a_{n-1})+a_n\left(\frac{n}{N}-b^{2}_n\right) \\
& -(-a_{n+1})b^{2}_{n+1}-\left(\frac{n+1}{N}-b^{2}_{n+1}\right)a_{n};
\end{align*}
this is the first of the edge-Toda equations.  The second edge Toda equation (\ref{EdgeTodaB}) can be proved by a similar calculation starting from the other Toda equation (\ref{MotzkinTodaB}).

\end{document}